%% file: main.tex
\documentclass{cup-elements}

\usepackage{blindtext}
\usepackage{multirow}
\usepackage{graphicx}
\graphicspath{{pics/}}
\usepackage{amsmath,amssymb}
\usepackage{xypic}
\usepackage{subcaption}
\usepackage{algpseudocode}
\usepackage{float}

\usepackage[draft]{tikzit}
\input{hypergraph.tikzdefs}
\input{hypergraph.tikzstyles}

\allowdisplaybreaks
\newenvironment{proof}{\textit{Proof:}}{\hfill$\square$}

\newtheorem{theorem}{Theorem}[section]
\newtheorem{definition}[theorem]{Definition}
\newtheorem{lemma}[theorem]{Lemma}

\newtheorem{remark}[theorem]{Remark}
\newtheorem{example}[theorem]{Example}
\newtheorem{exercise}[theorem]{Exercise}

\usepackage{apacite}
\usepackage{bussproofs}
\usepackage{listings}

\newcommand{\semic}{\mathbin{;}}

\usepackage[all]{xy}
\CUPseries{}%Cambridge Elements}
\CUPelements{}%Elements of Applied Category Theory}

\title{String Diagrams for $\lambda$-calculi \\ and Functional Computation  }

\author{Dan R. Ghica}
\affil{University of Birmingham and Huawei Central Software Institute}
\author{Fabio Zanasi}
\affil{University College London}
\begin{document}

%%%%-----------------------
\frontmatter  %% DO NOT DELETE
%%%%-----------------------
\maketitle

\begin{abstract}
This tutorial gives an advanced introduction to string diagrams and graph languages for higher-order computation.
The subject matter develops in a principled way, starting from the two dimensional syntax of key categorical concepts such as functors, adjunctions, and strictification, and leading up to Cartesian Closed Categories, the core mathematical model of the lambda calculus and of functional programming languages. 
This methodology inverts the usual approach of proceeding from syntax to a categorical interpretation, by rationally reconstructing a syntax from the categorical model.
The result is a graph syntax---more precisely, a \emph{hierarchical hypergraph syntax}---which in many ways is shown to be an improvement over the conventional linear term syntax.
The rest of the tutorial focuses on applications of interest to programming languages: operational semantics, general frameworks for type inference, and complex whole-program transformations such as closure conversion and automatic differentiation.
\end{abstract}


%\JEL{tbd}

\copyrightauthor{Dan R. Ghica and Fabio Zanasi, 2023}

\newpage
\mbox{ }
\vspace{1ex}
\begin{center}
\includegraphics[scale=.75]{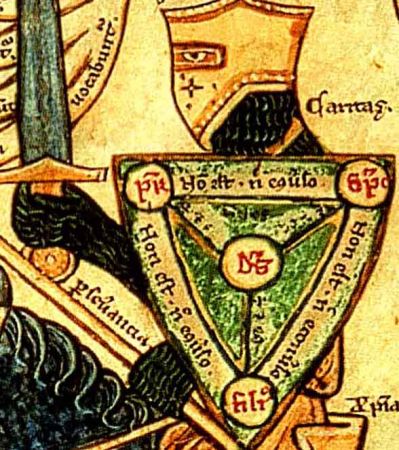}
\end{center}
\vspace{10ex}
%%%%-----------------------
\setcounter{tocdepth}{4} % This command is to see subsections in the TOC. can be removed later, but it's useful to have at this stage to see the organisation of the book.
\mainmatter  %% DO NOT DELETE
%%%%-----------------------

\section{Introduction}
\label{sec:intro}

The representation of syntax of programming languages, and even more so of natural languages, requires complex data structures. 
When writing programs, specifications of programs, or syntactic models of programming languages, such as operational semantics or abstract machines, it is customary to deploy syntax in its \emph{serialised form}: a linear sequence of symbols which we call \emph{text} or, to use more technical vocabulary, \emph{terms}.
Terms are, however, an awkward data structure for algorithmic purposes. 
Hence, compilers and other programming language tools never work on terms directly but immediately convert them, through processes known as \emph{lexing} and \emph{parsing}, into more useful data structures.
The most common, if not the most sophisticated, is called the \emph{abstract syntax tree} (AST). 

It is generally admitted that ASTs, although an improvement on terms, are far from the ideal data structure for representing syntax.
Various graph-based alternatives have been proposed, with some of them implemented in production compilers. 
Such data structures fall under the general label of \emph{intermediate representation} (IR). The design space for data structures to represent syntax is large.
Sound methodology requires that an IR is best derived in a principled way, from more general mathematical concepts, as informed by pragmatic considerations where needed.

This tutorial takes \emph{string diagrams} as the starting mathematical concept, and showcases their utility in representing and reasoning about programs. String diagrams originate in category theory, as a \emph{two-dimensional} (graphical) syntax for morphisms of monoidal categories. Even though in the literature string diagrams are often conflated with their concrete, combinatorial representations as graphs, we find it profitable to maintain the distinction and to insist on the former being syntax, albeit two-dimensional or planar. This allows us to isolate the IRs stemming from string diagrammatic syntax: these are called \emph{hierarchical hypergraphs}, and have uniquely helpful theoretical and algorithmic properties.

In Figure~\ref{fig:sdintro} we show, as a teaser, the kind of string diagram we will deal with in the sequel. 
Concretely, it represents the composition of currying and uncurrying in the $\lambda$ calculus. 
Without too many spoilers, in this diagram wires represent variables and boxes operations; boxes with round corners are thunks and their bound variables are wires attached to the frame. 
\begin{figure*}
\[
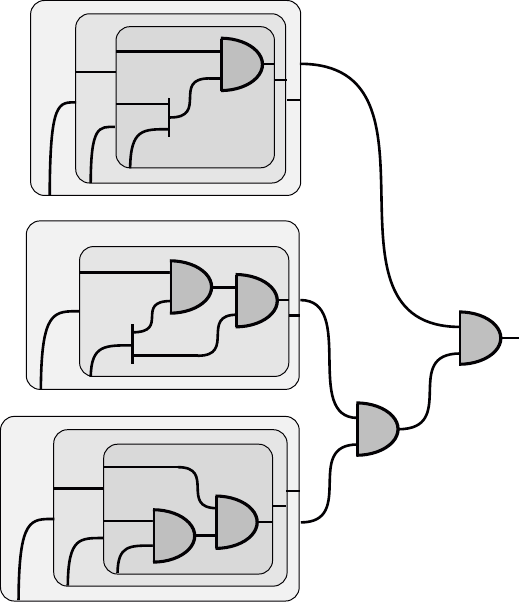
\]
\caption{An example of a string diagram}
\label{fig:sdintro}
\end{figure*}

Most mathematical syntax is one-dimensional, or linear, so one may reasonable wonder whether there are any advantages to two-dimensional string diagrammatic syntax. 
%Or, perhaps, it is preserved merely out of deference to tradition. 
One way to look at the question is observing the relationship between syntax and equations holding on it. One-dimensional syntax can \emph{absorb} quotienting by associativity, thus reducing unpleasant use of parentheses. 
We need not distinguish between $a\cdot(b\cdot c)$ and $(a\cdot b)\cdot c$ when the operation $\cdot$ is associative. 
This is why we are entitled to simply write $a\cdot b\cdot c$. 
If $\cdot$ is a common operation we can even represent it by simple concatenation, so that
\[
abc = a(bc) = (ab)c.
\]
If the operation is simply concatenation it would make sense to represent the unit of the operation, should it exist, by empty space (\textvisiblespace), as empty space is the unit of concatenation, which would give the following unit rules: 
\[
a\text{\textvisiblespace}=\text{\textvisiblespace}a=a
\]
This is, of course, not common practice, but if used it would also allow unit rules to be absorbed into linear notation. 

Two dimensional notation can be useful in absorbing quotienting by further equations. 
Consider for instance fractions written two-dimensionally 
\[
a \div b = \frac a b.
\]
Written in planar form, certain unobvious linear equations are simply absorbed by the notation:
\[
(a\div c)(b\div d)=\frac a c \frac b d = \frac{ab}{cd}=(ab)\div (cd).
\]

This suggests a criterion to compare notations: 
\begin{quote}\em
A notation is better if it absorbs more equations, by which we mean that instead of using an explicit equation to identify two terms they become syntactically equal in the improved notation. 
\end{quote}

Planar notation has an additional advantage over linear notations in that it can make certain notational invariants obvious. 
Consider for example matrices, written most commonly in two dimensional syntax:
\[
M = 
\begin{pmatrix}
a & b & c\\
u & v & w\\
x & y & z
\end{pmatrix}.
\]
For formal reasoning (in proof assistants) or algorithmic processing (in computer libraries) two-dimensional syntax is not possible so matrices are linearised as vectors-of-vectors, with $M$ above written either as list of columns
\[
M = \bigl(
(a, u, x), (b, v, y), (c, w, z)
\bigr), 
\]
or as list of rows. 
The key invariant that all columns (or rows, respectively) have the same size is no longer enforced by the notation: in the planar notation it is impossible to write a matrix with mismatched rows or columns. 

This suggests an additional criterion for comparing notations:
\begin{quote}\em
A superior notation absorbs more invariants, so that malformed terms cannot be written as such. 
\end{quote}

The planar notation for matrices also absorbs equational properties, for instance the fact that the transpose is an involution. 

In the case of more complex languages, with variables and binding, moving from linear to planar syntax can absorb further such equational properties. 
Consider the following terms, where \emph{let} introduces and binds a variable:
\lstset{language=Haskell, basicstyle=\small, keywordstyle=\bfseries, identifierstyle=\ttfamily}
\begin{lstlisting}
  let x=a in let y=b in let z=c in x+y+z
  let u=a in let v=b in let z=c in u+v+z
  let v=b in let u=a in let z=c in u+v+z
\end{lstlisting}
These terms are all equivalent. 
The first two are equivalent because bound variables can be systematically renamed, namely $x$ to $u$ and $y$ to $v$. 
The second and third are equivalent because definitions that do not depend on each other can be swapped. 
Neither of these equivalences are obvious in the linear syntax. 
The first one is the well known $\alpha$ equivalence and the second one is the lesser known \emph{graph equivalence} of calculi with explicit substitution. 
The equivalence are not difficult to work with mathematically, but what makes them challenging is that they are supposed to be \emph{congruences}, i.e. terms need to be quotiented by these equivalences so that definitions apply to equivalence classes of terms rather than terms. 
On the other hand, as we shall see in the sequel, in a two dimensional graph syntax all these terms would correspond to the same graph, or rather to isomorphic graphs, namely:

\[
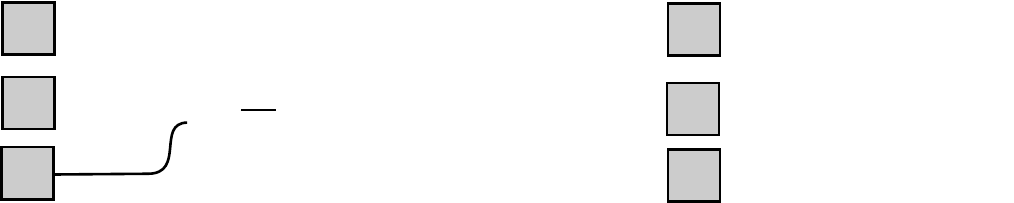
\]

In this representation, bound variables are graph edges, whereas constants and operations are boxes, so it should be intuitively clear that graph isomorphism absorbs both kinds of equivalences we mentioned above. 

To conclude this brief detour into two-dimensional syntax, hopefully it can be now accepted that moving from linear to planar has certain non-trivial advantages. 
However, moving from syntactic to combinatorial graph-like representations will have further, deeper advantages which we will explore in the rest of this tutorial. 
Yet, neither terms, nor diagrams, nor graphs completely supersede each other. 
They are all useful forms of syntax, fit for different purposes. 

\subsection{Syntactic trinitarianism}

The ultimate aim of this tutorial is to persuade the reader to accept a new and multifaceted vision of syntax which we (rather grandly) call \emph{syntactic trinitarianism}: terms, diagrams, graphs. 
They are all ultimately interchangeable representations of an abstract, Platonic if you like, ideal of syntax, each serving distinct roles:
\begin{description}
\item[Terms.] Despite the obvious appeal of graphical syntaxes, in particular to beginners, term syntax remains the workhorse of formal writing, be it programming, mathematics, or logic. 
The serial sequence of symbols, easily grouped by the brain into tokens, wrapping around the end of the page, decorated with whitespace and font variations remains the tried-and-tested repository for formal text. 
Our facility as humans with this representation arguably goes back to language itself evolving first as a spoken form of communication: a serial sequence of auditory symbols, decorated with pitch and tone, modulated by volume and speed. 
Programs ranging from student projects of tens of lines to software infrastructure such as operating systems or compilers of millions of lines of code can be written, understood (to some extent), and evolved using this format. 
For all its shortcomings term syntax is here to stay.
\item[Graphs.] Algorithms do not work very well on terms.
They are an awkward data structure. 
But graphs are efficient, allowing a far greater diversity of structure compared to the list-like terms. 
They can represent sharing and they can allow fast access from one region to the other by following edges. 
However, the more expressive structure comes at a conceptual cost. 
When modifying graphs, maintaining the right invariants is difficult. 
Writing algorithms and reasoning about properties of graphs is harder since they lack inductive structure. 
Even coming up with the right graph structure for a problem is difficult since the design space can be broad and not obviously constrained. 
\item[Diagrams.] The conceptual bridge between terms and graphs are diagrams. 
A new development, with roots in theoretical physics and category theory, string diagrams are ideally suited to this purpose. 
They are syntactic in the sense that they can be generated by grammars, but they are two-dimensional rather than linear, much like the standard notation for matrices is a two-dimensional syntax. 
The two-dimensional syntax, if written in a certain way, makes the graph representation intuitive and its formalisation unsurprising. 
\end{description}
The particular concept that bridges diagrams and graphs is that of \emph{foliation}, to which we dedicate Section~\ref{sec:foliations}. 
A foliation of a diagram is obtained by representing it first as a graph, then representing the graph back into a diagram of a particular form, which we can thing of as a quasi-normal form. 
These are not true normal forms because they are not unique, yet they are useful in greatly simplifying the structure of diagrams, ultimately allowing the formulation of inductive algorithms on essentially non-inductive data structures. 
These algorithms are not efficient, but they can be reasoned about using convenient inductive methods. 
Moreover, these algorithms, which are essentially executable specifications, can be further refined into graph-algorithms which are efficient. 

\begin{center}
\S
\end{center}

Our presentation assumes some basic understanding of the relevant categorical concepts, particularly in order to motivate the entire project. 
Some understanding of basics of programming language theory, such as operational semantics and abstract machines, as well as basic compiler developments should further enhance the appreciation of the material. 
Despite the theoretical underpinnings, the presentation here should lead to a practical understanding of hierarchical graphs as a realistic and useful data structure for representing syntax in compilers, interpreters, and program analysis tools.

The style of presentation is neither purely theoretical nor purely practical, but a kind of `technology transfer'. 
Readers with a background in category theory should be able to find genuine practical applications which hopefully will inform their work and allow them to converse and collaborate with compiler writers. 
Readers with a background in compiler development should find themselves alerted to the possibilities of using category theory not only as a semantic framework, not only as a type-theoretic framework, but also as a syntactic framework. 
Finally, this work is intended to be read as a tutorial rather than a comprehensive survey or a reference handbook. 
To minimise distractions we will save all references and discussions for special sub-sections. 

The structure of the presentation is linear and the reader is advised to proceed section-by-section.
We start with a general introduction to string diagrams in Section~\ref{sec:sdssmc} in which we pay attention to some constructions which are generally disregarded in the standard literature, in particular adjunctions and strictification, but are essential for dealing with closed monoidal structure. 
Building closely on this we introduce, in Section~\ref{sec:hsdcmc}, string diagrams for closed monoidal categories and Cartesian closed categories, with an immediate application to the syntax and equational theory of the $\lambda$ calculus. 
The connection between the two dimensional syntax of string diagrams and the graph rewriting intuitions are made precise in Section~\ref{sec:graphs}.
Section~\ref{sec:oslc} will turn to the operational semantics of the $\lambda$ calculus, pure but also extended with operations into more practical programming languages. 
Finally, in Section~\ref{sec:apps} we give three non-trivial applications of string diagrams: type inference, closure conversion, and reverse automated differentiation. 
They will arguably provide some evidence that string diagrams lead to clearer and more insightful versions of some known algorithms that, in their terms-based formulation, are considered sometimes exceedingly complicated.

\subsection{Further reading and related work}

We will only mention some related work in terms of syntax and intermediate representations. 
There is much more related work which we shall not mention here as it will be discussed in more detail in the following sections. 

The earliest and one of the most influential approaches to making syntax a data structure with better algorithmic properties is de Bruijn `nameless dummies', known widely as `de Bruijn indices'~\cite{de1972lambda}. 
In this notation variable names are replaced by natural numbers, assigned using a deterministic schema that takes into account the depth of the binder. 
This notation is widely used in formalisations of the $\lambda$ calculus, but it has the drawback that variables need to be re-indexed following applications. 
This is a symptom of it being non-compositional, as the integer assignment schema is global for the term. 

For representations of syntax two alternative approaches are significant. 
The first one has been promoted by Pitts, Gabbay, Urban, Cheney and others and is based on so-called `nominal techniques'~\cite{pitts2013nominal}.
The key goal here is to reconcile quotienting by alpha equivalence with recursive definitions and inductive proofs. 
This goal is achieved by introducing names as a data structure equipped with only two operations: one local, testing for equality, and one global, generating a fresh name. 
Because names are indistinguishable save for their identity they are dubbed `atoms'. 
An essential consequence of this is that nominal data structures are `equivariant', which means that two elements are indistinguishable if the underlying atoms are systematically changed. 

Graph-based representations are also implicitly equivariant, since graph isomorphism subsumes equivariance for the actual identity of nodes and edges. 
However, nominal techniques give a nice syntactic model of binding which does not appear in the graph representation, where it does not seem to be required. 
Instead of the standard structural induction or recursion that nominal representations support, the string-diagram representation offers a new approach, of induction or recursion on `foliations'. 

A different approach to syntax is `higher-order abstract syntax' (HOAS) \cite{pfenning1988higher}. 
It is a way of incorporating notions of binding in an abstract syntax tree when using a meta-language that itself supports binding. 
This is helpful for encoding programming languages into, for example, proof checkers, but it is not ideal as a data structure to be used, as is our aim, as an intermediate representation in a compiler. 
HOAS and de Bruijn notations complement each other, and can be used together~\cite{DBLP:conf/icfp/HickeyNYK06}.

The approach to intermediate representation in compilers is quite different from using de Bruijn or HOAS notations, which indicates a rift between theory and practice. 
There is also a sometimes manifested confusion between `intermediate representation' and `intermediate language', which we shall clarify before we proceed. 
The former is a (programming) language not intended for a human programmer. 
It is either an aggressive desugaring of the object language, as is the case for example with OCaml's FLambda intermediate language~\cite{ocaml}, or is an abstraction of the details of assembly, as is the case with LLVM~\cite{lattner2004llvm}. 
In both cases, the intermediate language has a much simplified syntax, which makes specifications less cumbersome. 

In contrast, an intermediate representation (IR) is a data structure, the simplest of which is the abstract syntax tree (AST). 
The AST is quite simple, for example missing any concept of direct sharing, and relying either on copying or on auxiliary data structures such as symbol tables. 
The graph representation unifies all this in a single, convenient data structure. 
Graph-based IRs have been used before with some success in commercial compilers \cite{DBLP:conf/irep/ClickP95}, sometimes under the name of sea-of-nodes representation \cite{paleczny2001java}, but they have rarely given the required importance. 
For instance, it is relatively common for intermediate languages to be formalised in mechanical proof assistants \cite{DBLP:conf/popl/ZhaoNMZ12} but we are unaware of efforts to mechanise the theory of graph-based IRs, as data structures. 
Although we do not present it here, an effort of formalising the graph representation of string diagrams is ongoing. 
Note that there exists an important distinction must be made between IRs and other graph-based representations which are used in compilers, such as control-flow graphs (CFGs) and single-static assignment (SSA), or register allocation algorithms. 

Finally, a special acknowledgement is required for `interaction nets', a diagrammatic representation of programming language which uses graph rewriting to give an operational semantics \cite{DBLP:conf/popl/Lafont90}. 
The only difference, which is nevertheless essential, between interaction nets and our string diagrams is that the former are flat and the latter are hierarchical. 
We will see later on why the flat structure is problematic when it comes to representing the binding structure common in programming languages. 

\newpage
\section{String Diagrams}
\label{sec:sdssmc}

We assume a familiarity of the reader with basic concepts of category theory and string diagrams, therefore in this first section we will move quickly through the material, emphasising only those aspects which are specific to the current presentation. 
For readers who require a more introductory presentation to category theory there are numerous tutorial presentations to start from, many of them aimed at computer scientists \cite{pierce1991basic}. 
Particularly relevant and useful are tutorial introductions to string diagrams, see the very recent \cite{PiedeleuZanasi2023} and the classic \cite{Selinger2011}.

\subsection{Categories and their graphical language}

% definitions
\newcommand{\category}[1]{\mathcal { #1 } }
\newcommand{\id}{\mathrm {id} }
\newcommand{\defeq}{:=}

\begin{definition}[Category] \label{def:category}
A (small) \emph{category} $\category C$ consists of 
\begin{itemize}
\item a set of \emph{objects} $\mathit{obj}(\category C)$ 
\item a set of \emph{morphisms} $hom(\category C)$ such that for each morphism $f$ there exist unique objects $A$ and $B$, respectively called the \emph{source} and \emph{target} of $f$,  written $f:A\to B$. 
\item a \emph{composition} operation $\circ$ on morphisms such that
\begin{description}
\item[closure] for all morphisms if $f:A\to B$ and $g: B\to C$ then $g\circ f:A\to C$ is a morphism
\item[associativity] for all morphisms $f:A\to B$, $g: B\to C$, $h:C\to D$, 
\[ h\circ(g\circ f) = (h\circ g)\circ f \]
\item[identity] for all objects $A$ there exists a morphism $\id_A:A\to A$ such that
\[ \id_B\circ f = f \circ \id_A = f\]
for any $f:A\to B$. 
\end{description}
\end{itemize}

We may freely generate a category $\category C$ from a \emph{signature} $\Sigma = (\Sigma_0,\Sigma_1)$ consisting of a set $\Sigma_0$ of generating objects and a set $\Sigma_1$ of generating morphisms with sources and targets in $\Sigma_0$. First, $\Sigma$-terms are defined inductively as follows:
\begin{itemize}
\item All morphisms $f \colon A \to B$ in $\Sigma_1$ and identities $\id_A \colon A \to A$ are $\Sigma$-terms.
\item If $f \colon A \to B$, $g \colon B \to C$ are $\Sigma$-terms, then $g\circ f:A\to C$ is a $\Sigma$-term. 
\end{itemize}
The category $\category C$ freely generated by $(\Sigma_0,\Sigma_1)$ is defined as having objects $\mathit{obj}(\category C) = \Sigma_0$ and morphisms  $hom(\category C)$ the $\Sigma$-terms quotiented by associativity and identity equations as above. 
\end{definition}
Morphisms are sometimes called \emph{maps} or \emph{arrows}. 
The name of the morphism is sometimes stacked onto the arrow:
\[
A\stackrel f \to B \defeq f:A\to B.
\]
If $A,B\in \mathit{obj}(\category C)$ are fixed we write the set of all $f:A\to B$ as $hom_{\category C}(A,B)$. 
Composition is sometimes defined using a reverse-order operator
\[
f \semic g\defeq g\circ f, 
\]
as it follows more naturally the left-to-right order of the arrow itself, i.e.
\[
(A\stackrel f \to B)\semic (B\stackrel g\to C) = A\stackrel {f\semic g}\longrightarrow C.
\]
The notation above is often streamlined into 
\[
A\stackrel f \to B\stackrel g\to C = A\stackrel {f \semic g}\longrightarrow C,
\]
usually rendered as a so-called \emph{commutative diagram}:
\[
\xymatrix{
    A \ar[dr]_{f \semic g} \ar[r]^f & B \ar[d]^g \\
                     & C }
\]
\begin{exercise}
Render the equations for associativity and identity as commutative diagrams. 
\end{exercise}

Because composition is associative, it is acceptable to elide the brackets when writing up more complex compositions, for example
\[
(f \semic g) \semic (h \semic (i \semic j)) = f \semic g \semic h \semic i \semic j.
\]
More specifically, if using the unbracketed form $f \semic g \semic h \semic i \semic j$ it does not matter how we insert brackets, because all bracketings will result in expressions which are equal even though syntactically distinct. 
So the unbracketed expression represents an \emph{equivalence class} of syntactic terms quotiented by associativity. 
This is a very common syntactic trick in mathematics for associative operations that we used automatically, without reflecting on it. 
But the idea of using notational convention to quotient classes of syntactic objects (terms) by their semantics is an attractive one as it inexorably leads to simpler, more abstract, notation. 
In this situation we say that the equivalence is \emph{absorbed} by the notation. 
The question is: can we generalise this idea to other syntactic structures? 
Using string diagrams, to be introduced momentarily, the answer is a clear `\emph{yes}'.

\begin{definition}
The \emph{graphical language} of a category $\category C$ is 
\begin{description}
\item[Objects] $A$ are represented as labelled wires \[ 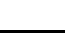 \]
\item[Morphisms] $f:A\to B$ are represented as boxes \[ 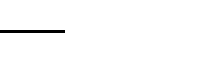 \]
\item[Identities] $\id_A:A\to A$ are represented as labelled wires \[ \hfill \input{pics/id.pdf_tex} \]
\item[Composition] $f \semic g$ is represented as \[ 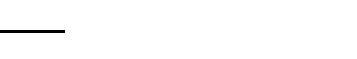 \]
\end{description}
We call terms of the graphical language \emph{string diagrams}.
\end{definition}
The graphical language is a bracket-free syntax so it presumes associativity of composition by construction. 
As we shall see soon, it is possible to use the equivalent of brackets (boxes) in the graphical language for disambiguation. 
However, it remains our goal to avoid this as much as possible. 
Reducing the need for brackets or boxes is one of the hallmarks of good syntax. 
\begin{remark}
The identity equations are absorbed into the graphical language. 
Note the following ambiguity regarding this string diagram,
\[
\input{pics/obj.pdf_tex}
\]
which could be interpreted as $id_A \semic f$ or $f \semic id_B$. 
However, 
\[
id_A \semic f=f \semic  id_B.
\]
Similarly, the string diagram
\[
\input{pics/id.pdf_tex}
\]
can be interpreted as $id_A$ or $id_A \semic id_A$ or $id_A \semic id_A \semic id_A$ and so on, all of which are equal. 
\end{remark}
The astute reader might rightfully complain that, actually
\begin{align*}
id_A &=\input{pics/id.pdf_tex} \\
id_A \semic id_A &=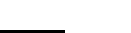\\
id_A \semic id_A \semic  id_A &=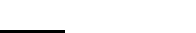. 
\end{align*}
And that is a valid objection which can be raised against a sleight of hand we have performed. 
We did not specify how the string diagrams are to be \emph{represented} concretely. 
The answer is going to be, at this stage, informal. 
Formalising the graphical language as a combinatorial object is possible, and it will capture the intuition that they are graph-like objects with well defined (albeit hidden) interfaces to the right and to the left of the diagram, anchoring what appears like `dangling' wires. 
As graphs, these objects will be quotiented by an appropriately defined notion of isomorphism, which renders the length and other attributes of the wires irrelevant, focussing on connectivity only. 
This is why the three wires of increasing length above are all \emph{equal} in the interpretation, as they can be formalised as isomorphic graph-like structures. 
We consider these formal details mostly as a sanity check, in the sense that they precisely described the desired intuitions, so we postpone them until Section~\ref{sec:graphs}. 

\subsection{Functors and boxes}
This section will introduce some more basic definitions and notations, to be illustrated with interesting examples later on. 

\begin{definition}
Let $\category C$ and $\category D$ be two categories. 
We define a \emph{map} $F:\category C \to \category D$ between these categories as a map on objects together with a map on morphisms such that source and target objects are preserved:
\[
F(f:A\to B) = Ff:FA\to FB. 
\]
\end{definition}
\begin{definition}
In the graphical language maps $F:\category C \to \category D$ are represented as $F$-labelled boxes: 
\[
 F(f:A\to B) \quad := \quad \raisebox{-.4cm}{\hbox{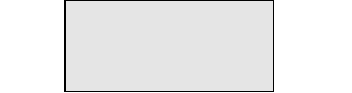}}
\]
\end{definition}
\begin{remark}
The \emph{box} serves the same role, in the two-dimensional landscape of string diagrams, that brackets serve in the one-dimensional world of terms. 
But, unlike a bracket, the box acts on the object: as the wire crosses the boundary of the box the objects on either side may be different. 
This can be emphasised by the use of further graphical conventions, such as use of distinct colours or shading.

Also note that wire labeled with $A$ ($B$ respectively) inside the box and $FA$ ($FB$ respectively) outside of the box does not `cross' the boundary since the type of the wire changes, so it is a different wire. 
This distinction is made clear in the concrete representation which we shall see later, but for now we take this as a given. 
This is why we can also draw the box in a way that make the distinctiveness of the two wires more clear: 
\[
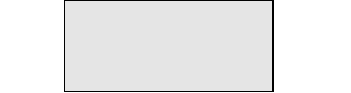
\]
In the semi-formal string-diagram notation will make this distinction wherever we think it might introduce ambiguities. 
\end{remark}
\begin{definition}[Functor]
A map $F:\category C \to \category D$ is a \emph{functor} if it preserves composition and identity, i.e.
\begin{align*}
F(id_A:A\to A) &= id_{FA}: FA\to FA \\
F(A\stackrel f\to B \stackrel g\to C) &= FA\stackrel {Ff}\to FB \stackrel {Fg}\to FC.
\end{align*}
\end{definition}
\begin{remark}
In the graphical language functoriality is expressed as the following equalities of string diagrams:
\[
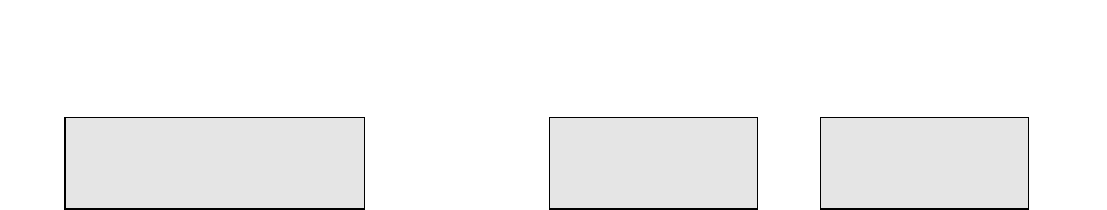
\]
\end{remark}
If clear from context we may elide some of the annotations in order to prevent clutter. 
For instance, if $\category C$ and $\category D$ are singletons (a single object) we can write the functoriality diagrams for some fixed map between them as:
\[
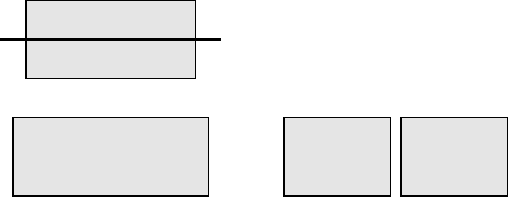
\]
\begin{remark}
The notion of map of categories generalises in an obvious way to maps from two categories, i.e. $F:\category C_1\times \category C_2\to \category D$. 
The graphical language representation of the box for $F(f_1, f_2)$ now includes two `compartments' corresponding to the two arguments. 
\[
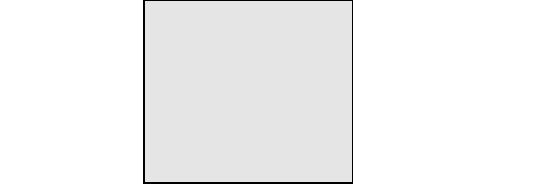
\]
The internal compartments are not drawn unless needed for disambiguation.
Maps of two categories which are also functors are called \emph{bifunctors}.
In the sequel, we will not use bifunctor boxes except briefly to introduce the definition of monoidal categories, after which we will not need them again. 
\end{remark}
Using different colours for the inside and outside of the box, indicating the two different categories, can also add an extra visual indicator helpful for error-checking, as we shall see in the next section. 
 
\subsubsection{Natural transformations}
Climbing further up the ladder of abstraction, natural transformations can be seen as \emph{maps between functors}. 
They are collections of morphisms indexed by objects in the source category, which are `uniform' over the target category. 
The uniformity property is defined below. 
\begin{definition}[Natural transformation]
If $F,G$ are functors from $\category C$ to $\category D$ then a \emph{natural transformation} $\eta:F\Rightarrow G$ between them is an object-indexed family of morphisms such that:
\begin{enumerate}
\item For each object $X$ in $\mathcal C$ there exists a morphism $\eta_X: FX\to GX$ in $\category D$. 
\item For each morphism $f:X\to Y$ in $\mathcal C$
\[
\xymatrix{
X\ar[d]_f & FX\ar[d]_{Ff} \ar[r]^{\eta_X} & GX\ar[d]_{Gf} \\
Y & FY\ar[r]_{\eta_Y} & GY
}
\]
\end{enumerate}
\end{definition}
The morphism $\eta_X$ is called the \emph{component of $\eta$ at $X$}.
It is usual, and convenient, to add the morphism $f:X\to Y$ to the commuting diagram as an annotation. 
If the components of a natural transformation are isomorphisms then the natural transformation is called a \emph{natural isomorphism}.
In this case, we write $\eta_X : FX \simeq GX.$

Natural transformations are, from the point of view of string diagrams, just collections of morphisms so no new graphical notation needs to be introduced. 
But the commuting diagram for components can be represented graphically as:
\[
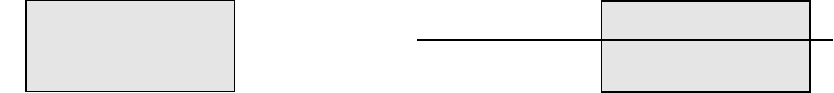
\]
Above we left out all the information that can be inferred from the diagram context, to reduce clutter. 
To further reduce clutter we may indicate the functor using colour-coding or by decorating the border of the box, e.g.
\[
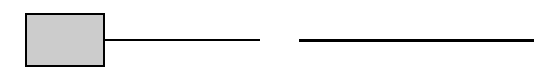
\]
and removing more information that can be inferred from the context of the diagram ($X$ and $Y$ can be recovered from $f:X\to Y$ etc.)
\begin{remark}
The graphical notation, we can argue, is already starting to pay dividends by now. 
String diagrams more directly suggest that naturality is a generalisation of commutativity, using functoriality and indexing as a clever way to make the source and the target of the map $f$ match with those of a well chosen component of the natural transformation. 
Furthermore, as we shall see shortly, when calculating with string diagrams it will be easier to identify formula `\emph{redexes}', i.e. spots in a formula where an equation can be applied. 
\end{remark}

\subsubsection{Adjunctions} \label{sec:adjunctions}
Two functors that stand in a certain relation to each other are said to be \emph{adjoint}.
Adjoint functors are common in mathematics and computer science. 
We shall not present examples at this moment, just introduce definitions and properties, with examples to follow later. 

Adjunctions can be characterised in several ways, but we will emphasise the definition most suitable to a nice rendering in the graphical language of string diagrams, namely the so-called unit-counit adjunction.
\begin{definition}[Adjunction]\label{def:adj}
A \emph{unit-counit adjunction} between two categories $\category C$ and $\category D$ consists of two functors $F:\category D\to \category C$ and $G:\category C\to\category D$ and two natural transformations $\epsilon:F\circ G\Rightarrow \id_{\category C}$ and $\eta:\id_{\category D}\Rightarrow G\circ F$ respectively called the co-unit and the unit of the adjunction such that for each object $A$ in $\category C$ and $X$ in~$\category D$: 
\begin{gather*}
 F\eta_{X} \semic  \epsilon_{FX} = \id_{FX}\\
 \eta_{GA} \semic  G\epsilon_A = \id_{GA}.
\end{gather*}
\end{definition}
If two functors are adjoint, as defined above, we say that $F$ is the \emph{left adjoint} and $G$ the \emph{right adjoint}, written $F\dashv G$. 

Using string diagrams, the two diagram equations become quite suggestive in a graphical rendering:
\[
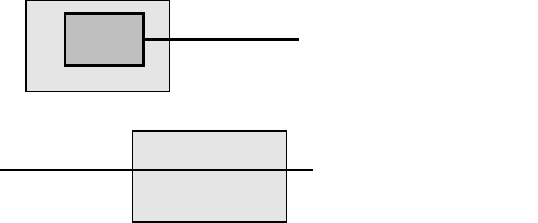
\]
To make the diagrams even more suggestive we can represent the unit and the counit as opposite facing half-circles. 
To avoid ambiguity, if necessary, they can be annotated with the indexing object. 
We use different colours to indicate morphisms in the two categories, and colour the functor boxes in the same colour as their codomain. 
With these further conventions the two string diagram equations become simply:
\[
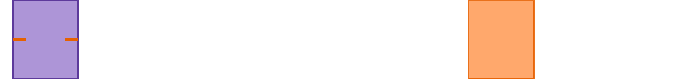
\] 
The colours (purple and orange) are used to identify what category the identity belongs to, and also the two functors by tagging them with the colour of their codomain category. 
All elided information (object labels, etc.) can be unambiguously inferred. 

\paragraph{Homset adjunctions}

The unit/co-unit presentation of adjunctions is arguably the best when using string diagrams. 
On the basis of this definition we can also show the equivalence of the alternative homset formulation of adjunctions.
\begin{definition}[Adjunctions (homset)]\label{def:homset}
An adjunction between categories $\category C$ and $\category D$ consists of two functors $F:\category D\to \category C$ and $G:\category C\to\category D$ such that for all objects $A$ in $\category D$ and $V$ in $\category C$ there exists a family of bijections, natural in $A$ and $V$, between morphisms sets
\[
hom_{\category C}(FA, V)\cong hom_{\category D}(A, GV).
\]
\end{definition}
Naturality here means that there are natural isomorphisms between the pair of functors ${\mathcal {C}}(F-,V):{\mathcal {D}}\to \mathrm {Set} $ and ${\mathcal {D}}(-,GV):{\mathcal {D}}\to \mathrm {Set}$ for a fixed $V$ in $\mathcal {C}$, and also the pair of functors ${ {\mathcal {C}}(FA,-):{\mathcal {C}}\to \mathrm {Set} } $ and ${\mathcal {D}}(A,G-):{\mathcal {C}}\to \mathrm {Set} $ for a fixed $A$ in $\mathcal {D}$.

We can now show that Definition~\ref{def:homset} is implied by Definition~\ref{def:adj}. The two maps, which will turn out to be inverses, are:
\[
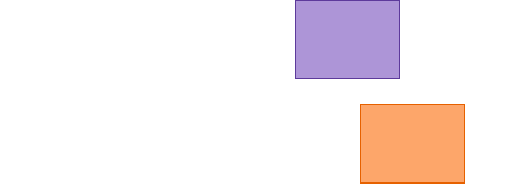
\]
where the orange box, decorated with a small outline mark, is the $G$ functor and the purple box, decorated with a small solid mark, the $F$ functor. 

This is how we show the first direction of these maps being inverses. 
\[
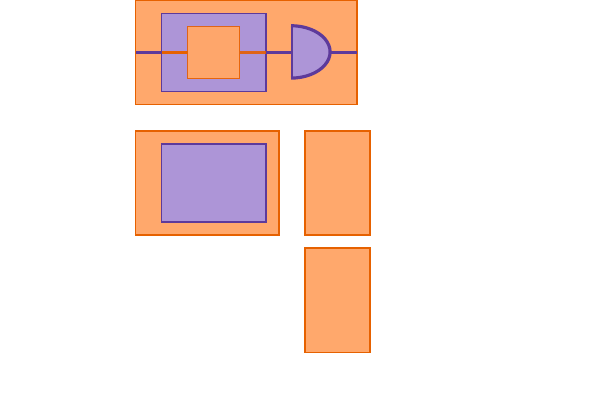
\]

\begin{exercise}
Show the other direction of the property of homset functions being inverses. 
\end{exercise}
\begin{remark}
The reader is encouraged to repeat the same exercise in the conventional term language of category theory or using conventional commuting diagrams in order to experience the relative convenience of the diagrammatic notation. 
For calculations with functors (which are not endofunctors) the use of colour serve as an additional `type-checking' mechanism, which adds an extra layer of notational comfort.
\end{remark}

\subsection{Monoidal categories}

A bifunctor of particular importance in many areas of mathematics, physics, and computer science is the \emph{monoidal tensor}, $\otimes:\category C\times \category C\to \category C$ as it can be used to model the construction of aggregate data types out of simpler data types. 
The monoidal tensor is \emph{associative up to isomorphism}, indicating the fact that data can be aggregated in ways that are different yet retrievable one form another in a canonical way. 
This property holds both for \emph{product} types and for \emph{sum} types, both of which are instances of monoidal tensor. 
If the tensor is associative \emph{on the nose}, i.e. the isomorphisms are identities, it is said to be strict. 
Once we introduce strict monoidal tensors in a category, the graphical language of string diagrams truly starts to achieve its full potential. 

\begin{definition}[Monoidal category]
A \emph{monoidal category} is a category $\category C$ equipped with the following structure:
\begin{description}
\item[Tensor] A bifunctor $\otimes:\category C\times \category C\to \category C$ 
\item[Unit] An object $I$, also called \emph{identity}
\item[Associators] A family of object-indexed natural isomorphisms (\emph{associators})
\[ \alpha_{A,B,C} : A\otimes(B\otimes C)\simeq (A\otimes B)\otimes C \]
\item[Unitors] Two families of object-indexed natural isomorphisms (left and right \emph{unitors})
\begin{align*}
\lambda_A:I\otimes A\simeq A \\
\rho_A:A\otimes I \simeq A,
\end{align*}
\end{description}
such that the following coherence conditions hold for all objects $A,B,C$:

\begin{gather*}
\xymatrix{
A\otimes(B\otimes (C\otimes D)) \ar[r]^{\alpha_{A,B,C\otimes D}}\ar[d]_{\id_A\otimes\alpha_{B,C,D}} & (A\otimes B)\otimes (C\otimes D) \ar[r]^{\alpha_{A\otimes B,C,D}} & ((A\otimes B)\otimes C)\otimes D\\
A\otimes((B\otimes C)\otimes D) \ar[rr]_{\alpha_{A,B\otimes C,D}}                     & & (A\otimes (B\otimes C))\otimes D \ar[u]_{\alpha_{A,B,C}\otimes \id_D}
} \\
\xymatrix{
A\otimes(I\otimes B) \ar[rr]^{\alpha_{A,I,B}}\ar[dr]_{\id_A\otimes \lambda_B} & & (A\otimes I)\otimes B \ar[dl]^{\rho_A\otimes \id_B}\\
& A\otimes B
}
\end{gather*}
\end{definition}
The first coherence is usually called the \emph{pentagon diagram} and the second the \emph{triangle diagram}, due to their respective shapes. 

\begin{remark}
It is quite common to have composition higher priority than tensor, so $f\otimes g\circ h$ can only be bracketed as $f\otimes (g\circ h)$. 
However, for most of the examples here it is more convenient that the tensor should be stronger than composition, so 
\[
f\otimes g \semic h = (f\otimes g) \semic  h.
\]
\end{remark}

The naturality of the associator isomorphism means that the following string diagram equation holds, where the box represents the $\otimes$ bifunctor:
\[
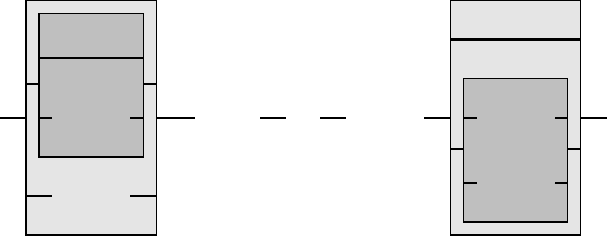
\]
The notation improves greatly if the tensor is \emph{strict}, that is the associators $\alpha$ are identities. 
In this case the equation becomes:
\[
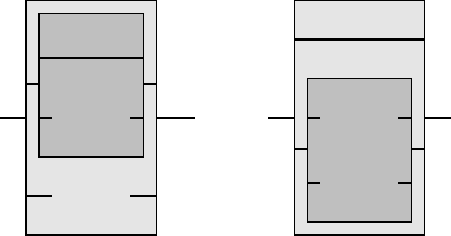
\]
\begin{remark}
We are now in a situation similar to that of the associativity of composition, in which 
\[
(f \semic g) \semic h=f \semic (g \semic h)=f \semic g \semic h.
\]
We recall that boxing is the two-dimensional counterpart of bracketing, so the diagram above expresses the property that the two possible boxings are equal. 
We are justified then to quotient out the boxes altogether and to absorb this equation into a new notational convention, representing the strictly associative tensor simply as:
\[
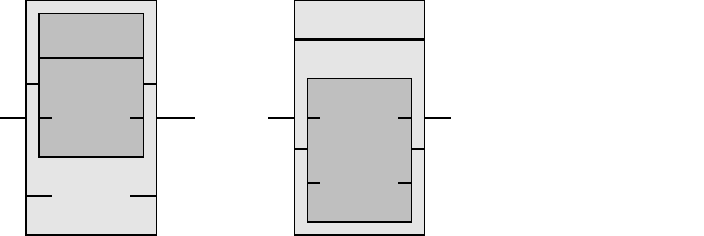
\]
Removing boxes of functors which are associative on the nose is the two-dimensional version of removing brackets of associative operations in the one-dimensional notation for terms. 
\end{remark}
A similar situation arises in the case of the unitors, where the naturality conditions, diagrammatically, are:
\[
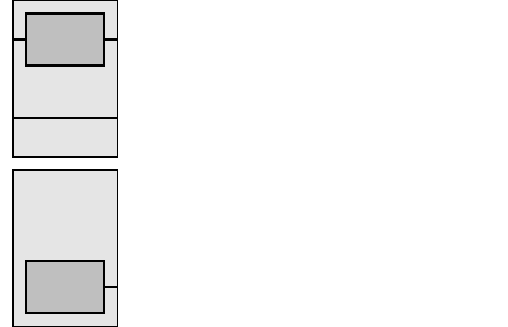
\]
On the right-hand-side of the diagram the functor is the identity functor. 

In the strict case the unitors are also identities, so the two diagrams become:
\[
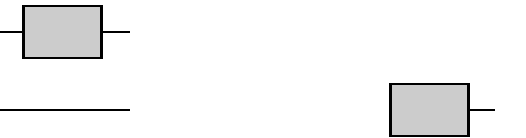
\]
which only makes sense if the unit identity is deleted altogether. 
\begin{remark}
Just like the identity morphism can be represented graphically in a way that makes it absorbed by composition, so the identity for the unit $\id_I$ can be absorbed into the tensor by representing it as just empty space.
\[
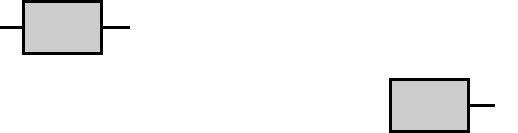
\]
In the above, the dotted box is simply an attempt to highlight the presence of empty space; it has no further meaning. 
\end{remark}

We have now defined the core diagrammatic notation of strict monoidal categories. 
To summarise, it is a syntax of labelled boxes and wires in which the boxes are the morphisms and the wires the objects. 
Morphisms are composed by connecting wires serially, while the monoidal tensor is a parallel stacking of diagrams. 
The identity of the unit of the tensor can be elided by representing it as empty space. 

In a strict monoidal category objects are essentially lists with the unit $I$ the empty list. 
We introduce notations to define lists and basic operations on lists, which are handy for working with morphisms in strict monoidal categories.
\begin{description}
\item[List of objects:] $[A_1, A_2, \ldots,A_k]$
\item[Empty: ] $[\ ]$
\item[Concatenation: ] $[A_1,\ldots,A_k]+\!\!+[B_1,\ldots,B_j]=[A_1, \ldots,A_k,B_1, \ldots,B_j]$.
\end{description}
This notation will come in handy occasionally. For instance, we may use it to freely obtain strict monoidal categories starting from a signature, similarly to the case of plain categories (Definition~\ref{def:category}).

\begin{definition}[Freely Generated Strict Monoidal Category] \label{def:freeMC} A \emph{monoidal signature} is a pair $\Sigma = (\Sigma_0,\Sigma_1)$ where $\Sigma_0$ is a class of generating objects and $\Sigma_1$ a pair of generating morphisms with sources and targets lists of objects in $\Sigma_0$. Monoidal $\Sigma$-terms are defined inductively as follows:
\begin{itemize}
\item All morphisms $f \colon L_1 \to L_2$ in $\Sigma_1$, the identity $id_{[]} \colon [] \to []$, and identities $\id_{[A]} \colon [A] \to [A]$ for all $A \in \Sigma_0$, are $\Sigma$-terms.
\item If $f \colon L_1 \to L_2$, $g \colon L_2 \to L_3$ are $\Sigma$-terms, then $f ; g \colon L_1 \to L_3$ is a $\Sigma$-term. 
\item If $f \colon L_1 \to L_2$, $g \colon L_3 \to L_4$ are $\Sigma$-terms, then $f \otimes g \colon (L_1 +\!\!+ L_3) \to (L_2 +\!\!+ L_4)$ is a $\Sigma$-term. 
\end{itemize}
The strict monoidal category $\category C$ freely generated by $(\Sigma_0,\Sigma_1)$ is defined as having objects $\mathit{obj}(\category C) = \Sigma_0^{\star}$ (i.e., lists of objects in $\Sigma_0$) and morphisms  $hom(\category C)$ the $\Sigma$-terms quotiented by the equations of strict monoidal categories, where $[]$ acts as the identity for $\otimes$, plus the following:
\[ \id_{L_1 +\!\!+ L_2} = \id_{L_1} \otimes \id_{L_2}. \] 
\end{definition}

A common situation is for a strict monoidal category to be freely generated by a signature $(\Sigma_0,\Sigma_1)$ where $\Sigma_0$ consists of a single object, denoted by $\bullet$. 
We call such strict monoidal categories \emph{single sorted}.
In this case the generated objects are 
\begin{align*}
\overline 0 = [\ ] &= I \\
\overline 1 = [\bullet] &= \bullet \\
\overline 2 = [\bullet, \bullet] &= \bullet \otimes \bullet \\
\overline 3 = [\bullet, \bullet, \bullet] &= (\bullet \otimes \bullet)\otimes \bullet = \bullet\otimes(\bullet\otimes\bullet) \\
&\vdots
\end{align*}
Every object in the strict monoidal category is fully characterised by its size so it can be written as $\overline n$ for $n$ some natural number. 
\begin{example}\label{ex:bool}
Suppose that a single-sorted strict monoidal category has
morphisms generated by $\mathsf t:\overline 0\to\overline 1, \mathsf f:\overline 0\to\overline 1$, and ${\land},{\lor}:\overline 2\to \overline 1$.
Since the category is constructed as strict, generators of type $\overline 2\to \overline 1$ are represented as having two wires on the left and one wire on the right.
Concretely, a term such as
\[
\mathsf t\otimes(\mathsf f\otimes \mathsf t) \semic ((\id\otimes \mathsf f \semic {\land})\otimes{\lor}) \semic {\lor}
\]
corresponds to the diagram:
\[
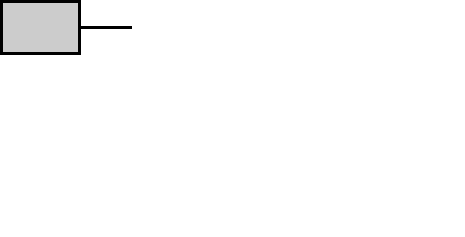
\]
\end{example}

\begin{example}\label{ex:funcmon}
Some of the equations of strict monoidal categories become trivialised by the being fully absorbed into the graphical notation. 
For instance, consider the functoriality of the tensor:
\[
(f\otimes g) \semic (f'\otimes g') = (f \semic f')\otimes (g \semic g')
\]
for any $f,f',g,g'$ so that the compositions are well defined. 
This is trivial since both sides of the equation correspond to the diagram
\[
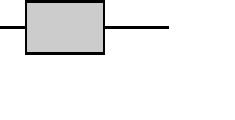
\]
\end{example}

\begin{exercise}
Show that the following equations hold by constructing the string diagrams of all the expressions involved:
\begin{equation*}
(f \semic \id)\otimes(\id \semic g) = f\otimes g = (\id \semic f)\otimes(g \semic \id).
\end{equation*}
\end{exercise}

\begin{example}
We refer the reader back to Example~\ref{ex:bool}. 
Let us equip the category with equations consistent with the standard interpretation of Boolean logic, i.e. $t\otimes t \semic  \land = t$, $t\otimes f \semic \land = f$ etc.
The concrete term used in the example reduces to a single Boolean value. 
A possible step-by-step calculation in the one-dimensional (term) syntax involves repeated applications of naturality to re-associate the term so that the equation of boolean connectors can be applied. 
On the other hand, the two-dimensional (string diagram) makes these redexes obvious. 
The reader is invited to fill in the details, which are edifying. 
\end{example}

\subsection{Strictification and string diagrams}\label{sec:strictification}

The simplified representation of the functorial box for the tensor makes calculation with string diagram for strict monoidal categories particularly pleasant. 
However, many interesting categories are not strict. 
The motivating examples of data types in programming languages (product, sum) are not strict. 
In a programming language tuples $(a,(b,c))$ and $((a,b),c)$ and $(a,b,c)$ are never considered equal, even though canonical isomorphisms between them are straightforward. 
Tupling, when non-strict, is a convenient mechanism for abstraction. 
For instance, consider a type of headers ($H$), a type of payloads $P$, from which we can construct a type of messages $M=H\otimes P$, from which we can construct a type of error-correcting messages $E=M\otimes C$, where $C$ is a checksum applied to a message of type $M$. 
But we could have, on the other hand, a type of error-correcting payloads $E'=P\otimes C$ where the checksum is applied to the message $M$ only, from which we can construct messages $M'=H\otimes E'$. 
The first data-type is $(H\otimes P)\otimes E$, the second $H\otimes (P\otimes E)$, and they should not be taken as equal even if isomorphic. 

As our overall intent is to define string diagrams for programming languages, we can not usually assume the tensor to be strict. 
But without strictness, the functorial box for the tensor loses its simplifying properties and consequently the graphical notation loses its elegance. 
This may seem as an insurmountable obstacle, but there exists a principled way out using \emph{strictification}. 
As we shall see, this construction introduces a deliberately strict version of a monoidal tensor in a way that result in a new category which is equivalent to the original. 
The diagrammatic presentation should give a clearer account of what is going on. 

This construction
%, originally due to Mac Lane~\cite{MACLANE} 
is interesting and powerful, and has been interpreted in various ways. 
The broadest, and rather misleading, common interpretation is that ``in a free monoidal category all diagrams made of unitors and associators commute.''

\begin{definition}
  \label{definition:catd}
  Given a monoidal category $(\category C, \otimes, I)$ we define
  $\overline{\category C}$ as the strict monoidal category freely generated by:
  \begin{itemize}
    \item Objects $\overline A$ for each object $A$ 
    \item Morphisms $\overline f:\overline A\to\overline B$ for each morphism $f:A\to B$ in $\category C$.
    \item De-strictifying generators:
    \begin{align*}
    \phi &: [\ ]\to \overline I & & 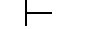\\
    \psi_{A, B} &: [\overline A, \overline B]\to \overline{A\otimes B} & & 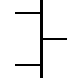
    \end{align*}
    \item Strictifying generators:
    \begin{align*}
    \phi^* &: \overline I \to [\ ] & & 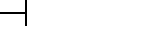\\
    \psi_{A,B}^* &: \overline{A\otimes B}\to [\overline A, \overline B] & & 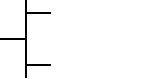
    \end{align*}
   \end{itemize}
   subject to the equations in Figure~\ref{fig:strictification}.
\end{definition}
\begin{figure}
\[
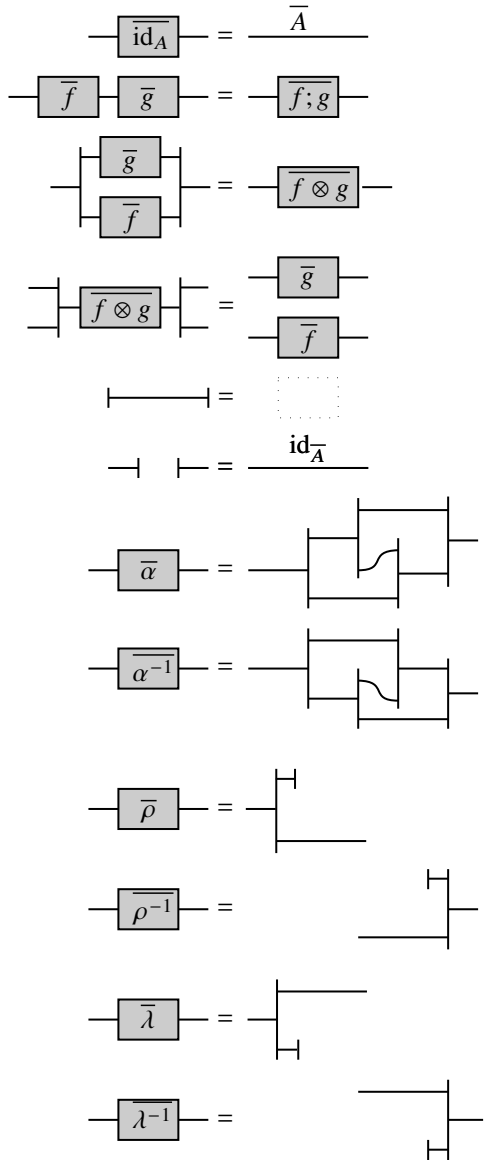
\]
\caption{Strictification equations (Definition~\ref{definition:catd})}
\label{fig:strictification}
\end{figure}
Since the category $\overline{\category C}$ is strict monoidal we are entitled to use the strict string diagrams without ambiguity. 
In Figure~\ref{fig:strictification} and elsewhere we omit the labels on the generator boxes since they can be unambiguously identified by shape and number and position of wires. 

This is a functorial construction, yielding a (monoidal) equivalence between
$\category C$ and $\overline{\category C}$.
The details involved in formulating the following precisely are outside the scope of this tutorial survey, so we will only state this property informally:
\begin{remark}\label{rem:strict}
The category $\overline{\category C}$ is essentially equivalent to $\category C$. 
Every construction in the former category is essentially the same as in the latter. 
\end{remark}
The second sentence above may seem puzzling, since the use of strictifiers and de-strictifiers appears to greatly complicate the diagrammatic language. 
We will show, with a simple example, that it is not the case, and shine a light, at the same time, on the implications of strictification as a syntactic device, allowing the rigorous use of string diagrams for non-strict tensors. 
\begin{example}
Recall example~\ref{ex:funcmon} in which an equation, namely
\[
(f\otimes g) \semic (f'\otimes g') = (f \semic f')\otimes (g \semic g'),
\]
is shown to hold by constructing the string diagram expression of both side and noticing that it is the same, when the tensor is strict. 
Let us conduct the same exercise in a non-strict diagram. 

The two diagrams are:
\[
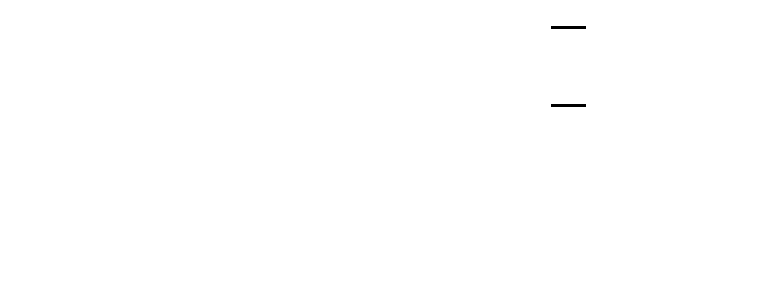
\]
\end{example}
Showing only the part of the diagram that differs, the first diagram can be reduced to the second using the first, then fourth, then first again (twice) equations from Figure~\ref{fig:strictification}:
\[
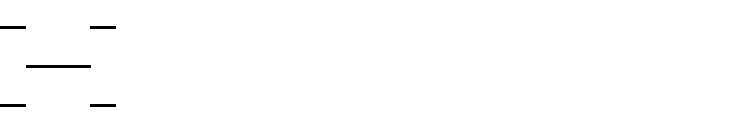
\]
The important idea in the proof of the property in Remark~\ref{rem:strict} is that any construction in the non-strict category can be expressed by a diagram that has the following general shape:
\[
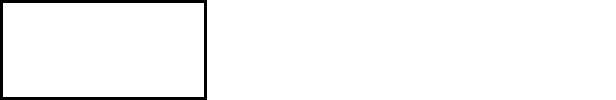
\]
\begin{description}
\item[strictification]: A prefix diagram ($\sigma^*$) which is a tree of strictifiers, flattening the tensors,
\item[diagram]: The actual diagram $f$ as if constructed in the strict setting,
\item[de-strictification]: A suffix diagram ($\tau$) specifying the association of the tensors.
\end{description}
So we can define string diagrams as if they were performed in a strict category, without loss of generality. 
However, we can still use the non-strict tensor whenever convenient to pack structure into the wires for the purpose of abstraction. 
An example will be given in the next section (Example~\ref{ex:butterfly}). 
More concretely speaking:
\begin{remark}
A morphism $f:A\otimes B\to C$ can be drawn without ambiguity in three equivalent ways:
\[
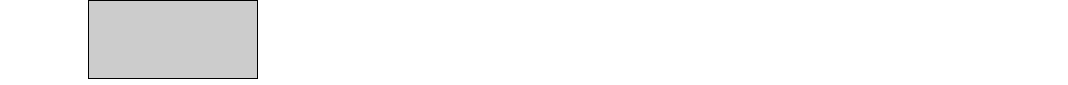
\]
\end{remark}
\begin{exercise}
Using diagrammatic reasoning and the equations of strictification prove some of the expected properties of $\overline{\category C}$, e.g. $\overline {\alpha} \semic \overline{\alpha^{-1}} = \id$.
\end{exercise}
\begin{remark}
The exercise above should make the point that diagrammatic reasoning is not a panacea, and simple equational reasoning is sometimes more handy than diagrammatic reasoning. 
Diagrammatic reasoning is just another arrow in your quiver, to be fired at the appropriate targets. 
\end{remark}
\begin{example}[Sorting network]
A \emph{sorting network} is a freely generated strict monoidal category, where the objects represent a fixed-size integer data type and there is one generator morphism $s_2:2\to 2$, a two-input sorter, plus a family of generator morphisms $k_i:0\to 1$ corresponding to constant integers  $i\in\mathbb N$, so that
\begin{align*}
(k_i\otimes k_j) \semic s_2 &= k_i\otimes k_j &\text{if }i\geq j\\
(k_i\otimes k_j) \semic s_2 &= k_j\otimes k_i &\text{if }i\le j.
\end{align*}
Larger sorting networks are constructed recursively. 
Two common ones are \emph{insertion sort} in which the first $n$ wires are sorted then the remaining value is inserted, and \emph{bubble sort} in which the largest value is propagated to the bottom, then the remaining $n$ values are sorted. 
The insert-sort network is defined by the following two recursive definitions, with base case $max_2=s_2$:
\[
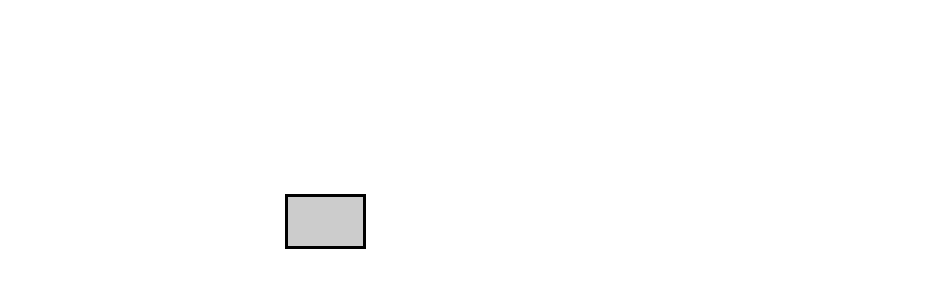
\]
\end{example}
\begin{exercise}
\begin{enumerate}
\item Construct $s_4$ and simplify according to the strictification equations.
\item Define the bubble-sorter, according to the informal specification above. 
\item Prove that the bubble-sorter and the insert-sorter are always equal. 
\item Define an efficient sorter, e.g. a \emph{bitonic sorter}~\cite{nassimi1979bitonic}. 
\end{enumerate}
\end{exercise}

\subsection{Symmetric monoidal categories}\label{sec:smc}
Remember that we gave the tensor ($\otimes$)  higher precedence than composition ($ \semic $), so we shall continue to do that in the sequel. 
\begin{definition}[Symmetric monoidal category]
A \emph{symmetric monoidal category} is a monoidal category $(\category C, \otimes, I)$ such that, for every pair $A, B$ of objects in $\category C$, there is an isomorphism $\sigma_{A,B}:A\otimes B\to B\otimes A$, called \emph{symmetry}, that is natural in both $A$ and $B$ and such that
\begin{align*}
\sigma_{A,I} \semic \lambda_A & =\rho_A \\
\sigma_{A,B}\otimes \id_C \semic \alpha_{B,A,C} \semic \id_B\otimes \sigma_{A,C} &=\alpha_{A,B,C} \semic \sigma_{A,B\otimes C} \semic \alpha_{B,C,A} \\
\sigma_{A,B} \semic \sigma_{B \semic A}&=\id_{A\otimes B}.
\end{align*}
\end{definition}
As discussed in the previous section, it is simpler to work in a strict setting and add strictifier and de-strictifier morphisms globally only if non-strict tensoring is required. 
\begin{remark}
In a strict monoidal category the symmetry equations are:
\begin{align*}
\sigma_{A,I}&=\id_A\\
\sigma_{A,B}\otimes \id_C \semic \id_B\otimes \sigma_{A,C} &=\sigma_{A,B\otimes C}\\
\sigma_{A,B} \semic \sigma_{B \semic A}&=\id_{A\otimes B}.
\end{align*}
\end{remark}
The string diagram language needs no extension, so the two equations above, along with the naturality equations are given in Figure~\ref{fig:eqsym}.
\begin{figure}
\[
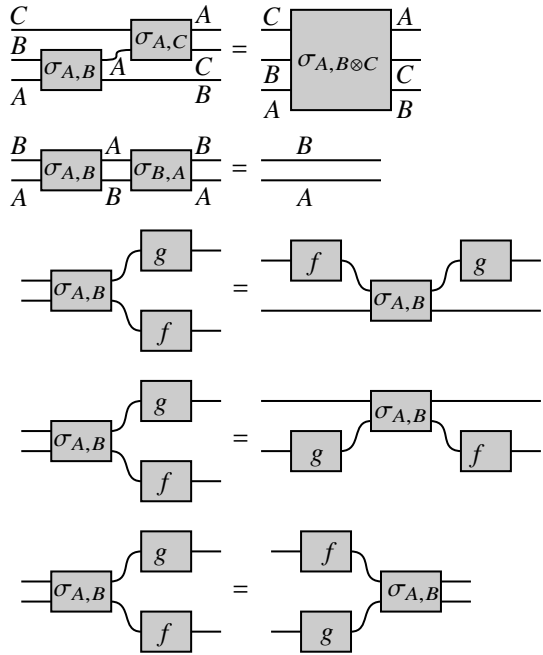
\]
\caption{Symmetry equations}
\label{fig:eqsym}
\end{figure}

Symmetry is an important property of many instances of the monoidal tensor (e.g. Cartesian product) and reasoning with the symmetry equations above is ubiquitous. 
A particularly clever development in the design of string diagram languages is to introduce special graphical notation for symmetry so that these equations become absorbed into the graphical notation. 
This is done by imagining symmetry $\sigma_{A,B}$ as a swapping of the wires for $A$ and $B$. 
With this new and improved graphical notation, the same equations become as in Figure~\ref{fig:eqsymg}, and they are absorbed into graphical equivalence by allowing wires to bend and cross and boxes to slide along them. 
These are informal considerations, but a suitable mathematical concept of graph can be expressed rigorously, and the rigorously defined notion of isomorphism of such graphs corresponds precisely to the intuition at work here. 
\begin{figure}
\[
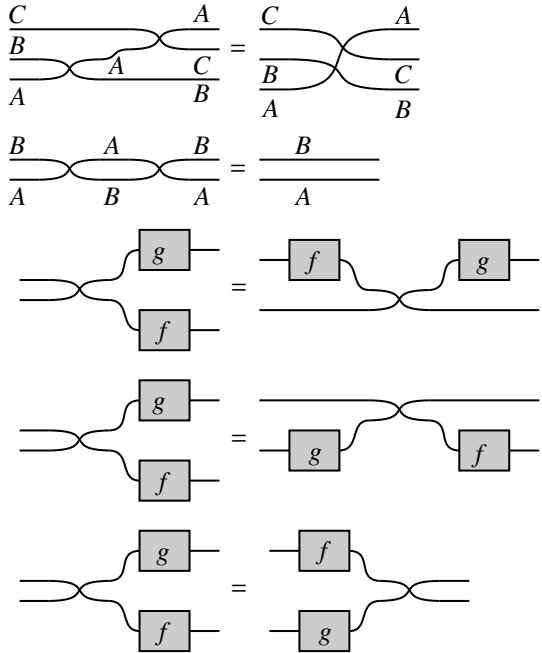
\]
\caption{Symmetry equations absorbed into the graphical notation}
\label{fig:eqsymg}
\end{figure}

The first equation of symmetry shows that $\sigma_{A,B\otimes C}$ can be constructed out of the simpler $\sigma_{A,B}$ and $\sigma_{B,C}$. 
Attempting to formally generalise this axiom to showing how general symmetries can be constructed from the basic symmetry in a single sorted strict monoidal category is an interesting exercise which shows the usefulness of strictification in an already strict category!
\begin{example}\label{ex:butterfly}
The following family of mutually recursive axioms show how $\sigma_{\overline m,\overline n}$ can be constructed from $\sigma_{\overline 1, \overline 1}$ in a symmetric single sorted strict symmetric monoidal category. 
These categories are a common setting for string diagrams, called PROPs (from \emph{PROducts and Permutations}). 
\[
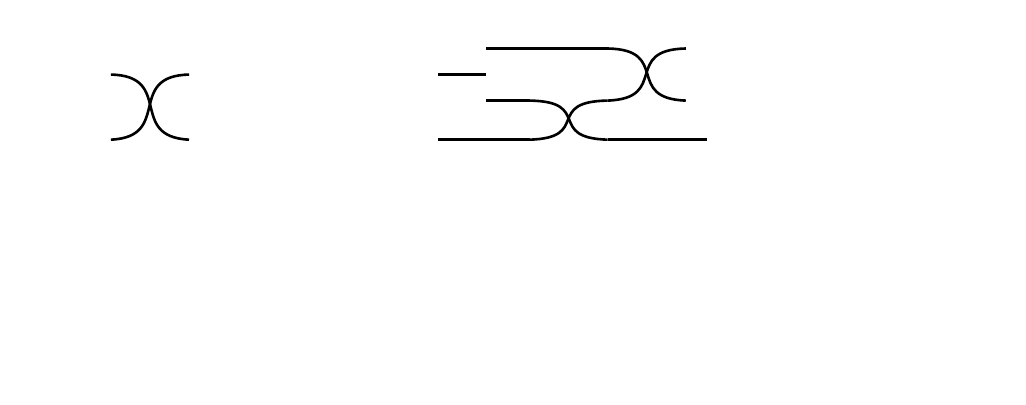
\]
Even in a strict monoidal category it is convenient to have a way to unpack a composite wire into its components, and to pack it again. 
Without such a facility composite symmetries are usually specified using ellipses in the diagrammatic notation, which injects an unpleasant element of informality. 
\end{example}
\begin{example}
We construct $\sigma_{2,2}$ according to the definition in Example~\ref{ex:butterfly}. 
Unfolding the definition produces the diagram
\[
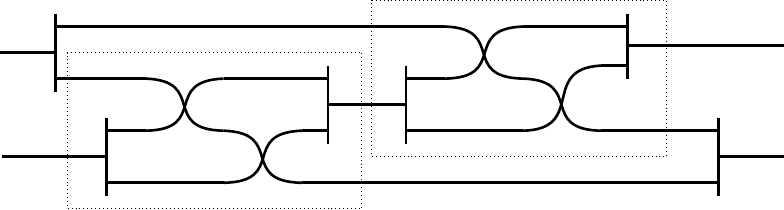
\]
in which the dotted boxes indicate instances of $\sigma_{2,1}$.
Using the strictification axioms the matching pair of a de-strictifier and a strictifier can be eliminated, resulting in:
\[
\input{pics/sigma22b.pdf_tex}
\]
\begin{exercise}
Find two different constructions for $\sigma_{2,2}$, then prove that the two definitions are equal in the equational theory of the PROP. 
\end{exercise}
\begin{exercise}
Prove that for any $m,n$, it is the case that $\sigma_{m,n}$ constructed as above is indeed the symmetry of the PROP.
\end{exercise}
\end{example}

Analogously to monoidal categories, we may freely generate symmetric strict monoidal categories from a monoidal signature $\Sigma= (\Sigma_0,\Sigma_1)$. In short, what changes with respect to Definition~\ref{def:freeMC} is that the inductive definition of $\Sigma$-term includes an additional clause for symmetries:
\begin{itemize}
	\item If $A$ and $B$ are in $\Sigma_0$, then the `symmetry' $\sigma_{[A],[B]} \colon [A,B] \to [B,A]$ is a $\Sigma$-term.
\end{itemize}
and the morphisms of the freely generated category are $\Sigma$-terms quotiented as in Definition~\ref{def:freeMC} plus a family of equations saying that arbitrary symmetries are constructed from symmetries only involving generating objects, analogously to what is shown in Example~\ref{ex:butterfly}.
\[
\star\star\star\star\star
\]

This concludes the presentation of the core language of string diagrams, for symmetric monoidal categories, on which much of the more elaborate graphical syntaxes rely. 
The trajectory we followed is rather different from the one the reader will find in most of the literature, with an emphasis on functorial boxes and non-strict tensoring. 
The hope is that this perspective has some pedagogical value, in itself and as a complement of the presentations found in the literature, in understanding the core graphical language of string diagrams. 
Our unusual emphasis on functorial boxes and non-strict tensoring was pedagogically motivated, but these constructs will turn out to be important in their own right. 

\subsection{Foliations}
\label{sec:foliations}
\newcommand{\Int}{\mathit{Int}}

As compared to the term notation the diagrammatic notation was seen to have several advantages. 
It is more abstract, as several terms correspond to the same diagram;
it is automatically quotiented by systematic renaming of variables ($\alpha$-equivalence);
it can more easily identify redexes, which allows the derivation of more efficient abstract machines through the fine-grained control of copying and sharing. 

In this section we shall see how the diagrammatic presentation can help derive a quasi-normal form for terms which leads to simpler inductive algorithms and proofs of correctness for certain classes of algorithms on graph-like data structures. 
To give concrete examples we will use a generic syntax for an OCaml-like functional language (call-by-value, algebraic data types, pattern matching). 
The techniques we illustrate here are part of the folklore of string diagram research but a deliberate and focussed, rather than incidental, presentation can be considered to be new. 

Let us begin with a common situation where a non-tree-like data structure is algorithmically justified, and examine the issues and the proposed approaches. 
In functional languages, tree-like data structures are ubiquitous since they can be expressed as an algebraic data type. 
Consider for example the type of binary trees with (integer) data stored in the nodes:

\lstset{language=Haskell, basicstyle=\small, keywordstyle=\bfseries, identifierstyle=\ttfamily}
\begin{lstlisting}
type tree = Node int tree tree | Empty
\end{lstlisting}

A standard function on trees would be \emph{map}, which applies a function $f$ to each node while preserving the overall tree shape:
\begin{lstlisting}
map f (Node n t t') = Node (f n) (map f t) (map f t')
map f Empty         = Empty 
\end{lstlisting}
The function is inductive (or \emph{structurally} recursive on the data structure) so that termination is guaranteed. 
For instance, this is the effect of mapping the function $\lambda x.x+1$ on some tree, represented graphically in the usual way:
\[
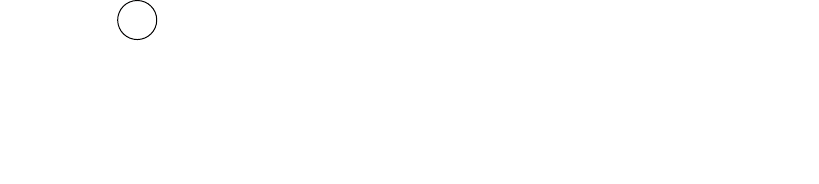
\]
Note that this particular tree has a certain amount of duplication. 
A common space-saving algorithmic device in this situation is to use directed acyclic graphs (DAGs) to represent sharing of common substructures. 
The tree above can have several optimised DAG representations, with the minimal one shown below:
\[
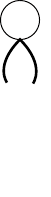
\]
Note that the DAG has exactly the same paths as the tree it efficiently represents. 

The problem we now face is that DAGs are not an inductive data structure. 
Efficient algorithms on DAGs can be written functionally, but they are rather sophisticated (attribute grammars and higher-order abstract syntax are two ways of doing it). 

Here we will present a concrete solution to this problem based on string diagrams, then we will give a more general framework for such algorithms. 
Consider a freely generated PROP with signature morphisms $node_k : 2\to 1 $ for $k\in\Int$ and $emp:0\to 1$. 
Both the tree and the DAG used in the examples below can be represented int this category as string diagrams (if we tilt the graphs to the right to fit our graphical conventions). 

To save on notation overload we introduce the following auxiliary notation:
\[
f^{(k)}=f\otimes\cdots\otimes f\qquad (k\text{ times})
\]
As terms, they are
\begin{align*}
\text{Tree}& : emp^{(8)} \semic node_3^{(4)} \semic node_2^{(2)} \semic node_1\\
\text{DAG}& : emp^{(2)} \semic node_3 \semic node_2 \semic node_1
\end{align*}
noting that these representations are not unique. 
\begin{exercise}
Find alternative representations for the tree and the DAG used in the example. 
\end{exercise}
The key observation here is that we wrote morphisms in a form (sequentially composed tensors of signature elements and basic structural morphism) that can always be represented as a list of lists:
\begin{multline*}
\text{tree:} \bigr[[emp \semic emp \semic emp \semic emp \semic emp \semic emp \semic emp \semic emp] \semic \\
[node_3 \semic node_3 \semic node_3 \semic node_3] \semic 
[node_2 \semic node_2] \semic 
[node_1]\bigl]
\end{multline*}
and
\[
\text{DAG:} \bigl[
[emp \semic emp] \semic [node_3] \semic [node_2] \semic [node_1]
\bigr]
\]
respectively.
The inner lists are aggregated by tensor ($\otimes$) and the outer list by composition ($ \semic $). 
We call this presentation of the morphism a \emph{foliation}. 

Of course, a morphism can be presented in a different form, which does not fit this shape, for instance this tree:
\[
(((emp\otimes emp) \semic node_1)\otimes((emp\otimes emp) \semic node_2) \semic node_3
\]
Its foliation is
\[
emp^{(4)} \semic (node_1\otimes node_2) \semic node_3.
\]
The advantage of this representation is that now we can write a general inductive map function for DAGs presented as foliations, remembering that the language of a freely generated PROP also includes identity ($id$) and symmetry ($sym$). 
However, since a map function preserves the shape of the DAG, it simply preserve these structural morphisms.
\begin{lstlisting}
map' f (emp    :: xs) = emp :: (map' f xs)
map' f (node k :: xs) = node (f k) :: (map' f xs)
map' f (id     :: xs) = id :: (map' f xs)
map' f (sym    :: xs) = sym :: (map' f xs)
map' f []             = []

map f xs :: xss = (map' f xs) :: (map f xss) 
map f []        = []
\end{lstlisting}
If \emph{List.map} is the usual map on lists, the \emph{DAG.map} above can be written more compactly as
\begin{lstlisting}
lift f (node k) = node (f k)
lift f x        = x

Dag.map f = List.map (List.map (lift f))
\end{lstlisting}
Or, with a function composition operator defined as 
\begin{lstlisting}
f ; g x = g(f x)
\end{lstlisting}
simply
\begin{lstlisting}
Dag.map = lift ; List.map ; List.map
\end{lstlisting}

Let us consider some other common functions on such binary graphs and DAGs. 
Map preserves the shape of a DAG, and functions that transform the shape in a uniform way are also easy to implement. 
\emph{Mirroring} a binary tree, as an inductive data type is:
\begin{lstlisting}
flip Node n t t' = Node n (flip t') (flip t)
flip Empty       = Empty
\end{lstlisting}

On DAG foliations, we simply reverse every list in the foliation:
\begin{lstlisting}
DAG.flip = List.reverse ; List.map
\end{lstlisting}
where \emph{List.reverse} is the usual list reversing function. 
Also functions that modify a DAG locally, by replacing a node with another DAG are easy to implement;  
they are left as an exercise. 

Functions that ignore the shape of the DAG are also trivial, as the DAG can be flattened into a list. 
The structural elements (identities and symmetries) can be filtered out if needed.
Finding an element $n$ in a DAG can simply done by searching for it in the flattened list of lists:
\begin{lstlisting}
Dag.find n = List.concat ; List.find (node n) 
\end{lstlisting}

\begin{exercise}
Find the minimum element of a DAG presented as a foliation. 
\end{exercise}

However, not all algorithms are naturally suited when a DAG is presented as a foliation. 
The following is difficult, but not impossible. 
\begin{exercise}[difficult]
Define a datatype of expressions (integers, addition, multiplication) as a tree and write a simple evaluator. 
Consider the types of expressions with shared sub-expressions as a DAG and write an evaluator on foliations.
\end{exercise}

Many algorithms which are standard on graph representations seem harder to write efficiently or naturally on foliations, for example traversals with a prescribed order (e.g. prefix, infix, postfix), searching for sub-graphs, or computing isomorphisms. 
Applying the operational semantics is a perfect example of algorithm apparently unsuitable for foliations, which would benefit from the greater flexibility of the graph representation. 
Note here that the unsuitability is not in terms of expressiveness but in terms of efficiency: the token can be represented as a node in the graph and the rewriting rules can be represented in a suitable way. 
However, the traversal of the identities may require several steps in the list-like notation, instead of the single step in the conventional graph notation. 
In general, any edge may be represented by a number of identities and symmetries in the list-like notation, which is likely to be inefficient in time and space. 

\begin{remark}
Simple inductive algorithms on foliations are not to be seen as practical from the point of view of computation. 
The best way to think of them is as simple inductive (and executable) specifications on DAGs which are easier to formalise and use in proofs than working with DAGs as a native data structure. 
This feature particularly comes in handy in defining and proving complex program transformations such as (reverse) automatic differentiation or closure conversion. 
However, the computational overhead of computing in terms of foliation rather than on a direct implementation of graph using pointers has not been studied. 
\end{remark}

\subsection{Properties of foliations}
The following theorem is category theory folklore:
\begin{theorem}[Foliation]
Any morphism $f$ in a freely generated (strict) monoidal category can be presented as
\[
f = \bigotimes_{i_0}f_{1,i_0} \semic \cdots \semic \bigotimes_{i_n}f_{1,i_n}
\]
where $f_{j,k}$ are elements of the signature or structural morphisms (identities, symmetries if the category is symmetric, etc.). 
\end{theorem}
This theorem is the corollary of a stricter version:
\begin{theorem}[Maximally sequential foliation]\label{thm:msf}
Any morphism $f$ in a freely generated (strict) monoidal category can be presented as
\[
f = id_{A_1}\otimes f_{1}\otimes id_{A_1'} \semic \cdots \semic id_{A_n}\otimes f_{n}\otimes id_{A_n'}
\]
where $f_k$ are elements of the signature or structural morphisms (identities, symmetries if the category is symmetric, etc.). 
\end{theorem}
\begin{proof} 
We sketch the proof below.

First we write $f$ as a string diagram, and we consider its concrete graph representation, which we topologically sort. 
We take the smallest node in the graph and call it $f_1$. 
Since this is a topological sort $f_1$ is not connected on the left to any other node but directly to the input interface. 
The interface itself is a list and we take the type of the edges below the input edges of $f_1$ to be $A_1$ and those above to be $A_1'$.
Note that edges do not cross (in braided or symmetric categories the braiding or the symmetry are non-identity morphisms). 
We have obtained the first term of the foliation, $id_{A_1}\otimes f_{1}\otimes id_{A_1'}$. 
Repeat the procedure for the smaller graph obtained by removing $f_1$ and its incoming edges, and adding its outgoing edges to the input interface.
Continue until the graph is empty. 
The result is the foliation. Proving that this foliation is equal to $f$ is a tedious but simple exercise. 
\end{proof}
\\[1.5ex]
Note that since the topological sort is not unique, the foliation is also not unique.

\subsection{Further reading and related work}

Among the surveys on string diagrammatic reasoning, we mention three: \cite{PiedeleuZanasi2023} is aimed at computer scientists, and present contemporary applications of string diagrams in several neighbouring fields; \cite{Selinger2011} offers a comprehensive catalogue of string diagrams for many different categorical structures, focussing on which geometric principles govern their equivalence; finally, \cite{baez2010physics} proposes a conceptual and high-level understanding of string diagrams in the context of physics, logic, topology and computer science.
The particular style of presentation we rely on, that of `functorial boxes' is due to \cite{DBLP:conf/csl/Mellies06}---although we have to point out that, somewhat curiously, this paper stops short of applying the technique to adjunctions, which is where we found it most useful. 
These papers together provide the best foundations for understanding the philosophy, motivation, and some of the technical details in this section. 
They were also an important source of inspiration for the first author in gaining a personal appreciation of the beauty and power of string diagrams. 
Also inspirational, although, as it will be seen in the next section, without an immediate technical connection, is \cite{coeckek}, which is a culmination of a series of papers starting with \cite{DBLP:conf/lics/AbramskyC04} which have firmly established diagrammatic methods as a serious and useful syntax for monoidal categories and their applications. 

The diagrammatic notation presented here is not the only kind of string diagrams out there. 
Unrelated, other than in spirit, are string diagram calculi using two \cite{marsden2014category} and higher dimensional structures \cite{DBLP:conf/lics/BarV17}. 
This should alert the reader to the fact that the term `string diagram' is rather broad and may cause some confusion. 

One important possible confusion worth explicating is the conflation of the term `string diagram' as a two-dimensional syntax with its graph representation, as a combinatorial mathematical object or concrete data structure. 
The distinction between the two is worth making, as the rest of this tutorial will hopefully show. 
The second author's research work on rewriting techniques used to formalise the equational theory of string diagrams represented as graphs insists on this distinction, as seen for example in \cite{DBLP:conf/lics/BonchiGKSZ16}. 

For monoidal categories we need to mention the original papers which introduce a graphical notation in the style of string diagrams \cite{JOYAL199155,joyal1995geometry}. 
Since these general monoidal categories are not presumed to be symmetric the notion of equivalence is not isomorphism but `homotopy', i.e. smooth deformations of the substrate in which the diagrams are embedded.
Symmetry is also considered in the original papers, with the interesting observation that isomorphism can be reduced to isotopy by increasing the number of dimensions of the substrate, i.e. the isomorphism of a 2-dimensional structure can be seen as isotopy in a 3-dimensional space. 

The original papers, and virtually the entire literature that follows, operate in the context of strict monoidal categories. 
The issue of non-strict categories is rather summarily dismissed: ``\emph{In principle, most results obtained with the hypothesis that a tensor category is strict can be reformulated and proved without this condition.}'' (p. 59, \emph{loc. cit.}).  
Mathematically this is perhaps true, but the details are clearly worth spelling out if we are to seriously consider string diagrams as a formal syntax. 
This careful specification of strictification for string diagrams is done in \cite{wilson}. 

\newpage

\section{Hierarchical String Diagrams and the $\lambda$ Calculus}
\label{sec:hsdcmc}

A \emph{symmetric closed monoidal category} (SCMC) is the essential mathematical structure for the interpretation of higher-order computation. 
It is the foundation on which models of $\lambda$ calculi are constructed. 
\begin{definition}[Symmetric closed monoidal category]
A symmetric monoidal category $\category C$ is \emph{closed} if for all objects $X$ the tensor product functor $F_X(A)=X\otimes A: \category  C\to\category C$ has a right adjoint functor $G_X(A)$, usually written as $X\multimap A$:
\[
X\otimes A \quad \dashv \quad X\multimap -. 
\]
\end{definition}
Objects $A\multimap B$ are called \emph{exponential objects} and are sometimes also written as $B^A$. 

We use the functorial boxes and adjunction notations to represent all the equations entailed by this definition, first noting that the functor $F_X$ can be represented explicitly:
\[
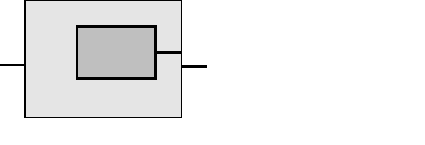
\]
The equations are shown in Figure~\ref{fig:adjclosed}: note that in order to express these equations rigorously, in particular the last one, the availability of explicit strictification and de-strictification comes in handy. 
The co-unit of the adjunction is usually called \emph{eval} (evaluation) so, for the sake of symmetry, we will call the unit \emph{co-eval}. 
\begin{figure} 
\begin{subfigure}[t]{\linewidth}
\[
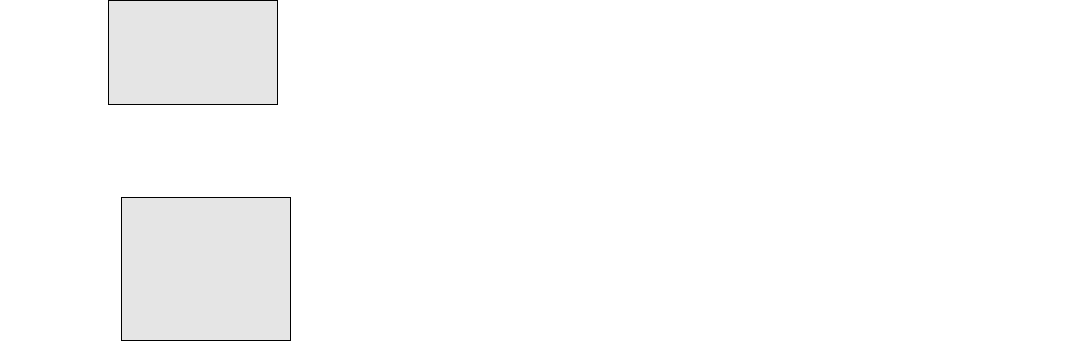
\]
\caption{Naturality of unit and co-unit}
\end{subfigure}

\begin{subfigure}[t]{\linewidth}
\[
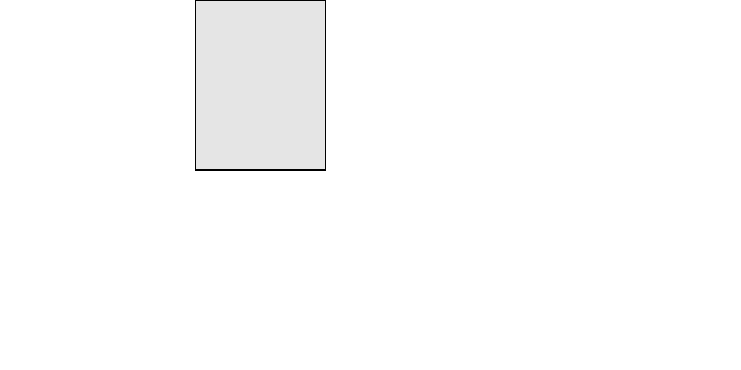
\]
\caption{Unit and co-unit equations}
\end{subfigure}

\caption{Symmetric closed monoidal adjunction equations}\label{fig:adjclosed}
\end{figure}
As we shall see both in models of $\lambda$ calculi and in the lemma below, a construction that proves to be useful is the following. 
\begin{definition}\label{def:lambdanotation}
For any morphism $f:X\otimes A\to Y$ we define its \emph{abstraction} as the morphism $\Lambda_X(f):A\to X\multimap Y$, denoted graphically as below:
\[
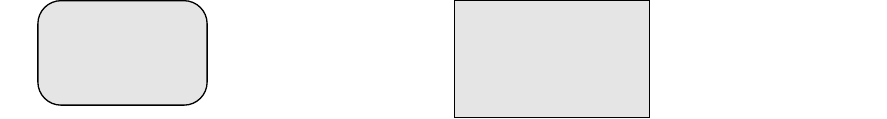
\]
\end{definition}
Note that the rounded box with a wire attached to the edge is merely a graphical convention, and the abstraction box only has a precise meaning as expanded according to its definition. 
Most importantly, it is not a functorial box. 

A useful property of abstraction is that it can `swallow' other abstractions it is sequenced with:
\begin{lemma}\label{lem:sbstlam0}
\[
f;\Lambda(g) = \Lambda\bigl((f\otimes id);g\bigr).
\]
\end{lemma}
\begin{proof}
\[
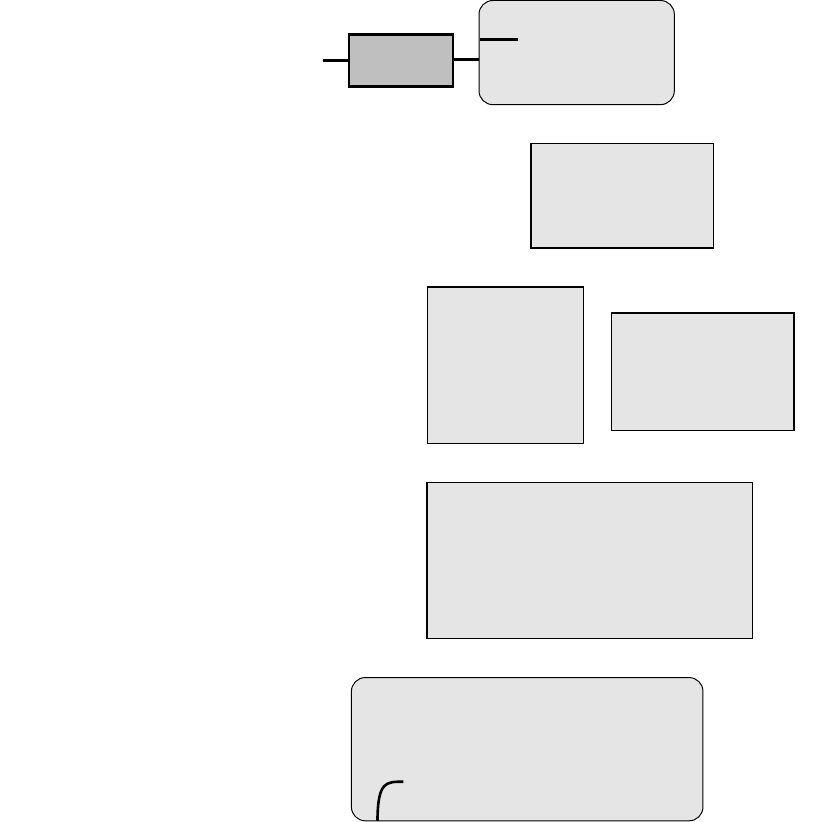
\]
\end{proof}

The following is the defining property of SCMCs.
\begin{lemma}
In a SCMC $\mathcal C$
\[ hom_{\category C}(X\otimes Y,Z)\simeq hom_{\category C}(X, Y\multimap Z)\]
 and 
 \[ X\otimes Y\multimap Z\simeq X\multimap (Y\multimap Z).\] 
\end{lemma}
\begin{proof}
The two isomorphisms are:
\[
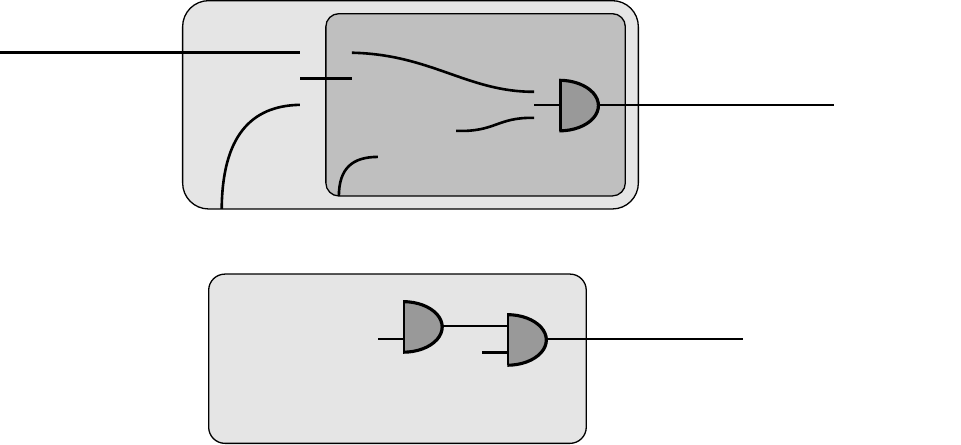
\]
The proof that the first morphism composed to the second is the identity is sketched in Figure~\ref{fig:curryproof}. 
The steps are:
\begin{enumerate}
\item the initial morphism composition,
\item\label{step2} apply Lemma~\ref{lem:sbstlam0},
\item elaborate the definition of abstraction,
\item naturality of unit,
\item naturality of unit, 
\item unit/co-unit cancellation, 
\item unit/co-unit cancellation,
\item symmetry cancellation,
\item strictification/de-strictification cancellation,
\item unit and co-unit cancellation.
\end{enumerate}
The diagrammatic redexes involved in a particular step are highlighted in red. 

The details of verifying all the other details that these are indeed natural isomorphisms are left as an exercise.
\end{proof}

\begin{figure}
\[
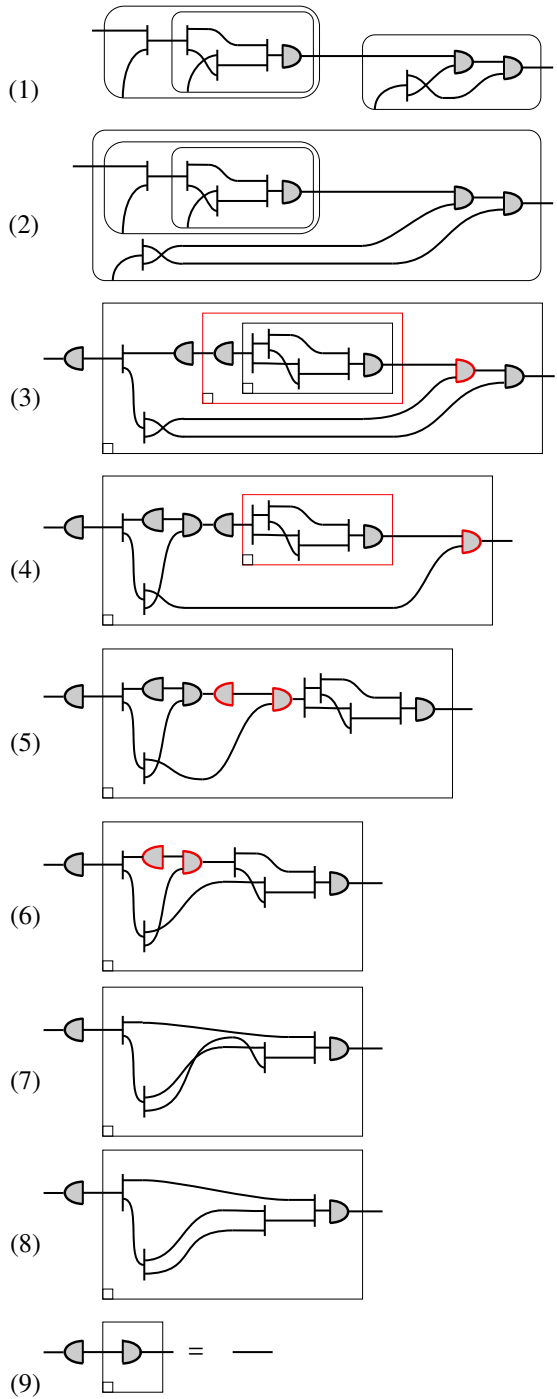
\]
\caption{Proof of identity of composition}
\label{fig:curryproof}
\end{figure}

{
The notion of a category freely generated from a signature, as seen for monoidal and symmetric monoidal strict monoidal categories in Section~\ref{sec:sdssmc}, can be adapted to the case of SCMCs. A difference is that, in order to properly express the conditions of closedness, we need a more fine grained representation for the objects: because of the presence of the operator $\multimap$, objects cannot be abstractly represented as lists of generators. 

\begin{definition}[Freely Generated Symmetric Monoidal Closed Category] \label{def:freeSCMC} A \emph{closed monoidal signature} is a pair $\Sigma = (\Sigma_0,\Sigma_1)$ where $\Sigma_0$ is a set of generating objects and $\Sigma_1$ a set of generating morphisms with sources and targets elements of the set $\mathit{obj}_{\Sigma_0}$, defined inductively as follows:
\begin{itemize}
	\item A designated object $I$ (the unit) and each object in $\Sigma_0$ is in $\mathit{obj}_{\Sigma_0}$.
	\item If $A_1$, $A_2$ are in $\mathit{obj}_{\Sigma_0}$, then $A_1 \otimes A_2$ and $A_1 \multimap A_2$ are in  $\mathit{obj}_{\Sigma_0}$.
\end{itemize}
The closed monoidal $\Sigma$-terms are defined inductively as follows:
\begin{itemize}
\item All morphisms $f \colon A_1 \to A_2$ in $\Sigma_1$ are $\Sigma$-terms.
\item If $A$ and $B$ are in $\mathit{obj}_{\Sigma_0}$, then the identity $\id_{A} \colon A \to A$ and symmetry $\sigma_{A,B} \colon {A \otimes B} \to B \otimes A$ are $\Sigma$-terms. 
\item If $f \colon A_1 \to A_2$, $g \colon A_2 \to A_3$ are $\Sigma$-terms, then $f ; g \colon A_1 \to A_3$ is a $\Sigma$-term. 
\item If $f \colon A_1 \to A_2$, $g \colon A_3 \to A_4$ are $\Sigma$-terms, then $f \otimes g \colon A_1 \otimes A_3 \to A_2 \otimes A_4$ is a $\Sigma$-term. 
\item If $X$ and $A$ are in $\mathit{obj}_{\Sigma_0}$, then the `evaluation' map $eval_{X,A} \colon ((X \multimap A) \otimes X) \to A$ is a $\Sigma$-term, represented graphically in the same way as the unit of the adjunction.
\item If $h \colon X \otimes A \to Y$ is a $\Sigma$-term, then its `abstraction' $\Lambda_X(h) \colon A \to (X \multimap Y)$ is a $\Sigma$-term, represented graphically as in Definition \ref{def:lambdanotation}. 
\end{itemize}
The closed monoidal category $\category C$ freely generated by $\Sigma$ is defined as having objects $\mathit{obj}(\category C) := \mathit{obj}_{\Sigma_0}$ and morphisms  $hom(\category C)$ the $\Sigma$-terms quotiented by:
\begin{itemize}
	\item the equations of symmetric monoidal categories, where $\otimes$ acts as the monoidal product, and $I$ acts as the identity for $\otimes$;
	\item the following three equations, describing the behaviour of evaluation and abstraction (where we adopt the notation of Definition~\ref{def:lambdanotation} for $\Lambda_X(h)$):
\[
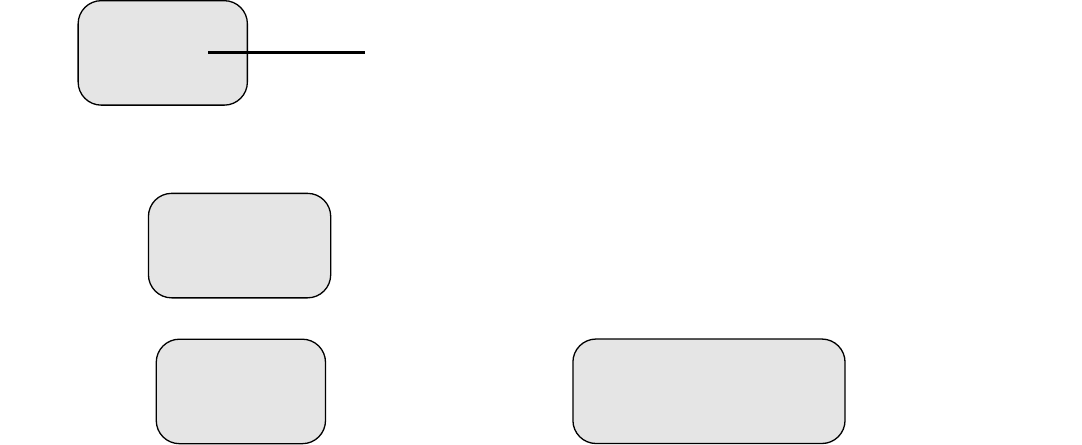
\]
\end{itemize}
\end{definition} 
Observe that the last equation is the same equation proven in Lemma~\ref{lem:sbstlam0}.

Considering SCMCs were introduced as SMCs with an adjunction, it is not completely obvious why the category $\category C$ obtained in Definition~\ref{def:freeSCMC} would meet the requirements of being monoidal closed. It is a pleasant exercise in diagrammatic reasoning to verify this fact. 

\begin{lemma}
	Let $\category C$ be freely generated from a signature $\Sigma = (\Sigma_0,\Sigma_1)$ as in Definition~\ref{def:freeSCMC}. Then $\category C$ is a symmetric monoidal closed category.
\end{lemma}
\begin{proof}
First, we define the right adjoint functor $G_X(A)$ to the tensor product functor $F_X(A) = X \otimes A \colon C \to C$, as follows. On objects, we set $G_X(A)$ to be $X \multimap A$. On a morphism $f \colon A \to B$ we define it as
\begin{equation}\label{def:GXadj}
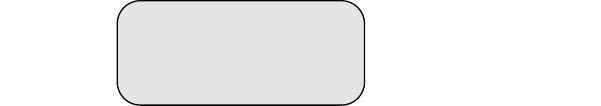
\end{equation}
We claim $G_X$ is right adjoint to $F_X$. To show this, we define the unit $\eta_B \colon ((X \multimap B) \otimes X) \to B$ of the adjunction as $eval_{X,B}$, and the counit $\epsilon_B \colon B \to (X \multimap (X \otimes B))$ as
\begin{equation}\label{def:counitAdj}
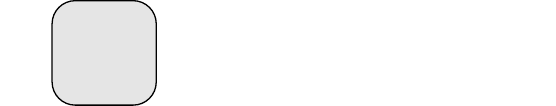
\end{equation}
It now remains to verify the four equations of Figure~\ref{fig:adjclosed}. In the derivations below, remember that the label $X \multimap$ in the functorial boxes appearing in the equations of Figure~\ref{fig:adjclosed} stands for the right adjoint $G_X$, and thus those boxes are defined according to~\eqref{def:GXadj}. The diagrammatic reasoning is in Figure~\ref{fig:psmc}.
\begin{figure}
\[
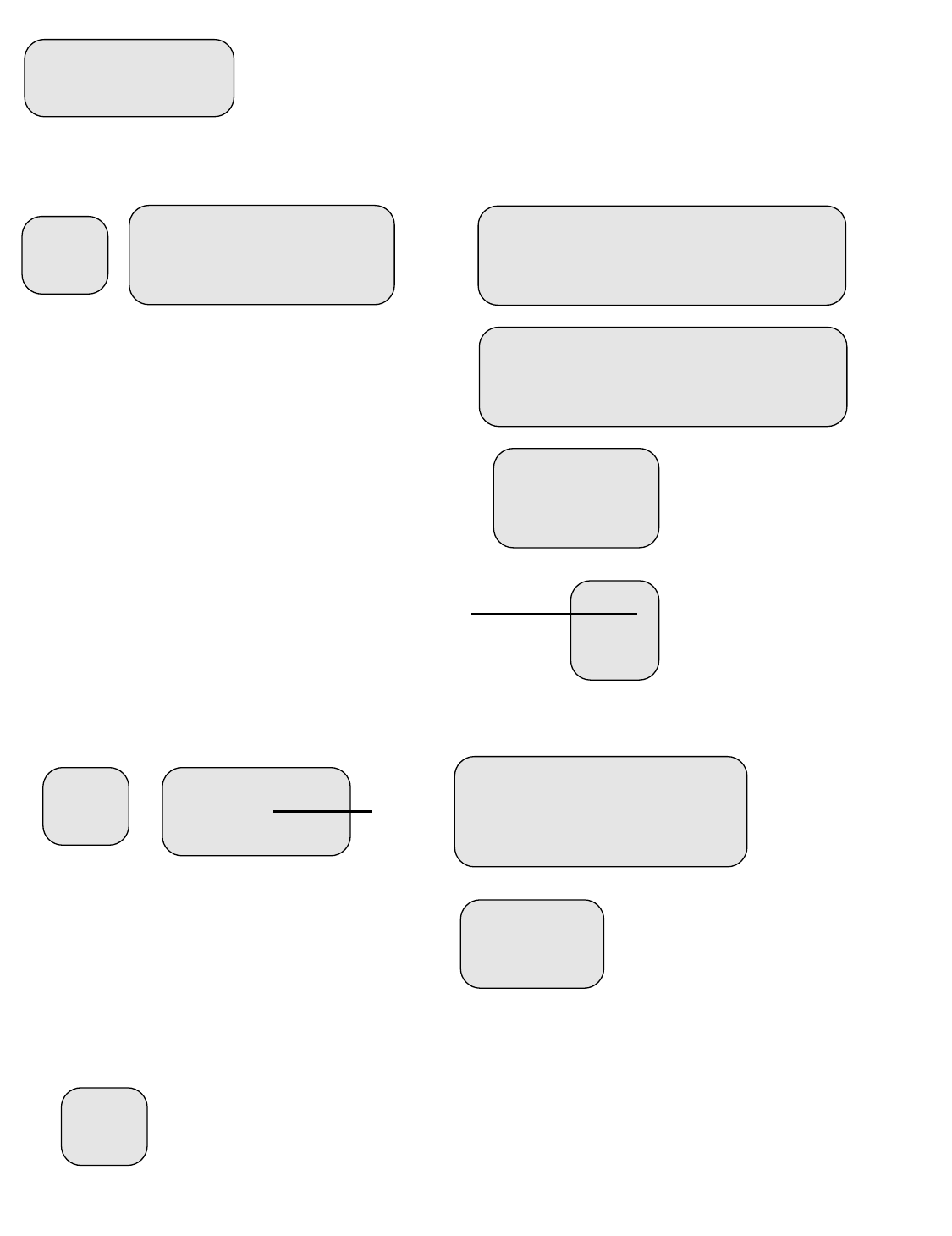
\]
\caption{Proof of monoidal closure}
\label{fig:psmc}
\end{figure}
\end{proof}

We may also show the converse implication, namely that any SCMC exhibits the structures defined by the construction of Definition~\ref{def:freeSCMC}. The reader may verify it via the following exercise.

\begin{exercise}
	Let $\category C$ be an SCMC. Consider objects $X, A, Y$ in $\category C$, and a morphism $h \colon X \otimes A \to Y$. Define an evaluation map $eval_{X,A} \colon ((X \multimap A ) \otimes X) \to A$ and abstraction $\Lambda_X(h) \colon A \to (X \multimap Y)$ satisfying the equations in Definition~\ref{def:freeSCMC}.
\end{exercise}

\begin{remark}
In contrast with the case of monoidal and symmetric monoidal categories, we constructed the free SCMC as a \emph{non-strict} monoidal category. Intuitively, this is due to the fact that, in order to properly define the monoidal closed structure, we need wires in string diagrams of a SCMC to be labeled with arbitrary objects, rather than just the generating ones. This is particularly evident both in Definition~\ref{def:lambdanotation}, where we need to be able to abstract over \emph{any} object $X$ that may be tensored with $A$, and in the definition~\eqref{def:counitAdj} of the co-unit. Because we need both wires for, say, $X$ and $A$, and for $X \otimes A$, we also need (de)strictifiers to switch from one representation to the other, as it happens for instance in~\eqref{def:counitAdj}. It is thus essential to not impose strictness in the construction of Definition~\ref{def:freeSCMC}. This is also coherent with our motivating examples of SCMCs (to be explored in Section~\ref{sec:lambda} below) which are usually non-strict.
\end{remark}
}

\medskip

We conclude this section by pointing out that also the concept of foliation, introduced in Section~\ref{sec:foliations} for  monoidal categories, can be extended to a monoidal closed category, in the obvious way:
\begin{theorem}[Hierarchical foliation]
Any morphism $f$ in a freely generated monoidal closed category can be presented as
\[
f = \bigotimes_{i_0}f_{1,i_0} \semic \cdots \semic \bigotimes_{i_n}f_{1,i_n}
\]
where $f_{j,k}$ are 
\begin{itemize}
\item elements of the signature 
\item structural morphisms (identities, symmetries, etc.)
\item $\Lambda(f')$ for some morphism $f'$ which is also a hierarchical foliation. 
\end{itemize}
\end{theorem}
The proof is similar to that of Theorem~\ref{thm:msf} except that the graph-label of each hierarchical box is foliated beginning with the deepest one. 
This is sound since the depth of nesting of $\Lambda()$ is finite. 

Note that the proof of the foliation theorem(s) also works as a proof of definability, showing that any well-formed hierarchical graph can be converted into a term. 
Starting with an arbitrary morphism in a (closed) monoidal category, therefore we can represent it as a string diagram, the retrieve an equal foliated morphism. 
So foliations, in particular maximally sequential ones, can be seen as a quasi-canonical form of morphisms in such categories, and generating them a form of \emph{normalisation by evaluation}.

We remark that, in obtaining the foliated form, the string diagram notation and its graph concrete representation are essential. 
An algorithm for deriving a foliation out of a morphism using terms seems extremely difficult to formulate directly. 

% THE NOTION OF PATH IS ESSENTIAL IN SOME ALGOS, E.G. REDEXES

\subsection{Cartesian product and Cartesian closed categories}\label{def:cartesian}
We introduce a family of natural transformations called \emph{copying} and \emph{deletion} respectively, written as $\delta_A : A \to A \otimes A $ and $\omega_A : A \to I $, respectively. 
Because deletion and copying play such an important role in a Cartesian category we employ graphical syntax for them, with the caveat that this graphical syntax is merely meant to prevent clutter but, unlike e.g. symmetry, does not absorb any equations. 
The naturality equations are graphically rendered as:
\[
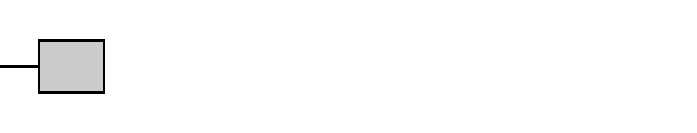
\]
\begin{remark}
Cartesian product is not strict but to avoid clutter we present the diagrammatic equations as if it is. 
This imprecision can be easily corrected using explicit strictification and de-strictification as seen in Section~\ref{sec:strictification}.
\end{remark}

\begin{definition}[Cartesian product]
A symmetric monoidal tensor is a (Cartesian) product if for each object $A$ in the category the following monoidal natural transformations exist:
\begin{align*}
\omega_A &: A \to I  & \text{(deletion)}\\
\delta_A &: A \to A \otimes A & \text{(copying)}
\end{align*}
such that 
\begin{align*}
\omega_{A\otimes B} &= \omega_A\otimes\omega_B \\
\delta_{A\otimes B} &= \delta_A\otimes\delta_B  \semic  
\id_A\sigma_{A,B}\otimes \id_B
\end{align*}
and
\[
\delta \semic \id\otimes \omega = \delta \semic \omega\otimes\id = \id
\]
\end{definition}
Graphically, this is
\[
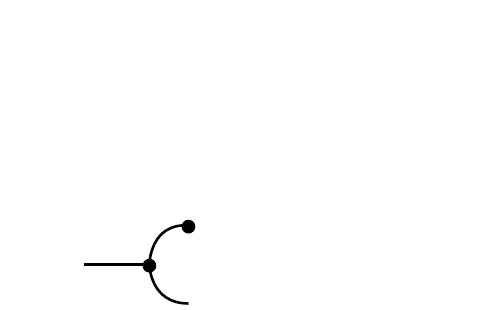
\]
If a symmetric monoidal category is Cartesian it is customary to write the tensor as $A\times B$. 
If it is also closed it is customary to write the exponentiation as $A\Rightarrow B$;  the composite morphisms $\pi_1:=\id\otimes\omega$ and $\pi_2:=\omega\otimes \id$ are called the first and second projections, respectively. 
Such a category is said to be a \emph{Cartesian closed category} (CCC)
Finally, it is usual to introduce notation for 
\[
A\xrightarrow{\langle f,g\rangle}B\times C := A\xrightarrow{\delta_A}A\times A\xrightarrow{f\otimes g} B\times C, 
\]
called the \emph{pairing} of $f$ and $g$. 

\paragraph{Further syntactic conventions}

The principle of absorbing equations into the notation can be pushed farther. 
It is convenient to generalise the notation for binary to $n$-ary copying: 
\[
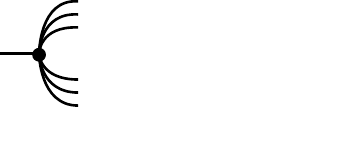
\]
It can be easily shown that because of the laws of the Cartesian product any tree constructed of copying, deletion, and symmetry can be given a normal form, e.g.
\[
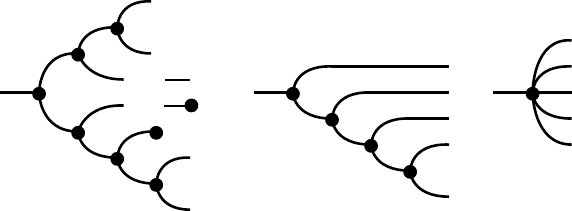
\]
\begin{remark}
A more formal construction of $n$-ary copying, in the style of Exercise~\ref{ex:butterfly} is left as an exercise to the reader. 
\end{remark}
In general any term in a monoid or, as in the case of copying a co-monoid, can be reduced to this dot-like notation. 
As a more advance graphical device, which we will not employ in the sequel, the presence or absence of commutativity (or co-commutativity, respectively) can be made graphically perspicuous by making it impossible to discern the order of incoming (or out-going) wires. 
Thus, an $n$-ary addition operation ($+$), which is commutative, would be represented differently from an $n$-ary list constructor ($::$), which is not, as seen below (for $n=5$): 
\[
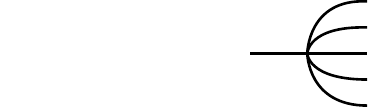
\]

\paragraph{Universal products}

The formulation of a Cartesian category in Definition~\ref{def:cartesian} is not the usual one found in category theory literature. 
It is more common to define it using a \emph{universal property}:
\begin{definition}[Categorical product (universal)]\label{def:universal}
Given a category $\category C$ and two objects $A_1,A_2$, then their product $A_1\times A_2$ is an object equipped with morphisms (projections) $p_i:A_1\times A_2\to A_i$ for $i=1,2$ such that it is universal with this property, i.e. given any other object $X$ with morphisms $f_i:X\to A_i$ for $i=1,2$ there exists a unique morphism (pairing) $\langle f_1,f_2\rangle: X\to A_1\times A_2$ such that $f_i=\langle f_1, f_2\rangle \semic p_i$. 
\end{definition}
This definition generalises to arbitrary sets $i\in I$. 

We can show that Definition~\ref{def:cartesian} implies that Definition~\ref{def:universal} also holds. 
The projections and the pairing are:
\[
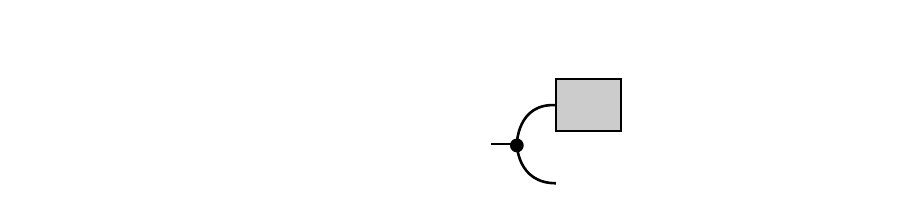
\]
The equation governing the interaction between the first projection and pairing is
\[
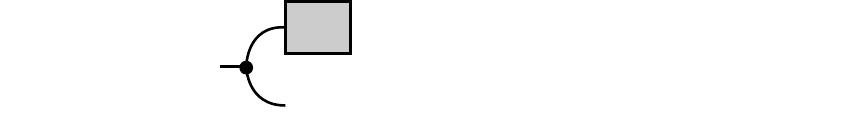
\]
The equation for the second projection is the obvious counterpart. 

The proof of uniqueness is slightly more interesting. 
Suppose there is a map $u$ satisfying the same equational properties as pairing, i.e.
\[
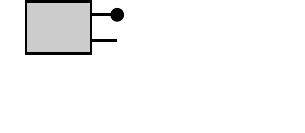
\]
The equation $u=\langle f_1,f_2\rangle$ follows from copy-delete interaction, naturality of copy, and the hypothesis:
\[
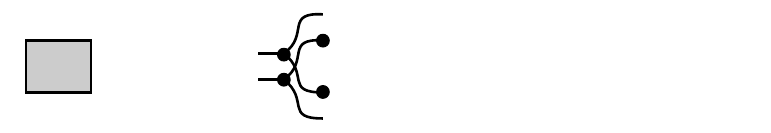
\]

\paragraph{Product as adjunction}

Mathematical concepts can be defined in many ways, something which is particularly true of category theory. 
So let us define product in a third way, this time using adjunctions!

First consider the category $\mathcal C\times\mathcal C$ which is constructed so that its objects are pairs of objects of $\mathcal C$ and its morphisms are pairs of morphisms of $\mathcal C$, which we shall write as $(A,B)$ and $(f,g)$, respectively.

Note that this is a new category, but one which we can relate to $\mathcal C$ using an obvious \emph{diagonal} functor that pairs objects, and respectively morphisms with themselves, $\Delta:\mathcal C\rightarrow \mathcal C\times\mathcal C$. 
We can show that the functor ${-}\times{-} : \mathcal C\times\mathcal C\rightarrow \mathcal C$ is adjoint to the diagonal, $\Delta\dashv{-}\times{-}$.
To do that, we define 
\[
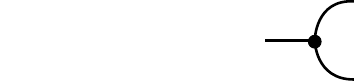
\]
and 
\[
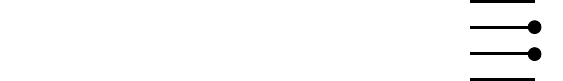
\] 
The details of checking that the above define an adjunction are routine and are left as an exercise to the reader. 

\subsection{The simply-typed $\lambda$ calculus}\label{sec:lambda}

\subsubsection{Term syntax}

The simply-typed $\lambda$ calculus (STLC) is at the core of the design of the functional fragment of most modern programming languages. 
The types $T$ are constructed from a fixed set of base types, for instance natural numbers ($N$), and exponentiation $T\to T$. 
Terms $u$ are constructed from variables $x$,  abstraction $\lambda x{:}T.u$, and application $u\, u$. 
This is the \emph{pure} fragment of the calculus, that can be extended with constants $c$, resulting in the \emph{applied} STLC. 
For now we focus on the pure fragment. 
If the type is unimportant we may omit it from a term and just write $\lambda x.u$. 

An important syntactic operation is that of substituting a term $v$ for a variable $x$ in some other term $u$, written as $u[x/v]$. 
For instance
\[
(\lambda x.f\, x)[f/\lambda y.y] = \lambda x.(\lambda y.y)\, x.
\]
This operation is not mere textual replacement, and must be defined carefully in order to avoid \emph{variable capture}: 
A $\lambda x.u$ term defines $x$ locally, in a definite scope $u$;  $x$ is said to be \emph{bound}.
Otherwise, a variable is said to be \emph{free}, such as $f$ in the example above. 
Moreover, repeated definitions are allowed, so variables can \emph{overshadow} other variables. 
For instance, the term $\lambda x.(\lambda \underline x.\underline x)\, x$, shows four occurrences of the string $x$, functioning sometimes as a binder and sometimes as a bound variable. 
The correct matchups between binders and bound variables are identified using underlining. 
\begin{definition}[Free variables]
The set of \emph{free variables} $\mathcal F(-)$ of a $\lambda$ term is defined inductively on its syntax as
\begin{align*}
\mathcal F(x) = \{ x\}
\quad
\mathcal F(u\,u') = \mathcal F(u)\cup\mathcal F(u')
\quad
\mathcal F(\lambda x{:}T.u) = \mathcal F(u)\setminus\{x\}.
\end{align*}
\end{definition}
To see why substitution is not merely textual replacement, consider the following invalid substitution:  
\[
(\lambda x.f\, x)[f/x] \not= \lambda x.x\, x. 
\]
The status of the occurrence of the substituted variable changes from free $f$ to captured $x$ in the resulting term. 
The operation of substitution is, however, always well defined because bound variables can be systematically renamed in any $\lambda$ term so, the correct application of substitution is:
\[
(\lambda x.f\, x)[f/x] = (\lambda y.f\, y)[f/x] = \lambda y.x\, y. 
\]
The first equation is the renaming of $x$ to $y$, the next is the actual performing of the substitution. 
To get there we define first \emph{transposition} of variables, then $\alpha$-equivalence of terms up to renaming of bound variables. 

The transposition of variables $x,y$ in $u$, written as $(x\,y)\cdot u$, is defined inductively on terms taken as text, swapping all occurrences ignoring the bindings. 
\begin{align*}
(x\,y)\cdot x &= y\\
(x\,y)\cdot y &= x\\
(x\,y)\cdot z &= z &\text{if } x\neq z\neq y\\
(x\,y)\cdot (u\,u') &= \bigl((x\,y)\cdot u\bigr)\,\bigl((x\,y)\cdot u'\bigr)\\
(x\,y)\cdot (\lambda z.u) &= \lambda (x\,y)\cdot z. (x\,y)\cdot u.
\end{align*}
Using it, we can define
\begin{definition}[$\alpha$-equivalence]
The relation of $\alpha$-equivalence between $\lambda$ terms, written as $u\equiv u'$, is defined as
\[
\frac{}{x\equiv x} \qquad
\frac{u_1\equiv u_1' \qquad u_2\equiv u_2'}{u_1\,u_2\equiv u_1'\,u_2'} \quad
\frac{(z\, x)\cdot u\equiv (z\, y)\cdot u' \quad x\neq z\neq y}{\lambda x{:}T.u\equiv\lambda y{:}T.u'}
\]
\end{definition}
%This is not the usual definition of $\alpha$-equivalence but it strikes a good balance of rigour and readability. 
From now on we consider $\lambda$ terms quotiented by $\equiv$, so e.g. $\lambda x.x=\lambda y.y$. 
We can now define (capture-avoiding) substitution, inductively on the syntax of $\lambda$ terms up to $\alpha$ equivalence:
\begin{align*}
x[x/u] & = u \\
y[x/u] & = y & x\neq y\\ 
(u_1\,u_2)[x/u] &= (u_1[x/u])\,(u_2[x/u])\\
(\lambda y{:}T.u')[x/u] &= \lambda y{:}T.u'[x/u] &  x\neq y \text{ and } y\not\in \mathcal F(u).
\end{align*}
\begin{remark}
On a superficial glance the definition seems ill-defined. 
What if $x\neq y$ and $y\in\mathcal F(u)$? 
What if $x=y$?
Are there cases missing from the definition of substitution?
But this is exactly the situation of the variable capture discussed earlier!
Using the mechanics of $\alpha$-equivalence $y$ can be simply replaced by a suitably named variable not bound in $u$, just as we replaced $x$ in $(\lambda x.f\, x)[f/x]$. 
\begin{exercise}
Prove that $(\lambda x{:}T.u)[x/u'] = \lambda x{:}T.u$.
\end{exercise}
\end{remark}
Types are assigned to terms using judgements of the form
\[
\Gamma \vdash u:T
\]
where $\Gamma = x_1:T_1, \ldots, x_k:T_k$ is a \emph{variable type assignment}.
The judgement above is interpreted as \emph{term $u$ has type $T$ if each $x_i$ has type $T_i$ as given by $\Gamma$}. 
The rules for deriving type judgements can be given in the style of natural deduction:

\begin{gather}
\Gamma_1, x:T,\Gamma_2\vdash x:T 
 \tag{var}\\[1.5ex]
\frac
{\Gamma_1,x:T,\Gamma_2\vdash u:T' }
{\Gamma_1,\Gamma_2\vdash \lambda x{:}T.u:T\to T'}
\tag{abs}\\[1.5ex]
\frac
{\Gamma\vdash u_1:T_1 \quad \Gamma\vdash u_2:T_1\to T_2}
{\Gamma\vdash u_2\,u_1:T_2}  
\tag{app}
\end{gather}
%The following rules are called \emph{structural} and have to do with general properties of variables in the type assignments. 
%For a fuller formalisation it is useful to spell them out. 
%\begin{gather}
%\frac
%{\Gamma,x:T,x':T'\vdash u:T}
%{\Gamma,x':T',x:T\vdash u:T}
%\tag{symmetry} \\[1.5ex]
%\frac
%{\Gamma,x:T,x':T\vdash u:T'}
%{\Gamma,x:T\vdash u[x/x']:T'}
%\tag{contraction}\\[1.5ex]
%\frac
%{\Gamma,x:T\vdash u:T''\quad x':T'\not\in\Gamma}
%{\Gamma,x:T,x':T'\vdash u:T''}
%\tag{weakening}
%\end{gather}

\subsection{Categorical and string diagram syntax}

\newcommand{\seval}[1]{\llbracket #1 \rrbracket}
\newcommand{\conc}{+\!\!+}
\newcommand{\lst}[1]{\overrightarrow{ #1 }}

The STLC can be given a canonical interpretation in any Cartesian closed category $\category C$.
Types are interpreted as objects of the category ($\seval N=N$, $\seval{T\to T'}=\seval T\Rightarrow\seval{T'}$). 
To reduce clutter we shall just write $T$ instead of $\seval T$, without ambiguity. 
%Variable type assignments $\Gamma$ are more convenient to be given an interpretation in the strictified category $\overline{\category C}$. 
%Let us write $\lst T := T_1,\ldots,T_n$ and $\lst {x:T} := x_1:T_1,\ldots,x_n:T_n$. 
\begin{align*}
\seval{\Gamma}=\seval{\lst{x:T}}=\seval{x_1:T_1,\ldots,x_n:T_n} = T_1\otimes\cdots\otimes T_n=\bigotimes_nT_i.
\end{align*}
To avoid clutter we also write just $\Gamma$ instead of $\seval \Gamma$. 
Interpretation is given to type judgements, inductively on its derivation, with types interpreted as objects and judgements as morphisms:
\[
\seval{\Gamma\vdash u:T'}:\Gamma\to T'. 
\]
The definitions are usually given as: 
\begin{align*}
\seval{x:T,\Gamma\vdash x:T} &= \id_T\otimes\omega_{\Gamma}\\
\seval{x:T,\Gamma\vdash \lambda x{:}T.u:T'} &= \Lambda_T\bigl(\seval{x:T,\Gamma\vdash u:T'}\bigr)\\
\seval{\Gamma\vdash u_2\,u_1: T_2} &= \langle\seval{\Gamma\vdash u_1:T_1}, \seval{\Gamma\vdash u_2:T_1\to T_2}\rangle \semic eval_{T_1,T_2}.
\end{align*}
\begin{remark}
These definitions, which can be encountered with various small variations in the literature, treat the product used in the interpretation of $\Gamma$ as if it is strict, by quietly re-associating (and even applying symmetry) its components. 
This is necessary since variables are not necessarily in the first position in the environment. 
This is not likely to lead to serious confusions, due to the strictification theorem.
Not unless, that is, there is a product type in the language as well. 
But it is possible to make the definitions formal if working in the strictified category.  
Of course, it is possible to make the definitions as formal in the non-strict category as well, but that is an even more tedious exercise, which we shall leave to the reader. 
\end{remark}
Let us change the definition of the interpretation of $\Gamma$ to a strict tensor
\[
\seval{\Gamma}=\seval{x_1:T_1,\ldots x_n:T_n}=[T_1,\ldots,T_n]=[\Gamma]
\]
and change the definition of interpretation accordingly to
\[
\seval{\Gamma\vdash u:T'}:[\Gamma]\to T'.
\]
We write $[\Gamma]$ to avoid ambiguity with the non-strict interpretation, written just~$\Gamma$. 

It is convenient to define
\begin{align*}
\psi_{\Gamma} : [\Gamma] \to \Gamma\qquad
\psi^*_{\Gamma} : \Gamma \to [\Gamma].
\end{align*}
As the strictifiers and de-strictifiers of $\Gamma$, respectively.

The more precise definitions can be now stated. 
\begin{align*}
\seval{\Gamma_1,x:T,\Gamma_2\vdash x:T} = [\Gamma_1,T,\Gamma_2]& \xrightarrow{\psi_{[\Gamma_1,T],[\Gamma_2]}}[\Gamma_1,T]\otimes[\Gamma_2]\\
&\xrightarrow{\psi_{[\Gamma_1],T}\otimes\id_{[\Gamma_2]}}[\Gamma_1]\otimes (T\otimes[\Gamma_2])\\
&\xrightarrow{\alpha_{[\Gamma_2],T,[\Gamma_2]}}[\Gamma_1]\otimes T\otimes[\Gamma_2]\\
&\xrightarrow{\omega_{[\Gamma_1]}\otimes id_T\otimes\omega_{[\Gamma_2]}} T \\[1.5ex]
\seval{\Gamma_1,x:T,\Gamma_2\vdash \lambda x{:}T.u:T'} &= [\Gamma_1,\Gamma_2]\xrightarrow{\Lambda_T(f)}T\Rightarrow T'\\
&\text{where}\\
f = T\otimes[\Gamma_1]\otimes[\Gamma_2] & \xrightarrow{\sigma_{T,[\Gamma_1]}\otimes \id{[\Gamma_2]}}[\Gamma_1]\otimes T\otimes[\Gamma_2]\\
&\xrightarrow{\alpha^{-1}_{[\Gamma_2],T,[\Gamma_2]}}[\Gamma_1]\otimes (T\otimes[\Gamma_2])\\
&\xrightarrow{\psi^*_{[\Gamma_1],T}\otimes\id_{[\Gamma_2]}}[\Gamma_1,T]\otimes[\Gamma_2]\\
& \xrightarrow{\psi^*_{[\Gamma_1,T],[\Gamma_2]}} [\Gamma_1,T,\Gamma_2] \\
& \xrightarrow{\seval{\Gamma_1,x:T,\Gamma_2\vdash u:T'}}T'\\[1.5ex]
\seval{\Gamma\vdash u_2\,u_1: T_2} = [\Gamma]&\xrightarrow{\langle\seval{\Gamma\vdash u_1:T_1}, \seval{\Gamma\vdash u_2:T_1\to T_2}\rangle}T_1\otimes (T_1\Rightarrow T_2)\\*
&\xrightarrow{\mathit{eval}_{T_1, T_2}} T_2.
\end{align*}
\begin{remark}
It is understandable why most presentation eschew the tedium of the structural manipulation of the type-assignment $\Gamma$, preferring to trade informality for readability. 
With string diagrams, as we shall see, there is no price to pay for full rigour, as the bureaucracy of structural manipulation is for the most part absorbed into the notation. 
The diagrammatic interpretation corresponding to the terms above is:
\[
%the pdf exported by inkscape has wrong margin because the text is framed badly:
\hspace{-15ex}
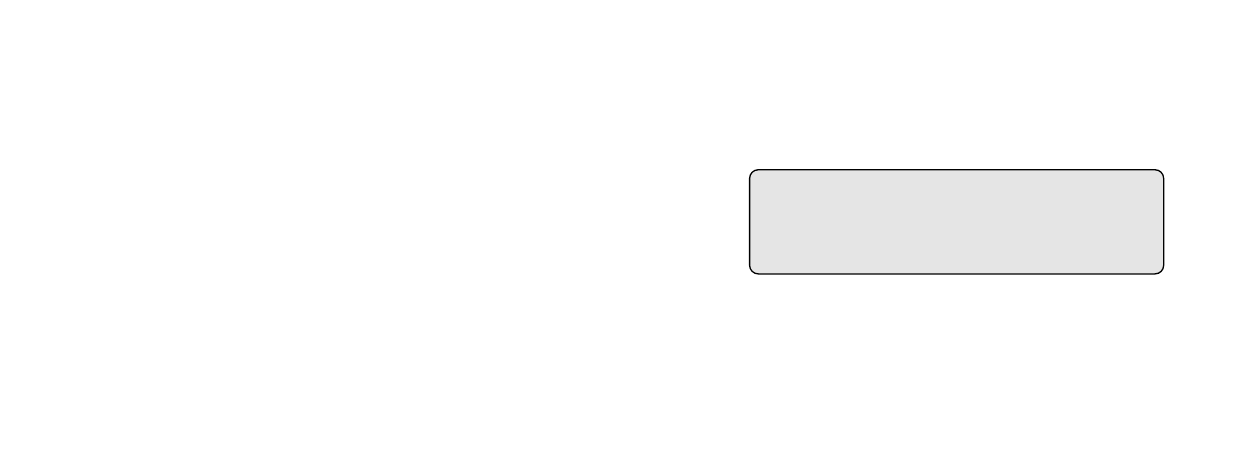\\
\]
The full calculation resulting in the simplified interpretations above is routine. 
For instance, the full unsimplified diagram with all the strictifiers for the interpretation of abstraction is given in Figure~\ref{fig:fullabs}.
\begin{figure}
\resizebox{\textwidth}{!}{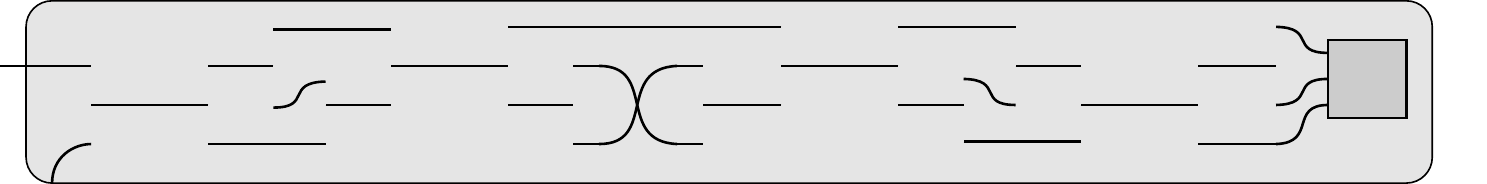}
\caption{Diagram of abstraction before simplifying strictification operators}\label{fig:fullabs}
\end{figure}

By a slight, although unambiguous, overloading of the notation we will combine strictification and symmetry in a single box, to write abstraction as: 
\[
\hspace{-25ex}
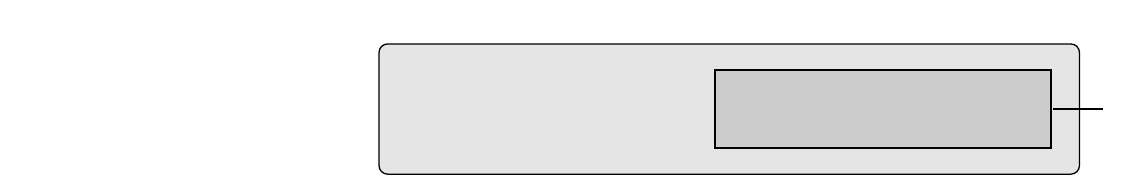
\]
The object annotations on the composite box indicate that it can be uniquely reconstructed as
\[
\hspace{-15ex}
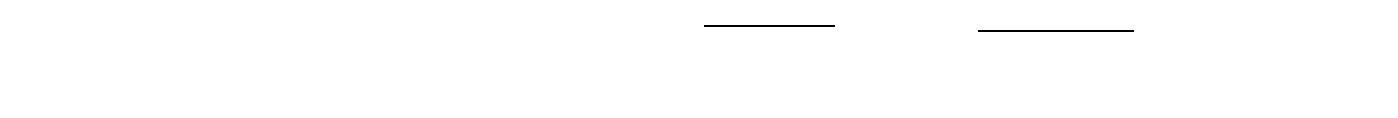
\]
\end{remark}

\begin{example}\label{ex:idid}
Consider the identity applied to itself $(\lambda x.x)(\lambda y.y)$. 
Its string diagram representation is:
$
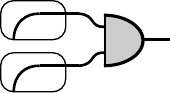
$
\end{example}
\begin{example}
Consider the term  
\[
\lambda n{:}((o\to o)\to o \to o).\lambda f{:}o\to o.\lambda x:o .f (\,n f\, x)
\]
The diagram constructed from the earlier definitions is:
\[
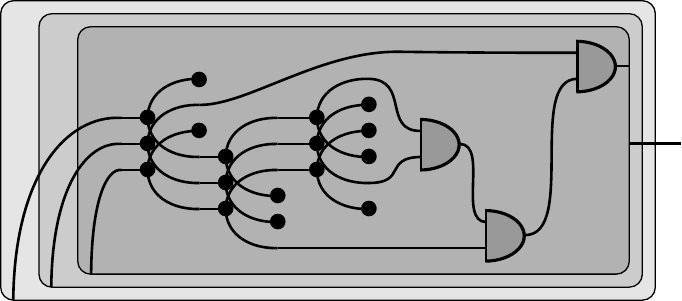
\] 
The copying-deletion pairs introduced by the definition of application and variable, respectively can be simplified, resulting in:
\[
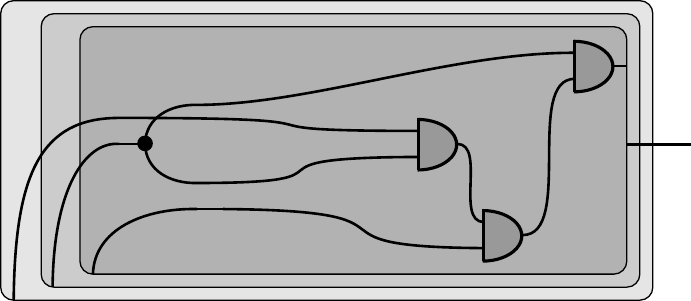
\] 
\end{example}

\subsubsection{$\lambda$ calculus with product types}
Extending the calculus with products is a simple exercise. 
The language of types is extended with a new constructor, $T:=\cdots\mid T\times T$. 
They term syntax is extended with the ability to construct pairs $(u_1,u_2)$, as well as a $\lambda$ abstraction on a pair, $\lambda(x,y).u$. 
Note that this is a syntax that implies pattern matching on the argument. 
Projections can be defined as a family of terms $\pi_i = \lambda(x_1{:}T_1,x_2{:}T_2).x_i:T_i$ with $i=1,2$. 

The two new typing rules are:
\begin{gather*}
\frac{\Gamma\vdash u_i:T_i\quad i=1,2}{\Gamma\vdash{(u_1,u_2)}:T_1\times T_2}
\\[1.5ex]
\frac{\Gamma_1,x_1:T_1,
\Gamma_2,x_2:T_2, 
\Gamma_3\vdash u:T_3}{\Gamma_1,\Gamma_2,\Gamma_3\vdash \lambda(x_1,x_2){:}T_1\times T_2.u:(T_1\times T_2)\to T_3}.
\end{gather*}
The diagrammatic syntax is:
\[
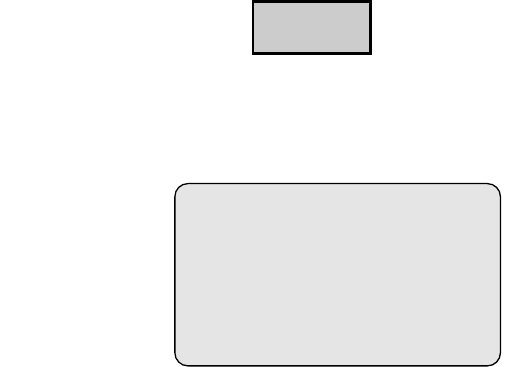
\]

\subsubsection{The un(i)typed $\lambda$ calculus}\label{sec:unityped}
We can define $\lambda$ terms without typing information, i.e. \emph{untyped}, and we call it UTLC for brevity.  
However, in order to implement this calculus in a CCC types are essential, as they correspond to the category's objects. 
So the interpretation will be \emph{uni-typed}, assigning all objects a single type, interpreted by a single object $U$.
\begin{definition}
$U$ is a \emph{reflexive object} in a CCC if it there are morphisms 
\[
\xymatrix{
U \ar@<1ex>[r]^{\iota\quad} &  \ar@<1ex>[l]^{\rho\quad}
(U\Rightarrow U) 
}
\]
such that 
\[
\rho \semic \iota = \id_{U\Rightarrow U}.
\]
\end{definition}
In this case we say that $U\Rightarrow U$ is a \emph{retract} of $U$. 

Let us use the following graphical syntax for the special morphisms $\iota$ and $\rho$:
\[
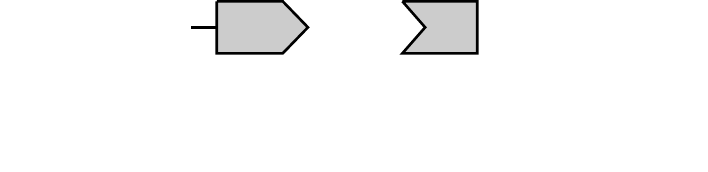
\]
Note that it is redundant to indicate the type of terms and variables in $\Gamma$, as they are always $U$. 
With these, the interpretation of UTLC terms can be given as
\[
\hspace{-15ex}
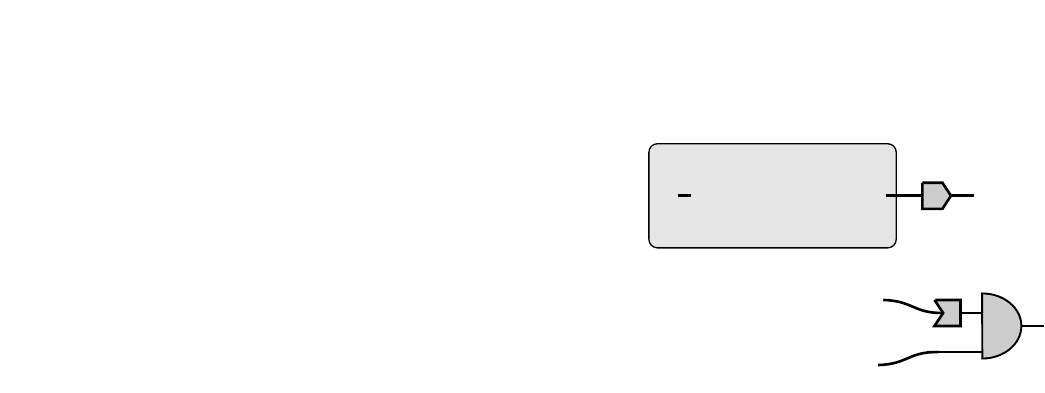
\]
\begin{example}\label{ex:idid2}
We can revisit Example~\ref{ex:idid}. 
Consider the identity applied to itself $(\lambda x.x)(\lambda y.y)$. 
Its string diagram representation is, after cancelling out the retract, the same as before:
\[
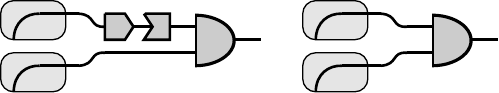
\]
\end{example}
\begin{example}
In the UTLC we can have more terms than in STLC, for example diverging terms, $(\lambda x.x\,x)\,(\lambda x.x\,x)$:
\[
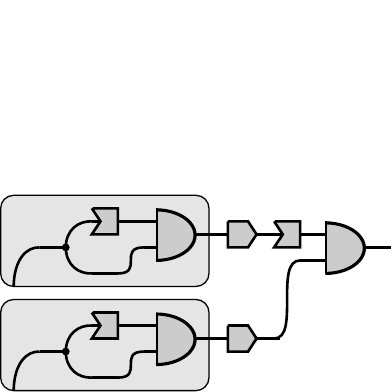
\]
\end{example}

\subsection{The theory $\lambda$}
We write the equality of two $\lambda$ calculus terms, typed or un(i)typed, as 
\[
\Gamma \vdash u = v
\]
where $\Gamma \vdash u:T$ and $\Gamma \vdash v:T$ for some type $T$.

The equality of  consists of the usual equality axioms (reflexivity, symmetry, transitivity) plus  three \emph{congruences}:
\[
\frac{\Gamma \vdash u=u'}{\Gamma \vdash u\,v=u'\,v}
\qquad
\frac{\Gamma \vdash v=v'}{\Gamma \vdash u\,v=u\,v'}
\qquad
\frac{\Gamma,x:U,\Gamma' \vdash u=u'}{\Gamma,\Gamma' \vdash \lambda x.u=\lambda x.u'}
\]
and a main rule, which makes the $\lambda$ calculus special, relating \emph{application} and \emph{substitution}:
\[
\Gamma \vdash (\lambda x.u)v = u[x/v],
\]
where $[-/-]$ is the capture-avoiding substitution defined earlier. 
It is called \emph{the $\beta$ rule}: 

The calculus is often extended with an additional rule called $\eta$:
\[
\Gamma \vdash \lambda x.(u\,x)=u\qquad x\not\in \Gamma.
\]
\begin{remark}
The first set of inferences are called congruences because they immediately imply that if two terms are equivalent then they are so in any \emph{context}, i.e. term-with-a-hole $\mathcal U[-]$:
\[
\frac{v=v'}{\mathcal U[v]=\mathcal U[v']}.
\]
\end{remark}
The categorical and implicitly string-diagram syntax is constructed out of adjunctions which have their own equational properties. 
Ideally, the theory $\lambda$ on the terms should coincide with the equational theory of the CCC as applied to the interpreted terms. 

Before that, we will derive the interpretation of substitution in the CCC:
\begin{lemma}\label{lem:subst}
\begin{multline*}
\seval{\Gamma_1,\Gamma_2\vdash u[x/v]:T} = \\
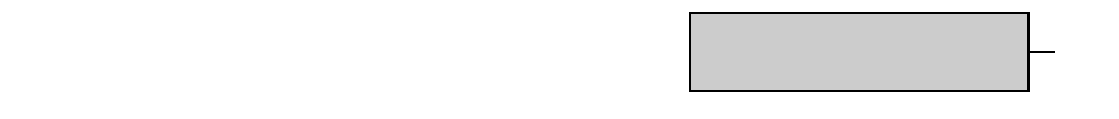
\end{multline*}
\end{lemma}
\begin{proof}
The proof is by induction on the syntax of $u$, and we show the equations diagrammatically
\begin{enumerate}
\item $u=x$, which implies $T=T'$, and $u[x/v]=x[x/v]=v$. 
\[
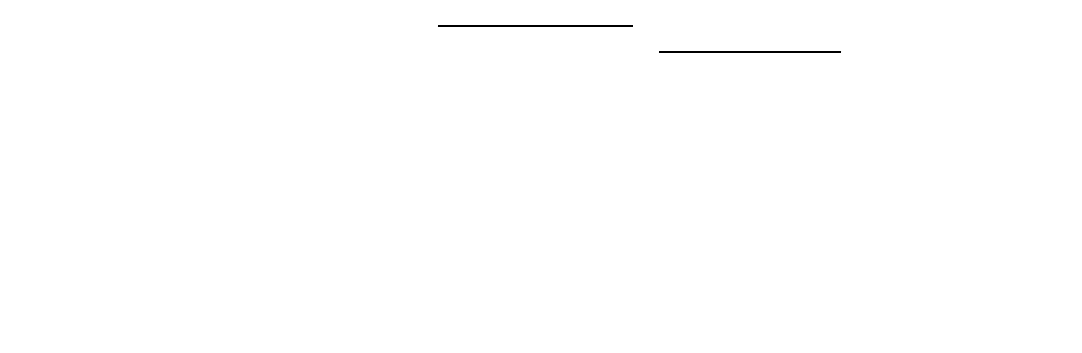
\]
\item \label{subst:2}
$u=y$, $y\neq x$ which implies $u[x/v]=y[x/v]=y$.
The proof is in Figure~\ref{fig:sbsty}.
\begin{figure}
\[
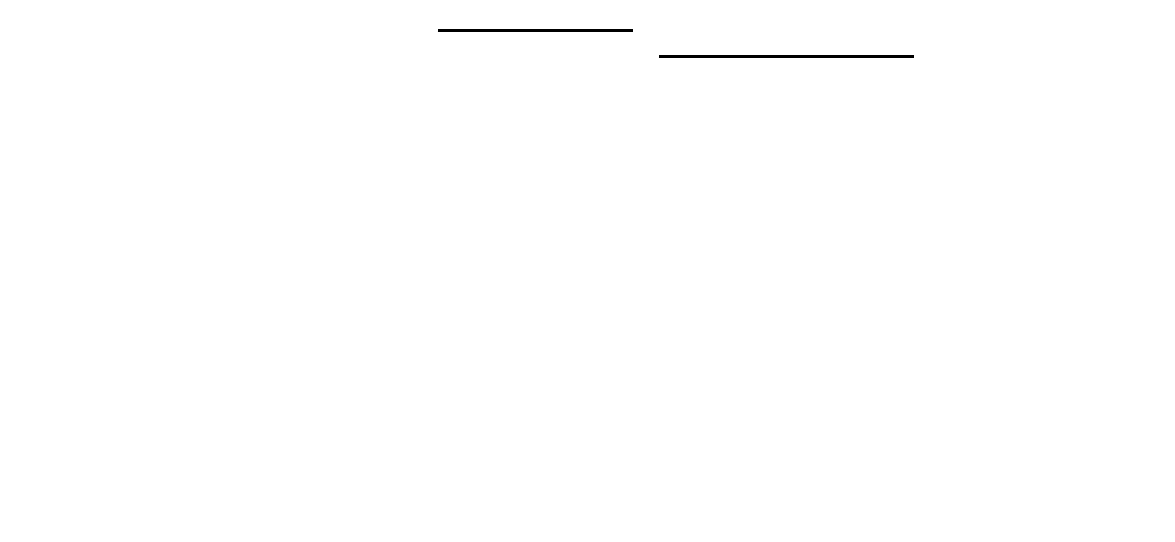
\] 
\caption{Proof of Lemma \ref{lem:subst}(\ref{subst:2})}
\label{fig:sbsty}
\end{figure}
\item The case $u_1\,u_2$ is an immediate consequence of the naturality of copying, left as an exercise to the reader. 
\item The case $\lambda y.u$, $y\neq x$  which implies the right-hand-side is $(\lambda y.u)[x/v]=\lambda y.(u[x/v]$.
It follows from naturality of coeval and functoriality of the exponential, as shown in Figure~\ref{fig:sbstlam} (which is essentially the same as Lemma~\ref{lem:sbstlam0}).
\begin{figure}
\[
\hspace{-20ex}
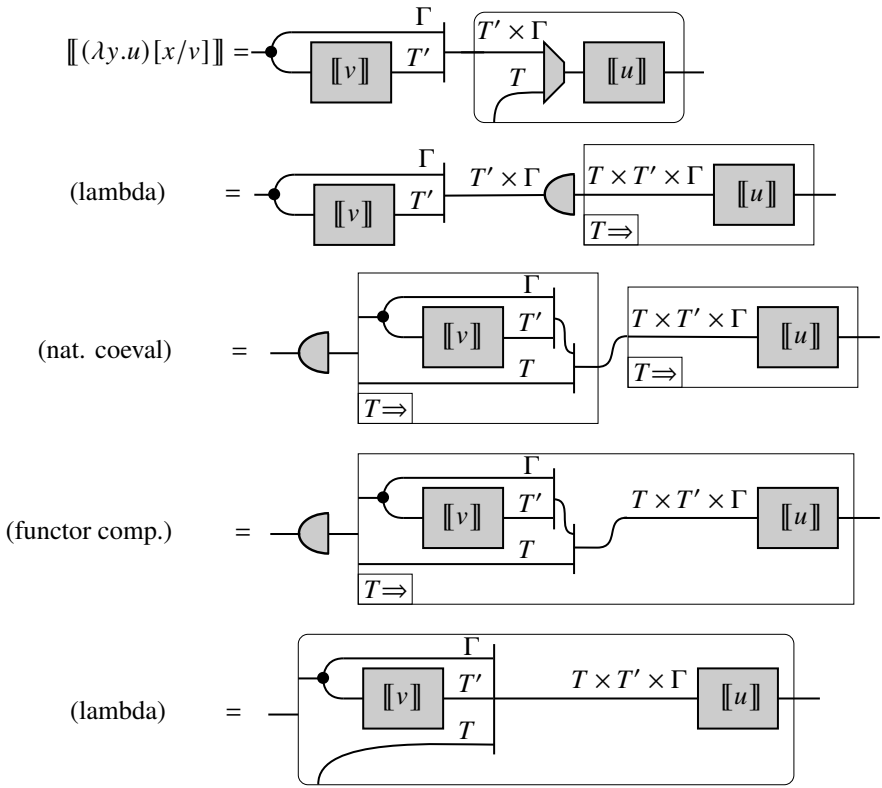
\]
\caption{Substitution in abstraction}
\label{fig:sbstlam}
\end{figure}
\end{enumerate}
\end{proof}
\begin{theorem}\label{thm:snd}
For any $\Gamma\vdash u:T$, $\Gamma\vdash v:T$,
if $\Gamma \vdash u=v$ in the theory $\lambda$, then 
\[\seval{\Gamma\vdash u:T}=\seval{\Gamma\vdash v:T}\] 
in the CCC. 
\end{theorem}
\begin{proof}
We use induction on the rules of the theory $\lambda$. 
For the equality and congruence axioms the statement is obviously true. 
The non-obvious rules are $\beta$ and $\eta$, which we derive graphically. 
\begin{description}
\item[$\beta$ rule]: \label{snd:beta}
$(\lambda x.u)v=u[v/x]$, proof in Figure \ref{fig:betaeq}.
\begin{figure}
\[
\hspace{-20ex}
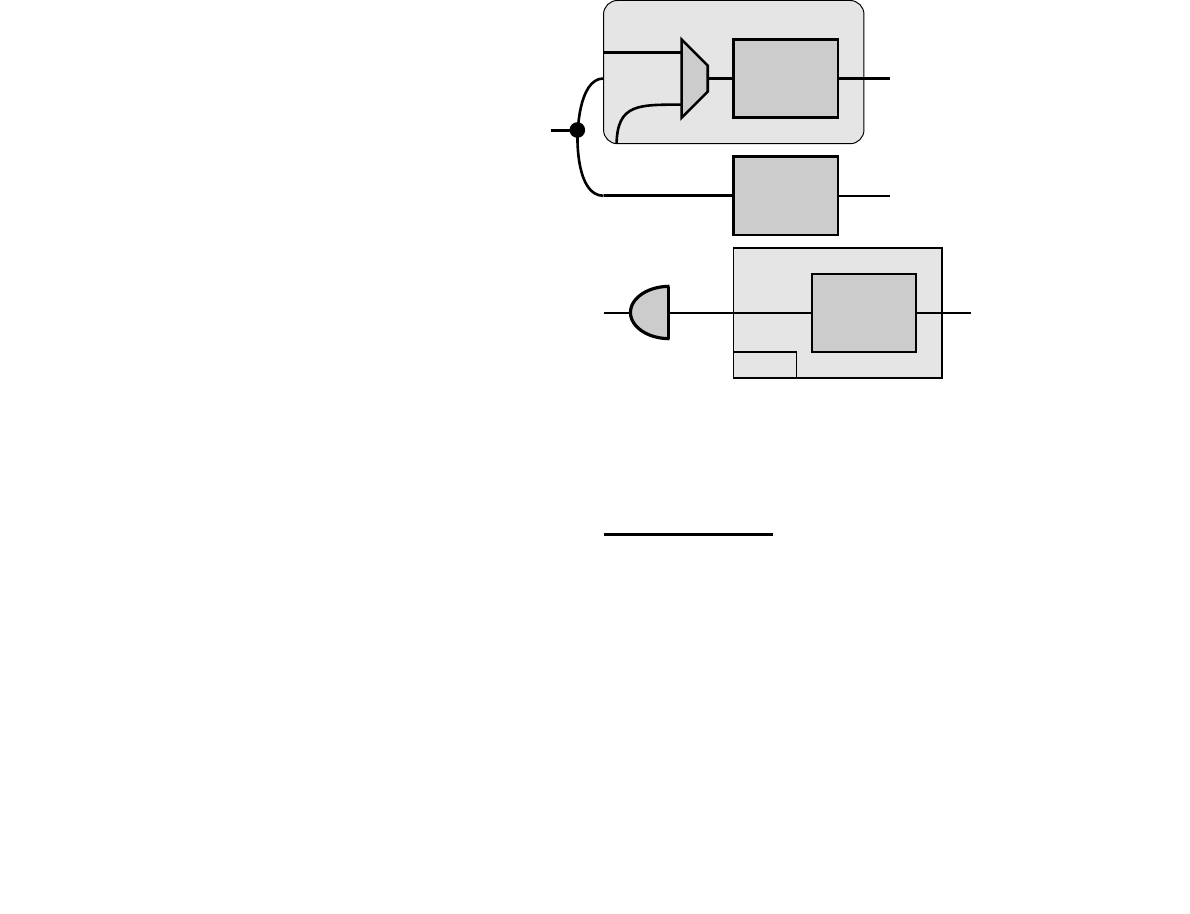
\]
\caption{Proof of $\beta$ rule}
\label{fig:betaeq}
\end{figure}
\item[$\eta$ rule]: $\lambda x.ux = u$, with $x\not\in\mathcal F(u)$ is shown in Figure~\ref{fig:etaeq}.
\begin{figure}
\[
\hspace{-20ex}
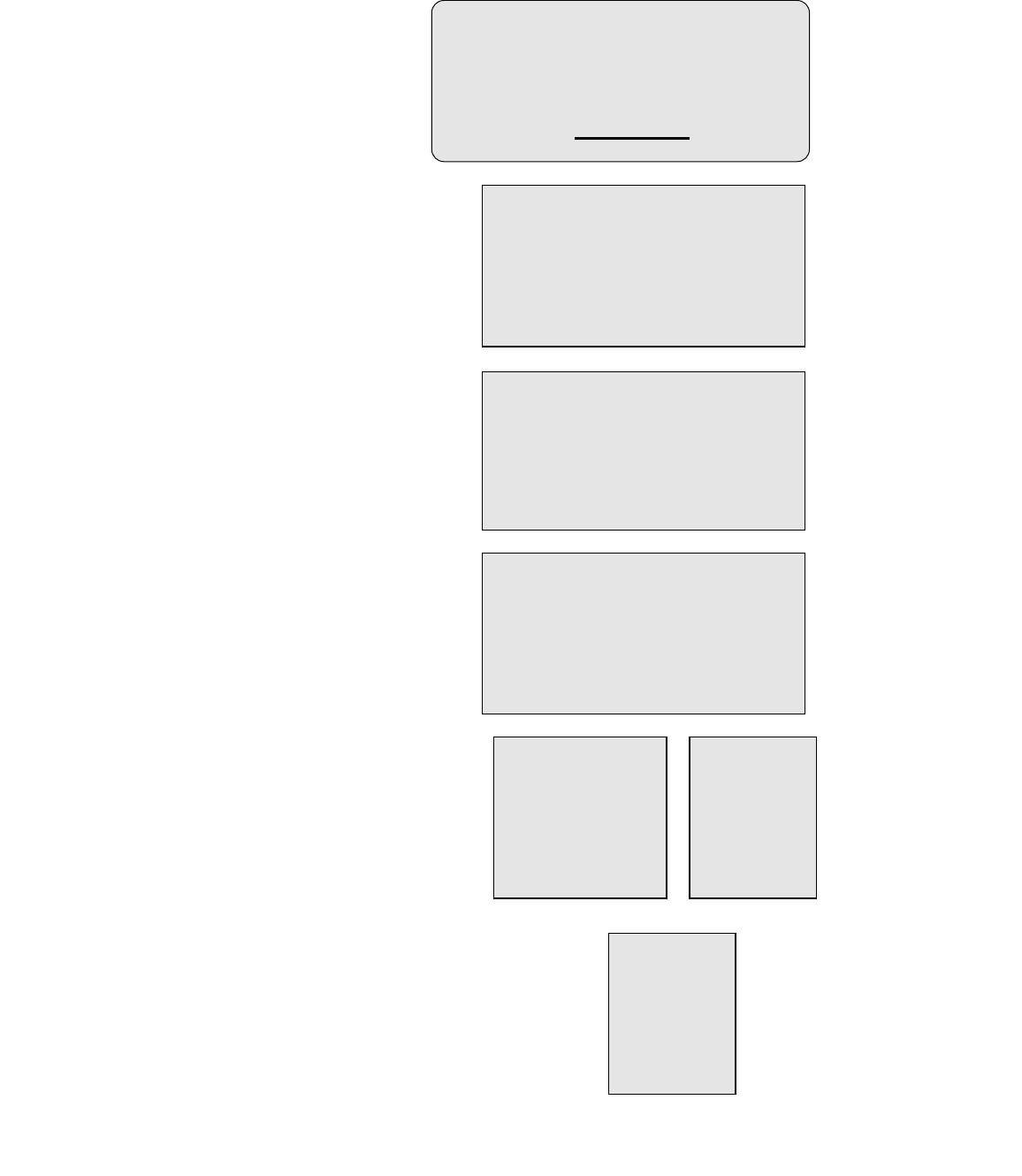
\]
\caption{Proof of the $\eta$ law}
\label{fig:etaeq}
\end{figure}
\end{description}
\end{proof}

\subsection{The Geometry of Synthesis construction}

In this section we will examine how Cartesian closed categories can be constructed out of more elemental categories, especially out of compact closed categories. 
The latter are natural models of systems with a relational semantics, for instance electrical circuits, whereas the former are, as we have just seen, frameworks for models of programming languages. 
``Compiling'' conventional programming languages into circuits is a long standing area of engineering interest, generally called \emph{high level synthesis}. 
Indeed, this process can be achieved systematically by understanding how Cartesian closed categories can be constructed out of compact closed categories, a construction dubbed \emph{the Geometry of Synthesis}~\cite{DBLP:conf/popl/Ghica07}.
The name of the construction is a homage to a famous construction in semantics of linear logic proofs called \emph{the Geometry of Interaction}, which shows how compact closed categories arise out of traced monoidal categories \cite{DBLP:journals/mscs/AbramskyHS02} and which can also be given an intuitively diagrammatic construction. 

We will describe it briefly, starting with the diagram language of compact closed categories, which are (symmetric for this case) monoidal categories such that every object $A$ has: 
\begin{itemize}
\item a \emph{dual} $A^*$, unique up to canonical isomorphism;
\item a \emph{unit} $\eta_A:I\mapsto A^*\otimes A$
\item a \emph{counit} $\epsilon_A:A\otimes A^*\to A$
\end{itemize}
represented diagramatically as: 
\[
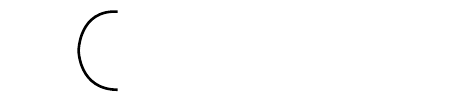
\]
such that the following equations hold, expressed directly in the diagrammatic language: 
\[
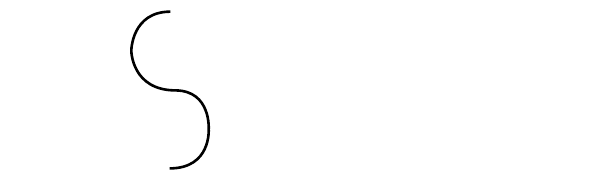
\]
Using the machinery of the (symmetric) compact closed category, a monoidal closed category can be constructed in a standard way:
\[
\hspace{18ex} 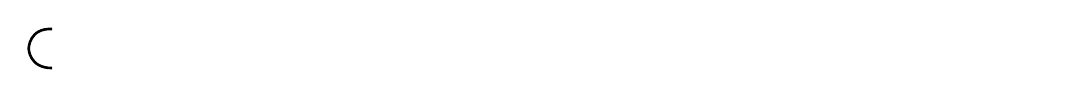
\]
Evaluation and the proof of correctness of evaluation (the $\beta$ law) are given below:
\[
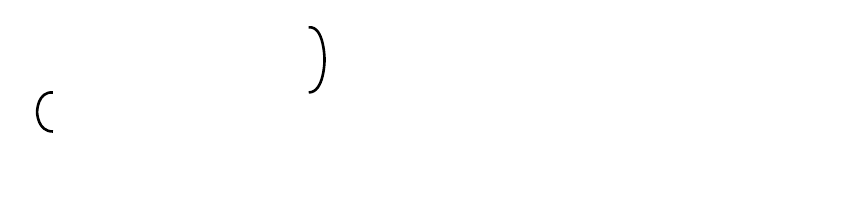
\]
The first equation relies on strictifier/de-strictifier cancellation and symmetry. 
The second equation is the `yanking' property of compact closed category in which the unit/co-unit pair cancel out, allowing the straightening of the wire. 

However, this simple solution is not satisfactory because of a ``no go'' result of category theory, that in a compact closed category the product is automatically a co-product. 
It is not important to dwell at this stage on what a co-product is, but having the two coincide is a degeneracy of the framework: in other words, there are now too many equations in the category, allowing certain constructions that should be distinct to be identified. 
Therefore, this construction cannot be used to eliminate the need of a direct construction for a Cartesian closed category, as seen so far. 

This construction is useful in the context of `high-level synthesis', with a caveat: the category cannot have unrestricted Cartesian product. 
Restricting the scope of the product while maintaining the usefulness of the language require certain type-system trickery inspired by bounded linear logic, which is beyond the scope of this tutorial \cite{DBLP:conf/popl/GhicaS11}.

\subsection{Abstract syntax graphs}
\label{sec:asg}

This section is informal, intended to reconnect the readers, especially those less versed into category theoretical concepts, with more basic and more available intuitions. 
As such, the section can be skipped without loss of continuity, especially by the reader confident in their mathematical intuitions. 

It is common, especially in the area of compilers, to cling to the oversimplified view that syntax is represented by a \emph{tree}. 
In fact, syntax is a far more sophisticated data structure: a \emph{hierarchical graph}.
In this section we shall see how the diagrams we used so far can be informally reconstructed as a data structure which enhances in a rather intuitive way the more commonly used yet often unsatisfactory `\emph{abstract syntax trees}' (AST). 

As an expression of algebraic expressions, abstract syntax trees are perfectly adequate. 
Consider an expression such as $1+(2+3)$. 
Its tree representation is conventionally written top-to-bottom like this: 
\[
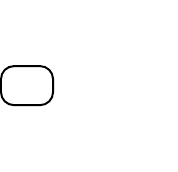
\]
This syntax tree is `abstract' in the sense that not all the lexical tokens of the original are represented in the tree, namely the open and closed brackets in this example. 

The representation becomes unsatisfactory when variables are introduced, as in the expression $x+(2+x)$. 
The AST of this expression is:
\[
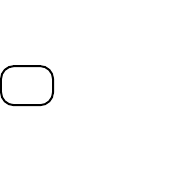
\]
The problem with this representation, as a data structure, is that it offers poor support to the most common operations performed on variables: substitution. 
Replacing $x$ with, say, 0 requires searching the entire graph for occurrences of $x$. 
By contrast, using string diagrams the same expression is represented such that only one logical occurrence of the variable $x$, the leftmost wire, is shared by two nodes in the expression:
\[
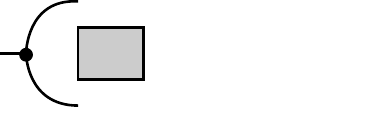
\]
To perform a substitution we no longer require to search the tree for occurrences of $x$, as there is only a single occurrence of the variable. 
In fact, in any reasonable concrete implementation of string diagrams the left and right interface wires need to be explicitly recorded, which makes substitution possible in constant time, relative to the size of the overall string diagram. 

Once binding syntax is introduced, the shortcomings of the AST become more perspicuous. 
Consider the expression $\mathit{let}\ x=0\ in\ x+(2+x)$. 
The AST form of the expression is:
\[
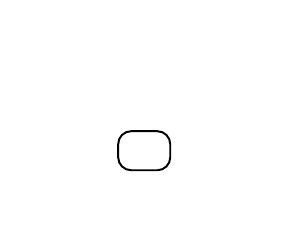
\]
Note that the binding occurrence of $x$ and the regular occurrence are not distinguished, and there is no easy way to connect them. 
By contrast, the same expression, using a string diagram representation, is:
\[
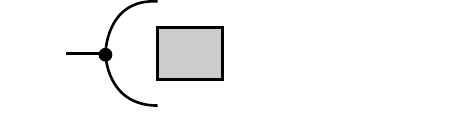
\]
The overall structure is clearer, and it corresponds more closely to the \emph{meaning} of the expression. 
The problem of relating variable occurrences to the binder is solved be making variables into wires rather than boxes. 
This distinction between variables and operations is useful in several ways, perhaps the most important being that the string diagram is automatically quotiented by alpha equivalence. 

The convenience of this approach becomes more perspicuous in the case of variable shadowing, as in 
\[
\mathit{let}\  x = 0\ in\ (\mathit{let}\  x = 1\ in\ x + 2) + x
\]
The AST is of course not ambiguous, but connecting variable occurrences to their binders is even more difficult:
\[
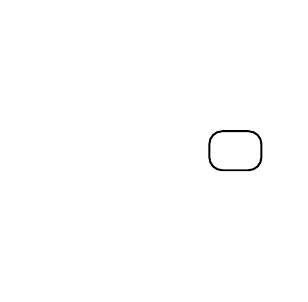
\]
This is why for purposes of mechanisation this representation of variables is usually avoided.
The more streamlined, if harder for human consumption, notation of deBruijn indices is the preferred alternative. 
The string diagram representation of this expression is easy to read and it avoids all the pitfalls of variable shadowing:
\[
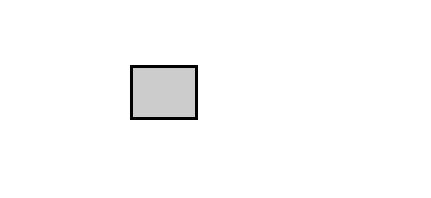
\]
The rather remarkable thing to note here is that this string diagram is \emph{equal} to that of just $(1+2)+0$!
So the string diagram notation is significantly more abstract than `abstract' syntax trees, as not just alpha equivalence but also linear substitution are absorbed in the notation. 

The version of string diagrams used up to this point can be called, in their concrete representation, `abstract syntax DAGs' (directed acyclic graphs). 
The hierarchical aspect appears with the introduction of \emph{thunks}. 
Categorical considerations aside, thunks are have uses both semantically and syntactically. 
Semantically they indicate to the evaluator that the evaluation of a piece of code must be postponed;  first, a thunk needs to be `forced', i.e. extracted from its protective bubble, by some other operations. 
We have already seen how the if-then-else statement forces one of its argument thunks and how application forces the thunk that corresponds to the applied function. 

Representing the nullary (i.e. which binds zero variables) thunks of the if-then-else statement explicitly may not seem justified enough. 
But the small notational overhead of indicating the thunk explicitly is relevant for generic analyses that want to take the flow of control into account, as the presence of the thunk immediately indicates that the program can control and direct the evaluation. 
But these are rather subtle and partly manufactured reasons;  the real reasons are conceptual, the consistency and integrity of the string diagram notation, as presented in the previous sections. 

A more direct case can be made for the hierarchical notation in the case of thunks which also bind variables, as is the case with functions. 
Taking the ubiquitous $\lambda$-abstraction, the same issues of binding and alpha equivalence discuss earlier will come into play again. 
Notably, the binding structures above did not involve thunks. 
Why do we need thunks for $\lambda$ and why is the hierarchical structure a reasonable evolution of the AST? 

The most important thing is that thunks must be seen as undivided whole, in particular when it comes to their interaction with the Cartesian equations for copying and discarding. 
Let us illustrate this with a simple non-example, in which we attempt, and fail, to represent thunks without resorting to hierarchical structure: $\lambda x.x+x$. 
The conventional AST we have already agreed that it is not ideal, so we will use the variable conventions of string diagrams while treating the $\lambda$ as if it were an operation:
\[
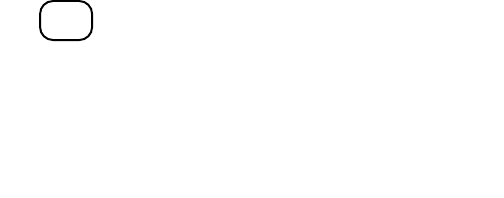
\]
This is an approach similar to that taken in \emph{Interaction nets}. 

But this notation is clearly problematic both in the presence of copy and discard operation. 
A term such as $\mathit{let}\ f=\lambda x.x\ in\ f\,f$ is represented in a way that applying the copying axioms results in an ill-formed string diagram which corresponds to no legal term:
\[
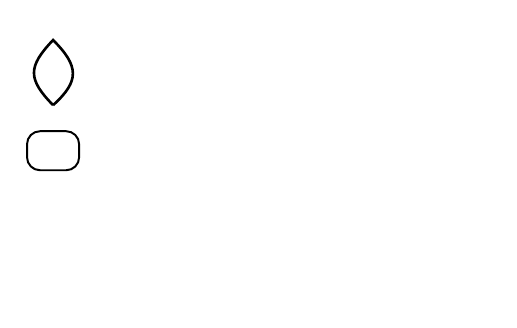
\]
To overcome this problem interaction nets uses a different kind of node when $\lambda$ nodes are duplicated, with $C$ for copy and $\delta$ as the special node resulting from the copying of a $\lambda$, so the copying rule for abstraction is:
\[
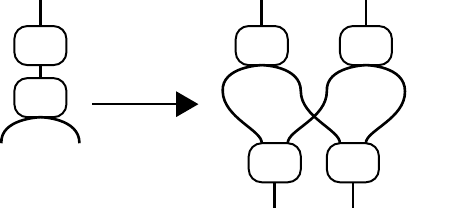
\]
This, however, is not a satisfactory solution since it ruins the naturality of copying. 
The hierarchical graph structure prevents this problem by treating a thunk as a single indivisible thing.
Indeed, a thunk is always a node, with the inner graph its label: 
\[
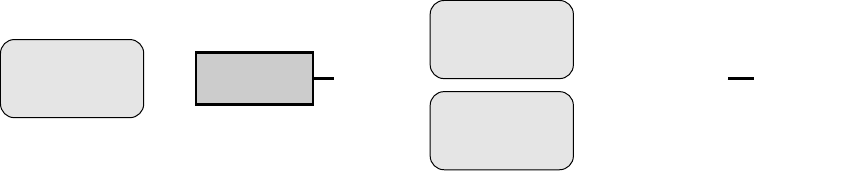
\]
\begin{exercise}
Consider the way in which applying the discard equation is not compatible with the naive representation of $\lambda$ nodes are operations. 
\end{exercise}

Before we conclude, a note on expressing the the beta law or, more generally, the unit-counit cancellation in string diagrams.
The widely used string diagrams for strict monoidal compact closed categories absorb the cancellation rule for adjunctions.
Without going into the technical detail, the unit and the co-unit can be represented as `loops' so that their cancellation equations
\begin{gather*}
\id_A\otimes\eta  \semic  \epsilon\otimes\id_A = id_A \\
\eta\otimes\id_{A^*} \semic  \id_{A^*}\otimes\epsilon = id_{A^*}
\end{gather*}
have neat topological interpretations:
\[
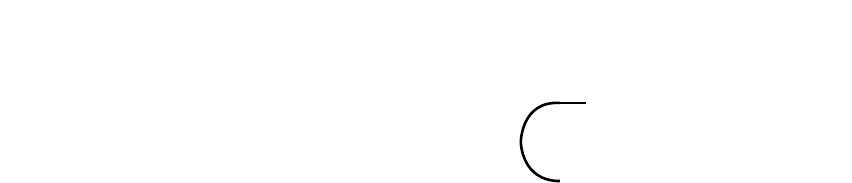
\]
Such notations do not seem to be readily available for the unit/co-unit pair involved in the beta law. 
A notable attempt has been made by Baez and Stay, who proposed a `bubble'-`claps' pair in their seminal `Rosetta Stone' paper. 
However, their proposal was informal, a side issue to the main thrust of their paper, which focusses on the compact closed structure. 
We will attempt to refine our notation here to make it more similar to bubble-clasp diagrams, but only for the sake of explanation. 
In the sequel we will use the graphical syntactic sugar already introduced. 

Currying in the Rosetta notation, vs. our notation is:
\[
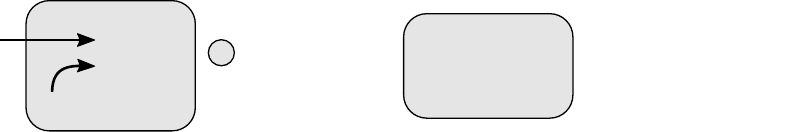
\]
We can see that this notation is trying to maintain the strict monoidal convention that each wire is labelled by an object, so that composite wires can be decomposed in their constituents. 
However, in the case of closed monoidal structure it is not a good idea to give direct access to both $X$ and $Y$ separately in $X\multimap Y$ unless additional graphical conventions are used to prevent nonsensical use of the clasp. 

The Rosetta notation `peeks' inside the evaluation, so that the beta law is written as the following string diagram equation. 
\[
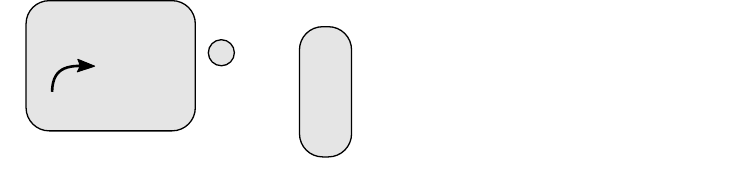
\]
This is graphically more compelling as it establishes a visual parallel with the way compact closed categories cancel out the two adjunctions. 
By using the clasp in a way similar to the strictifier/de-strictifier morphisms, and by pairing it with an `anti-clasp' we can combine our notation with the Rosetta in a way that is perhaps graphically more informative without raising the possibility of ill-formed diagrams:
\[
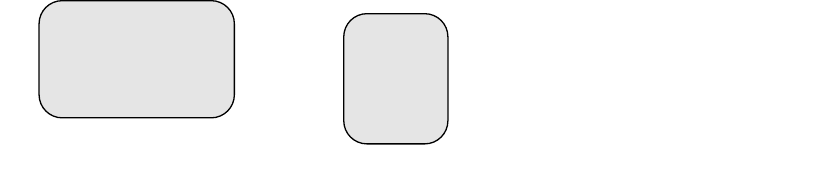
\]
However, the clasps and the anti-clasp now ruin somewhat the visual appeal. 
But this is the price we may need to pay to avoid ill-formed diagrams. 
The additional purchase of the clasps is not worth the additional clutter, so we will stick to the simpler notation. 
The reader who disagrees may adopt it, with no cost in terms of rigour. 

We conclude here this brief and informal section. 
In it we have temporarily suspended our category-theory-driven approach and attempted to persuade the reader that string diagrams, and their concrete incarnation as hierarchical abstract syntax graphs (or more precisely abstract syntax DAGs) could be recovered as an evolution of the more conventional AST representation in which variables, binding, sharing, copying, and discarding are dealt with in such a way so that many useful properties are absorbed in the notation. 

% Maybe not needed for this section:
%\subsection{Further reading and related work}

%\section{Traced monoidal categories}

\newpage
\section{String Diagram Rewriting}\label{sec:graphs}

\newcommand{\rew}{\rightsquigarrow}
\newcommand{\downmapsto}{\rotatebox[origin=c]{-90}{\mapsto}\mkern2mu}

Even though we specified the reduction rules of the $\lambda$ calculus as equations, they display a clear directionality. For instance, the $\beta$ rule may be seen as the process of applying a term $\lambda x .u$ to a value $v$, resulting into the term $u[v/x]$. More formally, this amounts to orienting the equation $(\lambda x. u) = u[v/x]$ into a \emph{rewrite rule} $(\lambda x. u)v \rew u[v/x]$. We may apply such a rule inside a $\lambda$-term $t$ when the left-hand side $(\lambda x. u)v$ appears as a sub-term of $t$ --- in this case say that $t$ contains a \emph{redex} for the rewrite rule. Applying the rule then results into the substitution of $(\lambda x. u)v$ with $u[v/x]$ inside $t$. 

The rewriting theory of the $\lambda$ calculus is fundamental to describe the computational behaviour of $\lambda$-terms via an \emph{operational semantics} --- a perspective that we will pursue in Section~\ref{sec:oslc} below. The subject of this section is to illustrate how rewriting works when applied to string diagrams rather than terms.

\subsection{Rewriting of string diagrams in symmetric monoidal categories}\label{sec:rewSMC}

We first consider the case of string diagrams in a symmetric strict monoidal category, as introduced in Section~\ref{sec:sdssmc}. A naive attempt of transposing the above definitions to the context of string diagrammatic syntax poses immediate challenges. Consider the following example. Suppose we have a rewrite rule of string diagrams

\begin{equation}\label{eq:rewrule-stringdiag}
\raisebox{-.5cm}{\hbox{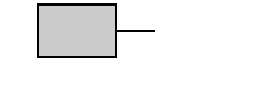}} \qquad \rew \qquad \raisebox{-.5cm}{\hbox{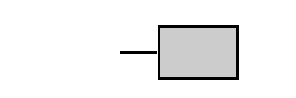}}
\end{equation}

\noindent which we want to apply inside a string diagram $t$ of the form

\begin{equation}\label{eq:rewcontext-stringdiag}
	{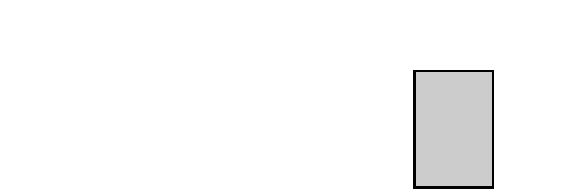}.
\end{equation}
`Morally', $t$ contains a redex for the rule. But, strictly speaking, we can isolate the left-hand side of the rule as a sub-term of $t$ only if we first use naturality of the symmetry (Figure~\ref{fig:eqsym}) to rearrange $t$ into a string diagram $t'$, as follows:
\begin{equation}\label{eq:rearrangeforrew}
	\raisebox{-.7cm}{\hbox{\scalebox{.75}{\input{pics/diagrewlhs.pdf_tex}}}} \quad = \quad \raisebox{-.7cm}{\scalebox{.75}{\hbox{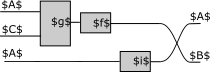}}} 
\end{equation}
Note $t$ and $t'$ in~\eqref{eq:rearrangeforrew} are equivalent modulo the laws of symmetric monoidal categories. We now have a clear indication of how to apply the rewrite rule: simply replace in $t'$ the left-hand side of the rule with its right-hand side:
 \begin{equation}\label{eq:rewstepstringdiag}
	\raisebox{-.7cm}{\scalebox{.75}{\hbox{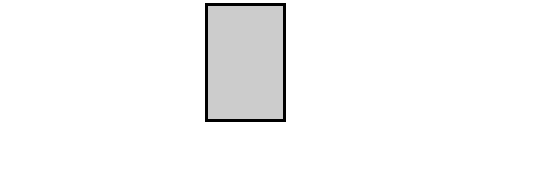}}} \quad \rew \quad \raisebox{-.7cm}{\scalebox{.75}{\hbox{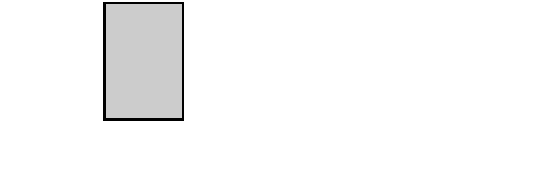}}} 
\end{equation}
This example suggests that the notion of redex becomes way subtler once moving from terms to string diagrams. A string diagram is an equivalence class of terms, and to find a redex in a string diagrams means to find at least one member of such an equivalence class which contains a redex in the traditional sense. More formally:

\begin{definition}[String diagram rewrite step] \label{def:rewritestepsyn}
Consider a rewrite rule $$\mathcal{R} \qquad = \qquad   \ \raisebox{-.2cm}{\hbox{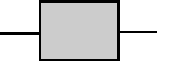}} \rew  \ \raisebox{-.2cm}{\hbox{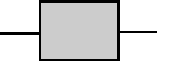}}.$$ 
We say that a string diagram  \ \raisebox{-.2cm}{\hbox{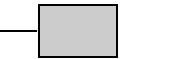}} has a \emph{redex} for the  rule $\mathcal{R}$ if there exist string diagrams 
$$ \ \raisebox{-.4cm}{\hbox{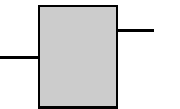}}\quad \text{ and } \quad \raisebox{-.4cm}{\hbox{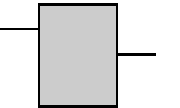}}$$ such that
	\begin{equation*}
		\ \raisebox{-.2cm}{\hbox{\input{pics/diagf.pdf_tex}}} \quad = \quad \raisebox{-.5cm}{\hbox{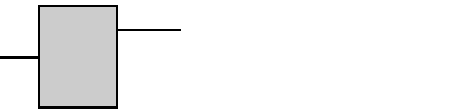}}
	\end{equation*}
	modulo the laws of symmetric monoidal categories. If $\ \raisebox{-.2cm}{\hbox{\input{pics/diagf.pdf_tex}}}$ has such a redex, then rewriting $\ \raisebox{-.2cm}{\hbox{\input{pics/diagf.pdf_tex}}}$ with $\mathcal{R}$ produces
		\begin{equation*}
		\ \raisebox{-.2cm}{\hbox{\input{pics/diagf.pdf_tex}}}  \quad  \rew\quad \raisebox{-.5cm}{\hbox{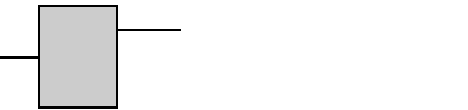}}
	\end{equation*}
\end{definition}
Although mathematically correct, this definition of string diagram rewriting is unsatisfactory from a practical viewpoint. If we aim at studying the operational behaviour of programs using rewriting, finding a redex should be a computationally feasible task. When it amounts to finding a sub-term inside a $\lambda$-term, this is the case. However, imagine searching for a redex in a string diagram: one would need to explore the space of all equivalent terms represented by such diagram, trying to find one with a redex. This is typically an extremely expensive process to implement.

For this reason, we need to rely on an interpretation of string diagrams which is more efficient for the purpose of rewriting. The data structure which we use to interpret string diagrams is the one of a \emph{open hypergraph}.

In a nutshell, open hypergraphs generalise standard (directed) graphs in two ways. First, edges are replaced with \emph{hyperedges}: whereas an edge has just one ingoing and one outgoing node, a hyperedge may have multiple of them (or none), organised as lists. Second, the structures we consider are `open', in the sense that some of their nodes act as an \emph{interface}, with which they may be combined together with other hypergraphs. Thanks to the interface, open hypergraphs may mirror sequential and parallel composition of string diagrams, and are thus adapted to interpret them.

An example showing our graphical representation of open hypergraphs is displayed in~\eqref{ex:hyp} below. First, we proceed with the formal definition of these structures. Labels for nodes and hyperedges will come from a \emph{monoidal signature}, as introduced in Section~\ref{sec:smc} when discussing freely generated symmetric strict monoidal categories.

\begin{definition}[Open Hypergraph] Fix a monoidal signature $\Sigma = (\Sigma_0,\Sigma_1)$ of objects $\Sigma_0$ and morphisms $\Sigma_1$.  A \emph{($\Sigma$-labelled) hypergraph} is a tuple $(N,E,v,l_n,l_e,)$ where 
\begin{itemize}
	\item $N$ is a set of nodes
	\item $E$ is a set of hyperedges
	\item $v \colon E \to N^{\star} \times N^{\star}$ is a function assigning to each hyperedge $e$ a list $\pi_1(v(e))$ of ingoing nodes and a list $\pi_2(v(e))$ of outogoing nodes  
	\item $l_n \colon N \to \Sigma_0$ is a function labelling each node with a generating object in $\Sigma_0$
	\item $l_e \colon E \to \Sigma_1$ is a function labelling each hyperedge with a generating object in $\Sigma_0$, with the requirement that $l_e(e)$ must have arity $\pi_1(v(e))$ and coarity $\pi_2(v(e))$ in $\Sigma_0$.
\end{itemize}
Given $\Sigma$-labelled hypergraphs $G = (N,E,v,l_n,l_e,)$ and $G' = (N',E',v',{l'}_n,{l'}_e,)$, a \emph{hypergraph morphism} $f \colon G \to G'$ consists of functions $f_N \colon N \to N'$ and $f_E \colon E \to E'$ respecting the labelling of nodes and hyperedges and the type of hyperedges. In other words, ${l'}_n(f_N(n)) = v(n)$, ${l'}_e(f_E(e)) = l_e(e)$, and $v'(f_E(e)) = (\hat{f}_N(\pi_1(v(e)),\hat{f}_N(\pi_2(v(e)) )$, where $\hat{f}_N \colon N^{\star} \to N^{\star}$ is the lifting of $f_N \colon N \to N'$, defined point-wise in the obvious way. Such a morphism is a \emph{monomorphism} if $f_N$ and $f_E$ are both injective functions.

    A ($\Sigma$-labelled) \emph{open} hypergraph is a tuple $(L, G, R, f_L, f_R)$ where 
    \begin{itemize}
    	\item $G$ is a $\Sigma$-labelled hypergraph, called the \emph{carrier}
    	\item $L$ is a discrete\footnote{Recall that a graph is discrete when is has only nodes, and no (hyper)edges.} $\Sigma$-labelled hypergraph, called the \emph{left interface} of $G$
    	\item $R$ is a discrete $\Sigma$-labelled hypergraph, called the \emph{right interface} of $G$
    	\item $f_L \colon L \to G$ and $f_R \colon R \to G$ are hypergraph morphisms.
    \end{itemize} 
   A \emph{open hypergraph morphism} $h \colon (L,G,R, f_L, f_R) \to (L',G',R', f'_L, f'_R)$ consists of hypergraph morphisms $h_G \colon G \to G'$, $h_L \colon L \to L'$, $h_R \colon R \to R'$ commuting with the interface morphisms, i.e. such that $h_L  \semic  f'_L = f_L  \semic  h_G$ and $h_R  \semic  f'_R = f_R  \semic  h_G$. 
\end{definition}

\begin{example}\label{ex:hypergraph}
	Fix $\Sigma = (\{A,B,C\}, \{f \colon A \otimes B \to C\}, g \colon C \to B\})$. Here is an example of open $\Sigma$-hypergraph $(G,L,R, f_L, f_R)$, displayed both as a tuple (top) and in its graphical representation --- in~\eqref{ex:hyp} below. 
	
	\begin{minipage}{0.7\textwidth}
	\begin{gather*} 
		\hspace{-.5cm}		G = \left( \begin{matrix}
			N = \{n_1,n_2,n_3,n_4\}, \ E = \{e_1,e_2\} \\
			{\small v \colon e_1 \mapsto ([n_1,n_2],[n_3]), \ e_2 \mapsto ([n_3], [n_4])}\\
			l_n \colon n_1 \mapsto A, \ n_2 \mapsto B,\  n_3 \mapsto C , \ n_4 \mapsto B \\
			l_e \colon e_1 \mapsto f , \ e_2 \mapsto g
		\end{matrix}\right) 
		\qquad 
	\hspace{-.5cm}	\begin{matrix}
			f_L \colon m_1 \mapsto n_1, \ m_2 \mapsto n_2 \\ \\ f_R \colon p \mapsto n_4
		\end{matrix}
	\end{gather*}
		\end{minipage} \\
	\begin{minipage}{0.4\textwidth}
	\begin{gather*}
		L = \left( \begin{matrix}
			N = \{m_1, m_2\},  E = \emptyset \\
			l_n \colon m_1 \mapsto A, \ m_2 \mapsto B 
		\end{matrix}\right) 
      \qquad 
		R = \left( \begin{matrix}
			N = \{p\}, \ E = \emptyset \\
			l_n \colon p \mapsto B  
		\end{matrix}\right) 
	\end{gather*}
	\end{minipage}
	\\[2em]
	\begin{equation}\label{ex:hyp}
\scalebox{0.9}{\input{./pics/rewriting-pics/exhyp.tikz}}
	\end{equation}
	
	 Hyperedges are represented as boxes with round corners, labeled with the generating morphisms in $\Sigma$. Nodes are represented as black dots, labeled with the generating objects in $\Sigma$. The assignment function associating hyperedges to their ingoing and outgoing nodes is displayed with black lines. The left and right interface functions are displayed with red lines.  To emphasise the different role played by the three hypergraphs involved, we use a grey background for the carrier and a blue background for the left and right interfaces.  
\end{example}

%The presence of interfaces allows us to understand open hypergraphs as \emph{composable} entities. We can form a category $\Ohyp{\Sigma}$ where the set of objects is $\Sigma_0^{\star}$, and morphisms of type $w \to v$ are $\Sigma$-labelled open hypergraphs and their morphisms form a symmetric monoidal category $\Ohyp{\Sigma}$. Given 

As mentioned, the presence of interfaces is essential to be able to \emph{compose} open hypergraphs sequentially, mimicking the way string diagrams are composed. Given open hypergraphs $H_1 = (L_1, G_1, R_1, f_{L_1}, f_{R_1})$ and $H_2 = L_2, G_2, R_2, f_{L_2}, f_{R_2})$ such that $R_1 = L_2$, we can form a composite open hypergraph $H_1  \semic  H_2 = (L_1, G, R_2, f'_{L_1}, f_{R_2})$, where $G$ has been obtained by `gluing' together $G_1$ and $G_2$ along the common interface $R_1 = L_2$. For instance:
\begin{gather*}
	H_1 = \scalebox{0.9}{\input{./pics/rewriting-pics/glueL.tikz}} \ \qquad H_2 = \scalebox{0.9}{\input{./pics/rewriting-pics/glueR.tikz}} \\
	H_1  \semic  H_2 \ = \ \scalebox{0.9}{\input{./pics/rewriting-pics/exhyp.tikz}}
\end{gather*}
This operation can be defined in full generality using category theory. $\Sigma$-labelled hypergraphs form a category, in which open hypergraphs $H_1$ and $H_2$ as above can be identified with \emph{cospans} $L_1 \xrightarrow{f_{L_1}} G_1 \xleftarrow{f_{R_1}} R_1$ and $L_2 \xrightarrow{f_{L_2}} G_2 \xleftarrow{f_{R_2}} R_2$ respectively. The composite open hypergraph $H_1 \semic H_2$ is then defined by the cospan $L_1 \xrightarrow{f_{L_1} \semic p_1} G \xleftarrow{f_{R_2} \semic p_2} R_2$ obtained via pushout along the shared interface $R_1 =L_2$:

\newcommand{\pushoutcorner}[1][u]{\save*!/#1-1.2pc/#1:(-1,1)@^{|-}\restore}

\begin{eqnarray*}
	\xymatrix{
	L_1 \ar^{f_{L_1}}[dr] & & \ar_{f_{R_1}}[dl]  R_1 \ar^{f_{L_2}}[dr] & & \ar_{f_{R_2}}[dl] R_2\\ 
	& H_1 \ar^{p_1}[dr]  & & \ar_{p_2}[dl] H_2 & \\
	&& \pushoutcorner G && 
	}
\end{eqnarray*}
Moreover, we can easily define a monoidal product operation between open hypergraphs of arbitrary interfaces, which pictorially amounts to stacking them vertically. For instance, with $H_1$ and $H_2$ as above, we have:

\begin{eqnarray*}
\qquad \qquad H_1  \otimes  H_2 & = & \vcenter{\scalebox{0.9}{\input{./pics/rewriting-pics/tnshyp.tikz}}}
\end{eqnarray*}

\medskip

The reason to introduce open hypergraphs is that there is a well-behaved, natural interpretation of string diagrams into them. In a nutshell, each box (morphism) of a string diagram corresponds to a hyperedge, and each wire (object) corresponds to a node. Nodes on the interfaces correspond to the domain (left interface) and the codomain (right interface) of the string diagram. Sequential and parallel composition of string diagrams is interpreted as composition of open hypergraphs by $\semic$ and $\otimes$ respectively, as defined above.

\newcommand{\gint}[1]{[\![#1]\!]}

\begin{definition}[Hypergraph interpretation] \label{def:hypint}
	Let $\Sigma$ be a signature. We define inductively an interpretation $\gint{ \cdot }$ mapping symmetric monoidal $\Sigma$-terms into $\Sigma$-labelled open hypergraphs as follows:
	\begin{gather*}
		{\input{pics/id.pdf_tex}} \ \mapsto \ \scalebox{0.9}{\input{./pics/rewriting-pics/idhyp.tikz}} \qquad\qquad \raisebox{-.5cm}{\hbox{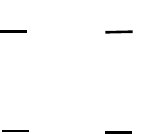}}
 \ \mapsto \ \scalebox{0.9}{\input{./pics/rewriting-pics/symhyp.tikz}} \\[2em]
		\raisebox{-.7cm}{\hbox{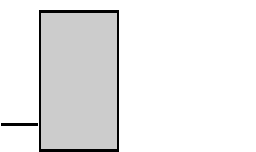}} \!\!\!\!\!\!\!\!\!\!\!\! \mapsto \ \scalebox{0.9}{\input{./pics/rewriting-pics/ophyp.tikz}} \qquad \text{ for each $f$ in $\Sigma_1$} \\[2em]
		\raisebox{-.2cm}{\hbox{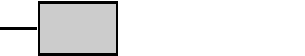}} \mapsto \gint{f} \semic \gint{g} \qquad \raisebox{-.7cm}{\hbox{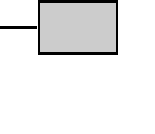}} \mapsto \gint{f} \otimes \gint{g}
	\end{gather*}
\end{definition} 
Note the interpretation is defined inductively on terms, but it is perfectly valid on string diagrams, because any two terms that are equal modulo the laws of symmetric monoidal categories are mapped to the same open hypergraph. In fact, even though we introduced this interpretation primarily for the purpose of implementing string diagram rewriting, it is of independent interest: in Section~\ref{sec:sdssmc}, we first introduced string diagrams as \emph{syntactic} object, but we now have a formal justification to reason about them \emph{combinatorially}, through the lenses of their hypergraph interpretation.

\begin{example}
	The following string diagram, on the same signature $\Sigma$ considered in Example~\ref{ex:hypergraph}, is mapped via $\gint{ \cdot }$ onto the open hypergraph of equation~\eqref{ex:hyp}.
	\[
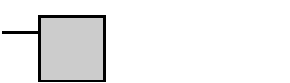
\]
\end{example}

As it turns out, the interpretation $\gint{ \cdot }$ is injective, but not surjective. There are open hypergraphs that are not in the image of any string diagrams. For instance:
\begin{equation}\label{eq:nonmonogamous}
\begin{gathered}
	\scalebox{0.9}{\input{./pics/rewriting-pics/nonmog1.tikz}} \qquad \scalebox{0.9}{\input{./pics/rewriting-pics/nonmog2.tikz}} \\[1em]
	\scalebox{0.9}{\input{./pics/rewriting-pics/nonmog3.tikz}}
\end{gathered}
\end{equation}
In these open hypergraphs, nodes deviate from their standard behaviour in the string diagram interpretation. In the first example, one node appears twice on the left interface. In the second example, one node does not appear on the right interface, even though it has no outgoing link to any hyperedge. In the third example, the same node is linked twice to the same hyperedge.

Inspired by these examples, we may give a general characterisation of precisely which open hypergraphs are in the image of the interpretation.
\begin{definition}[Monogamous Open Hypergraph]
An open hypergraph \\ $(G, L, R, f_L, f_R)$ is monogamous if
	\begin{enumerate}
		\item $f_L$ and $f_R$ are monomorphisms of hypergraphs
		\item for all nodes $n \in N$ of $G = (N,E,v,l_n,l_e)$:
	   			\begin{itemize}
	   				\item If $n$ is in the image of $f_L$, $n$ is not an outgoing node of any hyperedge, i.e. $n \not\in \pi_2(v(e))$ for all $e \in E$
	   				\item If $n$ is in the image of $f_R$, $n$ is not an ingoing node of any hyperedge, i.e. $n \not\in \pi_1(v(e))$ for all $e \in E$
	   				\item Otherwise, $n$ is not in the above images. In this case, $n$ must be an outgoing node of exactly one hyperedge $e$ and an ingoing node of exactly one hyperedge $e'$.
	   			\end{itemize}
	\end{enumerate} 
\end{definition}

Intuitively, monogamicity ensures that nodes on the left/right interface cannot be linked on the same side to a hyperedge, and nodes not on the interface are linked exactly to two hyperedges, one on their left and one on their right. For instance, the first open hypergraph in~\eqref{eq:nonmonogamous} is not monogamous because its left interface morphism $f_L$ is not a monomorphism. The second is not monogamous because a node is not on any interface but it is linked to no hyperedge on its right. Finally, the third is not monogamous because a node is not on any interface but it is linked more than once to hyperedges on its left (the fact that it is linked to the same hyperedge twice or to different ones is immaterial to the criterion).

\begin{theorem}
	An open hypergraph is monogamous if and only if it is in the image of the interpretation $\gint{ \cdot }$.
\end{theorem}
\begin{proof}
	One direction is straightforward: by induction on $\Sigma$-terms we can verify that, if an open hypergraph is in the image of $\gint{ \cdot }$, then it is monogamous. The converse direction requires more work: it amounts to show that any monogamous open hypergraph can be factorised as the composite of `atomic' open hypergraph, where each is either in the image of a structural $\Sigma$-term (identity or symmetry) or of a generating morphism in $\Sigma_1$. Working out the general case of this lemma, albeit not conceptually deep, is laborious: we refer to the literature in Section~\ref{sec:relatedwork-rewriting} for a full proof. Here we confine ourselves to an example:
	\[\scalebox{0.9}{\input{./pics/rewriting-pics/decompL.tikz}}\]
	decomposes as follows into elementary open hypergraphs in the image of $\gint{\cdot}$.
	\[		\scalebox{0.6}{\input{./pics/rewriting-pics/decompR.tikz}} \]
\end{proof}

Having introduced open hypergraphs and explained their correspondence with string diagrams, we are ready to turn to our initial question: rewriting. From the purpose of rewriting, the fundamental appeal of the interpretation $\gint{\cdot}$ is that, whereas a string diagram represents an equivalence class of terms, we can reason about the corresponding open hypergraph as a single entity. In other words, the  interpretation $\gint{\cdot}$ `absorbs' all the laws of symmetric monoidal categories, allowing us to forego any consideration about equivalence when inspecting the subparts of these objects in search of a redex. We can encapsulate this observation as the following lemma.

\begin{lemma}
    A string diagram $c$ has a redex for a rule $l \rew r$ in the sense of Definition~\ref{def:rewritestepsyn} if and only if the carrier of the open hypergraph $\gint{c}$ contains the carrier of $\gint{l}$ as a convex sub-hypergraph.
\end{lemma}

Before turning attention to the proof of this lemma, there is one notion that needs clarification: convexity. 

\begin{definition}[Convex sub-hypergraph]
A \emph{path} $p$ from a hyperedge $e$ to a hyperedge $e'$ in a hypergraph $G = (N,E,v,l_n,l_e)$ is a sequence $e_1, \dots, e_n$ of hyperedges such that $e_1 = e$, $e_n = e'$, and $\pi_2(v(e_i)) \cap \pi_1(v(e_{i+1})) \neq \emptyset$ for all $i$ with $1 \leq i \leq n$. In other words, each hyperedge $e_i$ in the path is linked to $e_{i+1}$ via at least one node, which is outgoing of $e_i$ and ingoing to $e_{i+1}$.

	A sub-hypergraph $H$ of $G$ is \emph{convex} if, for any nodes $n, n'$ in $H$ and any path $p$ from $n$ to $n'$, every hyperedge in $p$ is also in $H$.
\end{definition}

Intuitively, convexity ensures that a sub-hypergraph $H$ of $G$ has no `gaps'. 

\begin{example}
	Consider the following open hypergraph, which interprets the string diagram $t$ from~\eqref{eq:rewcontext-stringdiag}, our motivating example. 
	\begin{equation*}
		\gint{t} = \vcenter{\scalebox{0.9}{\input{./pics/rewriting-pics/graphwithredex.tikz}}}
	\end{equation*}
	Following the example, our goal was to find a redex for the rule~\eqref{eq:rewrule-stringdiag} in $t$. This amounts to identifying the following \emph{convex} subgraph (highlighted in green):
	\begin{equation*}
		\scalebox{0.9}{\input{./pics/rewriting-pics/redexGreen.tikz}}
	\end{equation*}
	Note that, unlike on string diagrams, no reasoning modulo equivalence is required to identify the redex. We can reason directly on the carrier of $\gint{t}$. Here is an example of a subgraph (in green) of such carrier that is not convex, because it includes nodes labeled with $A$ and $B$ but not the hyperedge in between:
	\begin{equation*}
				\scalebox{0.9}{\input{./pics/rewriting-pics/nonconvexsubgraph.tikz}}
	\end{equation*}
\end{example}

We have thus established that, in order to find a redex for a rule in string diagram rewriting, it suffices to inspect the corresponding open hypergraphs. How do we rewrite with them though? There is a long tradition in graph rewriting of using ``double-pushout'' (DPO) constructions in the category of graphs in order to formally define  rewriting of these structures. It turns out we can sensibly define a notion of DPO rewriting in the category of open hypergraphs, called `convex DPO rewriting with interfaces' (CDPOI), and make it correspond precisely to string diagram rewriting.

\begin{theorem}
	Let $l \rew r$ be a string diagram rewrite rule, and $t$ a string diagram on the same signature. Then $t$ rewrites into $t'$ with rule $l \rew r$ in the sense of Definition~\ref{def:rewritestepsyn} if and only if $\gint{t}$ rewrites into $\gint{t'}$ with rule $\gint{l} \rew \gint{r}$ via CDPOI rewriting.
\end{theorem}

\begin{example}
	We give a sense of the theorem via an example. The rewriting step~\eqref{eq:rewstepstringdiag} is interpreted to the following CDPOI rewriting step. 
		\begin{equation*}
				\scalebox{0.7}{\input{./pics/rewriting-pics/dpoexample.tikz}}
	\end{equation*}
	The top row features the left-hand side of the rewrite rule (left), its right-hand side (right), and their interface (center). In the middle row we have $\gint{t}$ (left) and the result of rewriting it via the rule (right). At the centre, the pushout complement: $\gint{t}$ with a `hole' replacing the redex subgraph. The redex itself is identified by the morphism $F$. The bottom row is occupied by the interface: its preservation during rewriting is ensured by commutativity of the diagram --- in which, moreover, the two squares are pushouts.
\end{example}

There are a few subtleties involved in the theorem, namely on how CDPOI ensures soundness of hypergraph rewriting with respect to string diagram rewriting - a property which is not guaranteed by standard DPO rewriting. We do not elaborate on the precise definition of CDPOI rewriting, nor on the proof of the theorem, as that would bring us too far from the focus of this tutorial. The interested reader is referred to Section~\ref{sec:relatedwork-rewriting} below for pointers to the relevant literature.

\subsection{Rewriting of hierarchical string diagrams}

As seen in Section~\ref{sec:hsdcmc}, higher-order computation may be modelled by moving from string diagrams, living in SMCs, to hierarchical string diagrams, living in SCMCs. How does rewriting theory extend to such setting? The key step is to identify a generalisation of hypergraphs suitable for modelling the layered structure of hierarchical string diagrams. This leads to the notion of (open) \emph{hierarchical hypergraph}. Labels for these structures come from closed monoidal signatures, see Definition~\ref{def:freeSCMC}.

\begin{definition}[Open Hierarchical Hypergraph]
	Fix a closed monoidal signature $\Sigma = (\Sigma_0,\Sigma_1)$ of objects $\Sigma_0$ and morphisms $\Sigma_1$.  A \emph{($\Sigma$-labelled) hierarchical hypergraph} is a $\Sigma$-labelled hypergraph $(N,E,v,l'_n,l'_e)$ together with a pair $(p_E, p_V)$ of functions where 
\begin{itemize}
	\item $l'_n \colon N \to obj_{\Sigma_0}$ extends the usual node labelling function with the possibility of assigning arbitrary objects in $obj_{\Sigma_0}$, instead of just the ones in $\Sigma_0$.
	\item $l'_e \colon E \to \Sigma_1 \cup \{\bot \}$ extends the usual hyperedge labelling function with the possibility of assigning no label to a hyperedge --- i.e. the case $l'_e(f) = \bot$.
	\item $v \colon E \to N^{\star} \times N^{\star}$ is a function assigning to each hyperedge $e$ a list $\pi_1(v(e))$ of ingoing nodes and a list $\pi_2(v(e))$ of outogoing nodes
	\item $p_N \colon N \to E + \{\bot\}$ and $p_E \colon E \to E + \{\bot\}$ assign to each node and hyperedge a `parent' hyperedge, or no parent ($\bot$)
	\item $p_N$ and $p_E$ satisfy the following constraints: (i) for each $e \in E$, each ingoing and outgoing node of $e$ must have the same parent as $e$ itself;  (ii) the parent relation must be acyclic, in the sense that $(p_{E,\bot})^k (e) = \bot$ for some $k \geq 1$, where $p_{E,\bot} \colon E+1 \to E+1$ is just the extension of $p_E$ adding $p_{E,\bot} (\bot) = \bot$.
\end{itemize}
A \emph{morphism} between hierarchical hypergraphs $G = (N,E,v,l_n,l_e,p_E,p_V)$ and $G'= (N',E',v',l'_n,l'_e,p'_E,p'_V)$ is a hypergraph morphism $(f_N \colon N \to N', f_E \colon E \to E')$ respecting the hierarchical structure of $G$, in the following sense: 
\begin{eqnarray*}
	(p_{N} \circ f_N)(n) = (f_E \circ p'_N) (n) && \text{ if } p_{N}(n) \neq \bot \\
	(p_{E} \circ f_E)(e) = (f_E \circ p'_E) (e) && \text{ if } p_{E}(e) \neq \bot 
\end{eqnarray*}
    A ($\Sigma$-labelled) \emph{open} hierarchical hypergraph is a tuple $(L, G, R, f_L, f_R)$ where 
    \begin{itemize}
    	\item $G$ is a $\Sigma$-labelled hierarchical hypergraph, called the \emph{carrier}
    	\item $L$ is a discrete $\Sigma$-labelled hierarchical hypergraph, called the \emph{left interface} of $G$
    	\item $R$ is a discrete $\Sigma$-labelled hierarchical hypergraph, called the \emph{right interface} of $G$
    	\item $f_L \colon L \to G$ and $f_R \colon R \to G$ are hierarchical hypergraph morphisms.
    \end{itemize} 
   A \emph{open hierarchical hypergraph morphism} $h \colon (L,G,R, f_L, f_R) \to (L',G',R', f'_L, f'_R)$ consists of hierarchical hypergraph morphisms $h_G \colon G \to G'$, $h_L \colon L \to L'$, $h_R \colon R \to R'$ commuting with the interface morphisms, i.e. such that $h_L  \semic  f'_L = f_L  \semic  h_G$ and $h_R  \semic  f'_R = f_R  \semic  h_G$. %\footnote{Note that we do not require that outermost vertices and edges are sent to outermost vertices and edges. This is relevant for rewriting, since a matching would generally not satisfy such requirement.}
\end{definition}

In essence, a hierarchical hypergraph is a hypergraph with layers. In the above definition, layers are determined by the parent-child relation. The `outermost' layer $0$ is formed by those nodes and hyperedges that have no parent (i.e., their parent is $\bot$). Nodes and hyperedges in layer $n+1$ are those with parent hyperedge sitting in layer $n$.

Graphically, we may represent an open hierarchical hypergraph using `bubbles' to indicate layers, in a way that echoes the notation of hierarchical string diagrams. Here is an example.

\begin{equation}\label{eq:exhierarchicalhyp}
				\scalebox{0.9}{\input{./pics/rewriting-pics/exhiearchicalhyp.tikz}}
\end{equation}

 Note layer borders cannot `cross' links between nodes and hyperedges, because of constraint (i) in the definition. The open hierarchical hypergraph in~\eqref{eq:exhierarchicalhyp} has three hyperedges; two are labelled with $f$ and $g$ respectively, whereas the third is unlabelled (labelled with $\bot$). Note the unlabelled hyperedge has one ingoing node, labeled with $A$, and one outgoing node, labeled with $B \multimap B$. %--- strictly speaking, all hyperedges are labelled, but we leave unlabelled those labeled with the special symbol $\bot$. 
 As for layers, the unlabelled hyperedge and $f$-labeled hyperedge have no parent and thus sit in layer $0$ (they are outermost). The unlabelled hyperedge is parent of the hyperedge labelled $g$ (and its sources and targets), which thus sits in layer $1$. We signal the different layers by depicting the $g$-labelled hyperedge inside the unlabelled one: this notation echoes the bubble operation in hierarchical string diagrams. 
 
 As for the interfaces, in open hierarchical hypergraphs we distinguish an \emph{outer} interface, which consists of the nodes on the left/right interface that sit in layer $0$. However, contrarily to standard hypergraphs, they may also have \emph{inner} interfaces, given by the nodes on the interfaces that sit in layers $n > 0$.    In~\eqref{eq:exhierarchicalhyp} we use a vertical dotted line to separate the nodes of the layer $0$ interface from the nodes of the layer $1$ interface, and red/blue colouring for the interface maps to emphasise the different layers. In general, nodes of the inner interfaces are easily recognisable as they are the `dangling wires' of the hyperedges inside bubbles. In fact, for each unlabelled hyperedge $f$ we may distinguish an \emph{input interface}, given by nodes that have $f$ as parent and are in the left interface of the whole hypergraph, and an \emph{output interface}, given by nodes that have $f$ as parent and are in the right interface of the whole hypergraph. For instance, the input interface of the unlabelled hyperedge in~\eqref{eq:exhierarchicalhyp} is the list consisting of two nodes, labeled $A$ and $B$ respectively, and the output interface is the list consisting of a single node labelled $B$.

\medskip

The definition of hierarchical hypergraphs is agnostic on the set of objects labelling the nodes, as the only requirement is that they match the type of the morphisms labelling the hyperedges. However, in modelling higher-order computation, we want to restrict attention to those hierarchical hypergraphs where layers enforce scoping discipline. Roughly speaking, this means that, for instance, if a `bubble' receives inputs of type $A$, and the hyperedges inside it have ingoing nodes labeled with $A$, $B$, and outgoing nodes labeled with $C$, $E$, then the output of the bubble should be of type $B \multimap (C \otimes E)$. This scoping discipline leads to the notion of \emph{hypernet}.

\begin{definition}
	A hypernet is a open hierarchical hypergraph $(L, G, R, f_L, f_R)$, with $G = (N,E,v,l_n, l'_e)$ such that
	\begin{itemize}
		\item it is a monogamous open hypergraph when forgetting about the hierarchical structure.
		\item if $l'_e(f) \not\eq \bot$, then $f$ has no children (it is not parent of any node nor hyperedge).
		\item if $l_e(f) = \bot$, then the hyperedge $f$ is a well-typed abstraction. That is, there exists some $B \in \Sigma_0$ such that $f$ has $[A,B]$ as list of labels of nodes on the input interface, where $A$ is the label of the ingoing node of $f$\footnote{Note $f$ has exactly one ingoing and one outgoing node, since the hypergraph is assumed monogamous.}, and the outgoing node from $f$ is labeled with $B \multimap C$, where $C$ is the list of labels of nodes on the output interface of $f$.	\end{itemize}
\end{definition}

We are now in position to give a combinatorial interpretation of hierarchical string diagrams in terms of the structures we just introduced. { We take as starting point the SCMC $\category C$ freely generated by a closed monoidal signature $(\Sigma_0,\Sigma_1)$, as described in Definition~\ref{def:freeSCMC}. As discussed therein, objects of a SCMC have a more complex structure than objects in a SMC. Therefore, to properly interpret string diagrams in $\category C$, we label nodes with arbitrary objects of $\category C$, not just the generating ones. As far as morphisms are concerned, we extend the interpretation of Definition~\ref{def:hypint} with clauses following the free construction of $\category C$. That means, we include two extra clauses, for evaluation and for abstraction, as follows.

\begin{eqnarray*}
	eval_{X,A}  \colon ((X \multimap A) \otimes X) \to A & \mapsto & \raisebox{.5cm}{\hbox{\scalebox{0.8}{\input{./pics/rewriting-pics/evalint.tikz}}}}\\ \Lambda_X(h) \colon A \to (X \multimap Y) &\mapsto & \raisebox{.3cm}{\hbox{\scalebox{0.7}{\input{./pics/rewriting-pics/lambdaint.tikz}}}}
\end{eqnarray*}
Note this interpretation is only sound when considering $\category C$ as an SMC. Indeed, it equates any two string diagrams that are equivalent modulo the laws of SMCs, but does not respect the laws of SCMCs: the string diagrams of Figure~\ref{fig:adjclosed}, and those featuring in Definition~\ref{def:freeSCMC}, are mapped onto distinct hypernets. Even though the interpretation does not completely capture equivalence in SCMCs, this is intended as a feature rather a bug: analogously to reductions of $\lambda$-terms, the laws describing the behaviour of evaluation and abstraction have an operational meaning --- as discussed in Section~\ref{sec:asg}, they are akin to the $\eta$ and $\beta$ rules of the $\lambda$-calculus --- and thus their application should be tangible in diagrammatic reasoning. In other words, they should be treated as proper rewrite rules --- as opposed to the structural laws of string diagrammatic representation, which should be absorbed in the graph-theoretic interpretation. }

We conclude this section by sketching some consideration about hierarchical string diagram rewriting. With respect to CDPOI rewriting of hypergraphs (see Section~\ref{sec:rewSMC}), the situation is way more intricate, and not as well-behaved. A first subtlety is with the notion of matching: whereas in CDPOI rewriting we may match the left-hand side of a rewrite rule with a sub-hypergraph with the same interface, for hypernets we just need the \emph{outer} interface to be the same. Indeed, the inner interfaces only important to enforce the well-typedness of abstractions --- they are interfaces to the `bubbles' appearing in the hypernet, but do not play a role in how the hypernet interfaces with the context. 
\begin{figure*}
				\scalebox{0.7}{\input{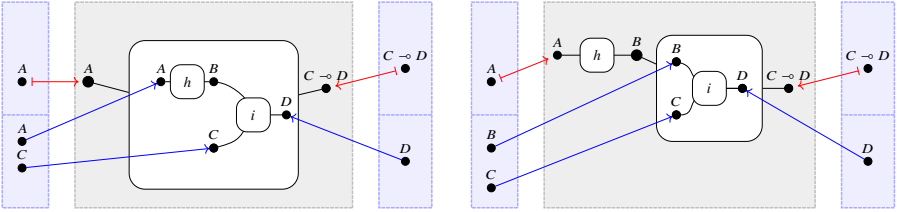}}
	\caption{Left and right-hand side of the `slide' rewrite rule, which intuitively allows to move an operation outside a bubble. Note the two hypernets have the same outer interface (labelled with $A$ on the left, $C \multimap D$ on the right), but different inner interfaces.}
\end{figure*}

A second, more problematic issue is that, unlike hypergraphs, pushouts of hypernets do not exist in general --- not even along monomorphisms. Intuitively, this is because hierarchical hypergraphs morphisms do not pose restrictions on which layer one may map the outer layer of a given graph. As a result, two embeddings of a hypergraphs into two parts of different hypergraphs may be unmergeable, because we may not be able to discern which should be the parents of outermost nodes and edges of the original hypergraph.  On the other hand, if we guaranteed the existence of pushouts by enforcing a stricter notion of morphism, for instance one sending outermost nodes and edges to outermost nodes and edges, then it would not be particularly useful, as rule matching typically does not satisfy such requirement --- intuitively, redexes should be allowed to appear inside bubbles, not just as sub-hypergraphs on layer $0$, as for example in Figure~\ref{fig:hypernetrew} below. 

Nonetheless, one may define a sensible variation of CDPOI rewriting which works for hypernets. The relevant pushouts exist under certain conditions, which rely on the specific shape of hypernet rewrite rules and properties of matching morphisms. There is a sound and complete correspondence between hypernet rewriting and hierarchical string diagram rewriting, but unlike for CDPOI hypergraph rewriting, it is not exact; instead, one rewrite step of hierarchical string diagrams may be simulated via multiple steps of hypernet rewriting, and vice versa. We refer to the discussion on related work below for pointers on this correspondence.
\begin{figure}
				\scalebox{0.55}{\input{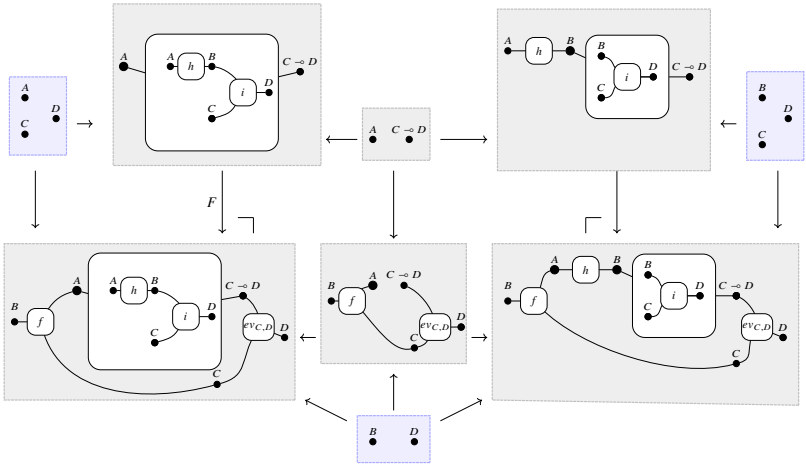}}
	\caption{An example of application of the slide rewrite rule. Note matching is given by a morphism $F$ that embeds the left-hand side of the rule into layer $1$ of the given hypernet.} \label{fig:hypernetrew}
\end{figure}

\subsection{Further reading and related work}\label{sec:relatedwork-rewriting}

For rewriting of string diagrams in symmetric monoidal categories, the most comprehensive reference is the series of works~\cite{DBLP:journals/jacm/BonchiGKSZ22,BonchiGKSZ22b,BonchiGKSZ22c}, which details the correspondence with DPO rewriting of open hypergraphs. We refer to this source for all the details which we omitted in this survey chapter. 

We mentioned how the variant that is both sound and complete for string diagram rewriting is \emph{convex} DPO-rewriting. We did not mention that, conversely, the standard notion of DPO-rewriting also has a natural string diagrammatic counterpart: it is sound and complete for so-called hypergraph categories, i.e. SMCs with a chosen Frobenius algebra structure on each object. Rewriting for string diagrams in categories with structure intermediate between symmetric monoidal and Frobenius has also been studied:~\cite{DBLP:journals/corr/abs-2204-04274} gives a sound and complete DPO-rewriting interpretation for SMCs with a commutative monoid structure on each object, and~\cite{GhicaKayeTracedComonoid} does the same for traced comonoid structure. For the case of monoidal closed categories, and rewriting of hierarchical string diagrams, we refer to~\cite{DBLP:conf/fscd/Alvarez-Picallo22}. All these work crucially rest on the fact that pushouts in categories of (hyper)graphs are well-behaved, and adapted to DPO rewriting. These features have been studied on a case-by-case basis for a long time, see e.g.~\cite{Corradinietal-graphtransf}, and at a later stage encapsulated more abstractly in the notion of \emph{adhesive category}, in~\cite{DBLP:conf/fossacs/LackS04}. The various notions of hypergraphs we introduced all form adhesive categories. For hierarchical hypergraphs, proving this is not at all trivial; it has been worked out only recently, in~\cite{DBLP:conf/fossacs/CastelnovoGM22}.

Whereas the literature on rewriting in categories that are \emph{at least} symmetric monoidal is quite rich, far less is known on rewriting in categories without symmetries. The case of monoidal categories has only been investigated with a `native' approach, in which the laws of monoidal categories are considered as additional rewrite rules, rather than structural --- see in particular~\cite{DBLP:conf/ac/Lafont93,lafontpeaks}, and~\cite{DBLP:journals/corr/Mimram14} for an overview on this approach. To the best of our knowledge, a graph interpretation for this class of string diagrams has not been studied yet, and the same is true for braided monoidal categories, which are intermediate between monoidal and symmetric monoidal. 

Together with the theoretical developments, various tools have been proposed to perform diagrammatic reasoning, which exploit the correspondence with DPO rewriting in its implementation. \textsc{Quantomatic}~\cite{kissinger_quantomatic:_2015} is perhaps the earliest example, which assumes that string diagrams may be manipulated as if they were in a compact-closed category --- this assumption is stronger than symmetric monoidal, but (slightly) weaker than hypergraph categories. The more recent \textsc{Cartographer}~\cite{DBLP:conf/calco/SobocinskiWZ19} works instead at the level of generality of string diagrams in symmetric monoidal categories. Last we mention~\textsc{Globular}~\cite{globular}, and the associated project~\textsc{Homotopy.io}, which is intended as a visual aid for reasoning in higher categories; (symmetric) monoidal categories are a special case, which however is treated without absorbing any structural rules --- all the laws of the category are regarded as rewrite rules. The tool landscape for hypernets is relatively less mature, including only one tool of note, \textsc{SD Visualiser}, aimed at rendering and exploring large diagrams but with no rewriting support at the moment\footnote{\url{https://sdvisualiser.github.io/sd-visualiser/}}.

\newpage

\section{Operational Semantics of $\lambda$ Calculi}
\label{sec:oslc}

Repeated applications of the $\beta$ rule may transform a term into another term where further $\beta$ reductions are impossible. 
Such irreducible terms are sometimes called \emph{normal forms} or \emph{values}, depending of context. 
The rewriting theory of the $\lambda$ calculus is \emph{confluent}, in the sense that the ultimate normal form, if any, is not sensitive to the order in which reductions are applied. 
The study of confluence is a key concern of the theory of the $\lambda$ calculus, but we shall not pursue it here. 

In contrast, when the $\lambda$ calculus is used as a \emph{programming language} the equations are not applied arbitrarily wherever possible, but according to a schedule, or \emph{evaluation strategy}. 
Recall that a \emph{redex} is a sub-term that can be rewritten in a way that is consistent with the left-hand-side of an equation. 
The \emph{evaluation strategy} is the way in which a term is repeatedly scanned for a redex, then rewritten according to the equation that matches the redex.
We call this step-by-step transformation of a term a \emph{reduction sequence}, and the system of rules that governs it an \emph{operational semantics}. 
Operational semantics can be specified in several ways, but here we will give what is in some sense the simplest presentation, the so-called \emph{big-step} operational semantics. 

Perhaps the most common evaluation strategy is the \emph{call-by-value} (CBV) strategy in which the argument of a function is evaluated before the function is applied. 
To define it we first define the concept of \emph{value} as a syntactic form:
\[
\mathcal V ::= x \mid \lambda x.u.
\]
We use $w\in\mathcal V$ to range over values. 
Evaluation of a term $u$ will stop if and only if it results in a value $w$,
\newcommand{\vto}{\Downarrow_v}
\newcommand{\nto}{\Downarrow_n}
written as $u\vto w$.
\begin{definition}[CBV $\lambda$ calculus]
\[
\frac{w\in\mathcal V}{w\vto w}\qquad
\frac{u\vto \lambda x.u'\quad v\vto w \quad u'[x/w]\vto w'}{u\,v\vto w'}.
\]
\end{definition}
Contrast this with the \emph{call-by-name} $\lambda$ calculus in which the argument of a function is not evaluated at application:
\begin{definition}[CBN $\lambda$ calculus]
\[
\frac{w\in\mathcal V}{w\nto w}\qquad
\frac{u\nto \lambda x.u'\quad u'[x/v]\vto w}{u\,v\vto w}.
\]
\end{definition}
This small change has vast repercussion over the behaviour of the $\lambda$ calculus used as a programming language, both in terms of time and space efficiency and equational properties. 
These matters are beyond the scope of this tutorial. 
\begin{example}\label{ex:idid3}
\begin{prooftree}
\AxiomC{$\lambda x.x\in\mathcal V$}
\UnaryInfC{$\lambda x.x\vto \lambda x.x$}
\AxiomC{$\lambda y.y\in\mathcal V$}
\UnaryInfC{$\lambda y.y\vto \lambda y.y$}
\AxiomC{$\lambda y.y\in\mathcal V$}
\UnaryInfC{$\lambda y.y\vto\lambda y.y$}
\doubleLine
\UnaryInfC{$x[x/\lambda y.y] \vto \lambda y.y$}
\TrinaryInfC{$(\lambda x.x)(\lambda y.y) \vto \lambda y.y$}
\end{prooftree}
The double line indicates that the terms involved in that relation are syntactically equal (up to $\alpha$ renaming), namely $x[x/\lambda y.y]=\lambda y.y$.
\end{example}

The concept of \emph{confluence} mentioned earlier is reflected by the following property:
\begin{theorem}[Confluence]
If $u$ is a closed term of STLC and $u\vto w$, $u\nto w'$, then $w=w'$. 
\end{theorem}
The proof is beyond the scope of this tutorial.

However, for untyped $\lambda$ calculi the confluence theorem does not hold. 
\begin{example}\label{ex:omega}
Consider the term $\omega:=(\lambda x.xx)(\lambda y.yy)$ and its CBV evaluation:

\begin{prooftree}
\AxiomC{$\lambda x.xx\in\mathcal V$}
\UnaryInfC{$\lambda x.xx\vto \lambda x.xx$}
\AxiomC{$\lambda y.yy\in\mathcal V$}
\UnaryInfC{$\lambda y.yy\vto \lambda y.yy$}
\AxiomC{\vdots}
\UnaryInfC{$(\lambda x.xx)(\lambda y.yy)\vto?$}
\doubleLine
\UnaryInfC{$(\lambda y.yy)(\lambda y.yy)\vto?$}
\doubleLine
\UnaryInfC{$(xx)[x/\lambda y.yy]\vto ?$}
\TrinaryInfC{$(\lambda x.xx)(\lambda y.yy)\vto?$}
\end{prooftree}
Note that the evaluation runs into some kind of a circular reference in which evaluating $\omega$ requires evaluating $\omega$ itself,
so the derivation tree cannot be finite. 
Using the CBN strategy would run into the same problem. 
\end{example}
These kind of terms that cannot be evaluated are said to be \emph{divergent}.

\begin{exercise}
CBV and CBN have distinct divergence properties. 
Check that $(\lambda x.y)\omega\nto y$ whereas there is no $w\in \mathcal V$ such that $(\lambda x.y)\omega\vto w$. 
\end{exercise}

The big-step style of presenting the operational semantics obscures to some extent the computational processes involved in the evaluation. 
This is particularly salient in the case of the divergent term in which we can only exhibit a partially constructed derivation tree. 
This partial construction indicates implicitly that in our attempt to construct the tree we expand certain branches of the tree first, using the results to then expand other branches. 
These computational processes are made more explicit in other styles of operational semantics, namely \emph{small-step} or \emph{abstract machine} operational semantics. 
We will give an operational semantics using string diagrams which as we shall see makes the computational aspects completely explicit. 
This style of semantics would allow a meaningful interpretation of divergent but useful programs, for example servers that endlessly wait for and handle connections. 

\subsection{String diagrams and operational semantics for CBV}
\label{def:cbvsd}
The presentation of operational semantics will be given in the informal but rigorous style we prefer, noting that the formalisation as graph rewriting in the previous section remains valid. 

We start by adding a \emph{decoration}, or a pointer, to a selected wire in a string diagram, which explicates the process of searching for a redex. 
The left-pointing triangle indicates a part of the diagram which is about to be evaluated, whereas the right-pointing triangle points away from a part of the diagram that was just evaluated. 
So the string diagram operational semantics will give rules for moving the pointer around in the diagram, plus the application of rewrites corresponding to the scheduled application of the particular equations:
\begin{definition}[CBV operational semantics in string diagrams]
\label{def:cbvsem}
\[
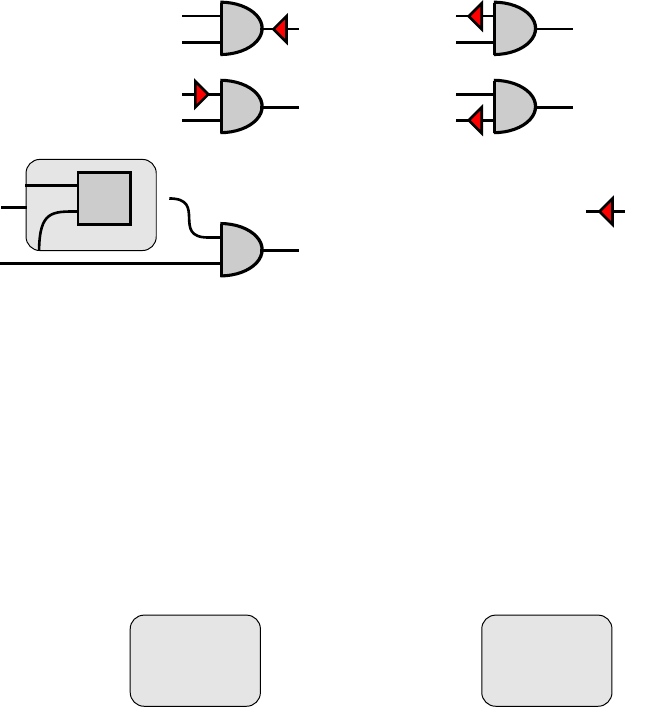
\]
Initially the diagram is annotated with a left-pointing triangle on the rightmost wire. 
The evaluation succeeds if a right-pointing triangle is present on the rightmost-wire.
\end{definition}
The rules can be read informally as follows:
\begin{description}
\item[Structural rule 1 (S1)] To evaluate an application, first evaluate the function.
\item[Structural rule 2 (S2)] After evaluating a function, evaluate the argument. 
\item[Beta rule ($\beta$))] After evaluating the argument apply the $\beta$ and evaluate the result. 
Note the rule here uses a more economical version of $\beta$ in which the argument is not actually involved. 
It is rather similar to the cancellation of the evaluation and co-evaluation in the adjunction. 
Also note the change of direction in the pointer, to re-evaluate the result of this rule. 
\item[Copying 1 (C1)] When encountering a copy node copy the node it connects to, in effect using the naturality of the Cartesian product. 
\item[Copying 2 (C2)] Just like the above, but from the other side of the copying node. 
\item[Value (V)] An abstraction is a value. 
Note the change of direction in the pointer. 
\end{description}

\begin{remark}
Something that is remarkable but easy to overlook is that the transition rules in Definition~\ref{def:cbvsd} are \emph{small}, i.e. they only involve local rewrites of the string diagram. 
This is obviously the case for the $\beta$ rule in which substitution is not really invoked, but rather essentially consists only of a cancellation of the evaluation/co-evaluation pair of adjunctions. 
The actual substitution is effected by the evaluation pointer reaching copy nodes, which trigger the replication of diagrams but in a controlled, stepwise fashion. 
This is a clearer execution model in terms of understanding the time and space costs involved in the evaluation, compared to relying on the meta-syntactic use of substitution. 
Also local are the structural rules, which show the way the evaluation pointer interacts with the evaluation operation independent of its arguments. 
\end{remark}

\begin{example}
In Figure~\ref{fig:exopsem} we reprise Example~\ref{ex:idid3} in the diagrammatic formulation. 
The dummy transitions labelled with $=$ are just re-drawings of the diagram, to emphasise a redex or just to tidy up. 
\begin{figure}
\[
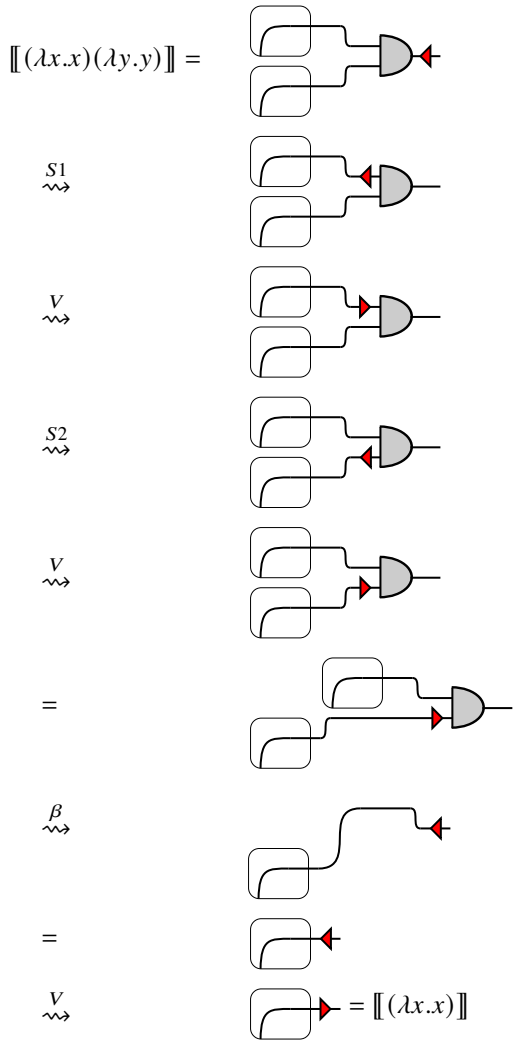
\]
\caption{$(\lambda x.x)(\lambda y.y)$}
\label{fig:exopsem}
\end{figure}
\end{example}
\begin{example}
If Figure~\ref{fig:omega2} we reprise Example~\ref{ex:omega} using string diagrams, in which we can clearly see now that the transition system has a cycle. 
\begin{figure}
\[
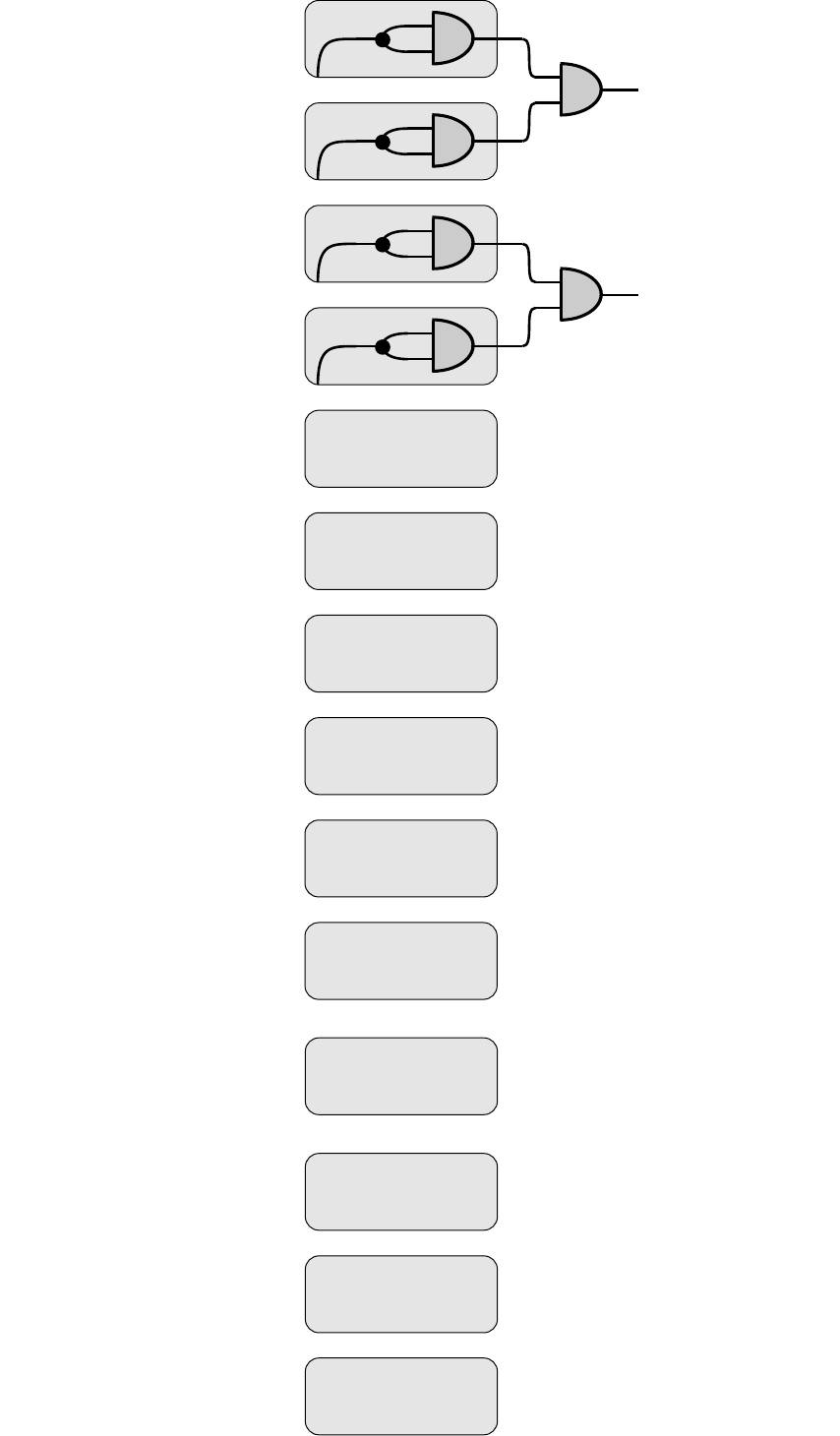
\]
\caption{Divergent evaluation}
\label{fig:omega2}
\end{figure}
\end{example}

\begin{remark}
The token used to enforce a certain evaluation strategy for the graph does not appear in the Section~\ref{sec:graphs}, but it can be handled within the graph-rewriting framework as a special hyperedge. 
However, this hyperedge should not be considered a `box' in the language of string diagrams, i.e. a morphism (or a family of morphisms) in the category, because it is not part of the language syntax. 
Certain diagrammatic equations involving the pointer should hold, for instance 
\[
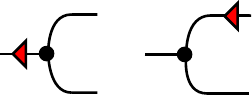
\]
However, such diagrams are ill formed, illegal configurations. 
This is not problematic beyond the obvious consequence that the strategic aspects of the operational semantics are not given a categorical interpretation. 
\end{remark}

\subsubsection{Calculi of explicit substitutions}
Substitution plays a critical role in the syntax and equational theory of the $\lambda$ calculus. 
At the same time substitution is a \emph{meta-syntactic} construct, which is makes formalisation more difficult, especially when it comes to understanding the computational cost of substitution.
The substitution rules in particular as a meta-operation can have a widely varying and complex cost involving the identification of variables, $\alpha$-renamings, and copying of terms of various sizes. 

Calculi of explicit substitution (ES) deal with this problem by making substitution a term former in the language, so that
\[
u::=x\mid u\,u\mid \lambda x.u \mid u[x/u].
\]
This formulation may raise some subtle syntactic issue regarding scope, binding, and alpha equivalence which are either swept under the carpet as bureaucratic annoyances or resolved by using DeBruijn indices instead of variable names. 
The latter makes syntax easier to formalise, but the downside of using indices is, of course, that the syntax becomes utterly unreadable. 

The rules of substitution are made into explicit language equations:
\begin{definition}[ES equations]
\label{def:eseqn}
Assuming $x\neq y$, 
\begin{align*}
x[x/u]  &= u \\
u[x/v] &= u &\text{ if }x\not\in\mathcal F(u) \\
(uu')[x/v] &= (u[x/v])(u'[x/v]) \\
(\lambda x.u)[y/v] &= \lambda x.u[y/v] & \text{ if } x\not\in\mathcal F(v)\\
\bigl(u[x/v]\bigr)[x'/v'] &= \bigl(u[x'/v']\bigr)\bigl[x/v[x'/v'] \bigr] & \text{ if } x'\in\mathcal F(v)\land x\in  \mathcal F(v')\\
\bigl(u[x/v]\bigr)[x'/v'] &= u\bigl[x/v[x'/v'] \bigr] & \text{ if } x'\in\mathcal F(v)\land x\not\in \mathcal F(v')\\
\bigl(u[x/v]\bigr)[x'/v'] &= \bigl(u[x'/v']\bigr)[x/v] & \text{ if } x'\not\in\mathcal F(v)\land x\not\in \mathcal F(v') 
\end{align*} 
\end{definition}

With the exception of the last equation, all the above are directed (left-to-right) and can immediately form the basis of a rewrite system. 
The last equation (dubbed the $C$-rule) is problematic computationally because it is undirected, so it has the same status as the $\alpha$ rule. 
Indeed, equivalence of terms is up to these two equations ($\alpha$ and $C$). 

We have seen that one of the key advantages of the string diagram notation is that it absorbs certain equations into isomorphic graphs. 
This is also the case for the $C$-rule;  using the interpretation of substitution (Lemma~\ref{lem:subst}) we can see that the left-hand side, $\bigl(u[x/v]\bigr)[x'/v']$ is, after cancelling out some copy-discard pairs, isomorphic to the graph below, which is in turn isomorphic to the right-hand side. 

\[
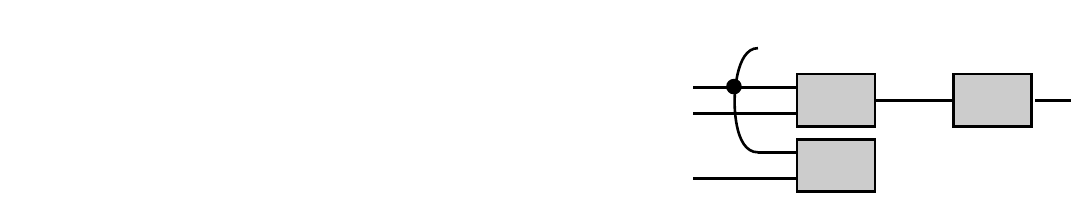
\]
\begin{exercise}
Prove the correctness of ES equations in the string diagram interpretation. 
\end{exercise}
\begin{remark}
The significant notational simplification introduced by the use of string diagrams in the calculus of explicit substitutions does not come at the expense of informality. 
It is true that the \emph{presentation} here is informal, for pedagogical reasons, but the formally minded reader has already seen how string diagrams can be presented as precisely defined combinatorial structures. 
Eliminating not one, but two quotient rules from the formal representation of the calculus is a significant notational advantage. 
\end{remark}

\subsection{Applied $\lambda$ calculi}

The $\lambda$ calculus provides the infrastructure upon which programming languages are built. 
Function calls represent the ``mortar'' that binds together the software ``bricks'' that constitute a program: the operations. 
Enhancing the $\lambda$ calculus with operations leads to an \emph{applied} $\lambda$ calculus, an idealised programming language. 
Practical programming languages represent the final elaboration, with numerous and complex operations and a more sophisticated syntax. 
Let us review several common such operations, their defining equations, and their operational semantics. 

\subsubsection{Arithmetic and logic}

\newcommand{\Add}{\mathit{add}}

Here and henceforth we shall treat matters of superficial syntax with a certain degree of informality which should not cause confusion. 
Whenever convenient we use human-readable syntax common in most languages (e.g. $1+2$) but when we want to avoid ambiguities we will switch to a de-sugared notation (e.g. $\Add\, (1, 2)$) or a categorical representation (e.g. $(1\otimes 2) \semic \Add$).

\newcommand{\Num}{\mathit{Num}}
\newcommand{\Float}{\mathit{Float}}
\newcommand{\Bool}{\mathit{Bool}}
\newcommand{\Neg}{\mathit{neg}}
\newcommand{\Op}{\mathit{op}}

Adding arithmetic and logic requires:
\begin{itemize}
\item one (or more) numerical type $\Num$ (e.g. $\Float$ or $\Int$ or $\mathit{Bool}$), 
\item numerical constants $m,n,p,\ldots:I\to \Num$ (e.g. $7:I\to \Int$ or $4.5:I\to \Float$), 
\item unary operators $\Op : \Num\to \Num$ (e.g. $\Neg: \Int\to \Int$),
\item arithmetic operators $\Op : \Num\times \Num\to \Num$ (e.g. $\Add:\Int\times \Int\to \Int$),
\item logical operators $\Op:\Num\times\Num\to\mathit{Bool}$ (e.g. $\mathit{less\_than}:\Int\times\Int\to\mathit{Bool}$),
\item conversion operations $\mathit{conv}:\Num\to\Num'$ (e.g. $\mathit{floor}:\Float\to\Int$), etc. 
\end{itemize}

The equations governing these operations are the obvious ones, and they are obviously directed (e.g. $1+1=2$) so they form a natural basis for an operational semantics. 
For binary operations, the structural rules $S1$ and $S2$ in Definition~\ref{def:cbvsd}, which give the order of evaluation (right-to-left) can be reused.
The value rule $V$ in the same Definition can also be reused, as numerical constants are values and require no further evaluation.
Finally, each operation has an associated reduction rule (sometimes called a $\delta$ rule):

\[
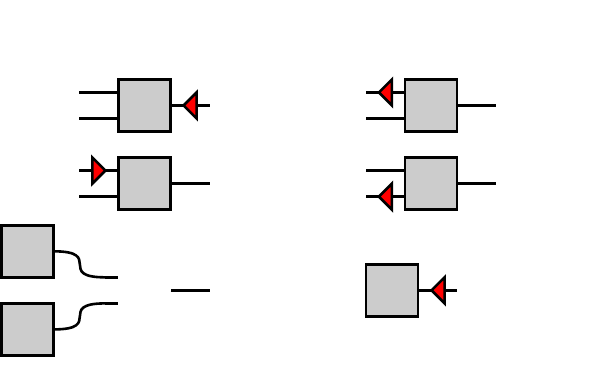
\]

\begin{remark}
It is obvious that each $\delta$ rule, which always produce a value, can be immediately combined with the value rule $V$ to give the composite rule
\[
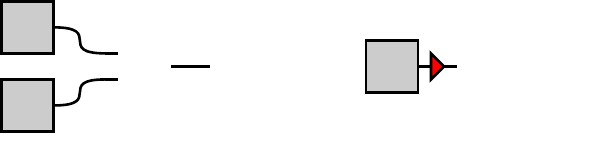
\]
We decline to perform this `optimisation' for the sake of keeping uniform the interaction between structural rules, which show how the pointer traverses the graph without changing it, with the rewrite rules which update the graph. 
Namely, a rewrite rule will always produce a pointer aiming to the left, back into the recently produced term, as rewrites do not generally produce values as results (e.g. $\beta$ reduction). 
\end{remark}
The following is left to the reader. 
\begin{exercise}
Define appropriate string-diagram operational semantics for unary and conversion operators. 
\end{exercise}
\subsubsection{If-then-else and other shortcut operators}

\newcommand{\Ite}{\mathit{ite}}

A naive if-then-else operator can be defined as $\Ite:\Bool\times\Num\times\Num\to\Num$, with equations
\begin{gather*}
\Ite(\mathit{true},v,v') =v\\
\Ite(\mathit{false},v,v') =v'.
\end{gather*}
But this does not give enough hints as to how the if-then-else operator is to be evaluated. 
An obvious extension of the structural rules $S1$ and $S2$ to rules for ternary operators is not good enough, because it would give the following undesired equalities, where $\omega$ is any diverging term:
\begin{gather*}
\Ite(\mathit{true},v,\omega) =\omega\\
\Ite(\mathit{false},\omega,v') =\omega,
\end{gather*}
thus violating our intended equations for the special case of $\omega$. 

There is one way the evaluation of the arguments after that of test expression can be prevented, which is the same way the evaluation of bodies of functions is prevented: thunking them. 
The thunks of the branches of an if-then-else statement take no arguments, so they are of the shape $\Lambda_I(u)$. 
Thus, an if-then-else statement written as $\mathit{if\ b\ then\ u\ else\ v}$ in the surface language is interpreted as
\[
\Ite(b,\Lambda_I(t), \Lambda_I(u)),
\]
noting that our problematic evaluations are not so anymore, with $\Lambda_I(\omega)$ being a value rather than a divergent computation. 
So the two equations are:
\begin{gather*}
\Ite(\mathit{true},\Lambda_I(t),\Lambda_I(u)) = t\\
\Ite(\mathit{false},\Lambda_I(t),\Lambda_I(u)) = u.
\end{gather*}
With this interpretation in place, structural rules $S1$ and $S2$, along with an obvious new rule $S3$, can be extended to handle ternary operators. 
The $V$ rule for thunks is still used as before, to prevent the evaluation of $t$ and $u$. 
The sole genuine extension are the two reduction rules for the if-then-else, shown in their graphical form, and involving both a selection of the thunk and its `forcing':
\[
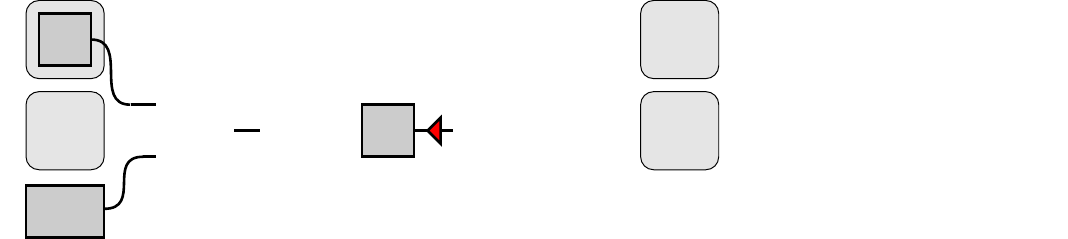
\]
\begin{remark}
As we move gradually from a \emph{mathematical} to a \emph{computational} model, various details that are mathematically irrelevant become meaningful from the point of view of cost of execution. 
For instance, in the rewrite of the if-then-else operation potentially large sub-graphs `disappear'. 
The reader is left to compare, in terms of possible cost model (time and space) the rules given above with rules which create inaccessible, from the point of view of evaluation, sub-graphs, usually called \emph{garbage} (only one of the rules shown):
\[
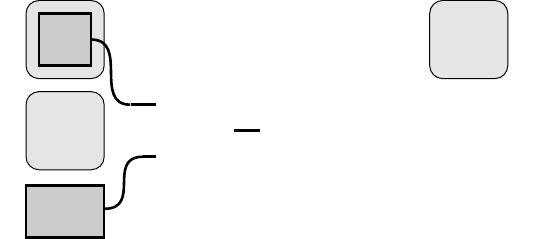
\]
\end{remark}
\begin{exercise}
Define and compare string-diagram semantics for ``lazy'' (for instance `$\&\&$' in Java) and ``eager'' (for instance `$\&$' in Java) boolean conjunction operators. 
\end{exercise}

\subsubsection{Recursion}

It is straightforward to include a standard recursion operator, with the rule
\[
\mathit{rec} (\lambda f.u) = u\bigr[f/\mathit{rec}(\lambda f.u)\bigl]
\]
noting that this is an ``expansion'' rather than a ``reduction'' rule. 

The recursion operator has a single argument, which is a thunk so it will use the unary version of the structural rule along with the rewrite:
\[
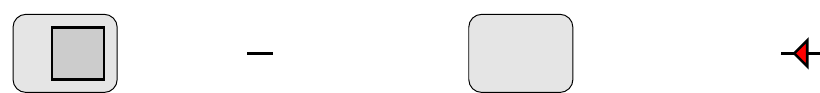
\]

\begin{remark}
The reader acquainted more casually with programming languages with recursion may feel slightly confused by the use of recursion as an operator, rather than as a binding construct. 
But the conventional recursive binding found in common programming languages can be easily \emph{interpreted} with the help of a recursion operator. 
A definition such as 
\[
f(x) = \cdots f(\cdots x\cdots)\cdots
\]
is interpreted as
\[
\mathit{rec}(\lambda f\lambda x.\cdots f(\cdots x\cdots)\cdots)
\]
\end{remark}

\subsubsection{Abortive continuations}

Prior to recursion, all operations we consider had ``reductive'' rules in which the result is smaller, as a graph, than the operands. 
A graph could still grow because of copying, but not because of the application of an operation. 
Recursion, in contrast, expands the graph via the application of an operation. 

However, something that all operations to this point share is that the result of their application does not create new thunks. 
But now we will see just this, a thunk-creating operation called ``\emph{call/cc}''. 
This operation is a very powerful control operator that can be used to implement a wide range of simpler and more common control operations, such as exceptions. 
To be effective it works in tandem with another operation, \emph{abort}. 

Another distinction between the \emph{call/cc} rules and the rules seen so far is that the previous rules are ``small'' in the sense that they only involve a part of the graph that is situated to the left of the pointer, whereas the \emph{call/cc} rules involve not only the operands, but also the \emph{context} in which the  operation is executed. 

The rules are as below:
\[
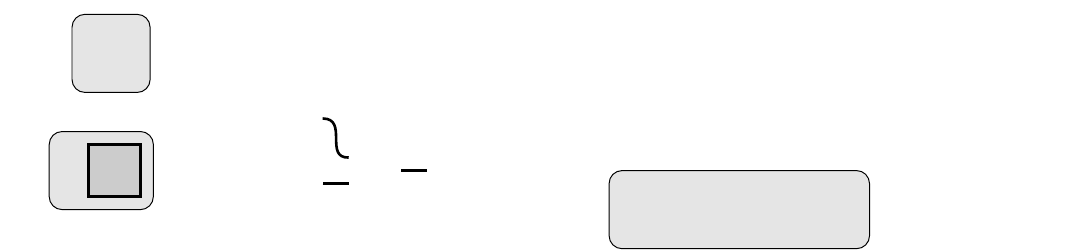
\]

The rules are global, in the sense that the sub-graph $t$, in both rules, is such that its root is the overall root of the graph. 
Such global rules are sometimes called \emph{program reductions}, in contrast with the more common, local, \emph{term reductions} seen earlier. 
The sub-graph term is usually called \emph{the context} of the evaluation and both rules manipulate it. 
The \emph{abort} rule is so called because it ``aborts'' its execution context $t$ and performs is argument instead. 
The \emph{call/cc} rule makes the context into a so-called ``abortive continuation'' creating a thunk in which the context $t$ is ``guarded'' by an \emph{abort} operation. 
Whenever this thunk is executed, the abort will cause the context at the time to be discarded and perform the continuation instead. 

\begin{remark}
Other control operations such as exceptions, delimited continuations, or effect handlers can be implemented in a similar style. 
These are left as an exercise. 
\end{remark}

\subsubsection{Store}

A common feature of programming languages is \emph{store}, also known as \emph{mutable variables} (variables which actually vary). 
One of the earliest systematic attempts at a language with higher order functions and store is \textsc{Algol60}. 
The extensions consist of a new type $Var$ for mutable variables along with the following operations
\newcommand{\newvar}{\mathit{newvar}}
\newcommand{\deref}{\mathit{deref}}
\newcommand{\assign}{\mathit{assign}}
\newcommand{\Var}{\mathit{Var}}
\begin{align*}
\newvar & : (\Var\to \Int)\to \Int & (\text{new variable declaration})\\
\deref & : \Var \to \Int & (\text{dereferencing})\\
\assign & : \Var\to \Int\to \Int & (\text{assignment})
\end{align*}
\begin{example}\label{ex:incx}
A term that in some conventional concrete syntax would be written as 
\begin{verbatim}
  var x ;  x = x + 1
\end{verbatim}
would be desugared as
\[
\newvar(\lambda x.\assign\ x\ (add\ (\deref\ x)\ 1))
\]
\end{example}
The configuration of the term reduction semantics, which has shape $t\to t'$, is expanded with an additional component $s$, a partial function $\mathcal L\rightharpoonup\mathcal V$ where $\mathcal L$ is a set of \emph{locations} (also known as \emph{references} or \emph{atoms}) and $\mathcal V$ is the set of values. 

Describing the behaviour of store equationally is possible but too complicated to delve into here and also rather unusual. 
The standard approach is to use \emph{operational semantics}, a reduction semantics in which configurations consists of terms and \emph{stores}, i.e. partial maps from a new data type of \emph{locations} (or \emph{references}) and values, $s:\mathcal L\rightharpoonup \mathcal V$.  

We will first give a string diagram semantics which is an instance of what we have seen so far, then enhance it for readability with additional graphical conventions.

The rule for new variable declaration is
\[
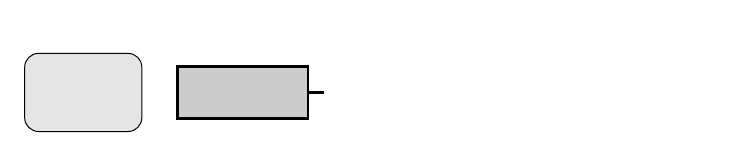
\]
where $a$ is \emph{fresh}, i.e. it is not used anywhere in the diagram and $v_0$ is a default initial value. 
\begin{remark}
The choice of $a$ is irrelevant, so long as it is some fresh value. 
No program construct can distinguish between two distinct choices of atoms $a\neq a'$, a property known as \emph{equivariance}. 
The data structure of locations is extremely simple, for which reason its elements are called \emph{atoms}. 
In this conceptual \emph{API} there are only two operations: creating fresh atoms and testing two atoms for equality. 
This data structure has been extensively studied within the framework of \emph{nominal sets}.
Equivariance is an essential property of more complex structures containing such atoms. 
\end{remark}
The rule for dereferencing is obvious, extracting the value $v$ associated with the dereferencing location $a$:
\[
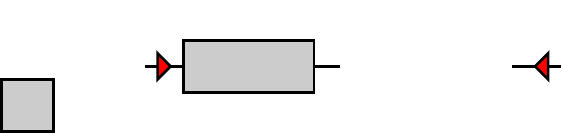
\]
The rule for assignment is more unusual because, as in the case of control, it is \emph{non-local} relative to the structure of the overall diagram:
\[
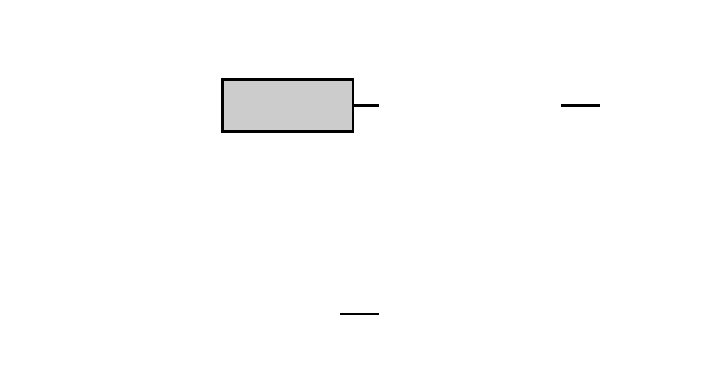
\]
The first rule states that evaluating the assignment operation returns a default value, not necessarily the same as the default initialisation value (in most languages it is the unique inhabitant of the singleton type). 

The second rule states that \emph{every subgraph} in which $v$ is associated with $a$ is rewritten so that $v$ is replaced by $v'$. 
As given, the rule is correct in terms of input-output behaviour of a programming language with store, but extremely inefficient. 
Every assignment requires searching for the relevant location $a$ in the diagram so that its corresponding value can be updated. 

We can improve our graphical notation with the following convention: identify all nodes labelled by the same atom $a$ in the diagram, i.e keep all the occurrences of the subgraph mapping the same $a$ to $v$ shared via contraction. 
This is possible because the category is still Cartesian and because it is an invariant of the diagram that the same $a$ is always mapped to the same $v$. 
\[
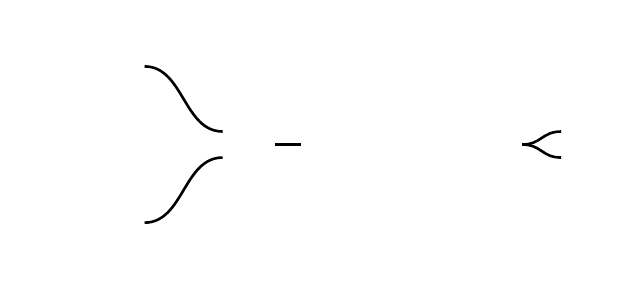
\]
Since the identity of the atom $a$ is irrelevant due to equivariance we can elide it, as we can elide the $\mapsto$ operation. 
What we are left with is:
\[
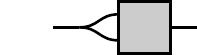
\]
With this streamlined notation the rewrite rules are:
\[
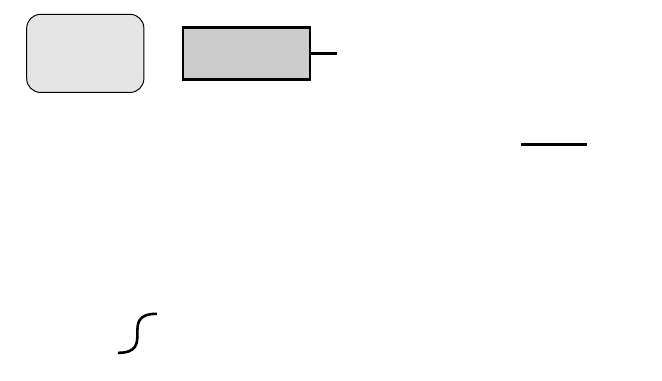
\]
\begin{exercise}
Prove the invariant that all occurrences of $a$ are associated with the same $v$ in a graph. 
\end{exercise}
\begin{example}
The key stages in the evaluation of the term in Example~\ref{ex:incx} are:
\[
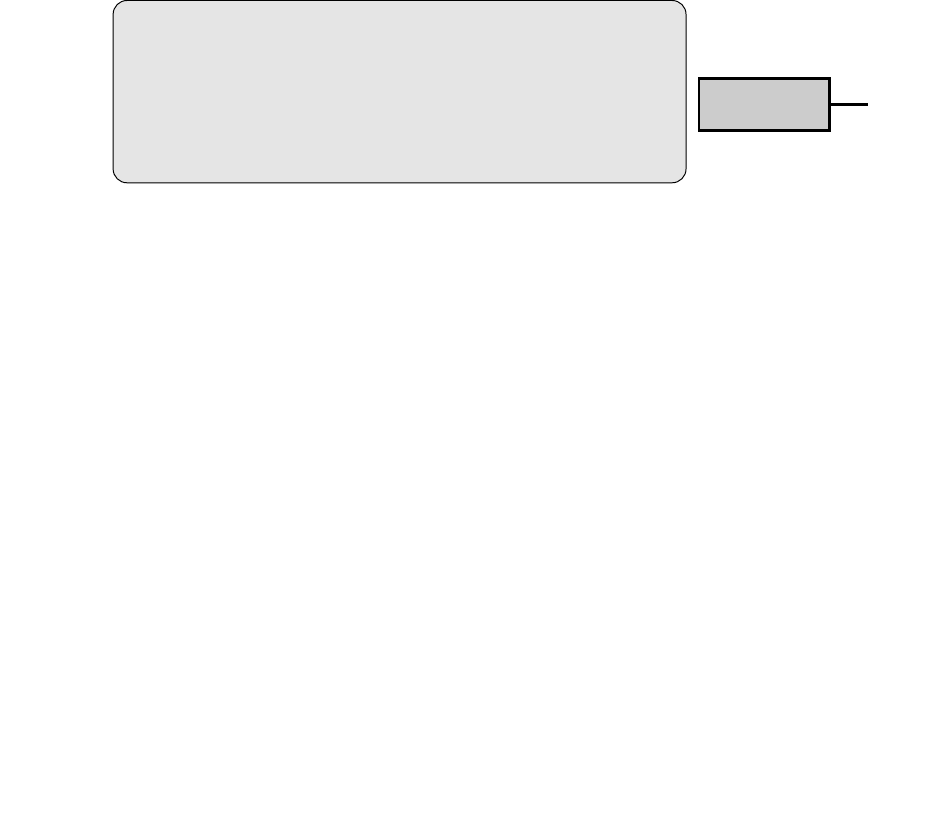
\]
For the purpose of this example we have assumed a default initialisation value of 0 and an assignment statement that returns the assigned value, in the style of the programming language C. 
\end{example}

\begin{remark}
The streamlined rule for assignment and dereferencing may seem not to apply, and consequently require a special case, when the atom is only reachable via a single edge in the diagram. 
In this situation the graph can be rewritten, using the properties of copy and discard to an equal graph which fits the required format, so it can be reduced:
\[
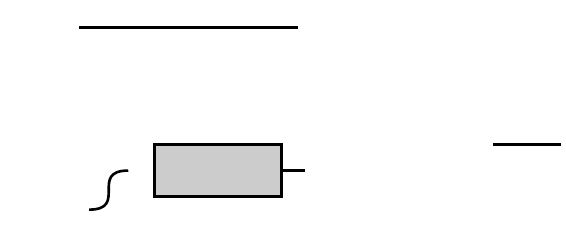
\]
Also using the properties of copy and discard the subgraph linked to the discard node, consisting of the atom and an occurrence of $v$ can be entirely removed. 
From an operational perspective this quiet automatic removal of a reference and its contents whenever it is unreachable from the main graph is an abstract counterpart of the so-called \emph{garbage collection} performed by runtime systems of programming languages with managed memory. 
\end{remark}

\subsection{Further reading and related work}
The material presented here is mostly based on the graph rewriting abstract machines of \cite{DBLP:journals/lmcs/MuroyaG19} and \cite{DBLP:journals/corr/abs-1907-01257}. 
An online interactive tool can help explore the step-by-step reduction of programs with many of the features discussed in this section\footnote{\url{https://tnttodda.github.io/Spartan-Visualiser/}}.
For the wider context of operational semantics there are many introductory tutorials, from the concise \cite{DBLP:conf/ac/Pitts00} to the extensive \cite{harper16}. 
The world of abstract machines in programming languages is also extensive and diverse, with an extensive annotated bibliography given in \cite{DBLP:journals/fgcs/DiehlHS00a}. 
The connection between operational semantics and abstract machines is quite technical, especially when the latter is meant to be extracted from the former, rather than proven \emph{post facto} as a correct implementation; a collection of such extractions is given in \cite{DBLP:conf/icfp/AccattoliBM14}.

We have only touched onto the issue of strategic graph rewriting, and in a naive way. More sophisticated approaches to strategic graph rewriting have been studied in the context of graph-oriented programming languages~\cite{DBLP:conf/lopstr/FernandezKN11}.

The calculus of explicit substitutions was introduced by \cite{DBLP:journals/jfp/AbadiCCL91} and it spawned a large literature. 
Of particular relevance are the connections with graph representations via proof nets \cite{DBLP:journals/tcs/Accattoli15} and interaction nets \cite{DBLP:conf/tlca/Sinot05}.

\newpage

\section{Case Studies}
\label{sec:apps}

String diagrams are a data structure for the representation of syntax with higher-order bindings. 
The best motivation for a new data structure are the algorithms it can enable and support. 
Case in point, we will consider three well-known such classes of algorithms: type inference, closure conversion, and reverse automatic differentiation. 
We choose these algorithms because in their original, term-based formulation they can be rather mysterious to the average compiler developer, who may struggle to tease out the intuition from what is usually a complex syntactic analysis or transformation. 

\subsection{Type inference}

Type inference is the process of determining whether a term in which the variables have not been assigned a type can be assigned a type. 
To give a trivial example, the term $x+1$ can be successfully assigned an integer type, which makes the term correctly typed presuming that 1 is itself an integer and addition requires two integer parameters. A type inference algorithm in this case will determine first that it is possible to assign a type to the variable and, if needed, produce that type. 
As a consequence, the variable $x$ occurring in the term will also be required to have an integer type. 

On the contrary, in the term $\mathit{if}\ x\ \mathit{then}\ x+1\ \mathit{else}\ 0$ it is not possible to assign a type to $x$ since the if statement demands it be a Boolean, whereas addition requires it to be an integer. 
Both these constraints cannot be satisfied simultaneously in a simple type system. 

As we can glean from these very simple examples, a type inference algorithms must perform three kinds of computations:
\begin{enumerate}
\item build a system of constraints determined by the syntactic structure of the term and the known types of operations or constants occurring in the term;  
\item check whether the system is consistent, i.e. there are no constraints that are not solvable; 
\item determine the missing types of variables if required. 
\end{enumerate}
The first step is usually computationally cheap, requiring a single pass over the term. 
The second step can be computationally expensive, asymptotically exponential time for a typical algorithm, although this worst case scenario is highly contrived. 
However, the potentially high complexity of the algorithm indicates that there is plenty of room for optimising heuristics that can make a difference in practice. 
And indeed this is where various type inference algorithms can take different approaches. 
Another source of significant optimisation is only computing the precise types of variables that need to be computed, avoiding potentially expensive type determinations if they are not strictly required. 

In this expository presentation of type inference using string diagrams we will not commit to any particular algorithm but we shall expose the essence of type inference so that it should be rather obvious how various heuristic choices can lead to various concrete algorithms. 

The first important idea is that when representing syntax as string diagrams type inference becomes an \emph{edge-labelling problem}. 
To stay focussed we will illustrate this algorithm for the concrete case of (unrecursive) PCF, i.e. $\lambda$ calculus plus if-then-else, and arithmetic-logic operators. 
The language of types is given by the grammar
\define{\bools}{\mathit{Bool}}
\define{\ints}{\mathit{Int}}
\[
T ::= \bools  \mid \ints \mid T\rightarrow T \mid \alpha
\]
i.e. Booleans, natural numbers, function types, and unknown or indeterminate types. 
Concrete variables for indeterminate types are written as (decorated) Greek alphabet letters. 

%https://en.wikipedia.org/wiki/Unification_(computer_science)#A_unification_algorithm

Before talking about type inference we need to briefly introduce the concept of \emph{unification}, which is the system of syntactic constraints that needs to be solved during inference. 
Generally speaking, unification is an algorithm that determines whether the variables of two expressions can be instantiated so that the expressions are equal. 
The theory of unification is out of the scope of this tutorial so we will rely on naive intuitions, which hold for this simple language of types. 

Let $\Gamma$ be an \emph{environment} mapping type variables to types. 
\begin{enumerate}
\item $\ints$ and $\ints$ unify in any environment; 
\item $\bools $ and $\bools $ unify in any environment; 
\item $\alpha$ and $T$ unify if $\alpha$ does not occur in $T$, in any environment in which $\alpha$ is instantiated to $T$; 
\item $T_1\to T_2$ and $T_1'\to T_2'$ unify in an environment if $T_i$ and $T_i'$ unify in that environment; 
\item if none of the rules apply the types cannot be unified. 
\end{enumerate}
For instance, in this given grammar of types:
\begin{itemize}
\item $\alpha$ and $\bools $ can be trivially unified whenever $\alpha$ is instantiated as $\bools $;  
\item $\alpha\to \bools $ and $\ints\to\beta$ can be unified whenever $\alpha$ and $\ints$ unify and so do $\beta$ and $\bools $; 
\item $\ints$ and $\bools$ cannot be unified; 
\item $\ints\to\alpha$ and $\bools \to\beta$ cannot be unified; 
\item $\alpha\to\beta$ and $\alpha$ cannot be unified. 
\end{itemize}

The syntax of terms is
\[
t ::= n \mid \mathit{tt} \mid \mathit{ff} \mid x \mid op \mid \mathit{if}\ t\ \mathit{then}\ t\ \mathit{else}\ t \mid t(t) \mid \lambda x.t
\]
that is, integer constants, boolean constants, variables, arithmetic-logic operators, if-then-else, function application and function abstraction. 

The abstract algorithm, i.e. without committing to any particular heuristics, can be expressed as the construction of an \emph{unification graph} out of the term graph. 
The emphasis will be on simplicity and clarity: a concrete, practical algorithm will require further optimisations and heuristics. 

The initial unification graph has vertices $T$, the types, and edges $E=\emptyset$. 
Even though the set of vertices is infinite, the set of edges will always remain finite, and we can safely ignore vertices which do not participate in any edge in $E$. 

The algorithm consists of adding new edges to the unification graph, corresponding to each edge of the diagram representing the term. 
We say that the outermost boxes are at \emph{level 0}, and the boxes contained by level-$n$ boxes are at level $n+1$. 
Intuitively, the algorithm adds a \emph{decoration} to each wire, depending on the type, which can be either known (e.g. $\ints$ for arithmetic operations) or unknown (case in which a fresh type variable is used) or partially known (the case of functions).

\begin{algorithmic}[1]
\Procedure{unify}{$t$} %\Comment{create the unification graph}
\State $n\gets$ highest level of $t$
\For{$i\gets n$ down to $0$}
	\For{$b\gets$ boxes at level $i$}
		\If{$b$ is constant $k$}
			\State $e\gets$ unique wire of $b$
			\State \Call{decorate}{$e,\ints$}
		\ElsIf{$b$ is constant $\mathit{tt,ff}$}
			\State $e\gets$ unique wire of $b$
			\State \Call{decorate}{$e,\bools $}
		\ElsIf{$b$ is $\mathit{add}$} %\Comment{similar for other ops}
			\State $(e_0,e_1,e_2)\gets$ wires of $b$
			\State \textbf{for }{$j\gets 0,2$}\textbf{ do } \Call{decorate}{$e_j,\ints$} 
		\ElsIf{$b$ is \textit{application}}
			\State $e_0\gets$ argument wire of $b$
			\State $\alpha_0\gets$ new fresh variable
			\State \Call{decorate}{$e_0,\alpha_0$}
			\State $e_1\gets$ result wire of $b$
			\State $\alpha_1\gets$ new fresh variable
			\State \Call{decorate}{$e_1,\alpha_1$}
			\State $e_2\gets$ function wire of $b$
			\State \Call{decorate}{$e_2,\alpha_0\to\alpha_1$}
		\ElsIf{$b$ is \textit{abstraction box}} 
			\State $e\gets$ result wire of $b$
			\State $\alpha_0\gets$ type of bound wire inside $b$
			\State $\alpha_1\gets$ type of root of diagram inside $b$
			\State \Call{decorate}{$e,\alpha_0\to\alpha_1$}
		\ElsIf{$b$ is \textit{contraction}} 
			\State $(e_0,e_1,e_2)\gets$ wires of $b$
			\State $\alpha\gets$ new fresh variable
			\State \textbf{for }{$j\gets 0,2$}\textbf{ do } \Call{decorate}{$e_j,\alpha$} 
		\EndIf
	\EndFor
\EndFor
\EndProcedure
\end{algorithmic}

Type information is extracted from the constants (Lines 5 and 8) and from arithmetic and logical operations (only addition is given on Line 11, all other operations are similar). 
The operation of abstraction however is \emph{polymorphic} in the sense that it may involve arguments of different types at different points in the program, so it will only introduce a constraint between the type of the two arguments and that of its result. 
Proceeding from the inner boxes outwards ensures that when we encounter an abstraction (Line 23) it has already been annotated at the interface with all type information, so the type information extracted from an abstraction may contain unknown types. 
Contraction, which corresponds to variables in the syntax, introduces only unknown type variables, which must be the same along all wires.

New edges are added to the unification graph using the procedure below. 
Decorations are a partial function $D$ from edges to types. 
\begin{algorithmic}[1]
\Procedure{decorate}{$e,T$}%\Comment{decorating an edge with a type}
\If{$D(e)$ is not defined}
	\State $D\gets(D\mid e\mapsto T)$ %\Comment{extend $D$ with a new decoration}
\Else
	\State $T'\gets D(e)$
	\State $E\gets E\cup(T,T')$
\EndIf
\EndProcedure
\end{algorithmic}
If an edge $e$ is not already annotated we add the type $T$ as an annotation. 
If an edge is already annotated we retrieve the existing annotation $T'$ and we add a new edge $(T,T')$ to the edge set $E$ of the unification graph. 
Note that because of how string diagrams are constructed it is impossible for a wire to receive more than two labels. 

This concludes the first part of the algorithm, the construction of the unification graph, in time linear in the size of the original diagram. 

\begin{example}
Consider the unification graph produced by applying the algorithm above to the diagram of the term $\mathit{if}\ x\ \mathit{then}\ x+1\ \mathit{else}\ 0$. 
The algorithm visits each interface wire once and each other wire exactly twice, so we decorate each wire with the corresponding types. 
In Figure~\ref{fig:unifgraph} we show the edge-decorated diagram and the resulting unification graph. 
\end{example}

\begin{figure}
\[
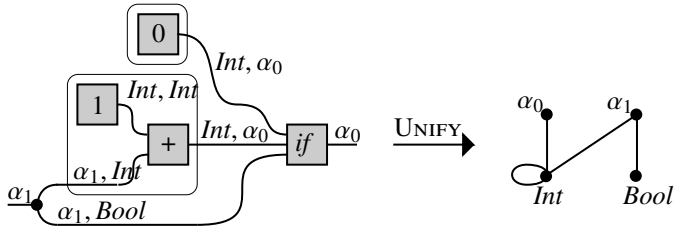
\]
\caption{Example term and its unification graph}
\label{fig:unifgraph}
\end{figure}

The second part of the algorithm is the computationally intensive part: checking whether an unification graph is valid. 
The example above already discloses the validation condition: \emph{there should be no path from $\ints$ to $\bools $ in the unification graph}. 
Note that the failure of this condition can be determined without solving all the unknown type assignments of the term. 
However, the solution is not always so straightforward when function types are involved. 
\begin{example}\label{ex:unifgraph2}
Consider the unification graph produced by applying the algorithm above to the diagram of the term $f(\mathit{tt})+f(1)$. 
Figure~\ref{fig:unifgraph2} shows the edge-decorated diagram and the resulting unification graph. 
\end{example}

\begin{figure}
\[
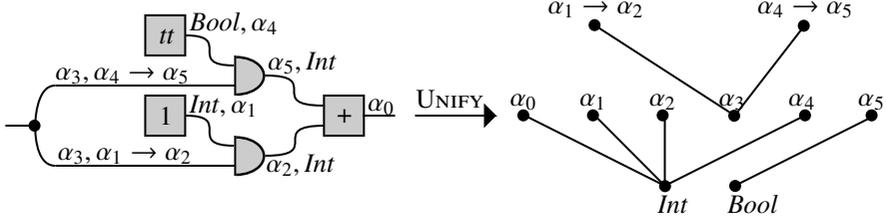
\]
\caption{Example term and its unification graph}
\label{fig:unifgraph2}
\end{figure}

Note that in Example~\ref{ex:unifgraph2} (Figure \ref{fig:unifgraph2}) the graph is disconnected, which means that there is not enough information to relate $\alpha_1, \alpha_2, \alpha_4, \alpha_5$. 
This is dealt with by the following algorithm.

\begin{algorithmic}[1]
\Procedure{saturate}{$E$}\Comment{infer new edges in the unification graph}
\State find new path $(v_1, v_2, \ldots, v_k)$ in $E$ such that 
\State $v_1 = t_1\to t_2$
\State $v_k = t_3\to t_4$
\If{fail}
	\State result $\gets E$
\Else
	\State \Call{saturate}{$E\cup(t_1, t_3)\cup(t_2, t_4)$}
\EndIf	
\EndProcedure
\end{algorithmic}

As applied to Example~\ref{ex:unifgraph2} the graph, the resulting graph is:

\[
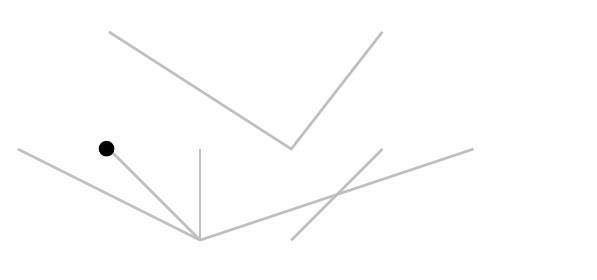
\]

Old edges are grey, new saturation edges are solid black. 
We can now detect an invalid path $\ints, \alpha_1, \alpha_4, \bools $ corresponding to a type error. 
Besides this one, there is another validation condition that must be checked:

\begin{algorithmic}[1]
\Procedure{validate}{$E$}\Comment{find type errors}
\State find path $(\ints, v_1, v_2, \ldots, \bools )$ in $E$ 
\State find path $(\alpha, v_1, v_2, \ldots, \alpha\to t)$ in $E$ or
\State find path $(\alpha, v_1, v_2, \ldots, t\to\alpha)$ in $E$ or
\If{fail}
	\State result $\gets$ type ok
\Else
	\State result $\gets$ type error
\EndIf	
\EndProcedure
\end{algorithmic}

The paths from a type variable to a type containing that variable arise from terms such as illegal self-application. 
It is a simple exercise to check that $f\ f$, an impossible term in PCF, leads to a violation of this condition. 

The \textsc{saturate} procedure is applied until it produces no new edges in the saturation graph. 
Finding saturating paths is computationally intensive and this is where heuristics come into play. 
In particular, it is not always necessary to add all saturating paths. 
Also, saturation and validity checking can be cleverly interleaved to lead to earlier error detection. 

\begin{exercise}
Determine whether the following term type checks: 
\[
f(0) + g(f) + g(\lambda x.x\ \&\ y)
\]
where $\&$ is the Boolean conjunction. 
\end{exercise}

\begin{exercise}
Extend the language of type with \emph{products} and that of terms with \emph{tuples} and \emph{projections} and extend the algorithms accordingly. 
Check that the term $f(f(1), f(2))$ cannot be typed. 
\end{exercise}

\paragraph{Further reading and related work}
Type inference is a central problem in theoretical computer science, starting with the pioneering work of~\cite{DBLP:conf/popl/DamasM82} in programming languages and~\cite{10.2307/1995158} in combinatory logic. 
The graph-oriented perspective is not common, but has been employed before, for example int the work of \cite{DBLP:conf/oopsla/PalsbergS91}. 
It is unclear whether pursuing this angle can lead to algorithmic improvements and therefore faster times in type inference, but the separation of concerns between collecting and solving constraints seems promising. 
At least for pedagogical reasons, we hope it is useful. 

\subsection{Closure conversion}

One challenge posed by the compilation of functional languages is the fact that certain operations, for instance partial application of curried functions, seem to lead to the creation of new functions during execution. 
For instance, if we write in a generic functional language $\mathit{let}\ f=(\lambda x.\lambda y. x+y) 1$ we effectively bound the new function $\lambda y.1+y$ to the variable $f$. 
For technical and pragmatic reasons which we shall not delve into here, dealing with the runtime creation of new functions is best avoided. 
The ideal situation is that of the programming language C (or its modernised version C++) in which a program consists of a list of functions, all in global scope. 
These functions can take other functions as arguments and can return other functions as result, so long as these functions are chosen from the already defined functions. 

Closure conversion is a transformation that allows the transformation of any program into this desirable form, bringing all functions into global scope. 
The 4th ES equation (Definition~\ref{def:eseqn}) tells us that in certain condition we can indeed `yank' terms outside an abstraction box:

\[
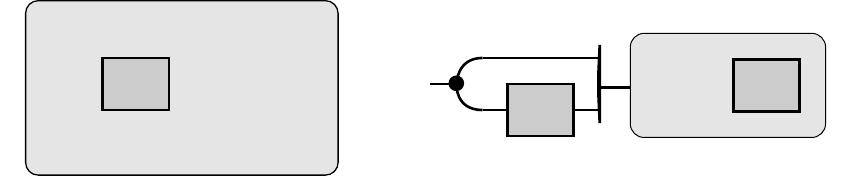
\]
However, this law requires the yanked term ($v$) to have no variables bound within the abstraction box it is extracted from. 
But how do we extract not $v$ but $u$? 
In other words, do we have `adaptor' terms $w$ and $w'$ as below so we can do this?
\[
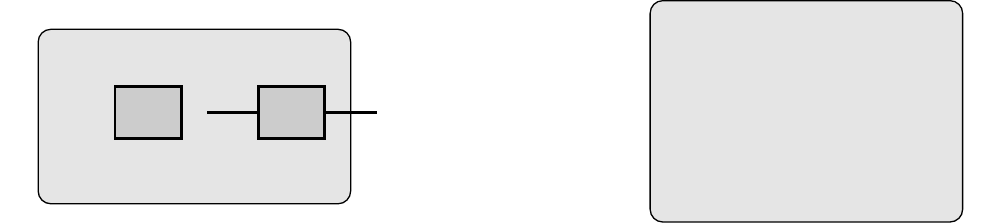
\]
The answer is `not quite, but close enough'. 
The first step is to `close' the term that is to be yanked out by treating its free variables as if they were bound variables, then to change applications to take into account these extra arguments. 
Once the term is closed, yanking it out of the enclosing abstraction is no longer problematic. 

Although not strictly necessary, we will use a $\mathit{let}\ x = u\ \mathit{in} v$ binder to define closure conversion. 
This is semantically equal to $(\lambda x.v)u$, but we will use this construction to distinguished between thunks created by closure conversion and thunks transformed by closure conversion. 
The diagrammatic syntax and operational semantics of the $\mathit{let}$ binder are:
\[
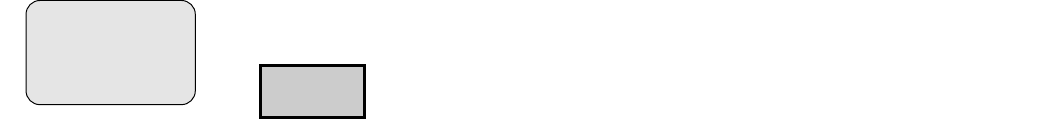
\]

The transformations, for the terms and for the applications, are given in Figure~\ref{fig:closureconv}.
Note that the transformation is not type-preserving as each function type $A\to C$ is replaced by a pair of a closure $B$ and a function with an additional closure argument $A\times B\to C$ .
\begin{figure}
\[
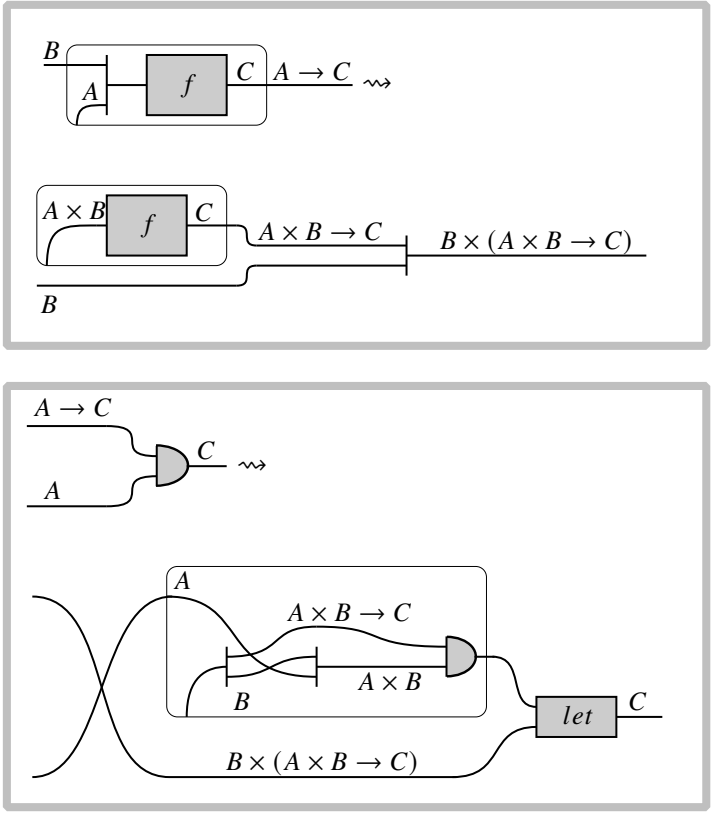
\]
\caption{Closure conversion rules for abstraction and application}
\label{fig:closureconv}
\end{figure}
We can easily check that any closure-converted beta redex reduces to the same result as the original beta redex (Figure~\ref{fig:closureprf}). 
At the top of the diagram we have the original beta law, and below it the closure converted beta redex. 
The equations used to derive the reduced term, at the bottom of the diagram are:
\begin{enumerate}
\item naturality of symmetry; 
\item beta-reduction for $\mathit{let}$; 
\item cancelation of strictification and de-strictification; 
\item beta law; 
\item graph isomorphism (which can be expressed as a series of tedious equational steps, omitted).
\end{enumerate}
\begin{figure}
\[
\scalebox{.75}{
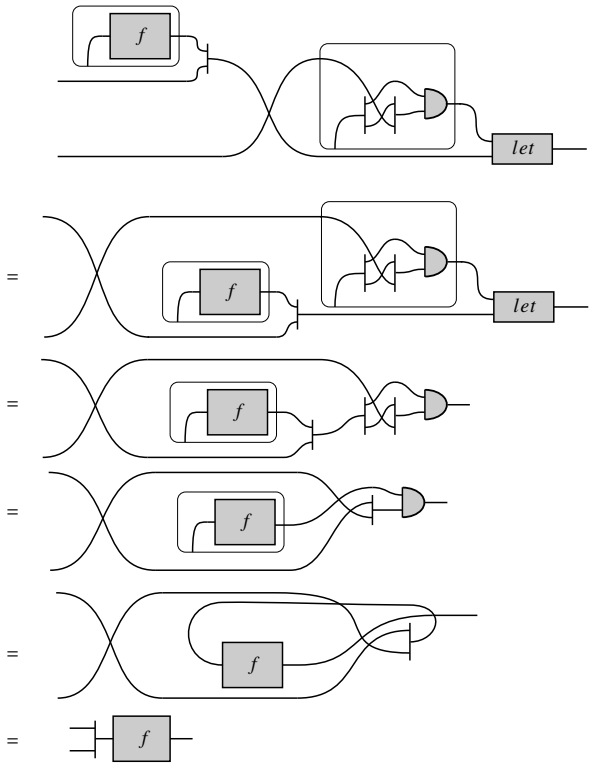}
\]
\caption{Correctness of beta-reduction for closure-converted functions}
\label{fig:closureprf}
\end{figure}
\begin{example}\label{ex:cc}
The transformation is illustrated with the partial application of a curried function, $(\lambda x\, y. x+y)\, 1\, 2$ so that the first closure is created statically. 
The diagrams in Figure~\ref{fig:cc} represent:
\begin{enumerate}
\item the representation of the original term; 
\item the closure-converted function, noting that one of the applications does not require closure conversion, only the one corresponding to the application of the inner function; 
\item yanking out the (closed) function into global scope, applying the 4th ES equation (Definition~\ref{def:eseqn}) and reorganising for clarity. 
\end{enumerate} 
\end{example}
\begin{exercise}
Define the closure conversion operation at term level and apply it to Example~\ref{ex:cc}.
Written back in term form in a language with let-bindings, tuples, and pattern matching, the closure converted term should be:
\begin{verbatim}
   let x0 = 2
   let x4(x5, x6) = x5 + x6
   let x2(x3) = (x3, x4)
   let (x1, x1') = x2(1)
   x1'(x0, x1)
\end{verbatim}

\end{exercise}
\begin{figure}
\[
\scalebox{.75}{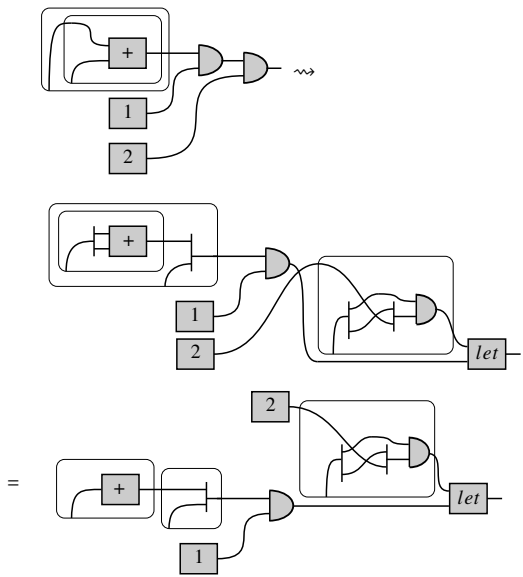}
\]
\caption{Example of closure-converted partial application}
\label{fig:cc}
\end{figure}

From the point of view of the abstract machine required to evaluate closure-converted code, it is possible to adapt the abstract machine derived from the operational semantics (see Definition~\ref{def:cbvsem}) so that no rewriting is required, at the expense of equipping the roving reduction token with extra data to allow to navigate contractions in reverse and to jump from bound variable edges to the arguments in function application.
If we think of the roving pointer as a `program counter', the associated information needed to evaluate a term without rewrite is a `call stack'.  

\paragraph{Related work and further reading}
Closure conversion is an old technique allowing the efficient compilation of functional languages \cite{DBLP:conf/acm/Reynolds72}. 
A related technique is the operation of replacing the free variables of an inner function with arguments in order to `yank' it into global scope; it is called ``$\lambda$ lifting'' and has also been used in early compilers of functional languages \cite{DBLP:conf/fpca/Johnsson85}. 
The difference between the two is that the former introduces an extra argument and changes the call sites uniformly, whereas the latter introduces an additional argument for each free variables, and so it changes all call sites in different ways, depending on the function. 
The simplistic closure conversion that we give here uses the tuple as the extra additional argument, so one could pedantically argue that it really is an uncurried form of $\lambda$ lifting and that a genuine realistic closure conversion operation would use a more performant data structure, for example some form of imperative state (e.g an object) \cite{DBLP:conf/popl/Leroy92}. 
These considerations are beyond the scope of this entry-level tutorial. 
For more in depth information on closure conversion and the difficult issues they raise in terms of typing in realistic languages we refer the reader to the rich available literature, e.g. \cite{DBLP:conf/popl/MinamideMH96}.  

\subsection{Reverse automatic differentiation}
Reverse automatic differentiation (RAD) is the workhorse of gradient-descent optimisation which is in turn the workhorse of the most successful and prevalent AI algorithms. 
The idea of RAD is that given a function $f:A\to B$ we want to compute $\delta a$ such that for some given $a,b,\delta b$ we have that $f(a+\delta a)=b+\delta b$. 
In other words, given a (small) perturbation $\delta b$ of the output $b$, what is the perturbation $\delta a$ of the output $a$ that would achieve it? 
Numerical approaches are not feasible because of the intrinsic errors introduced by operations on real-number approximations in conventional computers (e.g. `float') while symbolic approaches are not feasible because of fast increases in the sizes of the symbolic expressions required. 
What is left is the so-called `automatic' approach, an algorithmic approach to computing partial derivatives. 

RAD is an old technique, originally developed in the 1960s, and it was simple enough that the original algorithm could be described in a couple of pages. 
However, the idea become more difficult to translate into concrete algorithms when it involves programming languages with (higher order) functions. 
One of the most widely used such algorithms, dubbed by the authors `\emph{Lambda the Ultimate Backpropagator}' is notoriously difficult to explain, formalise, or reason about. 
In fact it was only by using string diagrams that a (slightly restricted) version of the algorithm was shown to be correct. 
The original, term based version is too complex to describe here, but the string diagram version of the algorithm is succinct enough. 
It is also an excellent illustration of the usefulness of foliations as a way of decomposing a diagram into a list of simple elements rather than syntactically. 

The RAD transformation involves three maps, represented as is usual using boxes. 
The target category must have a `reverse differential structure' which, without elaborating the technical details too much, is simply the requirement for all objects to have a monoid structures, i.e. for all objects $A$ there exists a zero morphism $0_A:I\to A$ and an addition morphism ${+}_A:A\otimes A \to A$, subject to the expected axioms. 
For simplicity we choose the unitype $\lambda$ calculus, as presented in Section \ref{sec:unityped}, so that we can omit types. 
This is consistent with the term-based presentation of the algorithm, but differs from the original string-diagram formulation. 

In addition to labels on the boxes we use colour to identify them: purple for the `adjoint' map ($\leftrightarrow$), blue for the `reverse' map ($\leftarrow$) and yellow for the `forward' map ($\rightarrow$). 
Given an arbitrary term, in its foliated form, the RAD transformation consists of applying the joint map, resulting in a forward and reverse components, elaborated by mapping the respective maps to each element in the foliation:

\[
\scalebox{.75}{
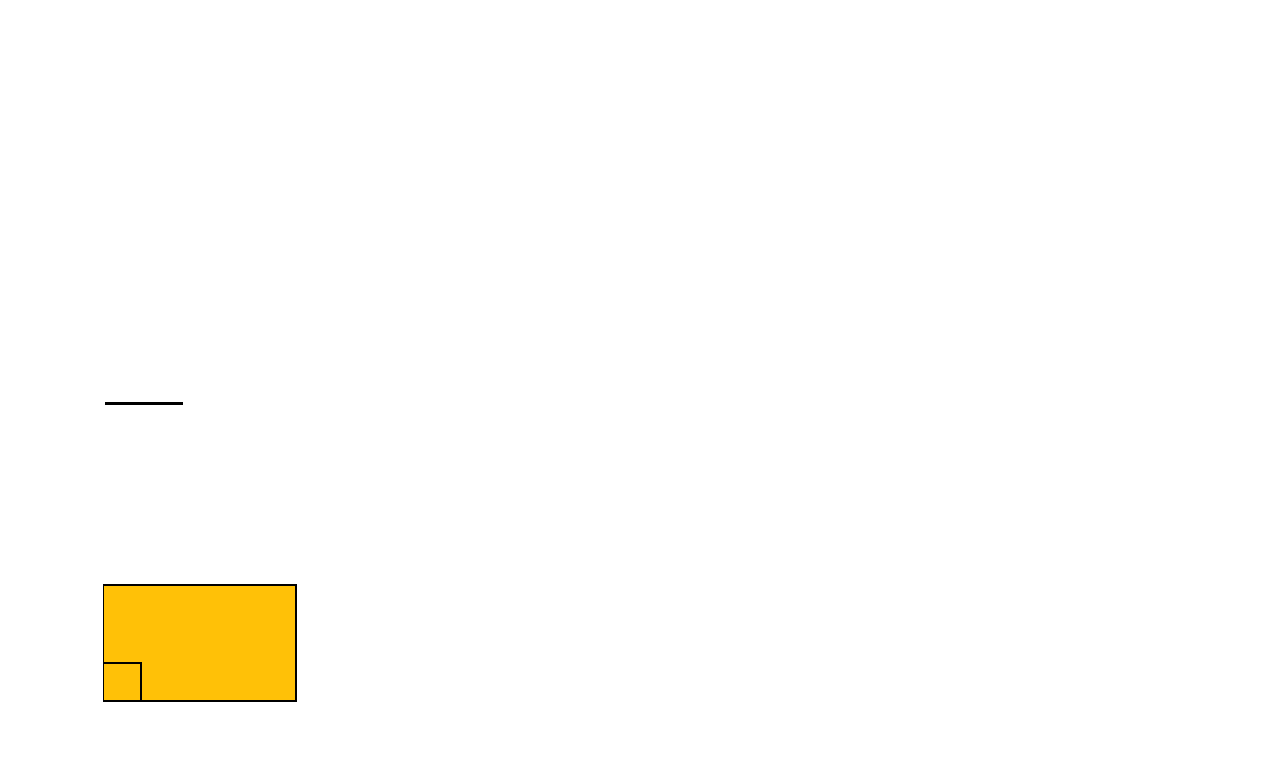
}
\]
We have not discussed strictification in a unitype setting, but unityped strictifiers/de-strictifiers with the same equational rules as those presented in Section \ref{sec:strictification} can be formulated. 
We leave this as an exercise. 

Because the original term is presented as a foliation, each $f_i$ consists of a box which is cannot be decomposed further, tensored with some identity wires. 
These wires are first factored out as follows:

\[
\scalebox{.75}{
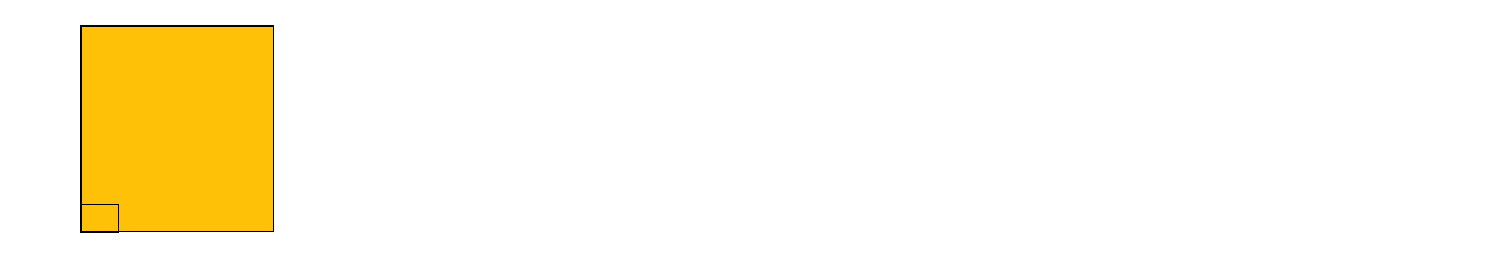
}
\] 

What is left to define is the action of the reverse and forward map on the signature and on the structural morphisms. 
For the latter, i.e. symmetry, copying and deletion, the definitions of the two maps are:

\[
\scalebox{.75}{
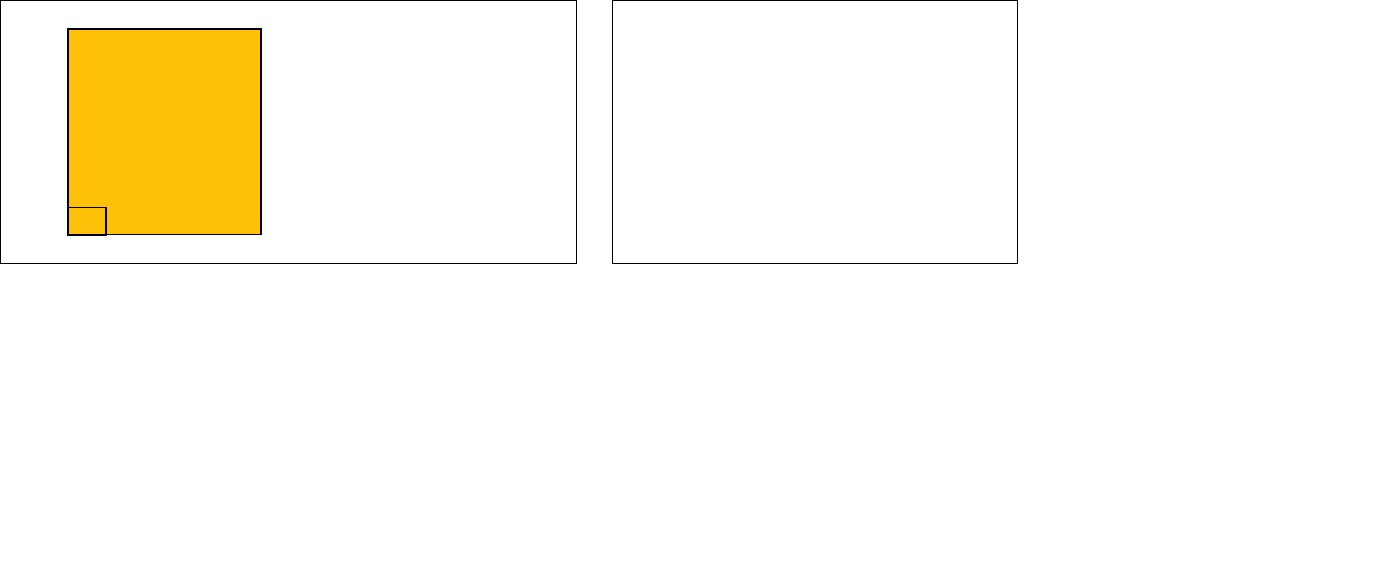
}
\]

The constants and the operations of the language, with a special provision for function application are transformed as below, where $\mathcal R[\mathit{op}]$ is a reverse-differential operator associated with any particular operation of the language. 
These need to be provided. 

\[
\scalebox{.75}{
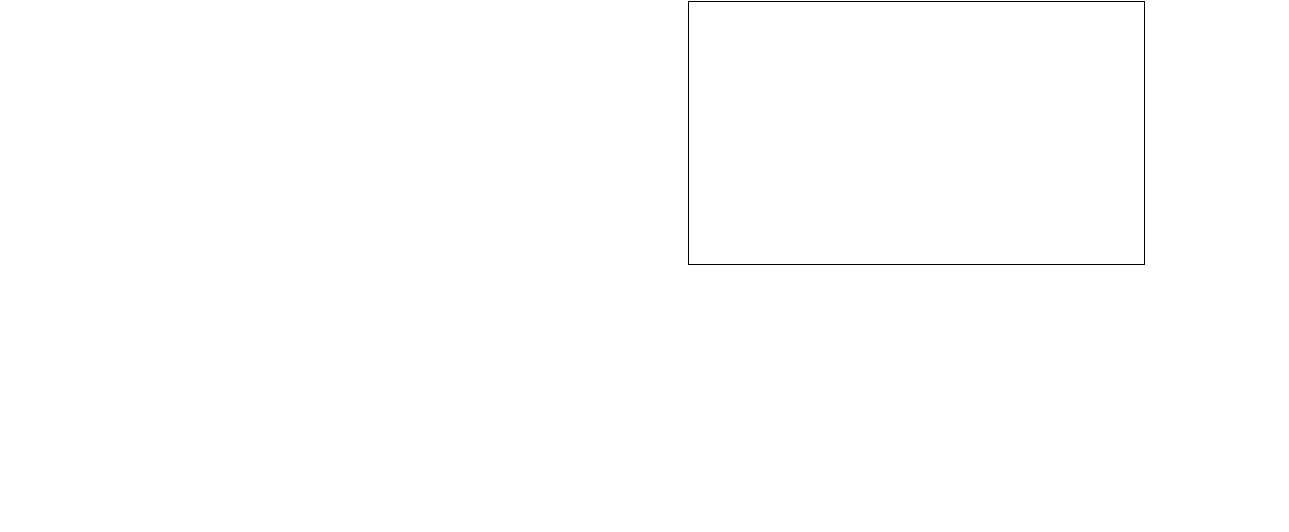
}
\]
\begin{exercise}
Using the algorithm above calculate the symbolic reverse differential of the term $(\lambda x.x\times y)(y)+y$. 
Use the equational properties of string diagrams to simplify the result. 
\end{exercise}
\emph{Hint:} The result can be found in \cite[Figure 9]{DBLP:conf/csl/Alvarez-Picallo23}.

\paragraph{Related work and further reading}

The literature on automatic differentiation is wide, particulary in the context of machine learning \cite{DBLP:journals/jmlr/BaydinPRS17}. 
The earliest algorithms are due to 
\cite{DBLP:journals/cacm/Wengert64},
and the particular algorithm we present is a simplification of
\cite{DBLP:journals/toplas/PearlmutterS08}.
The original string-diagram presentation, including a proof of correctness for the algorithm is in 
\cite{DBLP:conf/csl/Alvarez-Picallo23}.
Finally, the important technical details required to understand the categorical axiomatisation can be found in
\cite{DBLP:conf/csl/CockettCGLMPP20}

\newpage
\bibliography{main}

\end{document}

%% file: hypergraph.tikzstyles
\tikzstyle{hedge}=[fill=white, draw=black, shape=rectangle, rounded corners=2mm, inner sep=0.2mm, outer sep=-2mm, scale=0.8, minimum height=8mm, minimum width=8mm, tikzit category=hypergraph]
\tikzstyle{hedge blue}=[hedge, fill={rgb,255: red,102; green,204; blue,255}, draw=black, shape=rectangle, tikzit category=hypergraph]
\tikzstyle{node}=[fill=black, draw=black, shape=circle, minimum size=1.5mm, inner sep=0mm, tikzit category=hypergraph]
\tikzstyle{red node}=[fill=red, draw=black, shape=circle, minimum size=1.5mm, inner sep=0mm, tikzit category=hypergraph]
\tikzstyle{node highlight}=[fill=black, draw=blue, thick, shape=circle, minimum size=1.5mm, inner sep=0mm, tikzit category=hypergraph]
\tikzstyle{red node highlight}=[fill=red, draw=blue, thick, shape=circle, minimum size=1.5mm, inner sep=0mm, tikzit category=hypergraph]
\tikzstyle{green hedge}=[hedge, fill=green, draw=black, shape=rectangle]
\tikzstyle{yellow hedge}=[hedge, fill=green, draw=black, shape=circle]
\tikzstyle{small box}=[fill=white, draw=black, shape=rectangle, minimum height=6mm, minimum width=6mm, tikzit category=string diagram]
\tikzstyle{vsmall box}=[fill=black, draw=black, shape=rectangle, minimum height=4mm, minimum width=1mm, tikzit category=string diagram, inner sep=0]
\tikzstyle{medium box}=[fill=white, draw=black, shape=rectangle, minimum height=11mm, minimum width=6mm, tikzit category=string diagram]
\tikzstyle{semilarge box}=[fill=white, draw=black, shape=rectangle, minimum height=16mm, minimum width=6mm, tikzit category=string diagram]
\tikzstyle{large box}=[fill=white, draw=black, shape=rectangle, minimum height=21mm, minimum width=6mm, tikzit category=string diagram]
\tikzstyle{black dot}=[fill=black, draw=black, shape=circle, minimum size=2mm, inner sep=0mm, tikzit category=string diagram]
\tikzstyle{white dot}=[fill=white, draw=black, shape=circle, minimum size=2mm, inner sep=0mm, tikzit category=string diagram]
\tikzstyle{red dot}=[fill=red, draw=black, shape=circle, minimum size=2mm, inner sep=0mm, tikzit category=string diagram]
\tikzstyle{wlabel}=[fill=none, draw=none, shape=rectangle, tikzit category=string diagram, font={\footnotesize}, inner sep=0pt, tikzit fill={rgb,255: red,102; green,204; blue,255}, tikzit draw={rgb,255: red,102; green,204; blue,255}, yshift=0.3mm]
\tikzstyle{BRchange}=[draw=black, shape=diamond, tikzit shape=circle, tikzit fill={rgb,255: red,96; green,0; blue,0}, diamond split part fill={black,red}, inner sep=-5mm, minimum width=2.7mm, minimum height=1.7mm]
\tikzstyle{RBchange}=[draw=black, shape=diamond, tikzit shape=circle, tikzit fill={rgb,255: red,165; green,0; blue,0}, diamond split part fill={red,black}, inner sep=0, minimum width=2.7mm, minimum height=1.7mm]
\tikzstyle{dummy}=[fill=none, draw=none, shape=circle, font={\small}, inner sep=1pt, tikzit draw=blue, tikzit fill=white]
\tikzstyle{node label}=[fill=none, draw=none, shape=rectangle, tikzit fill=cyan, tikzit draw=cyan, font={\scriptsize}, tikzit shape=circle, inner sep=0pt]
\tikzstyle{empty diag}=[fill=white, draw={rgb,255: red,165; green,165; blue,165}, shape=rectangle, minimum size=1.2 cm, dashed, thick]
\tikzstyle{green dot}=[fill=green, draw=green, shape=circle, minimum size=1.5mm, inner sep=0mm, tikzit category=hypergraph]
\tikzstyle{bubble}=[fill=white, draw=black, shape=rectangle, minimum size=20mm, rounded corners=3mm, tikzit category=hypergraph]
\tikzstyle{longer bubble}=[fill=white, draw=black, shape=rectangle, minimum width=32mm, rounded corners=3mm, tikzit category=hypergraph, minimum height=28mm]

\tikzstyle{dashed edge}=[-, dashed, very thick]
\tikzstyle{alt sort}=[-, dashed, dash pattern=on 2pt off 0.5pt, thick, draw=red]
\tikzstyle{diredge}=[->, >={Latex[length=1.5mm]}]
\tikzstyle{diredge highlight}=[->, >={Latex[length=1.5mm]}, draw=blue, thick]
\tikzstyle{boundary frame}=[-, draw={rgb,255: red,170; green,170; blue,255}, dashed, fill={rgb,255: red,238; green,238; blue,255}, thick, dash pattern=on 2pt off 0.5pt]
\tikzstyle{graph frame}=[-, draw={rgb,255: red,191; green,191; blue,191}, dashed, fill={rgb,255: red,238; green,238; blue,238}, thick, dash pattern=on 2pt off 0.5pt]
\tikzstyle{def sort}=[-]
\tikzstyle{component}=[-, draw=red, thick]
\tikzstyle{map edge}=[{|->}, >=latex, shorten <=0.5mm, shorten >=0.5mm]
\tikzstyle{hypergraph map edge}=[{|->}, draw=red, shorten <=1mm, shorten >=1mm]
\tikzstyle{cdedge}=[->]
\tikzstyle{big cdedge}=[->, very thick, >=latex]
\tikzstyle{pointer edge}=[->, draw=gray, thick]
\tikzstyle{new edge style 0}=[-, fill=none, draw=green]
\tikzstyle{blue pointer}=[draw=blue, ->]

%% file: pics/id.pdf_tex
%% Creator: Inkscape 1.0.1 (c497b03c, 2020-09-10), www.inkscape.org
%% Accompanies image file 'id.pdf' (pdf, eps, ps)
\begingroup%
  \makeatletter%
  \providecommand\color[2][]{%
    \errmessage{(Inkscape) Color is used for the text in Inkscape, but the package 'color.sty' is not loaded}%
    \renewcommand\color[2][]{}%
  }%
  \providecommand\transparent[1]{%
    \errmessage{(Inkscape) Transparency is used (non-zero) for the text in Inkscape, but the package 'transparent.sty' is not loaded}%
    \renewcommand\transparent[1]{}%
  }%
  \providecommand\rotatebox[2]{#2}%
  \newcommand*\fsize{\dimexpr\f@size pt\relax}%
  \newcommand*\lineheight[1]{\fontsize{\fsize}{#1\fsize}\selectfont}%
  \ifx\svgwidth\undefined%
    \setlength{\unitlength}{19.90541255bp}%
    \ifx\svgscale\undefined%
      \relax%
    \else%
      \setlength{\unitlength}{\unitlength * \real{\svgscale}}%
    \fi%
  \else%
    \setlength{\unitlength}{\svgwidth}%
  \fi%
  \global\let\svgwidth\undefined%
  \global\let\svgscale\undefined%
  \makeatother%
  \begin{picture}(1,0.47113174)%
    \lineheight{1}%
    \setlength\tabcolsep{0pt}%
    \put(0,0){\includegraphics[width=\unitlength,page=1]{id.pdf}}%
    \put(0.24172691,0.1618568){\color[rgb]{0,0,0}\makebox(0,0)[lt]{\lineheight{1.25}\smash{\begin{tabular}[t]{l}$A$\end{tabular}}}}%
  \end{picture}%
\endgroup%

%% file: pics/obj.pdf_tex
%% Creator: Inkscape 1.0.1 (c497b03c, 2020-09-10), www.inkscape.org
%% Accompanies image file 'obj.pdf' (pdf, eps, ps)
\begingroup%
  \makeatletter%
  \providecommand\color[2][]{%
    \errmessage{(Inkscape) Color is used for the text in Inkscape, but the package 'color.sty' is not loaded}%
    \renewcommand\color[2][]{}%
  }%
  \providecommand\transparent[1]{%
    \errmessage{(Inkscape) Transparency is used (non-zero) for the text in Inkscape, but the package 'transparent.sty' is not loaded}%
    \renewcommand\transparent[1]{}%
  }%
  \providecommand\rotatebox[2]{#2}%
  \newcommand*\fsize{\dimexpr\f@size pt\relax}%
  \newcommand*\lineheight[1]{\fontsize{\fsize}{#1\fsize}\selectfont}%
  \ifx\svgwidth\undefined%
    \setlength{\unitlength}{61.17108743bp}%
    \ifx\svgscale\undefined%
      \relax%
    \else%
      \setlength{\unitlength}{\unitlength * \real{\svgscale}}%
    \fi%
  \else%
    \setlength{\unitlength}{\svgwidth}%
  \fi%
  \global\let\svgwidth\undefined%
  \global\let\svgscale\undefined%
  \makeatother%
  \begin{picture}(1,0.27592633)%
    \lineheight{1}%
    \setlength\tabcolsep{0pt}%
    \put(0,0){\includegraphics[width=\unitlength,page=1]{obj.pdf}}%
    \put(0.07865928,0.17528583){\color[rgb]{0,0,0}\makebox(0,0)[lt]{\lineheight{1.25}\smash{\begin{tabular}[t]{l}$A$\end{tabular}}}}%
    \put(0,0){\includegraphics[width=\unitlength,page=2]{obj.pdf}}%
    \put(0.75299741,0.17528618){\color[rgb]{0,0,0}\makebox(0,0)[lt]{\lineheight{1.25}\smash{\begin{tabular}[t]{l}$B$\end{tabular}}}}%
    \put(0.4438372,0.10177356){\color[rgb]{0,0,0}\makebox(0,0)[lt]{\lineheight{1.25}\smash{\begin{tabular}[t]{l}$f$\end{tabular}}}}%
  \end{picture}%
\endgroup%

%% file: pics/comp.pdf_tex
%% Creator: Inkscape 1.0.1 (c497b03c, 2020-09-10), www.inkscape.org
%% Accompanies image file 'comp.pdf' (pdf, eps, ps)
\begingroup%
  \makeatletter%
  \providecommand\color[2][]{%
    \errmessage{(Inkscape) Color is used for the text in Inkscape, but the package 'color.sty' is not loaded}%
    \renewcommand\color[2][]{}%
  }%
  \providecommand\transparent[1]{%
    \errmessage{(Inkscape) Transparency is used (non-zero) for the text in Inkscape, but the package 'transparent.sty' is not loaded}%
    \renewcommand\transparent[1]{}%
  }%
  \providecommand\rotatebox[2]{#2}%
  \newcommand*\fsize{\dimexpr\f@size pt\relax}%
  \newcommand*\lineheight[1]{\fontsize{\fsize}{#1\fsize}\selectfont}%
  \ifx\svgwidth\undefined%
    \setlength{\unitlength}{102.53045916bp}%
    \ifx\svgscale\undefined%
      \relax%
    \else%
      \setlength{\unitlength}{\unitlength * \real{\svgscale}}%
    \fi%
  \else%
    \setlength{\unitlength}{\svgwidth}%
  \fi%
  \global\let\svgwidth\undefined%
  \global\let\svgscale\undefined%
  \makeatother%
  \begin{picture}(1,0.16462146)%
    \lineheight{1}%
    \setlength\tabcolsep{0pt}%
    \put(0,0){\includegraphics[width=\unitlength,page=1]{comp.pdf}}%
    \put(0.04692921,0.10457794){\color[rgb]{0,0,0}\makebox(0,0)[lt]{\lineheight{1.25}\smash{\begin{tabular}[t]{l}$A$\end{tabular}}}}%
    \put(0,0){\includegraphics[width=\unitlength,page=2]{comp.pdf}}%
    \put(0.44924865,0.10457815){\color[rgb]{0,0,0}\makebox(0,0)[lt]{\lineheight{1.25}\smash{\begin{tabular}[t]{l}$B$\end{tabular}}}}%
    \put(0.2647994,0.06071951){\color[rgb]{0,0,0}\makebox(0,0)[lt]{\lineheight{1.25}\smash{\begin{tabular}[t]{l}$f$\end{tabular}}}}%
    \put(0,0){\includegraphics[width=\unitlength,page=3]{comp.pdf}}%
    \put(0.85156808,0.10457815){\color[rgb]{0,0,0}\makebox(0,0)[lt]{\lineheight{1.25}\smash{\begin{tabular}[t]{l}$C$\end{tabular}}}}%
    \put(0.66711883,0.06071951){\color[rgb]{0,0,0}\makebox(0,0)[lt]{\lineheight{1.25}\smash{\begin{tabular}[t]{l}$g$\end{tabular}}}}%
  \end{picture}%
\endgroup%

%% file: pics/id2.pdf_tex
%% Creator: Inkscape 1.0.1 (c497b03c, 2020-09-10), www.inkscape.org
%% Accompanies image file 'id2.pdf' (pdf, eps, ps)
\begingroup%
  \makeatletter%
  \providecommand\color[2][]{%
    \errmessage{(Inkscape) Color is used for the text in Inkscape, but the package 'color.sty' is not loaded}%
    \renewcommand\color[2][]{}%
  }%
  \providecommand\transparent[1]{%
    \errmessage{(Inkscape) Transparency is used (non-zero) for the text in Inkscape, but the package 'transparent.sty' is not loaded}%
    \renewcommand\transparent[1]{}%
  }%
  \providecommand\rotatebox[2]{#2}%
  \newcommand*\fsize{\dimexpr\f@size pt\relax}%
  \newcommand*\lineheight[1]{\fontsize{\fsize}{#1\fsize}\selectfont}%
  \ifx\svgwidth\undefined%
    \setlength{\unitlength}{37.50000094bp}%
    \ifx\svgscale\undefined%
      \relax%
    \else%
      \setlength{\unitlength}{\unitlength * \real{\svgscale}}%
    \fi%
  \else%
    \setlength{\unitlength}{\svgwidth}%
  \fi%
  \global\let\svgwidth\undefined%
  \global\let\svgscale\undefined%
  \makeatother%
  \begin{picture}(1,0.25008198)%
    \lineheight{1}%
    \setlength\tabcolsep{0pt}%
    \put(0,0){\includegraphics[width=\unitlength,page=1]{id2.pdf}}%
    \put(0.1283113,0.08591544){\color[rgb]{0,0,0}\makebox(0,0)[lt]{\lineheight{1.25}\smash{\begin{tabular}[t]{l}$A$\end{tabular}}}}%
    \put(0,0){\includegraphics[width=\unitlength,page=2]{id2.pdf}}%
  \end{picture}%
\endgroup%

%% file: pics/id3.pdf_tex
%% Creator: Inkscape 1.0.1 (c497b03c, 2020-09-10), www.inkscape.org
%% Accompanies image file 'id3.pdf' (pdf, eps, ps)
\begingroup%
  \makeatletter%
  \providecommand\color[2][]{%
    \errmessage{(Inkscape) Color is used for the text in Inkscape, but the package 'color.sty' is not loaded}%
    \renewcommand\color[2][]{}%
  }%
  \providecommand\transparent[1]{%
    \errmessage{(Inkscape) Transparency is used (non-zero) for the text in Inkscape, but the package 'transparent.sty' is not loaded}%
    \renewcommand\transparent[1]{}%
  }%
  \providecommand\rotatebox[2]{#2}%
  \newcommand*\fsize{\dimexpr\f@size pt\relax}%
  \newcommand*\lineheight[1]{\fontsize{\fsize}{#1\fsize}\selectfont}%
  \ifx\svgwidth\undefined%
    \setlength{\unitlength}{56.25bp}%
    \ifx\svgscale\undefined%
      \relax%
    \else%
      \setlength{\unitlength}{\unitlength * \real{\svgscale}}%
    \fi%
  \else%
    \setlength{\unitlength}{\svgwidth}%
  \fi%
  \global\let\svgwidth\undefined%
  \global\let\svgscale\undefined%
  \makeatother%
  \begin{picture}(1,0.16672132)%
    \lineheight{1}%
    \setlength\tabcolsep{0pt}%
    \put(0,0){\includegraphics[width=\unitlength,page=1]{id3.pdf}}%
    \put(0.08554087,0.05727696){\color[rgb]{0,0,0}\makebox(0,0)[lt]{\lineheight{1.25}\smash{\begin{tabular}[t]{l}$A$\end{tabular}}}}%
    \put(0,0){\includegraphics[width=\unitlength,page=2]{id3.pdf}}%
  \end{picture}%
\endgroup%

%% file: pics/map.pdf_tex
%% Creator: Inkscape 1.2.2 (732a01da63, 2022-12-09), www.inkscape.org
%% Accompanies image file 'map.pdf' (pdf, eps, ps)
\begingroup%
  \makeatletter%
  \providecommand\color[2][]{%
    \errmessage{(Inkscape) Color is used for the text in Inkscape, but the package 'color.sty' is not loaded}%
    \renewcommand\color[2][]{}%
  }%
  \providecommand\transparent[1]{%
    \errmessage{(Inkscape) Transparency is used (non-zero) for the text in Inkscape, but the package 'transparent.sty' is not loaded}%
    \renewcommand\transparent[1]{}%
  }%
  \providecommand\rotatebox[2]{#2}%
  \newcommand*\fsize{\dimexpr\f@size pt\relax}%
  \newcommand*\lineheight[1]{\fontsize{\fsize}{#1\fsize}\selectfont}%
  \ifx\svgwidth\undefined%
    \setlength{\unitlength}{101.25bp}%
    \ifx\svgscale\undefined%
      \relax%
    \else%
      \setlength{\unitlength}{\unitlength * \real{\svgscale}}%
    \fi%
  \else%
    \setlength{\unitlength}{\svgwidth}%
  \fi%
  \global\let\svgwidth\undefined%
  \global\let\svgscale\undefined%
  \makeatother%
  \begin{picture}(1,0.26296333)%
    \lineheight{1}%
    \setlength\tabcolsep{0pt}%
    \put(0,0){\includegraphics[width=\unitlength,page=1]{map.pdf}}%
    \put(0.23270788,0.17811712){\color[rgb]{0,0,0}\makebox(0,0)[lt]{\lineheight{1.25}\smash{\begin{tabular}[t]{l}$A$\end{tabular}}}}%
    \put(0,0){\includegraphics[width=\unitlength,page=2]{map.pdf}}%
    \put(0.64011525,0.17811734){\color[rgb]{0,0,0}\makebox(0,0)[lt]{\lineheight{1.25}\smash{\begin{tabular}[t]{l}$B$\end{tabular}}}}%
    \put(0.45333336,0.13370404){\color[rgb]{0,0,0}\makebox(0,0)[lt]{\lineheight{1.25}\smash{\begin{tabular}[t]{l}$f$\end{tabular}}}}%
    \put(0.19459861,0.01639286){\color[rgb]{0,0,0}\makebox(0,0)[lt]{\lineheight{1.25}\smash{\begin{tabular}[t]{l}$F$\end{tabular}}}}%
    \put(0,0){\includegraphics[width=\unitlength,page=3]{map.pdf}}%
    \put(-0.00432917,0.17811712){\color[rgb]{0,0,0}\makebox(0,0)[lt]{\lineheight{1.25}\smash{\begin{tabular}[t]{l}$FA$\end{tabular}}}}%
    \put(0.80378858,0.18075718){\color[rgb]{0,0,0}\makebox(0,0)[lt]{\lineheight{1.25}\smash{\begin{tabular}[t]{l}$FB$\end{tabular}}}}%
    \put(0,0){\includegraphics[width=\unitlength,page=4]{map.pdf}}%
  \end{picture}%
\endgroup%

%% file: pics/mapalign.pdf_tex
%% Creator: Inkscape 1.2.2 (732a01da63, 2022-12-09), www.inkscape.org
%% Accompanies image file 'mapalign.pdf' (pdf, eps, ps)
\begingroup%
  \makeatletter%
  \providecommand\color[2][]{%
    \errmessage{(Inkscape) Color is used for the text in Inkscape, but the package 'color.sty' is not loaded}%
    \renewcommand\color[2][]{}%
  }%
  \providecommand\transparent[1]{%
    \errmessage{(Inkscape) Transparency is used (non-zero) for the text in Inkscape, but the package 'transparent.sty' is not loaded}%
    \renewcommand\transparent[1]{}%
  }%
  \providecommand\rotatebox[2]{#2}%
  \newcommand*\fsize{\dimexpr\f@size pt\relax}%
  \newcommand*\lineheight[1]{\fontsize{\fsize}{#1\fsize}\selectfont}%
  \ifx\svgwidth\undefined%
    \setlength{\unitlength}{100.91571507bp}%
    \ifx\svgscale\undefined%
      \relax%
    \else%
      \setlength{\unitlength}{\unitlength * \real{\svgscale}}%
    \fi%
  \else%
    \setlength{\unitlength}{\svgwidth}%
  \fi%
  \global\let\svgwidth\undefined%
  \global\let\svgscale\undefined%
  \makeatother%
  \begin{picture}(1,0.26383441)%
    \lineheight{1}%
    \setlength\tabcolsep{0pt}%
    \put(0,0){\includegraphics[width=\unitlength,page=1]{mapalign.pdf}}%
    \put(0.23182248,0.17870693){\color[rgb]{0,0,0}\makebox(0,0)[lt]{\lineheight{1.25}\smash{\begin{tabular}[t]{l}$A$\end{tabular}}}}%
    \put(0,0){\includegraphics[width=\unitlength,page=2]{mapalign.pdf}}%
    \put(0.6405795,0.17870736){\color[rgb]{0,0,0}\makebox(0,0)[lt]{\lineheight{1.25}\smash{\begin{tabular}[t]{l}$B$\end{tabular}}}}%
    \put(0.45317889,0.13414672){\color[rgb]{0,0,0}\makebox(0,0)[lt]{\lineheight{1.25}\smash{\begin{tabular}[t]{l}$f$\end{tabular}}}}%
    \put(0.19358702,0.01644738){\color[rgb]{0,0,0}\makebox(0,0)[lt]{\lineheight{1.25}\smash{\begin{tabular}[t]{l}$F$\end{tabular}}}}%
    \put(0,0){\includegraphics[width=\unitlength,page=3]{mapalign.pdf}}%
    \put(-0.00599976,0.17870693){\color[rgb]{0,0,0}\makebox(0,0)[lt]{\lineheight{1.25}\smash{\begin{tabular}[t]{l}$FA$\end{tabular}}}}%
    \put(0.80479478,0.18135594){\color[rgb]{0,0,0}\makebox(0,0)[lt]{\lineheight{1.25}\smash{\begin{tabular}[t]{l}$FB$\end{tabular}}}}%
    \put(0,0){\includegraphics[width=\unitlength,page=4]{mapalign.pdf}}%
  \end{picture}%
\endgroup%

%% file: pics/fidlite.pdf_tex
%% Creator: Inkscape 1.2.2 (732a01da63, 2022-12-09), www.inkscape.org
%% Accompanies image file 'fidlite.pdf' (pdf, eps, ps)
\begingroup%
  \makeatletter%
  \providecommand\color[2][]{%
    \errmessage{(Inkscape) Color is used for the text in Inkscape, but the package 'color.sty' is not loaded}%
    \renewcommand\color[2][]{}%
  }%
  \providecommand\transparent[1]{%
    \errmessage{(Inkscape) Transparency is used (non-zero) for the text in Inkscape, but the package 'transparent.sty' is not loaded}%
    \renewcommand\transparent[1]{}%
  }%
  \providecommand\rotatebox[2]{#2}%
  \newcommand*\fsize{\dimexpr\f@size pt\relax}%
  \newcommand*\lineheight[1]{\fontsize{\fsize}{#1\fsize}\selectfont}%
  \ifx\svgwidth\undefined%
    \setlength{\unitlength}{150.00000094bp}%
    \ifx\svgscale\undefined%
      \relax%
    \else%
      \setlength{\unitlength}{\unitlength * \real{\svgscale}}%
    \fi%
  \else%
    \setlength{\unitlength}{\svgwidth}%
  \fi%
  \global\let\svgwidth\undefined%
  \global\let\svgscale\undefined%
  \makeatother%
  \begin{picture}(1,0.37750004)%
    \lineheight{1}%
    \setlength\tabcolsep{0pt}%
    \put(0,0){\includegraphics[width=\unitlength,page=1]{fidlite.pdf}}%
    \put(0.4503167,0.28441414){\color[rgb]{0,0,0}\makebox(0,0)[lt]{\lineheight{1.25}\smash{\begin{tabular}[t]{l}$=$\end{tabular}}}}%
    \put(0,0){\includegraphics[width=\unitlength,page=2]{fidlite.pdf}}%
    \put(0.63099999,0.06524988){\color[rgb]{0,0,0}\makebox(0,0)[lt]{\lineheight{1.25}\smash{\begin{tabular}[t]{l}$f$\end{tabular}}}}%
    \put(0,0){\includegraphics[width=\unitlength,page=3]{fidlite.pdf}}%
    \put(0.86099963,0.06524988){\color[rgb]{0,0,0}\makebox(0,0)[lt]{\lineheight{1.25}\smash{\begin{tabular}[t]{l}$g$\end{tabular}}}}%
    \put(0.45201223,0.06110967){\color[rgb]{0,0,0}\makebox(0,0)[lt]{\lineheight{1.25}\smash{\begin{tabular}[t]{l}$=$\end{tabular}}}}%
    \put(0,0){\includegraphics[width=\unitlength,page=4]{fidlite.pdf}}%
    \put(0.10600017,0.06524988){\color[rgb]{0,0,0}\makebox(0,0)[lt]{\lineheight{1.25}\smash{\begin{tabular}[t]{l}$f$\end{tabular}}}}%
    \put(0,0){\includegraphics[width=\unitlength,page=5]{fidlite.pdf}}%
    \put(0.28100016,0.06525016){\color[rgb]{0,0,0}\makebox(0,0)[lt]{\lineheight{1.25}\smash{\begin{tabular}[t]{l}$g$\end{tabular}}}}%
  \end{picture}%
\endgroup%

%% file: pics/map2.pdf_tex
%% Creator: Inkscape 1.2.2 (732a01da63, 2022-12-09), www.inkscape.org
%% Accompanies image file 'map2.pdf' (pdf, eps, ps)
\begingroup%
  \makeatletter%
  \providecommand\color[2][]{%
    \errmessage{(Inkscape) Color is used for the text in Inkscape, but the package 'color.sty' is not loaded}%
    \renewcommand\color[2][]{}%
  }%
  \providecommand\transparent[1]{%
    \errmessage{(Inkscape) Transparency is used (non-zero) for the text in Inkscape, but the package 'transparent.sty' is not loaded}%
    \renewcommand\transparent[1]{}%
  }%
  \providecommand\rotatebox[2]{#2}%
  \newcommand*\fsize{\dimexpr\f@size pt\relax}%
  \newcommand*\lineheight[1]{\fontsize{\fsize}{#1\fsize}\selectfont}%
  \ifx\svgwidth\undefined%
    \setlength{\unitlength}{159.55982223bp}%
    \ifx\svgscale\undefined%
      \relax%
    \else%
      \setlength{\unitlength}{\unitlength * \real{\svgscale}}%
    \fi%
  \else%
    \setlength{\unitlength}{\svgwidth}%
  \fi%
  \global\let\svgwidth\undefined%
  \global\let\svgscale\undefined%
  \makeatother%
  \begin{picture}(1,0.3313806)%
    \lineheight{1}%
    \setlength\tabcolsep{0pt}%
    \put(0,0){\includegraphics[width=\unitlength,page=1]{map2.pdf}}%
    \put(0.29028661,0.10946738){\color[rgb]{0,0,0}\makebox(0,0)[lt]{\lineheight{1.25}\smash{\begin{tabular}[t]{l}$A_1$\end{tabular}}}}%
    \put(0,0){\includegraphics[width=\unitlength,page=2]{map2.pdf}}%
    \put(0.54881032,0.10946751){\color[rgb]{0,0,0}\makebox(0,0)[lt]{\lineheight{1.25}\smash{\begin{tabular}[t]{l}$B_1$\end{tabular}}}}%
    \put(0.43028633,0.084843){\color[rgb]{0,0,0}\makebox(0,0)[lt]{\lineheight{1.25}\smash{\begin{tabular}[t]{l}$f_1$\end{tabular}}}}%
    \put(0,0){\includegraphics[width=\unitlength,page=3]{map2.pdf}}%
    \put(0.26610406,0.01040223){\color[rgb]{0,0,0}\makebox(0,0)[lt]{\lineheight{1.25}\smash{\begin{tabular}[t]{l}$F$\end{tabular}}}}%
    \put(0,0){\includegraphics[width=\unitlength,page=4]{map2.pdf}}%
    \put(-0.00379463,0.18085024){\color[rgb]{0,0,0}\makebox(0,0)[lt]{\lineheight{1.25}\smash{\begin{tabular}[t]{l}$F(A_1,A_2)$\end{tabular}}}}%
    \put(0.65089406,0.18332397){\color[rgb]{0,0,0}\makebox(0,0)[lt]{\lineheight{1.25}\smash{\begin{tabular}[t]{l}$F(B_1,B_2)$\end{tabular}}}}%
    \put(0,0){\includegraphics[width=\unitlength,page=5]{map2.pdf}}%
    \put(0.29028661,0.25467065){\color[rgb]{0,0,0}\makebox(0,0)[lt]{\lineheight{1.25}\smash{\begin{tabular}[t]{l}$A_2$\end{tabular}}}}%
    \put(0,0){\includegraphics[width=\unitlength,page=6]{map2.pdf}}%
    \put(0.43028633,0.22648796){\color[rgb]{0,0,0}\makebox(0,0)[lt]{\lineheight{1.25}\smash{\begin{tabular}[t]{l}$f_2$\end{tabular}}}}%
    \put(0.53114561,0.25467065){\color[rgb]{0,0,0}\makebox(0,0)[lt]{\lineheight{1.25}\smash{\begin{tabular}[t]{l}$B_2$\end{tabular}}}}%
    \put(0,0){\includegraphics[width=\unitlength,page=7]{map2.pdf}}%
  \end{picture}%
\endgroup%

%% file: pics/nat.pdf_tex
%% Creator: Inkscape 1.2.2 (732a01da63, 2022-12-09), www.inkscape.org
%% Accompanies image file 'nat.pdf' (pdf, eps, ps)
\begingroup%
  \makeatletter%
  \providecommand\color[2][]{%
    \errmessage{(Inkscape) Color is used for the text in Inkscape, but the package 'color.sty' is not loaded}%
    \renewcommand\color[2][]{}%
  }%
  \providecommand\transparent[1]{%
    \errmessage{(Inkscape) Transparency is used (non-zero) for the text in Inkscape, but the package 'transparent.sty' is not loaded}%
    \renewcommand\transparent[1]{}%
  }%
  \providecommand\rotatebox[2]{#2}%
  \newcommand*\fsize{\dimexpr\f@size pt\relax}%
  \newcommand*\lineheight[1]{\fontsize{\fsize}{#1\fsize}\selectfont}%
  \ifx\svgwidth\undefined%
    \setlength{\unitlength}{239.99999244bp}%
    \ifx\svgscale\undefined%
      \relax%
    \else%
      \setlength{\unitlength}{\unitlength * \real{\svgscale}}%
    \fi%
  \else%
    \setlength{\unitlength}{\svgwidth}%
  \fi%
  \global\let\svgwidth\undefined%
  \global\let\svgscale\undefined%
  \makeatother%
  \begin{picture}(1,0.11135965)%
    \lineheight{1}%
    \setlength\tabcolsep{0pt}%
    \put(0,0){\includegraphics[width=\unitlength,page=1]{nat.pdf}}%
    \put(0.05129871,0.07556516){\color[rgb]{0,0,0}\makebox(0,0)[lt]{\lineheight{1.25}\smash{\begin{tabular}[t]{l}$X$\end{tabular}}}}%
    \put(0,0){\includegraphics[width=\unitlength,page=2]{nat.pdf}}%
    \put(0.2231737,0.07556525){\color[rgb]{0,0,0}\makebox(0,0)[lt]{\lineheight{1.25}\smash{\begin{tabular}[t]{l}$Y$\end{tabular}}}}%
    \put(0.14437509,0.05682838){\color[rgb]{0,0,0}\makebox(0,0)[lt]{\lineheight{1.25}\smash{\begin{tabular}[t]{l}$f$\end{tabular}}}}%
    \put(0.03522133,0.00733773){\color[rgb]{0,0,0}\makebox(0,0)[lt]{\lineheight{1.25}\smash{\begin{tabular}[t]{l}$F$\end{tabular}}}}%
    \put(0,0){\includegraphics[width=\unitlength,page=3]{nat.pdf}}%
    \put(0.40472327,0.07355388){\color[rgb]{0,0,0}\makebox(0,0)[lt]{\lineheight{1.25}\smash{\begin{tabular}[t]{l}$GY$\end{tabular}}}}%
    \put(0.32323444,0.05659869){\color[rgb]{0,0,0}\makebox(0,0)[lt]{\lineheight{1.25}\smash{\begin{tabular}[t]{l}$\eta_Y$\end{tabular}}}}%
    \put(0.46082338,0.05571614){\color[rgb]{0,0,0}\makebox(0,0)[lt]{\lineheight{1.25}\smash{\begin{tabular}[t]{l}$=$\end{tabular}}}}%
    \put(0.74192354,0.07514317){\color[rgb]{0,0,0}\makebox(0,0)[lt]{\lineheight{1.25}\smash{\begin{tabular}[t]{l}$X$\end{tabular}}}}%
    \put(0,0){\includegraphics[width=\unitlength,page=4]{nat.pdf}}%
    \put(0.91379902,0.07514317){\color[rgb]{0,0,0}\makebox(0,0)[lt]{\lineheight{1.25}\smash{\begin{tabular}[t]{l}$Y$\end{tabular}}}}%
    \put(0.8287504,0.05640639){\color[rgb]{0,0,0}\makebox(0,0)[lt]{\lineheight{1.25}\smash{\begin{tabular}[t]{l}$f$\end{tabular}}}}%
    \put(0.72584621,0.00691565){\color[rgb]{0,0,0}\makebox(0,0)[lt]{\lineheight{1.25}\smash{\begin{tabular}[t]{l}$G$\end{tabular}}}}%
    \put(0,0){\includegraphics[width=\unitlength,page=5]{nat.pdf}}%
    \put(0.51119541,0.0722944){\color[rgb]{0,0,0}\makebox(0,0)[lt]{\lineheight{1.25}\smash{\begin{tabular}[t]{l}$FX$\end{tabular}}}}%
    \put(0.61698442,0.05617661){\color[rgb]{0,0,0}\makebox(0,0)[lt]{\lineheight{1.25}\smash{\begin{tabular}[t]{l}$\eta_X$\end{tabular}}}}%
  \end{picture}%
\endgroup%

%% file: pics/natlite.pdf_tex
%% Creator: Inkscape 1.0.1 (c497b03c, 2020-09-10), www.inkscape.org
%% Accompanies image file 'natlite.pdf' (pdf, eps, ps)
\begingroup%
  \makeatletter%
  \providecommand\color[2][]{%
    \errmessage{(Inkscape) Color is used for the text in Inkscape, but the package 'color.sty' is not loaded}%
    \renewcommand\color[2][]{}%
  }%
  \providecommand\transparent[1]{%
    \errmessage{(Inkscape) Transparency is used (non-zero) for the text in Inkscape, but the package 'transparent.sty' is not loaded}%
    \renewcommand\transparent[1]{}%
  }%
  \providecommand\rotatebox[2]{#2}%
  \newcommand*\fsize{\dimexpr\f@size pt\relax}%
  \newcommand*\lineheight[1]{\fontsize{\fsize}{#1\fsize}\selectfont}%
  \ifx\svgwidth\undefined%
    \setlength{\unitlength}{153.75000472bp}%
    \ifx\svgscale\undefined%
      \relax%
    \else%
      \setlength{\unitlength}{\unitlength * \real{\svgscale}}%
    \fi%
  \else%
    \setlength{\unitlength}{\svgwidth}%
  \fi%
  \global\let\svgwidth\undefined%
  \global\let\svgscale\undefined%
  \makeatother%
  \begin{picture}(1,0.14943401)%
    \lineheight{1}%
    \setlength\tabcolsep{0pt}%
    \put(0,0){\includegraphics[width=\unitlength,page=1]{natlite.pdf}}%
    \put(0.10341485,0.06365587){\color[rgb]{0,0,0}\makebox(0,0)[lt]{\lineheight{1.25}\smash{\begin{tabular}[t]{l}$f$\end{tabular}}}}%
    \put(0,0){\includegraphics[width=\unitlength,page=2]{natlite.pdf}}%
    \put(0.30456107,0.06329732){\color[rgb]{0,0,0}\makebox(0,0)[lt]{\lineheight{1.25}\smash{\begin{tabular}[t]{l}$\eta$\end{tabular}}}}%
    \put(0.49982183,0.06191969){\color[rgb]{0,0,0}\makebox(0,0)[lt]{\lineheight{1.25}\smash{\begin{tabular}[t]{l}$=$\end{tabular}}}}%
    \put(0,0){\includegraphics[width=\unitlength,page=3]{natlite.pdf}}%
    \put(0.84975715,0.06431458){\color[rgb]{0,0,0}\makebox(0,0)[lt]{\lineheight{1.25}\smash{\begin{tabular}[t]{l}$f$\end{tabular}}}}%
    \put(0,0){\includegraphics[width=\unitlength,page=4]{natlite.pdf}}%
    \put(0.66553681,0.06329732){\color[rgb]{0,0,0}\makebox(0,0)[lt]{\lineheight{1.25}\smash{\begin{tabular}[t]{l}$\eta$\end{tabular}}}}%
  \end{picture}%
\endgroup%

%% file: pics/adj.pdf_tex
%% Creator: Inkscape 1.2.2 (732a01da63, 2022-12-09), www.inkscape.org
%% Accompanies image file 'adj.pdf' (pdf, eps, ps)
\begingroup%
  \makeatletter%
  \providecommand\color[2][]{%
    \errmessage{(Inkscape) Color is used for the text in Inkscape, but the package 'color.sty' is not loaded}%
    \renewcommand\color[2][]{}%
  }%
  \providecommand\transparent[1]{%
    \errmessage{(Inkscape) Transparency is used (non-zero) for the text in Inkscape, but the package 'transparent.sty' is not loaded}%
    \renewcommand\transparent[1]{}%
  }%
  \providecommand\rotatebox[2]{#2}%
  \newcommand*\fsize{\dimexpr\f@size pt\relax}%
  \newcommand*\lineheight[1]{\fontsize{\fsize}{#1\fsize}\selectfont}%
  \ifx\svgwidth\undefined%
    \setlength{\unitlength}{154.71701742bp}%
    \ifx\svgscale\undefined%
      \relax%
    \else%
      \setlength{\unitlength}{\unitlength * \real{\svgscale}}%
    \fi%
  \else%
    \setlength{\unitlength}{\svgwidth}%
  \fi%
  \global\let\svgwidth\undefined%
  \global\let\svgscale\undefined%
  \makeatother%
  \begin{picture}(1,0.41446643)%
    \lineheight{1}%
    \setlength\tabcolsep{0pt}%
    \put(0,0){\includegraphics[width=\unitlength,page=1]{adj.pdf}}%
    \put(0.15609157,0.32987648){\color[rgb]{0,0,0}\makebox(0,0)[lt]{\lineheight{1.25}\smash{\begin{tabular}[t]{l}$\eta_X$\end{tabular}}}}%
    \put(0.05463609,0.25310564){\color[rgb]{0,0,0}\makebox(0,0)[lt]{\lineheight{1.25}\smash{\begin{tabular}[t]{l}$F$\end{tabular}}}}%
    \put(0,0){\includegraphics[width=\unitlength,page=2]{adj.pdf}}%
    \put(0.66491154,0.35994495){\color[rgb]{0,0,0}\makebox(0,0)[lt]{\lineheight{1.25}\smash{\begin{tabular}[t]{l}$\id_{FX}$\end{tabular}}}}%
    \put(0.39476111,0.32952018){\color[rgb]{0,0,0}\makebox(0,0)[lt]{\lineheight{1.25}\smash{\begin{tabular}[t]{l}$\epsilon_{FX}$\end{tabular}}}}%
    \put(0.57910631,0.32815115){\color[rgb]{0,0,0}\makebox(0,0)[lt]{\lineheight{1.25}\smash{\begin{tabular}[t]{l}$=$\end{tabular}}}}%
    \put(0,0){\includegraphics[width=\unitlength,page=3]{adj.pdf}}%
    \put(0.38392766,0.08749854){\color[rgb]{0,0,0}\makebox(0,0)[lt]{\lineheight{1.25}\smash{\begin{tabular}[t]{l}$\epsilon_A$\end{tabular}}}}%
    \put(0.2533861,0.01072769){\color[rgb]{0,0,0}\makebox(0,0)[lt]{\lineheight{1.25}\smash{\begin{tabular}[t]{l}$G$\end{tabular}}}}%
    \put(0,0){\includegraphics[width=\unitlength,page=4]{adj.pdf}}%
    \put(0.07482248,0.08714223){\color[rgb]{0,0,0}\makebox(0,0)[lt]{\lineheight{1.25}\smash{\begin{tabular}[t]{l}$\eta_{GA}$\end{tabular}}}}%
    \put(0,0){\includegraphics[width=\unitlength,page=5]{adj.pdf}}%
    \put(0.69399692,0.11756673){\color[rgb]{0,0,0}\makebox(0,0)[lt]{\lineheight{1.25}\smash{\begin{tabular}[t]{l}$\id_{GA}$\end{tabular}}}}%
    \put(0.6081917,0.08577321){\color[rgb]{0,0,0}\makebox(0,0)[lt]{\lineheight{1.25}\smash{\begin{tabular}[t]{l}$=$\end{tabular}}}}%
    \put(0,0){\includegraphics[width=\unitlength,page=6]{adj.pdf}}%
  \end{picture}%
\endgroup%

%% file: pics/adjlite.pdf_tex
%% Creator: Inkscape 1.2.2 (732a01da63, 2022-12-09), www.inkscape.org
%% Accompanies image file 'adjlite.pdf' (pdf, eps, ps)
\begingroup%
  \makeatletter%
  \providecommand\color[2][]{%
    \errmessage{(Inkscape) Color is used for the text in Inkscape, but the package 'color.sty' is not loaded}%
    \renewcommand\color[2][]{}%
  }%
  \providecommand\transparent[1]{%
    \errmessage{(Inkscape) Transparency is used (non-zero) for the text in Inkscape, but the package 'transparent.sty' is not loaded}%
    \renewcommand\transparent[1]{}%
  }%
  \providecommand\rotatebox[2]{#2}%
  \newcommand*\fsize{\dimexpr\f@size pt\relax}%
  \newcommand*\lineheight[1]{\fontsize{\fsize}{#1\fsize}\selectfont}%
  \ifx\svgwidth\undefined%
    \setlength{\unitlength}{195.7499915bp}%
    \ifx\svgscale\undefined%
      \relax%
    \else%
      \setlength{\unitlength}{\unitlength * \real{\svgscale}}%
    \fi%
  \else%
    \setlength{\unitlength}{\svgwidth}%
  \fi%
  \global\let\svgwidth\undefined%
  \global\let\svgscale\undefined%
  \makeatother%
  \begin{picture}(1,0.11685828)%
    \lineheight{1}%
    \setlength\tabcolsep{0pt}%
    \put(0,0){\includegraphics[width=\unitlength,page=1]{adjlite.pdf}}%
    \put(0.22782928,0.04863635){\color[rgb]{0,0,0}\makebox(0,0)[lt]{\lineheight{1.25}\smash{\begin{tabular}[t]{l}$=$\end{tabular}}}}%
    \put(0,0){\includegraphics[width=\unitlength,page=2]{adjlite.pdf}}%
    \put(0.82169911,0.04863635){\color[rgb]{0,0,0}\makebox(0,0)[lt]{\lineheight{1.25}\smash{\begin{tabular}[t]{l}$=$\end{tabular}}}}%
    \put(0,0){\includegraphics[width=\unitlength,page=3]{adjlite.pdf}}%
    \put(0.0548623,0.04240336){\makebox(0,0)[lt]{\lineheight{1.25}\smash{\begin{tabular}[t]{l}$X$\end{tabular}}}}%
    \put(0.71423337,0.04240336){\makebox(0,0)[lt]{\lineheight{1.25}\smash{\begin{tabular}[t]{l}$A$\end{tabular}}}}%
  \end{picture}%
\endgroup%

%% file: pics/homset.pdf_tex
%% Creator: Inkscape 1.2.2 (732a01da63, 2022-12-09), www.inkscape.org
%% Accompanies image file 'homset.pdf' (pdf, eps, ps)
\begingroup%
  \makeatletter%
  \providecommand\color[2][]{%
    \errmessage{(Inkscape) Color is used for the text in Inkscape, but the package 'color.sty' is not loaded}%
    \renewcommand\color[2][]{}%
  }%
  \providecommand\transparent[1]{%
    \errmessage{(Inkscape) Transparency is used (non-zero) for the text in Inkscape, but the package 'transparent.sty' is not loaded}%
    \renewcommand\transparent[1]{}%
  }%
  \providecommand\rotatebox[2]{#2}%
  \newcommand*\fsize{\dimexpr\f@size pt\relax}%
  \newcommand*\lineheight[1]{\fontsize{\fsize}{#1\fsize}\selectfont}%
  \ifx\svgwidth\undefined%
    \setlength{\unitlength}{152.01269854bp}%
    \ifx\svgscale\undefined%
      \relax%
    \else%
      \setlength{\unitlength}{\unitlength * \real{\svgscale}}%
    \fi%
  \else%
    \setlength{\unitlength}{\svgwidth}%
  \fi%
  \global\let\svgwidth\undefined%
  \global\let\svgscale\undefined%
  \makeatother%
  \begin{picture}(1,0.34783288)%
    \lineheight{1}%
    \setlength\tabcolsep{0pt}%
    \put(0,0){\includegraphics[width=\unitlength,page=1]{homset.pdf}}%
    \put(0.35603844,0.26037656){\makebox(0,0)[lt]{\lineheight{1.25}\smash{\begin{tabular}[t]{l}$\stackrel T\mapsto$\end{tabular}}}}%
    \put(0,0){\includegraphics[width=\unitlength,page=2]{homset.pdf}}%
    \put(0.11954809,0.26014054){\makebox(0,0)[lt]{\lineheight{1.25}\smash{\begin{tabular}[t]{l}$f$\end{tabular}}}}%
    \put(0.01674795,0.29726197){\makebox(0,0)[lt]{\lineheight{1.25}\smash{\begin{tabular}[t]{l}$A$\end{tabular}}}}%
    \put(0.21409986,0.29726197){\makebox(0,0)[lt]{\lineheight{1.25}\smash{\begin{tabular}[t]{l}$GV$\end{tabular}}}}%
    \put(0,0){\includegraphics[width=\unitlength,page=3]{homset.pdf}}%
    \put(0.63759692,0.26014054){\makebox(0,0)[lt]{\lineheight{1.25}\smash{\begin{tabular}[t]{l}$f$\end{tabular}}}}%
    \put(0.4607898,0.29726197){\makebox(0,0)[lt]{\lineheight{1.25}\smash{\begin{tabular}[t]{l}$FA$\end{tabular}}}}%
    \put(0.8949639,0.29726197){\makebox(0,0)[lt]{\lineheight{1.25}\smash{\begin{tabular}[t]{l}$V$\end{tabular}}}}%
    \put(0.35603851,0.06302458){\makebox(0,0)[lt]{\lineheight{1.25}\smash{\begin{tabular}[t]{l}$\stackrel T\mapsto$\end{tabular}}}}%
    \put(0,0){\includegraphics[width=\unitlength,page=4]{homset.pdf}}%
    \put(0.11954809,0.06278827){\makebox(0,0)[lt]{\lineheight{1.25}\smash{\begin{tabular}[t]{l}$g$\end{tabular}}}}%
    \put(-0.00298727,0.09990984){\makebox(0,0)[lt]{\lineheight{1.25}\smash{\begin{tabular}[t]{l}$FA$\end{tabular}}}}%
    \put(0.21409986,0.09990984){\makebox(0,0)[lt]{\lineheight{1.25}\smash{\begin{tabular}[t]{l}$V$\end{tabular}}}}%
    \put(0,0){\includegraphics[width=\unitlength,page=5]{homset.pdf}}%
    \put(0.76094192,0.06278842){\makebox(0,0)[lt]{\lineheight{1.25}\smash{\begin{tabular}[t]{l}$g$\end{tabular}}}}%
    \put(0.48052501,0.09990984){\makebox(0,0)[lt]{\lineheight{1.25}\smash{\begin{tabular}[t]{l}$A$\end{tabular}}}}%
    \put(0.8949639,0.09990984){\makebox(0,0)[lt]{\lineheight{1.25}\smash{\begin{tabular}[t]{l}$GV$\end{tabular}}}}%
    \put(0,0){\includegraphics[width=\unitlength,page=6]{homset.pdf}}%
  \end{picture}%
\endgroup%

%% file: pics/hom-inv.pdf_tex
%% Creator: Inkscape 1.2.2 (732a01da63, 2022-12-09), www.inkscape.org
%% Accompanies image file 'hom-inv.pdf' (pdf, eps, ps)
\begingroup%
  \makeatletter%
  \providecommand\color[2][]{%
    \errmessage{(Inkscape) Color is used for the text in Inkscape, but the package 'color.sty' is not loaded}%
    \renewcommand\color[2][]{}%
  }%
  \providecommand\transparent[1]{%
    \errmessage{(Inkscape) Transparency is used (non-zero) for the text in Inkscape, but the package 'transparent.sty' is not loaded}%
    \renewcommand\transparent[1]{}%
  }%
  \providecommand\rotatebox[2]{#2}%
  \newcommand*\fsize{\dimexpr\f@size pt\relax}%
  \newcommand*\lineheight[1]{\fontsize{\fsize}{#1\fsize}\selectfont}%
  \ifx\svgwidth\undefined%
    \setlength{\unitlength}{169.81206105bp}%
    \ifx\svgscale\undefined%
      \relax%
    \else%
      \setlength{\unitlength}{\unitlength * \real{\svgscale}}%
    \fi%
  \else%
    \setlength{\unitlength}{\svgwidth}%
  \fi%
  \global\let\svgwidth\undefined%
  \global\let\svgscale\undefined%
  \makeatother%
  \begin{picture}(1,0.70887207)%
    \lineheight{1}%
    \setlength\tabcolsep{0pt}%
    \put(0,0){\includegraphics[width=\unitlength,page=1]{hom-inv.pdf}}%
    \put(0.34417925,0.60828776){\makebox(0,0)[lt]{\lineheight{1.25}\smash{\begin{tabular}[t]{l}$f$\end{tabular}}}}%
    \put(0,0){\includegraphics[width=\unitlength,page=2]{hom-inv.pdf}}%
    \put(-0.00267415,0.38589264){\makebox(0,0)[lt]{\lineheight{1.25}\smash{\begin{tabular}[t]{l}$=$\end{tabular}}}}%
    \put(0,0){\includegraphics[width=\unitlength,page=3]{hom-inv.pdf}}%
    \put(0.34417925,0.38745543){\makebox(0,0)[lt]{\lineheight{1.25}\smash{\begin{tabular}[t]{l}$f$\end{tabular}}}}%
    \put(0,0){\includegraphics[width=\unitlength,page=4]{hom-inv.pdf}}%
    \put(0.76006871,0.39860229){\makebox(0,0)[lt]{\lineheight{1.25}\smash{\begin{tabular}[t]{l}(functoriality)\end{tabular}}}}%
    \put(-0.00267402,0.18714351){\makebox(0,0)[lt]{\lineheight{1.25}\smash{\begin{tabular}[t]{l}$=$\end{tabular}}}}%
    \put(0,0){\includegraphics[width=\unitlength,page=5]{hom-inv.pdf}}%
    \put(0.14543011,0.18870579){\makebox(0,0)[lt]{\lineheight{1.25}\smash{\begin{tabular}[t]{l}$f$\end{tabular}}}}%
    \put(0,0){\includegraphics[width=\unitlength,page=6]{hom-inv.pdf}}%
    \put(0.75964551,0.19942982){\makebox(0,0)[lt]{\lineheight{1.25}\smash{\begin{tabular}[t]{l}(naturality)\end{tabular}}}}%
    \put(-0.00267402,0.03256076){\makebox(0,0)[lt]{\lineheight{1.25}\smash{\begin{tabular}[t]{l}$=$\end{tabular}}}}%
    \put(0,0){\includegraphics[width=\unitlength,page=7]{hom-inv.pdf}}%
    \put(0.14543011,0.03412317){\makebox(0,0)[lt]{\lineheight{1.25}\smash{\begin{tabular}[t]{l}$f$\end{tabular}}}}%
    \put(0.75943397,0.04577449){\makebox(0,0)[lt]{\lineheight{1.25}\smash{\begin{tabular}[t]{l}(unit-counit)\end{tabular}}}}%
  \end{picture}%
\endgroup%

%% file: pics/assoc.pdf_tex
%% Creator: Inkscape 1.2.2 (732a01da63, 2022-12-09), www.inkscape.org
%% Accompanies image file 'assoc.pdf' (pdf, eps, ps)
\begingroup%
  \makeatletter%
  \providecommand\color[2][]{%
    \errmessage{(Inkscape) Color is used for the text in Inkscape, but the package 'color.sty' is not loaded}%
    \renewcommand\color[2][]{}%
  }%
  \providecommand\transparent[1]{%
    \errmessage{(Inkscape) Transparency is used (non-zero) for the text in Inkscape, but the package 'transparent.sty' is not loaded}%
    \renewcommand\transparent[1]{}%
  }%
  \providecommand\rotatebox[2]{#2}%
  \newcommand*\fsize{\dimexpr\f@size pt\relax}%
  \newcommand*\lineheight[1]{\fontsize{\fsize}{#1\fsize}\selectfont}%
  \ifx\svgwidth\undefined%
    \setlength{\unitlength}{174.7499726bp}%
    \ifx\svgscale\undefined%
      \relax%
    \else%
      \setlength{\unitlength}{\unitlength * \real{\svgscale}}%
    \fi%
  \else%
    \setlength{\unitlength}{\svgwidth}%
  \fi%
  \global\let\svgwidth\undefined%
  \global\let\svgscale\undefined%
  \makeatother%
  \begin{picture}(1,0.38841211)%
    \lineheight{1}%
    \setlength\tabcolsep{0pt}%
    \put(0,0){\includegraphics[width=\unitlength,page=1]{assoc.pdf}}%
    \put(0.48018953,0.1810452){\color[rgb]{0,0,0}\makebox(0,0)[lt]{\lineheight{1.25}\smash{\begin{tabular}[t]{l}$=$\end{tabular}}}}%
    \put(0,0){\includegraphics[width=\unitlength,page=2]{assoc.pdf}}%
    \put(0.1339056,0.18476387){\color[rgb]{0,0,0}\makebox(0,0)[lt]{\lineheight{1.25}\smash{\begin{tabular}[t]{l}$f_1$\end{tabular}}}}%
    \put(0,0){\includegraphics[width=\unitlength,page=3]{assoc.pdf}}%
    \put(0.1339056,0.29206019){\color[rgb]{0,0,0}\makebox(0,0)[lt]{\lineheight{1.25}\smash{\begin{tabular}[t]{l}$f_2$\end{tabular}}}}%
    \put(0,0){\includegraphics[width=\unitlength,page=4]{assoc.pdf}}%
    \put(0.1339056,0.0560086){\color[rgb]{0,0,0}\makebox(0,0)[lt]{\lineheight{1.25}\smash{\begin{tabular}[t]{l}$f_0$\end{tabular}}}}%
    \put(0,0){\includegraphics[width=\unitlength,page=5]{assoc.pdf}}%
    \put(0.34849793,0.18476399){\color[rgb]{0,0,0}\makebox(0,0)[lt]{\lineheight{1.25}\smash{\begin{tabular}[t]{l}$\alpha$\end{tabular}}}}%
    \put(0,0){\includegraphics[width=\unitlength,page=6]{assoc.pdf}}%
    \put(0.83347652,0.18476387){\color[rgb]{0,0,0}\makebox(0,0)[lt]{\lineheight{1.25}\smash{\begin{tabular}[t]{l}$f_1$\end{tabular}}}}%
    \put(0,0){\includegraphics[width=\unitlength,page=7]{assoc.pdf}}%
    \put(0.83347652,0.31351937){\color[rgb]{0,0,0}\makebox(0,0)[lt]{\lineheight{1.25}\smash{\begin{tabular}[t]{l}$f_2$\end{tabular}}}}%
    \put(0,0){\includegraphics[width=\unitlength,page=8]{assoc.pdf}}%
    \put(0.83347652,0.07746767){\color[rgb]{0,0,0}\makebox(0,0)[lt]{\lineheight{1.25}\smash{\begin{tabular}[t]{l}$f_0$\end{tabular}}}}%
    \put(0,0){\includegraphics[width=\unitlength,page=9]{assoc.pdf}}%
    \put(0.61888425,0.18476399){\color[rgb]{0,0,0}\makebox(0,0)[lt]{\lineheight{1.25}\smash{\begin{tabular}[t]{l}$\alpha$\end{tabular}}}}%
  \end{picture}%
\endgroup%

%% file: pics/assocstr.pdf_tex
%% Creator: Inkscape 1.2.2 (732a01da63, 2022-12-09), www.inkscape.org
%% Accompanies image file 'assocstr.pdf' (pdf, eps, ps)
\begingroup%
  \makeatletter%
  \providecommand\color[2][]{%
    \errmessage{(Inkscape) Color is used for the text in Inkscape, but the package 'color.sty' is not loaded}%
    \renewcommand\color[2][]{}%
  }%
  \providecommand\transparent[1]{%
    \errmessage{(Inkscape) Transparency is used (non-zero) for the text in Inkscape, but the package 'transparent.sty' is not loaded}%
    \renewcommand\transparent[1]{}%
  }%
  \providecommand\rotatebox[2]{#2}%
  \newcommand*\fsize{\dimexpr\f@size pt\relax}%
  \newcommand*\lineheight[1]{\fontsize{\fsize}{#1\fsize}\selectfont}%
  \ifx\svgwidth\undefined%
    \setlength{\unitlength}{129.74994425bp}%
    \ifx\svgscale\undefined%
      \relax%
    \else%
      \setlength{\unitlength}{\unitlength * \real{\svgscale}}%
    \fi%
  \else%
    \setlength{\unitlength}{\svgwidth}%
  \fi%
  \global\let\svgwidth\undefined%
  \global\let\svgscale\undefined%
  \makeatother%
  \begin{picture}(1,0.52312166)%
    \lineheight{1}%
    \setlength\tabcolsep{0pt}%
    \put(0,0){\includegraphics[width=\unitlength,page=1]{assocstr.pdf}}%
    \put(0.47331883,0.24383551){\color[rgb]{0,0,0}\makebox(0,0)[lt]{\lineheight{1.25}\smash{\begin{tabular}[t]{l}$=$\end{tabular}}}}%
    \put(0,0){\includegraphics[width=\unitlength,page=2]{assocstr.pdf}}%
    \put(0.18034696,0.24884389){\color[rgb]{0,0,0}\makebox(0,0)[lt]{\lineheight{1.25}\smash{\begin{tabular}[t]{l}$f_1$\end{tabular}}}}%
    \put(0,0){\includegraphics[width=\unitlength,page=3]{assocstr.pdf}}%
    \put(0.18034696,0.39335284){\color[rgb]{0,0,0}\makebox(0,0)[lt]{\lineheight{1.25}\smash{\begin{tabular}[t]{l}$f_2$\end{tabular}}}}%
    \put(0,0){\includegraphics[width=\unitlength,page=4]{assocstr.pdf}}%
    \put(0.18034696,0.07543357){\color[rgb]{0,0,0}\makebox(0,0)[lt]{\lineheight{1.25}\smash{\begin{tabular}[t]{l}$f_0$\end{tabular}}}}%
    \put(0,0){\includegraphics[width=\unitlength,page=5]{assocstr.pdf}}%
    \put(0.77572259,0.24884389){\color[rgb]{0,0,0}\makebox(0,0)[lt]{\lineheight{1.25}\smash{\begin{tabular}[t]{l}$f_1$\end{tabular}}}}%
    \put(0,0){\includegraphics[width=\unitlength,page=6]{assocstr.pdf}}%
    \put(0.77572259,0.42225453){\color[rgb]{0,0,0}\makebox(0,0)[lt]{\lineheight{1.25}\smash{\begin{tabular}[t]{l}$f_2$\end{tabular}}}}%
    \put(0,0){\includegraphics[width=\unitlength,page=7]{assocstr.pdf}}%
    \put(0.77572259,0.1043351){\color[rgb]{0,0,0}\makebox(0,0)[lt]{\lineheight{1.25}\smash{\begin{tabular}[t]{l}$f_0$\end{tabular}}}}%
  \end{picture}%
\endgroup%

%% file: pics/assocstrq.pdf_tex
%% Creator: Inkscape 1.2.2 (732a01da63, 2022-12-09), www.inkscape.org
%% Accompanies image file 'assocstrq.pdf' (pdf, eps, ps)
\begingroup%
  \makeatletter%
  \providecommand\color[2][]{%
    \errmessage{(Inkscape) Color is used for the text in Inkscape, but the package 'color.sty' is not loaded}%
    \renewcommand\color[2][]{}%
  }%
  \providecommand\transparent[1]{%
    \errmessage{(Inkscape) Transparency is used (non-zero) for the text in Inkscape, but the package 'transparent.sty' is not loaded}%
    \renewcommand\transparent[1]{}%
  }%
  \providecommand\rotatebox[2]{#2}%
  \newcommand*\fsize{\dimexpr\f@size pt\relax}%
  \newcommand*\lineheight[1]{\fontsize{\fsize}{#1\fsize}\selectfont}%
  \ifx\svgwidth\undefined%
    \setlength{\unitlength}{202.64273008bp}%
    \ifx\svgscale\undefined%
      \relax%
    \else%
      \setlength{\unitlength}{\unitlength * \real{\svgscale}}%
    \fi%
  \else%
    \setlength{\unitlength}{\svgwidth}%
  \fi%
  \global\let\svgwidth\undefined%
  \global\let\svgscale\undefined%
  \makeatother%
  \begin{picture}(1,0.33494913)%
    \lineheight{1}%
    \setlength\tabcolsep{0pt}%
    \put(0,0){\includegraphics[width=\unitlength,page=1]{assocstrq.pdf}}%
    \put(0.30306092,0.15612523){\color[rgb]{0,0,0}\makebox(0,0)[lt]{\lineheight{1.25}\smash{\begin{tabular}[t]{l}$=$\end{tabular}}}}%
    \put(0,0){\includegraphics[width=\unitlength,page=2]{assocstrq.pdf}}%
    \put(0.1154742,0.15933204){\color[rgb]{0,0,0}\makebox(0,0)[lt]{\lineheight{1.25}\smash{\begin{tabular}[t]{l}$f_1$\end{tabular}}}}%
    \put(0,0){\includegraphics[width=\unitlength,page=3]{assocstrq.pdf}}%
    \put(0.1154742,0.25185956){\color[rgb]{0,0,0}\makebox(0,0)[lt]{\lineheight{1.25}\smash{\begin{tabular}[t]{l}$f_2$\end{tabular}}}}%
    \put(0,0){\includegraphics[width=\unitlength,page=4]{assocstrq.pdf}}%
    \put(0.1154742,0.0482993){\color[rgb]{0,0,0}\makebox(0,0)[lt]{\lineheight{1.25}\smash{\begin{tabular}[t]{l}$f_0$\end{tabular}}}}%
    \put(0,0){\includegraphics[width=\unitlength,page=5]{assocstrq.pdf}}%
    \put(0.49668677,0.15933204){\color[rgb]{0,0,0}\makebox(0,0)[lt]{\lineheight{1.25}\smash{\begin{tabular}[t]{l}$f_1$\end{tabular}}}}%
    \put(0,0){\includegraphics[width=\unitlength,page=6]{assocstrq.pdf}}%
    \put(0.49668677,0.270365){\color[rgb]{0,0,0}\makebox(0,0)[lt]{\lineheight{1.25}\smash{\begin{tabular}[t]{l}$f_2$\end{tabular}}}}%
    \put(0,0){\includegraphics[width=\unitlength,page=7]{assocstrq.pdf}}%
    \put(0.49668677,0.06680463){\color[rgb]{0,0,0}\makebox(0,0)[lt]{\lineheight{1.25}\smash{\begin{tabular}[t]{l}$f_0$\end{tabular}}}}%
    \put(0,0){\includegraphics[width=\unitlength,page=8]{assocstrq.pdf}}%
    \put(0.67493491,0.15588564){\color[rgb]{0,0,0}\makebox(0,0)[lt]{\lineheight{1.25}\smash{\begin{tabular}[t]{l}$=$\end{tabular}}}}%
    \put(0,0){\includegraphics[width=\unitlength,page=9]{assocstrq.pdf}}%
    \put(0.86856076,0.15909245){\color[rgb]{0,0,0}\makebox(0,0)[lt]{\lineheight{1.25}\smash{\begin{tabular}[t]{l}$f_1$\end{tabular}}}}%
    \put(0,0){\includegraphics[width=\unitlength,page=10]{assocstrq.pdf}}%
    \put(0.86856076,0.2701253){\color[rgb]{0,0,0}\makebox(0,0)[lt]{\lineheight{1.25}\smash{\begin{tabular}[t]{l}$f_2$\end{tabular}}}}%
    \put(0,0){\includegraphics[width=\unitlength,page=11]{assocstrq.pdf}}%
    \put(0.86856076,0.06656504){\color[rgb]{0,0,0}\makebox(0,0)[lt]{\lineheight{1.25}\smash{\begin{tabular}[t]{l}$f_0$\end{tabular}}}}%
  \end{picture}%
\endgroup%

%% file: pics/unitors.pdf_tex
%% Creator: Inkscape 1.2.2 (732a01da63, 2022-12-09), www.inkscape.org
%% Accompanies image file 'unitors.pdf' (pdf, eps, ps)
\begingroup%
  \makeatletter%
  \providecommand\color[2][]{%
    \errmessage{(Inkscape) Color is used for the text in Inkscape, but the package 'color.sty' is not loaded}%
    \renewcommand\color[2][]{}%
  }%
  \providecommand\transparent[1]{%
    \errmessage{(Inkscape) Transparency is used (non-zero) for the text in Inkscape, but the package 'transparent.sty' is not loaded}%
    \renewcommand\transparent[1]{}%
  }%
  \providecommand\rotatebox[2]{#2}%
  \newcommand*\fsize{\dimexpr\f@size pt\relax}%
  \newcommand*\lineheight[1]{\fontsize{\fsize}{#1\fsize}\selectfont}%
  \ifx\svgwidth\undefined%
    \setlength{\unitlength}{146.25bp}%
    \ifx\svgscale\undefined%
      \relax%
    \else%
      \setlength{\unitlength}{\unitlength * \real{\svgscale}}%
    \fi%
  \else%
    \setlength{\unitlength}{\svgwidth}%
  \fi%
  \global\let\svgwidth\undefined%
  \global\let\svgscale\undefined%
  \makeatother%
  \begin{picture}(1,0.6435897)%
    \lineheight{1}%
    \setlength\tabcolsep{0pt}%
    \put(0,0){\includegraphics[width=\unitlength,page=1]{unitors.pdf}}%
    \put(0.12923093,0.06692305){\color[rgb]{0,0,0}\makebox(0,0)[lt]{\lineheight{1.25}\smash{\begin{tabular}[t]{l}$f$\end{tabular}}}}%
    \put(0,0){\includegraphics[width=\unitlength,page=2]{unitors.pdf}}%
    \put(0.33948725,0.14384628){\color[rgb]{0,0,0}\makebox(0,0)[lt]{\lineheight{1.25}\smash{\begin{tabular}[t]{l}$\rho$\end{tabular}}}}%
    \put(0,0){\includegraphics[width=\unitlength,page=3]{unitors.pdf}}%
    \put(0.66256418,0.14384628){\color[rgb]{0,0,0}\makebox(0,0)[lt]{\lineheight{1.25}\smash{\begin{tabular}[t]{l}$\rho$\end{tabular}}}}%
    \put(0.49684179,0.1394028){\color[rgb]{0,0,0}\makebox(0,0)[lt]{\lineheight{1.25}\smash{\begin{tabular}[t]{l}$=$\end{tabular}}}}%
    \put(0.12942908,0.25298986){\color[rgb]{0,0,0}\makebox(0,0)[lt]{\lineheight{1.25}\smash{\begin{tabular}[t]{l}$\id_I$\end{tabular}}}}%
    \put(0,0){\includegraphics[width=\unitlength,page=4]{unitors.pdf}}%
    \put(0.85230773,0.14384628){\color[rgb]{0,0,0}\makebox(0,0)[lt]{\lineheight{1.25}\smash{\begin{tabular}[t]{l}$f$\end{tabular}}}}%
    \put(0.12923085,0.56435906){\color[rgb]{0,0,0}\makebox(0,0)[lt]{\lineheight{1.25}\smash{\begin{tabular}[t]{l}$f$\end{tabular}}}}%
    \put(0,0){\includegraphics[width=\unitlength,page=5]{unitors.pdf}}%
    \put(0.33948717,0.47717956){\color[rgb]{0,0,0}\makebox(0,0)[lt]{\lineheight{1.25}\smash{\begin{tabular}[t]{l}$\lambda$\end{tabular}}}}%
    \put(0,0){\includegraphics[width=\unitlength,page=6]{unitors.pdf}}%
    \put(0.49684172,0.47273609){\color[rgb]{0,0,0}\makebox(0,0)[lt]{\lineheight{1.25}\smash{\begin{tabular}[t]{l}$=$\end{tabular}}}}%
    \put(0.12942901,0.43247684){\color[rgb]{0,0,0}\makebox(0,0)[lt]{\lineheight{1.25}\smash{\begin{tabular}[t]{l}$\id_I$\end{tabular}}}}%
    \put(0,0){\includegraphics[width=\unitlength,page=7]{unitors.pdf}}%
    \put(0.85230773,0.47717956){\color[rgb]{0,0,0}\makebox(0,0)[lt]{\lineheight{1.25}\smash{\begin{tabular}[t]{l}$f$\end{tabular}}}}%
    \put(0.66256404,0.47717956){\color[rgb]{0,0,0}\makebox(0,0)[lt]{\lineheight{1.25}\smash{\begin{tabular}[t]{l}$\lambda$\end{tabular}}}}%
  \end{picture}%
\endgroup%

%% file: pics/unitorsstr.pdf_tex
%% Creator: Inkscape 1.0.1 (c497b03c, 2020-09-10), www.inkscape.org
%% Accompanies image file 'unitorsstr.pdf' (pdf, eps, ps)
\begingroup%
  \makeatletter%
  \providecommand\color[2][]{%
    \errmessage{(Inkscape) Color is used for the text in Inkscape, but the package 'color.sty' is not loaded}%
    \renewcommand\color[2][]{}%
  }%
  \providecommand\transparent[1]{%
    \errmessage{(Inkscape) Transparency is used (non-zero) for the text in Inkscape, but the package 'transparent.sty' is not loaded}%
    \renewcommand\transparent[1]{}%
  }%
  \providecommand\rotatebox[2]{#2}%
  \newcommand*\fsize{\dimexpr\f@size pt\relax}%
  \newcommand*\lineheight[1]{\fontsize{\fsize}{#1\fsize}\selectfont}%
  \ifx\svgwidth\undefined%
    \setlength{\unitlength}{151.35096242bp}%
    \ifx\svgscale\undefined%
      \relax%
    \else%
      \setlength{\unitlength}{\unitlength * \real{\svgscale}}%
    \fi%
  \else%
    \setlength{\unitlength}{\svgwidth}%
  \fi%
  \global\let\svgwidth\undefined%
  \global\let\svgscale\undefined%
  \makeatother%
  \begin{picture}(1,0.26160479)%
    \lineheight{1}%
    \setlength\tabcolsep{0pt}%
    \put(0,0){\includegraphics[width=\unitlength,page=1]{unitorsstr.pdf}}%
    \put(0.80871641,0.04113353){\color[rgb]{0,0,0}\makebox(0,0)[lt]{\lineheight{1.25}\smash{\begin{tabular}[t]{l}$f$\end{tabular}}}}%
    \put(0,0){\includegraphics[width=\unitlength,page=2]{unitorsstr.pdf}}%
    \put(0.59163129,0.11241551){\color[rgb]{0,0,0}\makebox(0,0)[lt]{\lineheight{1.25}\smash{\begin{tabular}[t]{l}$=$\end{tabular}}}}%
    \put(0.80890802,0.22092935){\color[rgb]{0,0,0}\makebox(0,0)[lt]{\lineheight{1.25}\smash{\begin{tabular}[t]{l}$\id_I$\end{tabular}}}}%
    \put(0.1100093,0.18979477){\color[rgb]{0,0,0}\makebox(0,0)[lt]{\lineheight{1.25}\smash{\begin{tabular}[t]{l}$f$\end{tabular}}}}%
    \put(0,0){\includegraphics[width=\unitlength,page=3]{unitorsstr.pdf}}%
    \put(0.28683713,0.1111705){\color[rgb]{0,0,0}\makebox(0,0)[lt]{\lineheight{1.25}\smash{\begin{tabular}[t]{l}$=$\end{tabular}}}}%
    \put(0.12011152,0.07226811){\color[rgb]{0,0,0}\makebox(0,0)[lt]{\lineheight{1.25}\smash{\begin{tabular}[t]{l}$\id_I$\end{tabular}}}}%
    \put(0,0){\includegraphics[width=\unitlength,page=4]{unitorsstr.pdf}}%
    \put(0.43210808,0.11546436){\color[rgb]{0,0,0}\makebox(0,0)[lt]{\lineheight{1.25}\smash{\begin{tabular}[t]{l}$f$\end{tabular}}}}%
  \end{picture}%
\endgroup%

%% file: pics/unitorsstr2.pdf_tex
%% Creator: Inkscape 1.0.1 (c497b03c, 2020-09-10), www.inkscape.org
%% Accompanies image file 'unitorsstr2.pdf' (pdf, eps, ps)
\begingroup%
  \makeatletter%
  \providecommand\color[2][]{%
    \errmessage{(Inkscape) Color is used for the text in Inkscape, but the package 'color.sty' is not loaded}%
    \renewcommand\color[2][]{}%
  }%
  \providecommand\transparent[1]{%
    \errmessage{(Inkscape) Transparency is used (non-zero) for the text in Inkscape, but the package 'transparent.sty' is not loaded}%
    \renewcommand\transparent[1]{}%
  }%
  \providecommand\rotatebox[2]{#2}%
  \newcommand*\fsize{\dimexpr\f@size pt\relax}%
  \newcommand*\lineheight[1]{\fontsize{\fsize}{#1\fsize}\selectfont}%
  \ifx\svgwidth\undefined%
    \setlength{\unitlength}{148.96519908bp}%
    \ifx\svgscale\undefined%
      \relax%
    \else%
      \setlength{\unitlength}{\unitlength * \real{\svgscale}}%
    \fi%
  \else%
    \setlength{\unitlength}{\svgwidth}%
  \fi%
  \global\let\svgwidth\undefined%
  \global\let\svgscale\undefined%
  \makeatother%
  \begin{picture}(1,0.25677948)%
    \lineheight{1}%
    \setlength\tabcolsep{0pt}%
    \put(0,0){\includegraphics[width=\unitlength,page=1]{unitorsstr2.pdf}}%
    \put(0.82166847,0.04179231){\color[rgb]{0,0,0}\makebox(0,0)[lt]{\lineheight{1.25}\smash{\begin{tabular}[t]{l}$f$\end{tabular}}}}%
    \put(0,0){\includegraphics[width=\unitlength,page=2]{unitorsstr2.pdf}}%
    \put(0.6011066,0.11421591){\color[rgb]{0,0,0}\makebox(0,0)[lt]{\lineheight{1.25}\smash{\begin{tabular}[t]{l}$=$\end{tabular}}}}%
    \put(0.80584758,0.16203354){\color[rgb]{0,0,0}\makebox(0,0)[lt]{\lineheight{1.25}\smash{\begin{tabular}[t]{l}$\id_I$\end{tabular}}}}%
    \put(0.11177116,0.19283444){\color[rgb]{0,0,0}\makebox(0,0)[lt]{\lineheight{1.25}\smash{\begin{tabular}[t]{l}$f$\end{tabular}}}}%
    \put(0,0){\includegraphics[width=\unitlength,page=3]{unitorsstr2.pdf}}%
    \put(0.29143099,0.11295096){\color[rgb]{0,0,0}\makebox(0,0)[lt]{\lineheight{1.25}\smash{\begin{tabular}[t]{l}$=$\end{tabular}}}}%
    \put(0.10189622,0.05328657){\color[rgb]{0,0,0}\makebox(0,0)[lt]{\lineheight{1.25}\smash{\begin{tabular}[t]{l}$\id_I$\end{tabular}}}}%
    \put(0,0){\includegraphics[width=\unitlength,page=4]{unitorsstr2.pdf}}%
    \put(0.43902854,0.11731359){\color[rgb]{0,0,0}\makebox(0,0)[lt]{\lineheight{1.25}\smash{\begin{tabular}[t]{l}$f$\end{tabular}}}}%
    \put(0,0){\includegraphics[width=\unitlength,page=5]{unitorsstr2.pdf}}%
  \end{picture}%
\endgroup%

%% file: pics/exred.pdf_tex
%% Creator: Inkscape 1.1 (c68e22c387, 2021-05-23), www.inkscape.org
%% Accompanies image file 'exred.pdf' (pdf, eps, ps)
\begingroup%
  \makeatletter%
  \providecommand\color[2][]{%
    \errmessage{(Inkscape) Color is used for the text in Inkscape, but the package 'color.sty' is not loaded}%
    \renewcommand\color[2][]{}%
  }%
  \providecommand\transparent[1]{%
    \errmessage{(Inkscape) Transparency is used (non-zero) for the text in Inkscape, but the package 'transparent.sty' is not loaded}%
    \renewcommand\transparent[1]{}%
  }%
  \providecommand\rotatebox[2]{#2}%
  \newcommand*\fsize{\dimexpr\f@size pt\relax}%
  \newcommand*\lineheight[1]{\fontsize{\fsize}{#1\fsize}\selectfont}%
  \ifx\svgwidth\undefined%
    \setlength{\unitlength}{133.74048857bp}%
    \ifx\svgscale\undefined%
      \relax%
    \else%
      \setlength{\unitlength}{\unitlength * \real{\svgscale}}%
    \fi%
  \else%
    \setlength{\unitlength}{\svgwidth}%
  \fi%
  \global\let\svgwidth\undefined%
  \global\let\svgscale\undefined%
  \makeatother%
  \begin{picture}(1,0.51032601)%
    \lineheight{1}%
    \setlength\tabcolsep{0pt}%
    \put(0,0){\includegraphics[width=\unitlength,page=1]{exred.pdf}}%
    \put(0.06561654,0.43910183){\color[rgb]{0,0,0}\makebox(0,0)[lt]{\lineheight{1.25}\smash{\begin{tabular}[t]{l}$\mathsf  t$\end{tabular}}}}%
    \put(0,0){\includegraphics[width=\unitlength,page=2]{exred.pdf}}%
    \put(0.06561654,0.25964964){\color[rgb]{0,0,0}\makebox(0,0)[lt]{\lineheight{1.25}\smash{\begin{tabular}[t]{l}$\mathsf  f$\end{tabular}}}}%
    \put(0,0){\includegraphics[width=\unitlength,page=3]{exred.pdf}}%
    \put(0.06561662,0.03533409){\color[rgb]{0,0,0}\makebox(0,0)[lt]{\lineheight{1.25}\smash{\begin{tabular}[t]{l}$\mathsf  t$\end{tabular}}}}%
    \put(0,0){\includegraphics[width=\unitlength,page=4]{exred.pdf}}%
    \put(0.34601026,0.14749235){\color[rgb]{0,0,0}\makebox(0,0)[lt]{\lineheight{1.25}\smash{\begin{tabular}[t]{l}$\mathsf  f$\end{tabular}}}}%
    \put(0,0){\includegraphics[width=\unitlength,page=5]{exred.pdf}}%
    \put(0.34601026,0.41106218){\color[rgb]{0,0,0}\makebox(0,0)[lt]{\lineheight{1.25}\smash{\begin{tabular}[t]{l}$\lor$\end{tabular}}}}%
    \put(0,0){\includegraphics[width=\unitlength,page=6]{exred.pdf}}%
    \put(0.82267965,0.3830227){\color[rgb]{0,0,0}\makebox(0,0)[lt]{\lineheight{1.25}\smash{\begin{tabular}[t]{l}$\lor$\end{tabular}}}}%
    \put(0,0){\includegraphics[width=\unitlength,page=7]{exred.pdf}}%
    \put(0.57032525,0.1306683){\color[rgb]{0,0,0}\makebox(0,0)[lt]{\lineheight{1.25}\smash{\begin{tabular}[t]{l}$\land$\end{tabular}}}}%
  \end{picture}%
\endgroup%

%% file: pics/functens.pdf_tex
%% Creator: Inkscape 1.0.1 (c497b03c, 2020-09-10), www.inkscape.org
%% Accompanies image file 'functens.pdf' (pdf, eps, ps)
\begingroup%
  \makeatletter%
  \providecommand\color[2][]{%
    \errmessage{(Inkscape) Color is used for the text in Inkscape, but the package 'color.sty' is not loaded}%
    \renewcommand\color[2][]{}%
  }%
  \providecommand\transparent[1]{%
    \errmessage{(Inkscape) Transparency is used (non-zero) for the text in Inkscape, but the package 'transparent.sty' is not loaded}%
    \renewcommand\transparent[1]{}%
  }%
  \providecommand\rotatebox[2]{#2}%
  \newcommand*\fsize{\dimexpr\f@size pt\relax}%
  \newcommand*\lineheight[1]{\fontsize{\fsize}{#1\fsize}\selectfont}%
  \ifx\svgwidth\undefined%
    \setlength{\unitlength}{67.50000567bp}%
    \ifx\svgscale\undefined%
      \relax%
    \else%
      \setlength{\unitlength}{\unitlength * \real{\svgscale}}%
    \fi%
  \else%
    \setlength{\unitlength}{\svgwidth}%
  \fi%
  \global\let\svgwidth\undefined%
  \global\let\svgscale\undefined%
  \makeatother%
  \begin{picture}(1,0.51112862)%
    \lineheight{1}%
    \setlength\tabcolsep{0pt}%
    \put(0,0){\includegraphics[width=\unitlength,page=1]{functens.pdf}}%
    \put(0.23555559,0.37000896){\color[rgb]{0,0,0}\makebox(0,0)[lt]{\lineheight{1.25}\smash{\begin{tabular}[t]{l}$g$\end{tabular}}}}%
    \put(0,0){\includegraphics[width=\unitlength,page=2]{functens.pdf}}%
    \put(0.69111091,0.37000896){\color[rgb]{0,0,0}\makebox(0,0)[lt]{\lineheight{1.25}\smash{\begin{tabular}[t]{l}$g'$\end{tabular}}}}%
    \put(0,0){\includegraphics[width=\unitlength,page=3]{functens.pdf}}%
    \put(0.23555559,0.09223109){\color[rgb]{0,0,0}\makebox(0,0)[lt]{\lineheight{1.25}\smash{\begin{tabular}[t]{l}$f$\end{tabular}}}}%
    \put(0,0){\includegraphics[width=\unitlength,page=4]{functens.pdf}}%
    \put(0.69111107,0.09223109){\color[rgb]{0,0,0}\makebox(0,0)[lt]{\lineheight{1.25}\smash{\begin{tabular}[t]{l}$f'$\end{tabular}}}}%
  \end{picture}%
\endgroup%

%% file: pics/phi.pdf_tex
%% Creator: Inkscape 1.2.1 (9c6d41e4, 2022-07-14), www.inkscape.org
%% Accompanies image file 'phi.pdf' (pdf, eps, ps)
\begingroup%
  \makeatletter%
  \providecommand\color[2][]{%
    \errmessage{(Inkscape) Color is used for the text in Inkscape, but the package 'color.sty' is not loaded}%
    \renewcommand\color[2][]{}%
  }%
  \providecommand\transparent[1]{%
    \errmessage{(Inkscape) Transparency is used (non-zero) for the text in Inkscape, but the package 'transparent.sty' is not loaded}%
    \renewcommand\transparent[1]{}%
  }%
  \providecommand\rotatebox[2]{#2}%
  \newcommand*\fsize{\dimexpr\f@size pt\relax}%
  \newcommand*\lineheight[1]{\fontsize{\fsize}{#1\fsize}\selectfont}%
  \ifx\svgwidth\undefined%
    \setlength{\unitlength}{24.88678849bp}%
    \ifx\svgscale\undefined%
      \relax%
    \else%
      \setlength{\unitlength}{\unitlength * \real{\svgscale}}%
    \fi%
  \else%
    \setlength{\unitlength}{\svgwidth}%
  \fi%
  \global\let\svgwidth\undefined%
  \global\let\svgscale\undefined%
  \makeatother%
  \begin{picture}(1,0.3268107)%
    \lineheight{1}%
    \setlength\tabcolsep{0pt}%
    \put(-0.02432899,0.06466803){\color[rgb]{0,0,0}\makebox(0,0)[lt]{\lineheight{1.25}\smash{\begin{tabular}[t]{l}$\phi$\end{tabular}}}}%
    \put(0,0){\includegraphics[width=\unitlength,page=1]{phi.pdf}}%
  \end{picture}%
\endgroup%

%% file: pics/psi.pdf_tex
%% Creator: Inkscape 1.2.1 (9c6d41e4, 2022-07-14), www.inkscape.org
%% Accompanies image file 'psi.pdf' (pdf, eps, ps)
\begingroup%
  \makeatletter%
  \providecommand\color[2][]{%
    \errmessage{(Inkscape) Color is used for the text in Inkscape, but the package 'color.sty' is not loaded}%
    \renewcommand\color[2][]{}%
  }%
  \providecommand\transparent[1]{%
    \errmessage{(Inkscape) Transparency is used (non-zero) for the text in Inkscape, but the package 'transparent.sty' is not loaded}%
    \renewcommand\transparent[1]{}%
  }%
  \providecommand\rotatebox[2]{#2}%
  \newcommand*\fsize{\dimexpr\f@size pt\relax}%
  \newcommand*\lineheight[1]{\fontsize{\fsize}{#1\fsize}\selectfont}%
  \ifx\svgwidth\undefined%
    \setlength{\unitlength}{23.9922543bp}%
    \ifx\svgscale\undefined%
      \relax%
    \else%
      \setlength{\unitlength}{\unitlength * \real{\svgscale}}%
    \fi%
  \else%
    \setlength{\unitlength}{\svgwidth}%
  \fi%
  \global\let\svgwidth\undefined%
  \global\let\svgscale\undefined%
  \makeatother%
  \begin{picture}(1,0.93780266)%
    \lineheight{1}%
    \setlength\tabcolsep{0pt}%
    \put(-0.02523608,0.39778665){\color[rgb]{0,0,0}\makebox(0,0)[lt]{\lineheight{1.25}\smash{\begin{tabular}[t]{l}$\psi$\end{tabular}}}}%
    \put(0,0){\includegraphics[width=\unitlength,page=1]{psi.pdf}}%
  \end{picture}%
\endgroup%

%% file: pics/phistar.pdf_tex
%% Creator: Inkscape 1.2.1 (9c6d41e4, 2022-07-14), www.inkscape.org
%% Accompanies image file 'phistar.pdf' (pdf, eps, ps)
\begingroup%
  \makeatletter%
  \providecommand\color[2][]{%
    \errmessage{(Inkscape) Color is used for the text in Inkscape, but the package 'color.sty' is not loaded}%
    \renewcommand\color[2][]{}%
  }%
  \providecommand\transparent[1]{%
    \errmessage{(Inkscape) Transparency is used (non-zero) for the text in Inkscape, but the package 'transparent.sty' is not loaded}%
    \renewcommand\transparent[1]{}%
  }%
  \providecommand\rotatebox[2]{#2}%
  \newcommand*\fsize{\dimexpr\f@size pt\relax}%
  \newcommand*\lineheight[1]{\fontsize{\fsize}{#1\fsize}\selectfont}%
  \ifx\svgwidth\undefined%
    \setlength{\unitlength}{46.38964859bp}%
    \ifx\svgscale\undefined%
      \relax%
    \else%
      \setlength{\unitlength}{\unitlength * \real{\svgscale}}%
    \fi%
  \else%
    \setlength{\unitlength}{\svgwidth}%
  \fi%
  \global\let\svgwidth\undefined%
  \global\let\svgscale\undefined%
  \makeatother%
  \begin{picture}(1,0.18464339)%
    \lineheight{1}%
    \setlength\tabcolsep{0pt}%
    \put(0.19971148,0.03469264){\color[rgb]{0,0,0}\makebox(0,0)[lt]{\lineheight{1.25}\smash{\begin{tabular}[t]{l}$\phi^*$\end{tabular}}}}%
    \put(0,0){\includegraphics[width=\unitlength,page=1]{phistar.pdf}}%
  \end{picture}%
\endgroup%

%% file: pics/psistar.pdf_tex
%% Creator: Inkscape 1.2.1 (9c6d41e4, 2022-07-14), www.inkscape.org
%% Accompanies image file 'psistar.pdf' (pdf, eps, ps)
\begingroup%
  \makeatletter%
  \providecommand\color[2][]{%
    \errmessage{(Inkscape) Color is used for the text in Inkscape, but the package 'color.sty' is not loaded}%
    \renewcommand\color[2][]{}%
  }%
  \providecommand\transparent[1]{%
    \errmessage{(Inkscape) Transparency is used (non-zero) for the text in Inkscape, but the package 'transparent.sty' is not loaded}%
    \renewcommand\transparent[1]{}%
  }%
  \providecommand\rotatebox[2]{#2}%
  \newcommand*\fsize{\dimexpr\f@size pt\relax}%
  \newcommand*\lineheight[1]{\fontsize{\fsize}{#1\fsize}\selectfont}%
  \ifx\svgwidth\undefined%
    \setlength{\unitlength}{47.8077959bp}%
    \ifx\svgscale\undefined%
      \relax%
    \else%
      \setlength{\unitlength}{\unitlength * \real{\svgscale}}%
    \fi%
  \else%
    \setlength{\unitlength}{\svgwidth}%
  \fi%
  \global\let\svgwidth\undefined%
  \global\let\svgscale\undefined%
  \makeatother%
  \begin{picture}(1,0.47063454)%
    \lineheight{1}%
    \setlength\tabcolsep{0pt}%
    \put(0.24216185,0.2086704){\color[rgb]{0,0,0}\makebox(0,0)[lt]{\lineheight{1.25}\smash{\begin{tabular}[t]{l}$\psi^*$\end{tabular}}}}%
    \put(0,0){\includegraphics[width=\unitlength,page=1]{psistar.pdf}}%
  \end{picture}%
\endgroup%

%% file: pics/functensns.pdf_tex
%% Creator: Inkscape 1.2.1 (9c6d41e4, 2022-07-14), www.inkscape.org
%% Accompanies image file 'functensns.pdf' (pdf, eps, ps)
\begingroup%
  \makeatletter%
  \providecommand\color[2][]{%
    \errmessage{(Inkscape) Color is used for the text in Inkscape, but the package 'color.sty' is not loaded}%
    \renewcommand\color[2][]{}%
  }%
  \providecommand\transparent[1]{%
    \errmessage{(Inkscape) Transparency is used (non-zero) for the text in Inkscape, but the package 'transparent.sty' is not loaded}%
    \renewcommand\transparent[1]{}%
  }%
  \providecommand\rotatebox[2]{#2}%
  \newcommand*\fsize{\dimexpr\f@size pt\relax}%
  \newcommand*\lineheight[1]{\fontsize{\fsize}{#1\fsize}\selectfont}%
  \ifx\svgwidth\undefined%
    \setlength{\unitlength}{218.14369978bp}%
    \ifx\svgscale\undefined%
      \relax%
    \else%
      \setlength{\unitlength}{\unitlength * \real{\svgscale}}%
    \fi%
  \else%
    \setlength{\unitlength}{\svgwidth}%
  \fi%
  \global\let\svgwidth\undefined%
  \global\let\svgscale\undefined%
  \makeatother%
  \begin{picture}(1,0.39882524)%
    \lineheight{1}%
    \setlength\tabcolsep{0pt}%
    \put(0,0){\includegraphics[width=\unitlength,page=1]{functensns.pdf}}%
    \put(-0.00277556,0.30565951){\color[rgb]{0,0,0}\makebox(0,0)[lt]{\lineheight{1.25}\smash{\begin{tabular}[t]{l}$(f\otimes g);(f'\otimes g') =$\end{tabular}}}}%
    \put(0,0){\includegraphics[width=\unitlength,page=2]{functensns.pdf}}%
    \put(0.49872506,0.35515871){\color[rgb]{0,0,0}\makebox(0,0)[lt]{\lineheight{1.25}\smash{\begin{tabular}[t]{l}$g$\end{tabular}}}}%
    \put(0,0){\includegraphics[width=\unitlength,page=3]{functensns.pdf}}%
    \put(0.49872506,0.25201569){\color[rgb]{0,0,0}\makebox(0,0)[lt]{\lineheight{1.25}\smash{\begin{tabular}[t]{l}$f$\end{tabular}}}}%
    \put(0,0){\includegraphics[width=\unitlength,page=4]{functensns.pdf}}%
    \put(0.80815383,0.35515881){\color[rgb]{0,0,0}\makebox(0,0)[lt]{\lineheight{1.25}\smash{\begin{tabular}[t]{l}$g'$\end{tabular}}}}%
    \put(0,0){\includegraphics[width=\unitlength,page=5]{functensns.pdf}}%
    \put(0.80815383,0.25201569){\color[rgb]{0,0,0}\makebox(0,0)[lt]{\lineheight{1.25}\smash{\begin{tabular}[t]{l}$f'$\end{tabular}}}}%
    \put(0,0){\includegraphics[width=\unitlength,page=6]{functensns.pdf}}%
    \put(0.00410075,0.082183){\color[rgb]{0,0,0}\makebox(0,0)[lt]{\lineheight{1.25}\smash{\begin{tabular}[t]{l}$(f;f')\otimes (g;g') =$\end{tabular}}}}%
    \put(0,0){\includegraphics[width=\unitlength,page=7]{functensns.pdf}}%
    \put(0.50560136,0.1316822){\color[rgb]{0,0,0}\makebox(0,0)[lt]{\lineheight{1.25}\smash{\begin{tabular}[t]{l}$g$\end{tabular}}}}%
    \put(0,0){\includegraphics[width=\unitlength,page=8]{functensns.pdf}}%
    \put(0.50560136,0.02853918){\color[rgb]{0,0,0}\makebox(0,0)[lt]{\lineheight{1.25}\smash{\begin{tabular}[t]{l}$f$\end{tabular}}}}%
    \put(0,0){\includegraphics[width=\unitlength,page=9]{functensns.pdf}}%
    \put(0.81503013,0.1316823){\color[rgb]{0,0,0}\makebox(0,0)[lt]{\lineheight{1.25}\smash{\begin{tabular}[t]{l}$g'$\end{tabular}}}}%
    \put(0,0){\includegraphics[width=\unitlength,page=10]{functensns.pdf}}%
    \put(0.81503013,0.02853918){\color[rgb]{0,0,0}\makebox(0,0)[lt]{\lineheight{1.25}\smash{\begin{tabular}[t]{l}$f'$\end{tabular}}}}%
    \put(0,0){\includegraphics[width=\unitlength,page=11]{functensns.pdf}}%
  \end{picture}%
\endgroup%

%% file: pics/firstfourth.pdf_tex
%% Creator: Inkscape 1.2.1 (9c6d41e4, 2022-07-14), www.inkscape.org
%% Accompanies image file 'firstfourth.pdf' (pdf, eps, ps)
\begingroup%
  \makeatletter%
  \providecommand\color[2][]{%
    \errmessage{(Inkscape) Color is used for the text in Inkscape, but the package 'color.sty' is not loaded}%
    \renewcommand\color[2][]{}%
  }%
  \providecommand\transparent[1]{%
    \errmessage{(Inkscape) Transparency is used (non-zero) for the text in Inkscape, but the package 'transparent.sty' is not loaded}%
    \renewcommand\transparent[1]{}%
  }%
  \providecommand\rotatebox[2]{#2}%
  \newcommand*\fsize{\dimexpr\f@size pt\relax}%
  \newcommand*\lineheight[1]{\fontsize{\fsize}{#1\fsize}\selectfont}%
  \ifx\svgwidth\undefined%
    \setlength{\unitlength}{213.65220756bp}%
    \ifx\svgscale\undefined%
      \relax%
    \else%
      \setlength{\unitlength}{\unitlength * \real{\svgscale}}%
    \fi%
  \else%
    \setlength{\unitlength}{\svgwidth}%
  \fi%
  \global\let\svgwidth\undefined%
  \global\let\svgscale\undefined%
  \makeatother%
  \begin{picture}(1,0.17903479)%
    \lineheight{1}%
    \setlength\tabcolsep{0pt}%
    \put(0,0){\includegraphics[width=\unitlength,page=1]{firstfourth.pdf}}%
    \put(0.15852908,0.07893859){\color[rgb]{0,0,0}\makebox(0,0)[lt]{\lineheight{1.25}\smash{\begin{tabular}[t]{l}$=$\end{tabular}}}}%
    \put(0,0){\includegraphics[width=\unitlength,page=2]{firstfourth.pdf}}%
    \put(0.53102177,0.07893869){\color[rgb]{0,0,0}\makebox(0,0)[lt]{\lineheight{1.25}\smash{\begin{tabular}[t]{l}$=$\end{tabular}}}}%
    \put(0,0){\includegraphics[width=\unitlength,page=3]{firstfourth.pdf}}%
    \put(0.31398911,0.06775309){\color[rgb]{0,0,0}\makebox(0,0)[lt]{\lineheight{1.25}\smash{\begin{tabular}[t]{l}$\overline{\id\otimes\id}$\end{tabular}}}}%
    \put(0,0){\includegraphics[width=\unitlength,page=4]{firstfourth.pdf}}%
    \put(0.68125762,0.12040876){\color[rgb]{0,0,0}\makebox(0,0)[lt]{\lineheight{1.25}\smash{\begin{tabular}[t]{l}$\overline{\id}$\end{tabular}}}}%
    \put(0,0){\includegraphics[width=\unitlength,page=5]{firstfourth.pdf}}%
    \put(0.68125768,0.01509742){\color[rgb]{0,0,0}\makebox(0,0)[lt]{\lineheight{1.25}\smash{\begin{tabular}[t]{l}$\overline{\id}$\end{tabular}}}}%
    \put(0,0){\includegraphics[width=\unitlength,page=6]{firstfourth.pdf}}%
    \put(0.81318093,0.07893889){\color[rgb]{0,0,0}\makebox(0,0)[lt]{\lineheight{1.25}\smash{\begin{tabular}[t]{l}$=$\end{tabular}}}}%
    \put(0,0){\includegraphics[width=\unitlength,page=7]{firstfourth.pdf}}%
  \end{picture}%
\endgroup%

%% file: pics/drawing.pdf_tex
%% Creator: Inkscape 1.0.1 (c497b03c, 2020-09-10), www.inkscape.org
%% Accompanies image file 'drawing.pdf' (pdf, eps, ps)
\begingroup%
  \makeatletter%
  \providecommand\color[2][]{%
    \errmessage{(Inkscape) Color is used for the text in Inkscape, but the package 'color.sty' is not loaded}%
    \renewcommand\color[2][]{}%
  }%
  \providecommand\transparent[1]{%
    \errmessage{(Inkscape) Transparency is used (non-zero) for the text in Inkscape, but the package 'transparent.sty' is not loaded}%
    \renewcommand\transparent[1]{}%
  }%
  \providecommand\rotatebox[2]{#2}%
  \newcommand*\fsize{\dimexpr\f@size pt\relax}%
  \newcommand*\lineheight[1]{\fontsize{\fsize}{#1\fsize}\selectfont}%
  \ifx\svgwidth\undefined%
    \setlength{\unitlength}{177.25877487bp}%
    \ifx\svgscale\undefined%
      \relax%
    \else%
      \setlength{\unitlength}{\unitlength * \real{\svgscale}}%
    \fi%
  \else%
    \setlength{\unitlength}{\svgwidth}%
  \fi%
  \global\let\svgwidth\undefined%
  \global\let\svgscale\undefined%
  \makeatother%
  \begin{picture}(1,0.16234595)%
    \lineheight{1}%
    \setlength\tabcolsep{0pt}%
    \put(0,0){\includegraphics[width=\unitlength,page=1]{drawing.pdf}}%
    \put(0.1313366,0.07431317){\color[rgb]{0,0,0}\makebox(0,0)[lt]{\lineheight{1.25}\smash{\begin{tabular}[t]{l}$\sigma^*$\end{tabular}}}}%
    \put(0,0){\includegraphics[width=\unitlength,page=2]{drawing.pdf}}%
    \put(0.46325727,0.07431317){\color[rgb]{0,0,0}\makebox(0,0)[lt]{\lineheight{1.25}\smash{\begin{tabular}[t]{l}$f$\end{tabular}}}}%
    \put(0,0){\includegraphics[width=\unitlength,page=3]{drawing.pdf}}%
    \put(0.79517794,0.07431317){\color[rgb]{0,0,0}\makebox(0,0)[lt]{\lineheight{1.25}\smash{\begin{tabular}[t]{l}$\tau$\end{tabular}}}}%
  \end{picture}%
\endgroup%

%% file: pics/threeways.pdf_tex
%% Creator: Inkscape 1.2.1 (9c6d41e4, 2022-07-14), www.inkscape.org
%% Accompanies image file 'threeways.pdf' (pdf, eps, ps)
\begingroup%
  \makeatletter%
  \providecommand\color[2][]{%
    \errmessage{(Inkscape) Color is used for the text in Inkscape, but the package 'color.sty' is not loaded}%
    \renewcommand\color[2][]{}%
  }%
  \providecommand\transparent[1]{%
    \errmessage{(Inkscape) Transparency is used (non-zero) for the text in Inkscape, but the package 'transparent.sty' is not loaded}%
    \renewcommand\transparent[1]{}%
  }%
  \providecommand\rotatebox[2]{#2}%
  \newcommand*\fsize{\dimexpr\f@size pt\relax}%
  \newcommand*\lineheight[1]{\fontsize{\fsize}{#1\fsize}\selectfont}%
  \ifx\svgwidth\undefined%
    \setlength{\unitlength}{307.4062999bp}%
    \ifx\svgscale\undefined%
      \relax%
    \else%
      \setlength{\unitlength}{\unitlength * \real{\svgscale}}%
    \fi%
  \else%
    \setlength{\unitlength}{\svgwidth}%
  \fi%
  \global\let\svgwidth\undefined%
  \global\let\svgscale\undefined%
  \makeatother%
  \begin{picture}(1,0.07917282)%
    \lineheight{1}%
    \setlength\tabcolsep{0pt}%
    \put(0,0){\includegraphics[width=\unitlength,page=1]{threeways.pdf}}%
    \put(0.15371914,0.03578983){\color[rgb]{0,0,0}\makebox(0,0)[lt]{\lineheight{1.25}\smash{\begin{tabular}[t]{l}$f$\end{tabular}}}}%
    \put(0,0){\includegraphics[width=\unitlength,page=2]{threeways.pdf}}%
    \put(-0.00147721,0.04979592){\color[rgb]{0,0,0}\makebox(0,0)[lt]{\lineheight{1.25}\smash{\begin{tabular}[t]{l}$A\otimes B$\end{tabular}}}}%
    \put(0.25537612,0.05160312){\color[rgb]{0,0,0}\makebox(0,0)[lt]{\lineheight{1.25}\smash{\begin{tabular}[t]{l}$C$\end{tabular}}}}%
    \put(0.29739438,0.03804883){\color[rgb]{0,0,0}\makebox(0,0)[lt]{\lineheight{1.25}\smash{\begin{tabular}[t]{l}$=$\end{tabular}}}}%
    \put(0,0){\includegraphics[width=\unitlength,page=3]{threeways.pdf}}%
    \put(0.4611298,0.03578983){\color[rgb]{0,0,0}\makebox(0,0)[lt]{\lineheight{1.25}\smash{\begin{tabular}[t]{l}$f$\end{tabular}}}}%
    \put(0,0){\includegraphics[width=\unitlength,page=4]{threeways.pdf}}%
    \put(0.3352108,0.06199473){\color[rgb]{0,0,0}\makebox(0,0)[lt]{\lineheight{1.25}\smash{\begin{tabular}[t]{l}$B$\end{tabular}}}}%
    \put(0.56278674,0.05160312){\color[rgb]{0,0,0}\makebox(0,0)[lt]{\lineheight{1.25}\smash{\begin{tabular}[t]{l}$C$\end{tabular}}}}%
    \put(0.60480528,0.03804883){\color[rgb]{0,0,0}\makebox(0,0)[lt]{\lineheight{1.25}\smash{\begin{tabular}[t]{l}$=$\end{tabular}}}}%
    \put(0,0){\includegraphics[width=\unitlength,page=5]{threeways.pdf}}%
    \put(0.33521073,0.00344046){\color[rgb]{0,0,0}\makebox(0,0)[lt]{\lineheight{1.25}\smash{\begin{tabular}[t]{l}$A$\end{tabular}}}}%
    \put(0,0){\includegraphics[width=\unitlength,page=6]{threeways.pdf}}%
    \put(0.86125126,0.03578969){\color[rgb]{0,0,0}\makebox(0,0)[lt]{\lineheight{1.25}\smash{\begin{tabular}[t]{l}$f$\end{tabular}}}}%
    \put(0,0){\includegraphics[width=\unitlength,page=7]{threeways.pdf}}%
    \put(0.63774195,0.06199473){\color[rgb]{0,0,0}\makebox(0,0)[lt]{\lineheight{1.25}\smash{\begin{tabular}[t]{l}$B$\end{tabular}}}}%
    \put(0.96290785,0.05160298){\color[rgb]{0,0,0}\makebox(0,0)[lt]{\lineheight{1.25}\smash{\begin{tabular}[t]{l}$C$\end{tabular}}}}%
    \put(0,0){\includegraphics[width=\unitlength,page=8]{threeways.pdf}}%
    \put(0.63774195,0.00344046){\color[rgb]{0,0,0}\makebox(0,0)[lt]{\lineheight{1.25}\smash{\begin{tabular}[t]{l}$A$\end{tabular}}}}%
    \put(0,0){\includegraphics[width=\unitlength,page=9]{threeways.pdf}}%
    \put(0.70063398,0.05047362){\color[rgb]{0,0,0}\makebox(0,0)[lt]{\lineheight{1.25}\smash{\begin{tabular}[t]{l}$A\otimes B$\end{tabular}}}}%
    \put(0,0){\includegraphics[width=\unitlength,page=10]{threeways.pdf}}%
  \end{picture}%
\endgroup%

%% file: pics/sigma22.pdf_tex
%% Creator: Inkscape 1.2.1 (9c6d41e4, 2022-07-14), www.inkscape.org
%% Accompanies image file 'sigma22.pdf' (pdf, eps, ps)
\begingroup%
  \makeatletter%
  \providecommand\color[2][]{%
    \errmessage{(Inkscape) Color is used for the text in Inkscape, but the package 'color.sty' is not loaded}%
    \renewcommand\color[2][]{}%
  }%
  \providecommand\transparent[1]{%
    \errmessage{(Inkscape) Transparency is used (non-zero) for the text in Inkscape, but the package 'transparent.sty' is not loaded}%
    \renewcommand\transparent[1]{}%
  }%
  \providecommand\rotatebox[2]{#2}%
  \newcommand*\fsize{\dimexpr\f@size pt\relax}%
  \newcommand*\lineheight[1]{\fontsize{\fsize}{#1\fsize}\selectfont}%
  \ifx\svgwidth\undefined%
    \setlength{\unitlength}{225.65874331bp}%
    \ifx\svgscale\undefined%
      \relax%
    \else%
      \setlength{\unitlength}{\unitlength * \real{\svgscale}}%
    \fi%
  \else%
    \setlength{\unitlength}{\svgwidth}%
  \fi%
  \global\let\svgwidth\undefined%
  \global\let\svgscale\undefined%
  \makeatother%
  \begin{picture}(1,0.26714432)%
    \lineheight{1}%
    \setlength\tabcolsep{0pt}%
    \put(0,0){\includegraphics[width=\unitlength,page=1]{sigma22.pdf}}%
  \end{picture}%
\endgroup%

%% file: pics/sigma22B.pdf_tex
%% Creator: Inkscape 1.2.1 (9c6d41e4, 2022-07-14), www.inkscape.org
%% Accompanies image file 'sigma22B.pdf' (pdf, eps, ps)
\begingroup%
  \makeatletter%
  \providecommand\color[2][]{%
    \errmessage{(Inkscape) Color is used for the text in Inkscape, but the package 'color.sty' is not loaded}%
    \renewcommand\color[2][]{}%
  }%
  \providecommand\transparent[1]{%
    \errmessage{(Inkscape) Transparency is used (non-zero) for the text in Inkscape, but the package 'transparent.sty' is not loaded}%
    \renewcommand\transparent[1]{}%
  }%
  \providecommand\rotatebox[2]{#2}%
  \newcommand*\fsize{\dimexpr\f@size pt\relax}%
  \newcommand*\lineheight[1]{\fontsize{\fsize}{#1\fsize}\selectfont}%
  \ifx\svgwidth\undefined%
    \setlength{\unitlength}{123.74998866bp}%
    \ifx\svgscale\undefined%
      \relax%
    \else%
      \setlength{\unitlength}{\unitlength * \real{\svgscale}}%
    \fi%
  \else%
    \setlength{\unitlength}{\svgwidth}%
  \fi%
  \global\let\svgwidth\undefined%
  \global\let\svgscale\undefined%
  \makeatother%
  \begin{picture}(1,0.42424239)%
    \lineheight{1}%
    \setlength\tabcolsep{0pt}%
    \put(0,0){\includegraphics[width=\unitlength,page=1]{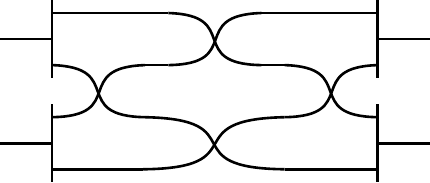}}%
  \end{picture}%
\endgroup%

%% file: pics/dagmin.pdf_tex
%% Creator: Inkscape 1.0.1 (c497b03c, 2020-09-10), www.inkscape.org
%% Accompanies image file 'dagmin.pdf' (pdf, eps, ps)
\begingroup%
  \makeatletter%
  \providecommand\color[2][]{%
    \errmessage{(Inkscape) Color is used for the text in Inkscape, but the package 'color.sty' is not loaded}%
    \renewcommand\color[2][]{}%
  }%
  \providecommand\transparent[1]{%
    \errmessage{(Inkscape) Transparency is used (non-zero) for the text in Inkscape, but the package 'transparent.sty' is not loaded}%
    \renewcommand\transparent[1]{}%
  }%
  \providecommand\rotatebox[2]{#2}%
  \newcommand*\fsize{\dimexpr\f@size pt\relax}%
  \newcommand*\lineheight[1]{\fontsize{\fsize}{#1\fsize}\selectfont}%
  \ifx\svgwidth\undefined%
    \setlength{\unitlength}{14.43683703bp}%
    \ifx\svgscale\undefined%
      \relax%
    \else%
      \setlength{\unitlength}{\unitlength * \real{\svgscale}}%
    \fi%
  \else%
    \setlength{\unitlength}{\svgwidth}%
  \fi%
  \global\let\svgwidth\undefined%
  \global\let\svgscale\undefined%
  \makeatother%
  \begin{picture}(1,3.91587742)%
    \lineheight{1}%
    \setlength\tabcolsep{0pt}%
    \put(0,0){\includegraphics[width=\unitlength,page=1]{dagmin.pdf}}%
    \put(0.230329,3.36822689){\makebox(0,0)[lt]{\lineheight{1.25}\smash{\begin{tabular}[t]{l}$1$\end{tabular}}}}%
    \put(0,0){\includegraphics[width=\unitlength,page=2]{dagmin.pdf}}%
    \put(0.23032825,1.80971527){\makebox(0,0)[lt]{\lineheight{1.25}\smash{\begin{tabular}[t]{l}$2$\end{tabular}}}}%
    \put(0,0){\includegraphics[width=\unitlength,page=3]{dagmin.pdf}}%
    \put(0.23535335,0.25124109){\makebox(0,0)[lt]{\lineheight{1.25}\smash{\begin{tabular}[t]{l}$3$\end{tabular}}}}%
    \put(0,0){\includegraphics[width=\unitlength,page=4]{dagmin.pdf}}%
  \end{picture}%
\endgroup%

%% file: pics/fclosed.pdf_tex
%% Creator: Inkscape 1.2.2 (732a01da63, 2022-12-09), www.inkscape.org
%% Accompanies image file 'fclosed.pdf' (pdf, eps, ps)
\begingroup%
  \makeatletter%
  \providecommand\color[2][]{%
    \errmessage{(Inkscape) Color is used for the text in Inkscape, but the package 'color.sty' is not loaded}%
    \renewcommand\color[2][]{}%
  }%
  \providecommand\transparent[1]{%
    \errmessage{(Inkscape) Transparency is used (non-zero) for the text in Inkscape, but the package 'transparent.sty' is not loaded}%
    \renewcommand\transparent[1]{}%
  }%
  \providecommand\rotatebox[2]{#2}%
  \newcommand*\fsize{\dimexpr\f@size pt\relax}%
  \newcommand*\lineheight[1]{\fontsize{\fsize}{#1\fsize}\selectfont}%
  \ifx\svgwidth\undefined%
    \setlength{\unitlength}{122.1547755bp}%
    \ifx\svgscale\undefined%
      \relax%
    \else%
      \setlength{\unitlength}{\unitlength * \real{\svgscale}}%
    \fi%
  \else%
    \setlength{\unitlength}{\svgwidth}%
  \fi%
  \global\let\svgwidth\undefined%
  \global\let\svgscale\undefined%
  \makeatother%
  \begin{picture}(1,0.33959873)%
    \lineheight{1}%
    \setlength\tabcolsep{0pt}%
    \put(0,0){\includegraphics[width=\unitlength,page=1]{fclosed.pdf}}%
    \put(0.25052417,0.20176133){\color[rgb]{0,0,0}\makebox(0,0)[lt]{\lineheight{1.25}\smash{\begin{tabular}[t]{l}$f$\end{tabular}}}}%
    \put(0.53314016,0.19395108){\color[rgb]{0,0,0}\makebox(0,0)[lt]{\lineheight{1.25}\smash{\begin{tabular}[t]{l}$=$\end{tabular}}}}%
    \put(0.74625576,0.011544){\color[rgb]{0,0,0}\makebox(0,0)[lt]{\lineheight{1.25}\smash{\begin{tabular}[t]{l}$\id_X$\end{tabular}}}}%
    \put(0.06062668,0.09224669){\color[rgb]{0,0,0}\makebox(0,0)[lt]{\lineheight{1.25}\smash{\begin{tabular}[t]{l}$F_X$\end{tabular}}}}%
    \put(0,0){\includegraphics[width=\unitlength,page=2]{fclosed.pdf}}%
    \put(0.77240322,0.20176115){\color[rgb]{0,0,0}\makebox(0,0)[lt]{\lineheight{1.25}\smash{\begin{tabular}[t]{l}$f$\end{tabular}}}}%
    \put(0,0){\includegraphics[width=\unitlength,page=3]{fclosed.pdf}}%
  \end{picture}%
\endgroup%

%% file: pics/lambda.pdf_tex
%% Creator: Inkscape 1.2.2 (732a01da63, 2022-12-09), www.inkscape.org
%% Accompanies image file 'lambda.pdf' (pdf, eps, ps)
\begingroup%
  \makeatletter%
  \providecommand\color[2][]{%
    \errmessage{(Inkscape) Color is used for the text in Inkscape, but the package 'color.sty' is not loaded}%
    \renewcommand\color[2][]{}%
  }%
  \providecommand\transparent[1]{%
    \errmessage{(Inkscape) Transparency is used (non-zero) for the text in Inkscape, but the package 'transparent.sty' is not loaded}%
    \renewcommand\transparent[1]{}%
  }%
  \providecommand\rotatebox[2]{#2}%
  \newcommand*\fsize{\dimexpr\f@size pt\relax}%
  \newcommand*\lineheight[1]{\fontsize{\fsize}{#1\fsize}\selectfont}%
  \ifx\svgwidth\undefined%
    \setlength{\unitlength}{255.22729185bp}%
    \ifx\svgscale\undefined%
      \relax%
    \else%
      \setlength{\unitlength}{\unitlength * \real{\svgscale}}%
    \fi%
  \else%
    \setlength{\unitlength}{\svgwidth}%
  \fi%
  \global\let\svgwidth\undefined%
  \global\let\svgscale\undefined%
  \makeatother%
  \begin{picture}(1,0.13370439)%
    \lineheight{1}%
    \setlength\tabcolsep{0pt}%
    \put(0,0){\includegraphics[width=\unitlength,page=1]{lambda.pdf}}%
    \put(0.36132628,0.05285488){\color[rgb]{0,0,0}\makebox(0,0)[lt]{\lineheight{1.25}\smash{\begin{tabular}[t]{l}$:=$\end{tabular}}}}%
    \put(0.52202131,0.08751242){\color[rgb]{0,0,0}\makebox(0,0)[lt]{\lineheight{1.25}\smash{\begin{tabular}[t]{l}$X\otimes A$\end{tabular}}}}%
    \put(0,0){\includegraphics[width=\unitlength,page=2]{lambda.pdf}}%
    \put(0.41031137,0.07539009){\color[rgb]{0,0,0}\makebox(0,0)[lt]{\lineheight{1.25}\smash{\begin{tabular}[t]{l}$A$\end{tabular}}}}%
    \put(0.74758947,0.07485669){\color[rgb]{0,0,0}\makebox(0,0)[lt]{\lineheight{1.25}\smash{\begin{tabular}[t]{l}$X\multimap Y$\end{tabular}}}}%
    \put(0,0){\includegraphics[width=\unitlength,page=3]{lambda.pdf}}%
    \put(0.52251548,0.00994162){\color[rgb]{0,0,0}\makebox(0,0)[lt]{\lineheight{1.25}\smash{\begin{tabular}[t]{l}$X{\multimap}$\end{tabular}}}}%
    \put(0,0){\includegraphics[width=\unitlength,page=4]{lambda.pdf}}%
    \put(0.65285741,0.06641058){\color[rgb]{0,0,0}\makebox(0,0)[lt]{\lineheight{1.25}\smash{\begin{tabular}[t]{l}$f$\end{tabular}}}}%
    \put(0.0463733,0.04933516){\color[rgb]{0,0,0}\makebox(0,0)[lt]{\lineheight{1.25}\smash{\begin{tabular}[t]{l}$X$\end{tabular}}}}%
    \put(0.19731807,0.08706917){\color[rgb]{0,0,0}\makebox(0,0)[lt]{\lineheight{1.25}\smash{\begin{tabular}[t]{l}$Y$\end{tabular}}}}%
    \put(0,0){\includegraphics[width=\unitlength,page=5]{lambda.pdf}}%
    \put(0.04871893,0.09975812){\color[rgb]{0,0,0}\makebox(0,0)[lt]{\lineheight{1.25}\smash{\begin{tabular}[t]{l}$A$\end{tabular}}}}%
    \put(0.24803472,0.07522198){\color[rgb]{0,0,0}\makebox(0,0)[lt]{\lineheight{1.25}\smash{\begin{tabular}[t]{l}$X\multimap Y$\end{tabular}}}}%
    \put(0,0){\includegraphics[width=\unitlength,page=6]{lambda.pdf}}%
    \put(0.13860983,0.06641041){\color[rgb]{0,0,0}\makebox(0,0)[lt]{\lineheight{1.25}\smash{\begin{tabular}[t]{l}$f$\end{tabular}}}}%
    \put(0,0){\includegraphics[width=\unitlength,page=7]{lambda.pdf}}%
  \end{picture}%
\endgroup%

%% file: pics/curryuncurry.pdf_tex
%% Creator: Inkscape 1.2.2 (732a01da63, 2022-12-09), www.inkscape.org
%% Accompanies image file 'curryuncurry.pdf' (pdf, eps, ps)
\begingroup%
  \makeatletter%
  \providecommand\color[2][]{%
    \errmessage{(Inkscape) Color is used for the text in Inkscape, but the package 'color.sty' is not loaded}%
    \renewcommand\color[2][]{}%
  }%
  \providecommand\transparent[1]{%
    \errmessage{(Inkscape) Transparency is used (non-zero) for the text in Inkscape, but the package 'transparent.sty' is not loaded}%
    \renewcommand\transparent[1]{}%
  }%
  \providecommand\rotatebox[2]{#2}%
  \newcommand*\fsize{\dimexpr\f@size pt\relax}%
  \newcommand*\lineheight[1]{\fontsize{\fsize}{#1\fsize}\selectfont}%
  \ifx\svgwidth\undefined%
    \setlength{\unitlength}{276.61175193bp}%
    \ifx\svgscale\undefined%
      \relax%
    \else%
      \setlength{\unitlength}{\unitlength * \real{\svgscale}}%
    \fi%
  \else%
    \setlength{\unitlength}{\svgwidth}%
  \fi%
  \global\let\svgwidth\undefined%
  \global\let\svgscale\undefined%
  \makeatother%
  \begin{picture}(1,0.46229048)%
    \lineheight{1}%
    \setlength\tabcolsep{0pt}%
    \put(0,0){\includegraphics[width=\unitlength,page=1]{curryuncurry.pdf}}%
    \put(-0.00001238,0.4195794){\color[rgb]{0,0,0}\makebox(0,0)[lt]{\lineheight{1.25}\smash{\begin{tabular}[t]{l}$X\otimes Y\multimap Z$\end{tabular}}}}%
    \put(0.616139,0.12182792){\color[rgb]{0,0,0}\makebox(0,0)[lt]{\lineheight{1.25}\smash{\begin{tabular}[t]{l}$X\otimes Y\multimap Z$\end{tabular}}}}%
    \put(0,0){\includegraphics[width=\unitlength,page=2]{curryuncurry.pdf}}%
    \put(0.67037599,0.36301214){\color[rgb]{0,0,0}\makebox(0,0)[lt]{\lineheight{1.25}\smash{\begin{tabular}[t]{l}$X\multimap (Y\multimap Z)$\end{tabular}}}}%
    \put(-0.00164166,0.14996143){\color[rgb]{0,0,0}\makebox(0,0)[lt]{\lineheight{1.25}\smash{\begin{tabular}[t]{l}$X\multimap (Y\multimap Z)$\end{tabular}}}}%
    \put(0.24229136,0.25188304){\color[rgb]{0,0,0}\makebox(0,0)[lt]{\lineheight{1.25}\smash{\begin{tabular}[t]{l}$X$\end{tabular}}}}%
    \put(0.37319572,0.26340025){\color[rgb]{0,0,0}\makebox(0,0)[lt]{\lineheight{1.25}\smash{\begin{tabular}[t]{l}$Y$\end{tabular}}}}%
    \put(0.38598006,0.0792459){\color[rgb]{0,0,0}\makebox(0,0)[lt]{\lineheight{1.25}\smash{\begin{tabular}[t]{l}$X$\end{tabular}}}}%
    \put(0.49156725,0.06552695){\color[rgb]{0,0,0}\makebox(0,0)[lt]{\lineheight{1.25}\smash{\begin{tabular}[t]{l}$Y$\end{tabular}}}}%
    \put(0,0){\includegraphics[width=\unitlength,page=3]{curryuncurry.pdf}}%
    \put(0.22291782,0.01566381){\color[rgb]{0,0,0}\makebox(0,0)[lt]{\lineheight{1.25}\smash{\begin{tabular}[t]{l}$X{\otimes}Y$\end{tabular}}}}%
    \put(0,0){\includegraphics[width=\unitlength,page=4]{curryuncurry.pdf}}%
  \end{picture}%
\endgroup%

%% file: pics/freesmc1.pdf_tex
%% Creator: Inkscape 1.2.2 (732a01da63, 2022-12-09), www.inkscape.org
%% Accompanies image file 'freesmc1.pdf' (pdf, eps, ps)
\begingroup%
  \makeatletter%
  \providecommand\color[2][]{%
    \errmessage{(Inkscape) Color is used for the text in Inkscape, but the package 'color.sty' is not loaded}%
    \renewcommand\color[2][]{}%
  }%
  \providecommand\transparent[1]{%
    \errmessage{(Inkscape) Transparency is used (non-zero) for the text in Inkscape, but the package 'transparent.sty' is not loaded}%
    \renewcommand\transparent[1]{}%
  }%
  \providecommand\rotatebox[2]{#2}%
  \newcommand*\fsize{\dimexpr\f@size pt\relax}%
  \newcommand*\lineheight[1]{\fontsize{\fsize}{#1\fsize}\selectfont}%
  \ifx\svgwidth\undefined%
    \setlength{\unitlength}{172.98454502bp}%
    \ifx\svgscale\undefined%
      \relax%
    \else%
      \setlength{\unitlength}{\unitlength * \real{\svgscale}}%
    \fi%
  \else%
    \setlength{\unitlength}{\svgwidth}%
  \fi%
  \global\let\svgwidth\undefined%
  \global\let\svgscale\undefined%
  \makeatother%
  \begin{picture}(1,0.1761327)%
    \lineheight{1}%
    \setlength\tabcolsep{0pt}%
    \put(0,0){\includegraphics[width=\unitlength,page=1]{freesmc1.pdf}}%
    \put(0.62433317,0.10161245){\color[rgb]{0,0,0}\makebox(0,0)[lt]{\lineheight{1.25}\smash{\begin{tabular}[t]{l}$X\multimap B$\end{tabular}}}}%
    \put(0,0){\includegraphics[width=\unitlength,page=2]{freesmc1.pdf}}%
    \put(0.20374653,0.05165147){\color[rgb]{0,0,0}\makebox(0,0)[lt]{\lineheight{1.25}\smash{\begin{tabular}[t]{l}$X$\end{tabular}}}}%
    \put(0.00433549,0.12302122){\color[rgb]{0,0,0}\makebox(0,0)[lt]{\lineheight{1.25}\smash{\begin{tabular}[t]{l}$X\multimap A$\end{tabular}}}}%
    \put(0,0){\includegraphics[width=\unitlength,page=3]{freesmc1.pdf}}%
    \put(0.48830902,0.07711416){\color[rgb]{0,0,0}\makebox(0,0)[lt]{\lineheight{1.25}\smash{\begin{tabular}[t]{l}$f$\end{tabular}}}}%
  \end{picture}%
\endgroup%

%% file: pics/freesmc2.pdf_tex
%% Creator: Inkscape 1.2.2 (732a01da63, 2022-12-09), www.inkscape.org
%% Accompanies image file 'freesmc2.pdf' (pdf, eps, ps)
\begingroup%
  \makeatletter%
  \providecommand\color[2][]{%
    \errmessage{(Inkscape) Color is used for the text in Inkscape, but the package 'color.sty' is not loaded}%
    \renewcommand\color[2][]{}%
  }%
  \providecommand\transparent[1]{%
    \errmessage{(Inkscape) Transparency is used (non-zero) for the text in Inkscape, but the package 'transparent.sty' is not loaded}%
    \renewcommand\transparent[1]{}%
  }%
  \providecommand\rotatebox[2]{#2}%
  \newcommand*\fsize{\dimexpr\f@size pt\relax}%
  \newcommand*\lineheight[1]{\fontsize{\fsize}{#1\fsize}\selectfont}%
  \ifx\svgwidth\undefined%
    \setlength{\unitlength}{159.84406156bp}%
    \ifx\svgscale\undefined%
      \relax%
    \else%
      \setlength{\unitlength}{\unitlength * \real{\svgscale}}%
    \fi%
  \else%
    \setlength{\unitlength}{\svgwidth}%
  \fi%
  \global\let\svgwidth\undefined%
  \global\let\svgscale\undefined%
  \makeatother%
  \begin{picture}(1,0.19061224)%
    \lineheight{1}%
    \setlength\tabcolsep{0pt}%
    \put(0,0){\includegraphics[width=\unitlength,page=1]{freesmc2.pdf}}%
    \put(0.30498475,0.11934984){\color[rgb]{0,0,0}\makebox(0,0)[lt]{\lineheight{1.25}\smash{\begin{tabular}[t]{l}$X\multimap (X\otimes B)$\end{tabular}}}}%
    \put(0.10319437,0.05589764){\color[rgb]{0,0,0}\makebox(0,0)[lt]{\lineheight{1.25}\smash{\begin{tabular}[t]{l}$X$\end{tabular}}}}%
    \put(0.00938414,0.13313457){\color[rgb]{0,0,0}\makebox(0,0)[lt]{\lineheight{1.25}\smash{\begin{tabular}[t]{l}$B$\end{tabular}}}}%
    \put(0,0){\includegraphics[width=\unitlength,page=2]{freesmc2.pdf}}%
  \end{picture}%
\endgroup%

%% file: pics/copydel.pdf_tex
%% Creator: Inkscape 1.0.1 (c497b03c, 2020-09-10), www.inkscape.org
%% Accompanies image file 'copydel.pdf' (pdf, eps, ps)
\begingroup%
  \makeatletter%
  \providecommand\color[2][]{%
    \errmessage{(Inkscape) Color is used for the text in Inkscape, but the package 'color.sty' is not loaded}%
    \renewcommand\color[2][]{}%
  }%
  \providecommand\transparent[1]{%
    \errmessage{(Inkscape) Transparency is used (non-zero) for the text in Inkscape, but the package 'transparent.sty' is not loaded}%
    \renewcommand\transparent[1]{}%
  }%
  \providecommand\rotatebox[2]{#2}%
  \newcommand*\fsize{\dimexpr\f@size pt\relax}%
  \newcommand*\lineheight[1]{\fontsize{\fsize}{#1\fsize}\selectfont}%
  \ifx\svgwidth\undefined%
    \setlength{\unitlength}{200.72255883bp}%
    \ifx\svgscale\undefined%
      \relax%
    \else%
      \setlength{\unitlength}{\unitlength * \real{\svgscale}}%
    \fi%
  \else%
    \setlength{\unitlength}{\svgwidth}%
  \fi%
  \global\let\svgwidth\undefined%
  \global\let\svgscale\undefined%
  \makeatother%
  \begin{picture}(1,0.19056153)%
    \lineheight{1}%
    \setlength\tabcolsep{0pt}%
    \put(0.25212835,0.08419884){\color[rgb]{0,0,0}\makebox(0,0)[lt]{\lineheight{1.25}\smash{\begin{tabular}[t]{l}$=$\end{tabular}}}}%
    \put(0,0){\includegraphics[width=\unitlength,page=1]{copydel.pdf}}%
    \put(0.08454414,0.08537788){\color[rgb]{0,0,0}\makebox(0,0)[lt]{\lineheight{1.25}\smash{\begin{tabular}[t]{l}$f$\end{tabular}}}}%
    \put(0,0){\includegraphics[width=\unitlength,page=2]{copydel.pdf}}%
    \put(0.43951163,0.1414255){\color[rgb]{0,0,0}\makebox(0,0)[lt]{\lineheight{1.25}\smash{\begin{tabular}[t]{l}$f$\end{tabular}}}}%
    \put(0,0){\includegraphics[width=\unitlength,page=3]{copydel.pdf}}%
    \put(0.43951174,0.02933069){\color[rgb]{0,0,0}\makebox(0,0)[lt]{\lineheight{1.25}\smash{\begin{tabular}[t]{l}$f$\end{tabular}}}}%
    \put(0,0){\includegraphics[width=\unitlength,page=4]{copydel.pdf}}%
    \put(0.86491437,0.08419884){\color[rgb]{0,0,0}\makebox(0,0)[lt]{\lineheight{1.25}\smash{\begin{tabular}[t]{l}$=$\end{tabular}}}}%
    \put(0,0){\includegraphics[width=\unitlength,page=5]{copydel.pdf}}%
    \put(0.7197493,0.08537799){\color[rgb]{0,0,0}\makebox(0,0)[lt]{\lineheight{1.25}\smash{\begin{tabular}[t]{l}$f$\end{tabular}}}}%
    \put(0,0){\includegraphics[width=\unitlength,page=6]{copydel.pdf}}%
  \end{picture}%
\endgroup%

%% file: pics/prodeq.pdf_tex
%% Creator: Inkscape 1.0.1 (c497b03c, 2020-09-10), www.inkscape.org
%% Accompanies image file 'prodeq.pdf' (pdf, eps, ps)
\begingroup%
  \makeatletter%
  \providecommand\color[2][]{%
    \errmessage{(Inkscape) Color is used for the text in Inkscape, but the package 'color.sty' is not loaded}%
    \renewcommand\color[2][]{}%
  }%
  \providecommand\transparent[1]{%
    \errmessage{(Inkscape) Transparency is used (non-zero) for the text in Inkscape, but the package 'transparent.sty' is not loaded}%
    \renewcommand\transparent[1]{}%
  }%
  \providecommand\rotatebox[2]{#2}%
  \newcommand*\fsize{\dimexpr\f@size pt\relax}%
  \newcommand*\lineheight[1]{\fontsize{\fsize}{#1\fsize}\selectfont}%
  \ifx\svgwidth\undefined%
    \setlength{\unitlength}{140.52677432bp}%
    \ifx\svgscale\undefined%
      \relax%
    \else%
      \setlength{\unitlength}{\unitlength * \real{\svgscale}}%
    \fi%
  \else%
    \setlength{\unitlength}{\svgwidth}%
  \fi%
  \global\let\svgwidth\undefined%
  \global\let\svgscale\undefined%
  \makeatother%
  \begin{picture}(1,0.63354563)%
    \lineheight{1}%
    \setlength\tabcolsep{0pt}%
    \put(0.39945861,0.07578202){\color[rgb]{0,0,0}\makebox(0,0)[lt]{\lineheight{1.25}\smash{\begin{tabular}[t]{l}$=$\end{tabular}}}}%
    \put(0,0){\includegraphics[width=\unitlength,page=1]{prodeq.pdf}}%
    \put(0.74636762,0.07366656){\color[rgb]{0,0,0}\makebox(0,0)[lt]{\lineheight{1.25}\smash{\begin{tabular}[t]{l}$=$\end{tabular}}}}%
    \put(0,0){\includegraphics[width=\unitlength,page=2]{prodeq.pdf}}%
    \put(0.50619994,0.31061218){\color[rgb]{0,0,0}\makebox(0,0)[lt]{\lineheight{1.25}\smash{\begin{tabular}[t]{l}$=$\end{tabular}}}}%
    \put(0,0){\includegraphics[width=\unitlength,page=3]{prodeq.pdf}}%
    \put(-0.00323139,0.31079316){\makebox(0,0)[lt]{\lineheight{1.25}\smash{\begin{tabular}[t]{l}$\delta_{A\otimes B}=$\end{tabular}}}}%
    \put(0,0){\includegraphics[width=\unitlength,page=4]{prodeq.pdf}}%
    \put(0.22741123,0.34594536){\makebox(0,0)[lt]{\lineheight{1.25}\smash{\begin{tabular}[t]{l}$A\otimes B$\\\end{tabular}}}}%
    \put(0.65928992,0.29962102){\makebox(0,0)[lt]{\lineheight{1.25}\smash{\begin{tabular}[t]{l}$A$\end{tabular}}}}%
    \put(0.65528337,0.38186904){\makebox(0,0)[lt]{\lineheight{1.25}\smash{\begin{tabular}[t]{l}$B$\end{tabular}}}}%
    \put(0.50619986,0.52943244){\color[rgb]{0,0,0}\makebox(0,0)[lt]{\lineheight{1.25}\smash{\begin{tabular}[t]{l}$=$\end{tabular}}}}%
    \put(0,0){\includegraphics[width=\unitlength,page=5]{prodeq.pdf}}%
    \put(-0.00323143,0.52961342){\makebox(0,0)[lt]{\lineheight{1.25}\smash{\begin{tabular}[t]{l}$\omega_{A\otimes B}=$\end{tabular}}}}%
    \put(0,0){\includegraphics[width=\unitlength,page=6]{prodeq.pdf}}%
    \put(0.22741115,0.56476562){\makebox(0,0)[lt]{\lineheight{1.25}\smash{\begin{tabular}[t]{l}$A\otimes B$\\\end{tabular}}}}%
    \put(0.65928992,0.51844127){\makebox(0,0)[lt]{\lineheight{1.25}\smash{\begin{tabular}[t]{l}$A$\end{tabular}}}}%
    \put(0.65528337,0.6006893){\makebox(0,0)[lt]{\lineheight{1.25}\smash{\begin{tabular}[t]{l}$B$\end{tabular}}}}%
  \end{picture}%
\endgroup%

%% file: pics/narycopy.pdf_tex
%% Creator: Inkscape 1.0.1 (c497b03c, 2020-09-10), www.inkscape.org
%% Accompanies image file 'narycopy.pdf' (pdf, eps, ps)
\begingroup%
  \makeatletter%
  \providecommand\color[2][]{%
    \errmessage{(Inkscape) Color is used for the text in Inkscape, but the package 'color.sty' is not loaded}%
    \renewcommand\color[2][]{}%
  }%
  \providecommand\transparent[1]{%
    \errmessage{(Inkscape) Transparency is used (non-zero) for the text in Inkscape, but the package 'transparent.sty' is not loaded}%
    \renewcommand\transparent[1]{}%
  }%
  \providecommand\rotatebox[2]{#2}%
  \newcommand*\fsize{\dimexpr\f@size pt\relax}%
  \newcommand*\lineheight[1]{\fontsize{\fsize}{#1\fsize}\selectfont}%
  \ifx\svgwidth\undefined%
    \setlength{\unitlength}{100.8676671bp}%
    \ifx\svgscale\undefined%
      \relax%
    \else%
      \setlength{\unitlength}{\unitlength * \real{\svgscale}}%
    \fi%
  \else%
    \setlength{\unitlength}{\svgwidth}%
  \fi%
  \global\let\svgwidth\undefined%
  \global\let\svgscale\undefined%
  \makeatother%
  \begin{picture}(1,0.45415683)%
    \lineheight{1}%
    \setlength\tabcolsep{0pt}%
    \put(0.25916218,0.27749534){\color[rgb]{0,0,0}\makebox(0,0)[lt]{\lineheight{1.25}\smash{\begin{tabular}[t]{l}$:=$\end{tabular}}}}%
    \put(0,0){\includegraphics[width=\unitlength,page=1]{narycopy.pdf}}%
    \put(0.20107785,0.3305274){\color[rgb]{0,0,0}\rotatebox{-89.846127}{\makebox(0,0)[lt]{\lineheight{1.25}\smash{\begin{tabular}[t]{l}...\end{tabular}}}}}%
    \put(0,0){\includegraphics[width=\unitlength,page=2]{narycopy.pdf}}%
    \put(0.94240531,0.00061904){\color[rgb]{0,0,0}\rotatebox{-0.61675973}{\makebox(0,0)[lt]{\lineheight{1.25}\smash{\begin{tabular}[t]{l}...\end{tabular}}}}}%
    \put(0,0){\includegraphics[width=\unitlength,page=3]{narycopy.pdf}}%
    \put(0.94240531,0.14932892){\color[rgb]{0,0,0}\rotatebox{-0.61675968}{\makebox(0,0)[lt]{\lineheight{1.25}\smash{\begin{tabular}[t]{l}...\end{tabular}}}}}%
    \put(0.94240531,0.22368365){\color[rgb]{0,0,0}\rotatebox{-0.61675968}{\makebox(0,0)[lt]{\lineheight{1.25}\smash{\begin{tabular}[t]{l}...\end{tabular}}}}}%
    \put(0.94240531,0.29803837){\color[rgb]{0,0,0}\rotatebox{-0.61675968}{\makebox(0,0)[lt]{\lineheight{1.25}\smash{\begin{tabular}[t]{l}...\end{tabular}}}}}%
    \put(0.94240531,0.3723931){\color[rgb]{0,0,0}\rotatebox{-0.61675968}{\makebox(0,0)[lt]{\lineheight{1.25}\smash{\begin{tabular}[t]{l}...\end{tabular}}}}}%
  \end{picture}%
\endgroup%

%% file: pics/narycopy2.pdf_tex
%% Creator: Inkscape 1.2.1 (9c6d41e4, 2022-07-14), www.inkscape.org
%% Accompanies image file 'narycopy2.pdf' (pdf, eps, ps)
\begingroup%
  \makeatletter%
  \providecommand\color[2][]{%
    \errmessage{(Inkscape) Color is used for the text in Inkscape, but the package 'color.sty' is not loaded}%
    \renewcommand\color[2][]{}%
  }%
  \providecommand\transparent[1]{%
    \errmessage{(Inkscape) Transparency is used (non-zero) for the text in Inkscape, but the package 'transparent.sty' is not loaded}%
    \renewcommand\transparent[1]{}%
  }%
  \providecommand\rotatebox[2]{#2}%
  \newcommand*\fsize{\dimexpr\f@size pt\relax}%
  \newcommand*\lineheight[1]{\fontsize{\fsize}{#1\fsize}\selectfont}%
  \ifx\svgwidth\undefined%
    \setlength{\unitlength}{164.65353691bp}%
    \ifx\svgscale\undefined%
      \relax%
    \else%
      \setlength{\unitlength}{\unitlength * \real{\svgscale}}%
    \fi%
  \else%
    \setlength{\unitlength}{\svgwidth}%
  \fi%
  \global\let\svgwidth\undefined%
  \global\let\svgscale\undefined%
  \makeatother%
  \begin{picture}(1,0.36905484)%
    \lineheight{1}%
    \setlength\tabcolsep{0pt}%
    \put(0,0){\includegraphics[width=\unitlength,page=1]{narycopy2.pdf}}%
    \put(0.37968288,0.19488262){\color[rgb]{0,0,0}\makebox(0,0)[lt]{\lineheight{1.25}\smash{\begin{tabular}[t]{l}$=$\end{tabular}}}}%
    \put(0.81696421,0.19488262){\color[rgb]{0,0,0}\makebox(0,0)[lt]{\lineheight{1.25}\smash{\begin{tabular}[t]{l}$=$\end{tabular}}}}%
    \put(0,0){\includegraphics[width=\unitlength,page=2]{narycopy2.pdf}}%
  \end{picture}%
\endgroup%

%% file: pics/plusvcons.pdf_tex
%% Creator: Inkscape 1.2.1 (9c6d41e4, 2022-07-14), www.inkscape.org
%% Accompanies image file 'plusvcons.pdf' (pdf, eps, ps)
\begingroup%
  \makeatletter%
  \providecommand\color[2][]{%
    \errmessage{(Inkscape) Color is used for the text in Inkscape, but the package 'color.sty' is not loaded}%
    \renewcommand\color[2][]{}%
  }%
  \providecommand\transparent[1]{%
    \errmessage{(Inkscape) Transparency is used (non-zero) for the text in Inkscape, but the package 'transparent.sty' is not loaded}%
    \renewcommand\transparent[1]{}%
  }%
  \providecommand\rotatebox[2]{#2}%
  \newcommand*\fsize{\dimexpr\f@size pt\relax}%
  \newcommand*\lineheight[1]{\fontsize{\fsize}{#1\fsize}\selectfont}%
  \ifx\svgwidth\undefined%
    \setlength{\unitlength}{105.76872781bp}%
    \ifx\svgscale\undefined%
      \relax%
    \else%
      \setlength{\unitlength}{\unitlength * \real{\svgscale}}%
    \fi%
  \else%
    \setlength{\unitlength}{\svgwidth}%
  \fi%
  \global\let\svgwidth\undefined%
  \global\let\svgscale\undefined%
  \makeatother%
  \begin{picture}(1,0.29103717)%
    \lineheight{1}%
    \setlength\tabcolsep{0pt}%
    \put(0,0){\includegraphics[width=\unitlength,page=1]{plusvcons.pdf}}%
    \put(0.52473026,0.13118362){\makebox(0,0)[lt]{\lineheight{1.25}\smash{\begin{tabular}[t]{l}vs.\end{tabular}}}}%
    \put(0,0){\includegraphics[width=\unitlength,page=2]{plusvcons.pdf}}%
    \put(0.80377276,0.12735715){\makebox(0,0)[lt]{\lineheight{1.25}\smash{\begin{tabular}[t]{l}$::$\end{tabular}}}}%
    \put(0,0){\includegraphics[width=\unitlength,page=3]{plusvcons.pdf}}%
    \put(0.12304221,0.12735674){\makebox(0,0)[lt]{\lineheight{1.25}\smash{\begin{tabular}[t]{l}$+$\end{tabular}}}}%
    \put(0,0){\includegraphics[width=\unitlength,page=4]{plusvcons.pdf}}%
  \end{picture}%
\endgroup%

%% file: pics/pairing.pdf_tex
%% Creator: Inkscape 1.1 (c68e22c387, 2021-05-23), www.inkscape.org
%% Accompanies image file 'pairing.pdf' (pdf, eps, ps)
\begingroup%
  \makeatletter%
  \providecommand\color[2][]{%
    \errmessage{(Inkscape) Color is used for the text in Inkscape, but the package 'color.sty' is not loaded}%
    \renewcommand\color[2][]{}%
  }%
  \providecommand\transparent[1]{%
    \errmessage{(Inkscape) Transparency is used (non-zero) for the text in Inkscape, but the package 'transparent.sty' is not loaded}%
    \renewcommand\transparent[1]{}%
  }%
  \providecommand\rotatebox[2]{#2}%
  \newcommand*\fsize{\dimexpr\f@size pt\relax}%
  \newcommand*\lineheight[1]{\fontsize{\fsize}{#1\fsize}\selectfont}%
  \ifx\svgwidth\undefined%
    \setlength{\unitlength}{259.23082682bp}%
    \ifx\svgscale\undefined%
      \relax%
    \else%
      \setlength{\unitlength}{\unitlength * \real{\svgscale}}%
    \fi%
  \else%
    \setlength{\unitlength}{\svgwidth}%
  \fi%
  \global\let\svgwidth\undefined%
  \global\let\svgscale\undefined%
  \makeatother%
  \begin{picture}(1,0.23396629)%
    \lineheight{1}%
    \setlength\tabcolsep{0pt}%
    \put(0.13652525,0.0630801){\makebox(0,0)[lt]{\lineheight{1.25}\smash{\begin{tabular}[t]{l}$\langle f_1,f_2\rangle:X\to A_1\times A_2=$\end{tabular}}}}%
    \put(-0.00233564,0.2014035){\color[rgb]{0,0,0}\makebox(0,0)[lt]{\lineheight{1.25}\smash{\begin{tabular}[t]{l}$p_1:A_1\times A_2\to A_1=$\end{tabular}}}}%
    \put(0,0){\includegraphics[width=\unitlength,page=1]{pairing.pdf}}%
    \put(0.63389023,0.10950575){\color[rgb]{0,0,0}\makebox(0,0)[lt]{\lineheight{1.25}\smash{\begin{tabular}[t]{l}$f_2$\end{tabular}}}}%
    \put(0,0){\includegraphics[width=\unitlength,page=2]{pairing.pdf}}%
    \put(0.63389031,0.02271069){\color[rgb]{0,0,0}\makebox(0,0)[lt]{\lineheight{1.25}\smash{\begin{tabular}[t]{l}$f_1$\end{tabular}}}}%
    \put(0,0){\includegraphics[width=\unitlength,page=3]{pairing.pdf}}%
    \put(0.51843578,0.20140416){\color[rgb]{0,0,0}\makebox(0,0)[lt]{\lineheight{1.25}\smash{\begin{tabular}[t]{l}$p_2:A_1\times A_2\to A_2=$\end{tabular}}}}%
    \put(0,0){\includegraphics[width=\unitlength,page=4]{pairing.pdf}}%
  \end{picture}%
\endgroup%

%% file: pics/eqpairing.pdf_tex
%% Creator: Inkscape 1.1 (c68e22c387, 2021-05-23), www.inkscape.org
%% Accompanies image file 'eqpairing.pdf' (pdf, eps, ps)
\begingroup%
  \makeatletter%
  \providecommand\color[2][]{%
    \errmessage{(Inkscape) Color is used for the text in Inkscape, but the package 'color.sty' is not loaded}%
    \renewcommand\color[2][]{}%
  }%
  \providecommand\transparent[1]{%
    \errmessage{(Inkscape) Transparency is used (non-zero) for the text in Inkscape, but the package 'transparent.sty' is not loaded}%
    \renewcommand\transparent[1]{}%
  }%
  \providecommand\rotatebox[2]{#2}%
  \newcommand*\fsize{\dimexpr\f@size pt\relax}%
  \newcommand*\lineheight[1]{\fontsize{\fsize}{#1\fsize}\selectfont}%
  \ifx\svgwidth\undefined%
    \setlength{\unitlength}{247.7337285bp}%
    \ifx\svgscale\undefined%
      \relax%
    \else%
      \setlength{\unitlength}{\unitlength * \real{\svgscale}}%
    \fi%
  \else%
    \setlength{\unitlength}{\svgwidth}%
  \fi%
  \global\let\svgwidth\undefined%
  \global\let\svgscale\undefined%
  \makeatother%
  \begin{picture}(1,0.1604541)%
    \lineheight{1}%
    \setlength\tabcolsep{0pt}%
    \put(-0.00244403,0.07587889){\makebox(0,0)[lt]{\lineheight{1.25}\smash{\begin{tabular}[t]{l}$\langle f_1,f_2\rangle;p_1=$\end{tabular}}}}%
    \put(0.5013928,0.0757331){\color[rgb]{0,0,0}\makebox(0,0)[lt]{\lineheight{1.25}\smash{\begin{tabular}[t]{l}$=$\end{tabular}}}}%
    \put(0,0){\includegraphics[width=\unitlength,page=1]{eqpairing.pdf}}%
    \put(0.34868908,0.12064227){\color[rgb]{0,0,0}\makebox(0,0)[lt]{\lineheight{1.25}\smash{\begin{tabular}[t]{l}$f_2$\end{tabular}}}}%
    \put(0,0){\includegraphics[width=\unitlength,page=2]{eqpairing.pdf}}%
    \put(0.34868916,0.02981913){\color[rgb]{0,0,0}\makebox(0,0)[lt]{\lineheight{1.25}\smash{\begin{tabular}[t]{l}$f_1$\end{tabular}}}}%
    \put(0,0){\includegraphics[width=\unitlength,page=3]{eqpairing.pdf}}%
    \put(0.80499529,0.0757331){\color[rgb]{0,0,0}\makebox(0,0)[lt]{\lineheight{1.25}\smash{\begin{tabular}[t]{l}$=$\end{tabular}}}}%
    \put(0,0){\includegraphics[width=\unitlength,page=4]{eqpairing.pdf}}%
    \put(0.65446179,0.02376449){\color[rgb]{0,0,0}\makebox(0,0)[lt]{\lineheight{1.25}\smash{\begin{tabular}[t]{l}$f_1$\end{tabular}}}}%
    \put(0,0){\includegraphics[width=\unitlength,page=5]{eqpairing.pdf}}%
    \put(0.90876646,0.07523061){\color[rgb]{0,0,0}\makebox(0,0)[lt]{\lineheight{1.25}\smash{\begin{tabular}[t]{l}$f_1$\end{tabular}}}}%
    \put(0,0){\includegraphics[width=\unitlength,page=6]{eqpairing.pdf}}%
  \end{picture}%
\endgroup%

%% file: pics/uniqpair.pdf_tex
%% Creator: Inkscape 1.1 (c68e22c387, 2021-05-23), www.inkscape.org
%% Accompanies image file 'uniqpair.pdf' (pdf, eps, ps)
\begingroup%
  \makeatletter%
  \providecommand\color[2][]{%
    \errmessage{(Inkscape) Color is used for the text in Inkscape, but the package 'color.sty' is not loaded}%
    \renewcommand\color[2][]{}%
  }%
  \providecommand\transparent[1]{%
    \errmessage{(Inkscape) Transparency is used (non-zero) for the text in Inkscape, but the package 'transparent.sty' is not loaded}%
    \renewcommand\transparent[1]{}%
  }%
  \providecommand\rotatebox[2]{#2}%
  \newcommand*\fsize{\dimexpr\f@size pt\relax}%
  \newcommand*\lineheight[1]{\fontsize{\fsize}{#1\fsize}\selectfont}%
  \ifx\svgwidth\undefined%
    \setlength{\unitlength}{87.57156161bp}%
    \ifx\svgscale\undefined%
      \relax%
    \else%
      \setlength{\unitlength}{\unitlength * \real{\svgscale}}%
    \fi%
  \else%
    \setlength{\unitlength}{\svgwidth}%
  \fi%
  \global\let\svgwidth\undefined%
  \global\let\svgscale\undefined%
  \makeatother%
  \begin{picture}(1,0.43678562)%
    \lineheight{1}%
    \setlength\tabcolsep{0pt}%
    \put(0,0){\includegraphics[width=\unitlength,page=1]{uniqpair.pdf}}%
    \put(0.15096156,0.32416128){\color[rgb]{0,0,0}\makebox(0,0)[lt]{\lineheight{1.25}\smash{\begin{tabular}[t]{l}$u$\end{tabular}}}}%
    \put(0,0){\includegraphics[width=\unitlength,page=2]{uniqpair.pdf}}%
    \put(0.15096156,0.06722856){\color[rgb]{0,0,0}\makebox(0,0)[lt]{\lineheight{1.25}\smash{\begin{tabular}[t]{l}$u$\end{tabular}}}}%
    \put(0,0){\includegraphics[width=\unitlength,page=3]{uniqpair.pdf}}%
    \put(0.45455507,0.31905614){\color[rgb]{0,0,0}\makebox(0,0)[lt]{\lineheight{1.25}\smash{\begin{tabular}[t]{l}$=$\end{tabular}}}}%
    \put(0.46025662,0.05569407){\color[rgb]{0,0,0}\makebox(0,0)[lt]{\lineheight{1.25}\smash{\begin{tabular}[t]{l}$=$\end{tabular}}}}%
    \put(0,0){\includegraphics[width=\unitlength,page=4]{uniqpair.pdf}}%
    \put(0.74190678,0.32416128){\color[rgb]{0,0,0}\makebox(0,0)[lt]{\lineheight{1.25}\smash{\begin{tabular}[t]{l}$f_1$\end{tabular}}}}%
    \put(0,0){\includegraphics[width=\unitlength,page=5]{uniqpair.pdf}}%
    \put(0.74190678,0.06722856){\color[rgb]{0,0,0}\makebox(0,0)[lt]{\lineheight{1.25}\smash{\begin{tabular}[t]{l}$f_2$\end{tabular}}}}%
    \put(0,0){\includegraphics[width=\unitlength,page=6]{uniqpair.pdf}}%
  \end{picture}%
\endgroup%

%% file: pics/uniqpair2.pdf_tex
%% Creator: Inkscape 1.1 (c68e22c387, 2021-05-23), www.inkscape.org
%% Accompanies image file 'uniqpair2.pdf' (pdf, eps, ps)
\begingroup%
  \makeatletter%
  \providecommand\color[2][]{%
    \errmessage{(Inkscape) Color is used for the text in Inkscape, but the package 'color.sty' is not loaded}%
    \renewcommand\color[2][]{}%
  }%
  \providecommand\transparent[1]{%
    \errmessage{(Inkscape) Transparency is used (non-zero) for the text in Inkscape, but the package 'transparent.sty' is not loaded}%
    \renewcommand\transparent[1]{}%
  }%
  \providecommand\rotatebox[2]{#2}%
  \newcommand*\fsize{\dimexpr\f@size pt\relax}%
  \newcommand*\lineheight[1]{\fontsize{\fsize}{#1\fsize}\selectfont}%
  \ifx\svgwidth\undefined%
    \setlength{\unitlength}{221.82154157bp}%
    \ifx\svgscale\undefined%
      \relax%
    \else%
      \setlength{\unitlength}{\unitlength * \real{\svgscale}}%
    \fi%
  \else%
    \setlength{\unitlength}{\svgwidth}%
  \fi%
  \global\let\svgwidth\undefined%
  \global\let\svgscale\undefined%
  \makeatother%
  \begin{picture}(1,0.17243591)%
    \lineheight{1}%
    \setlength\tabcolsep{0pt}%
    \put(0.16847432,0.0751938){\color[rgb]{0,0,0}\makebox(0,0)[lt]{\lineheight{1.25}\smash{\begin{tabular}[t]{l}$=$\end{tabular}}}}%
    \put(0,0){\includegraphics[width=\unitlength,page=1]{uniqpair2.pdf}}%
    \put(0.06635925,0.07725719){\color[rgb]{0,0,0}\makebox(0,0)[lt]{\lineheight{1.25}\smash{\begin{tabular}[t]{l}$u$\end{tabular}}}}%
    \put(0,0){\includegraphics[width=\unitlength,page=2]{uniqpair2.pdf}}%
    \put(0.4520267,0.0751938){\color[rgb]{0,0,0}\makebox(0,0)[lt]{\lineheight{1.25}\smash{\begin{tabular}[t]{l}$=$\end{tabular}}}}%
    \put(0,0){\includegraphics[width=\unitlength,page=3]{uniqpair2.pdf}}%
    \put(0.28274943,0.07725719){\color[rgb]{0,0,0}\makebox(0,0)[lt]{\lineheight{1.25}\smash{\begin{tabular}[t]{l}$u$\end{tabular}}}}%
    \put(0,0){\includegraphics[width=\unitlength,page=4]{uniqpair2.pdf}}%
    \put(0.62424024,0.12797364){\color[rgb]{0,0,0}\makebox(0,0)[lt]{\lineheight{1.25}\smash{\begin{tabular}[t]{l}$u$\end{tabular}}}}%
    \put(0,0){\includegraphics[width=\unitlength,page=5]{uniqpair2.pdf}}%
    \put(0.62424024,0.02654075){\color[rgb]{0,0,0}\makebox(0,0)[lt]{\lineheight{1.25}\smash{\begin{tabular}[t]{l}$u$\end{tabular}}}}%
    \put(0,0){\includegraphics[width=\unitlength,page=6]{uniqpair2.pdf}}%
    \put(0.73265782,0.0751938){\color[rgb]{0,0,0}\makebox(0,0)[lt]{\lineheight{1.25}\smash{\begin{tabular}[t]{l}$=$\end{tabular}}}}%
    \put(0,0){\includegraphics[width=\unitlength,page=7]{uniqpair2.pdf}}%
    \put(0.89810897,0.12797364){\color[rgb]{0,0,0}\makebox(0,0)[lt]{\lineheight{1.25}\smash{\begin{tabular}[t]{l}$f_2$\end{tabular}}}}%
    \put(0,0){\includegraphics[width=\unitlength,page=8]{uniqpair2.pdf}}%
    \put(0.89810897,0.02654075){\color[rgb]{0,0,0}\makebox(0,0)[lt]{\lineheight{1.25}\smash{\begin{tabular}[t]{l}$f_1$\end{tabular}}}}%
    \put(0,0){\includegraphics[width=\unitlength,page=9]{uniqpair2.pdf}}%
  \end{picture}%
\endgroup%

%% file: pics/pairingeta.pdf_tex
%% Creator: Inkscape 1.1 (c68e22c387, 2021-05-23), www.inkscape.org
%% Accompanies image file 'pairingeta.pdf' (pdf, eps, ps)
\begingroup%
  \makeatletter%
  \providecommand\color[2][]{%
    \errmessage{(Inkscape) Color is used for the text in Inkscape, but the package 'color.sty' is not loaded}%
    \renewcommand\color[2][]{}%
  }%
  \providecommand\transparent[1]{%
    \errmessage{(Inkscape) Transparency is used (non-zero) for the text in Inkscape, but the package 'transparent.sty' is not loaded}%
    \renewcommand\transparent[1]{}%
  }%
  \providecommand\rotatebox[2]{#2}%
  \newcommand*\fsize{\dimexpr\f@size pt\relax}%
  \newcommand*\lineheight[1]{\fontsize{\fsize}{#1\fsize}\selectfont}%
  \ifx\svgwidth\undefined%
    \setlength{\unitlength}{101.92216035bp}%
    \ifx\svgscale\undefined%
      \relax%
    \else%
      \setlength{\unitlength}{\unitlength * \real{\svgscale}}%
    \fi%
  \else%
    \setlength{\unitlength}{\svgwidth}%
  \fi%
  \global\let\svgwidth\undefined%
  \global\let\svgscale\undefined%
  \makeatother%
  \begin{picture}(1,0.2282722)%
    \lineheight{1}%
    \setlength\tabcolsep{0pt}%
    \put(0,0){\includegraphics[width=\unitlength,page=1]{pairingeta.pdf}}%
    \put(-0.00445538,0.08306066){\makebox(0,0)[lt]{\lineheight{1.25}\smash{\begin{tabular}[t]{l}$\eta_A:A\to A\times A=$\end{tabular}}}}%
  \end{picture}%
\endgroup%

%% file: pics/pairingepsilon.pdf_tex
%% Creator: Inkscape 1.1 (c68e22c387, 2021-05-23), www.inkscape.org
%% Accompanies image file 'pairingepsilon.pdf' (pdf, eps, ps)
\begingroup%
  \makeatletter%
  \providecommand\color[2][]{%
    \errmessage{(Inkscape) Color is used for the text in Inkscape, but the package 'color.sty' is not loaded}%
    \renewcommand\color[2][]{}%
  }%
  \providecommand\transparent[1]{%
    \errmessage{(Inkscape) Transparency is used (non-zero) for the text in Inkscape, but the package 'transparent.sty' is not loaded}%
    \renewcommand\transparent[1]{}%
  }%
  \providecommand\rotatebox[2]{#2}%
  \newcommand*\fsize{\dimexpr\f@size pt\relax}%
  \newcommand*\lineheight[1]{\fontsize{\fsize}{#1\fsize}\selectfont}%
  \ifx\svgwidth\undefined%
    \setlength{\unitlength}{161.86841077bp}%
    \ifx\svgscale\undefined%
      \relax%
    \else%
      \setlength{\unitlength}{\unitlength * \real{\svgscale}}%
    \fi%
  \else%
    \setlength{\unitlength}{\svgwidth}%
  \fi%
  \global\let\svgwidth\undefined%
  \global\let\svgscale\undefined%
  \makeatother%
  \begin{picture}(1,0.14363518)%
    \lineheight{1}%
    \setlength\tabcolsep{0pt}%
    \put(0,0){\includegraphics[width=\unitlength,page=1]{pairingepsilon.pdf}}%
    \put(-0.00280538,0.05521354){\makebox(0,0)[lt]{\lineheight{1.25}\smash{\begin{tabular}[t]{l}$\epsilon_{A,B}:A\times B,A\times B\to A\times B=$\end{tabular}}}}%
  \end{picture}%
\endgroup%

%% file: pics/stlc.pdf_tex
%% Creator: Inkscape 1.2.2 (732a01da63, 2022-12-09), www.inkscape.org
%% Accompanies image file 'stlc.pdf' (pdf, eps, ps)
\begingroup%
  \makeatletter%
  \providecommand\color[2][]{%
    \errmessage{(Inkscape) Color is used for the text in Inkscape, but the package 'color.sty' is not loaded}%
    \renewcommand\color[2][]{}%
  }%
  \providecommand\transparent[1]{%
    \errmessage{(Inkscape) Transparency is used (non-zero) for the text in Inkscape, but the package 'transparent.sty' is not loaded}%
    \renewcommand\transparent[1]{}%
  }%
  \providecommand\rotatebox[2]{#2}%
  \newcommand*\fsize{\dimexpr\f@size pt\relax}%
  \newcommand*\lineheight[1]{\fontsize{\fsize}{#1\fsize}\selectfont}%
  \ifx\svgwidth\undefined%
    \setlength{\unitlength}{361.38557476bp}%
    \ifx\svgscale\undefined%
      \relax%
    \else%
      \setlength{\unitlength}{\unitlength * \real{\svgscale}}%
    \fi%
  \else%
    \setlength{\unitlength}{\svgwidth}%
  \fi%
  \global\let\svgwidth\undefined%
  \global\let\svgscale\undefined%
  \makeatother%
  \begin{picture}(1,0.36460174)%
    \lineheight{1}%
    \setlength\tabcolsep{0pt}%
    \put(0,0){\includegraphics[width=\unitlength,page=1]{stlc.pdf}}%
    \put(0.56147847,0.31678691){\color[rgb]{0,0,0}\makebox(0,0)[rt]{\lineheight{1.25}\smash{\begin{tabular}[t]{r}$\seval{\Gamma_1,x:T,\Gamma_2\vdash x:T}=$\end{tabular}}}}%
    \put(0,0){\includegraphics[width=\unitlength,page=2]{stlc.pdf}}%
    \put(0.56792388,0.27627731){\color[rgb]{0,0,0}\makebox(0,0)[lt]{\lineheight{1.25}\smash{\begin{tabular}[t]{l}$\Gamma_1$\end{tabular}}}}%
    \put(0.56846296,0.35182537){\color[rgb]{0,0,0}\makebox(0,0)[lt]{\lineheight{1.25}\smash{\begin{tabular}[t]{l}$\Gamma_2$\end{tabular}}}}%
    \put(0.63894379,0.32622043){\color[rgb]{0,0,0}\makebox(0,0)[lt]{\lineheight{1.25}\smash{\begin{tabular}[t]{l}$T$\end{tabular}}}}%
    \put(0,0){\includegraphics[width=\unitlength,page=3]{stlc.pdf}}%
    \put(0.56090357,0.18187774){\makebox(0,0)[rt]{\lineheight{1.25}\smash{\begin{tabular}[t]{r}$\seval{\Gamma_1,\Gamma_2\vdash \lambda x{:}T.u:T\to T'}=$\end{tabular}}}}%
    \put(0,0){\includegraphics[width=\unitlength,page=4]{stlc.pdf}}%
    \put(0.78411134,0.18450895){\color[rgb]{0,0,0}\makebox(0,0)[t]{\lineheight{1.25}\smash{\begin{tabular}[t]{c}$\seval{\Gamma_1,x:T,\Gamma_2\vdash u:T'}$\end{tabular}}}}%
    \put(0,0){\includegraphics[width=\unitlength,page=5]{stlc.pdf}}%
    \put(0.56290794,0.14353755){\color[rgb]{0,0,0}\makebox(0,0)[lt]{\lineheight{1.25}\smash{\begin{tabular}[t]{l}$\Gamma_1$\end{tabular}}}}%
    \put(0.56399452,0.21888795){\color[rgb]{0,0,0}\makebox(0,0)[lt]{\lineheight{1.25}\smash{\begin{tabular}[t]{l}$\Gamma_2$\end{tabular}}}}%
    \put(0.61266531,0.12278715){\color[rgb]{0,0,0}\makebox(0,0)[lt]{\lineheight{1.25}\smash{\begin{tabular}[t]{l}$T$\end{tabular}}}}%
    \put(0.56217343,0.04357302){\makebox(0,0)[rt]{\lineheight{1.25}\smash{\begin{tabular}[t]{r}$\seval{\Gamma\vdash u_2\,u_1:T_2}=$\end{tabular}}}}%
    \put(0,0){\includegraphics[width=\unitlength,page=6]{stlc.pdf}}%
    \put(0.77580995,0.06206354){\color[rgb]{0,0,0}\makebox(0,0)[t]{\lineheight{1.25}\smash{\begin{tabular}[t]{c}$\seval{\Gamma\vdash u_2:T_1\to T_2}$\end{tabular}}}}%
    \put(0,0){\includegraphics[width=\unitlength,page=7]{stlc.pdf}}%
    \put(0.77580978,0.01433051){\color[rgb]{0,0,0}\makebox(0,0)[t]{\lineheight{1.25}\smash{\begin{tabular}[t]{c}$\seval{\Gamma\vdash u_1:T_1}$\end{tabular}}}}%
    \put(0,0){\includegraphics[width=\unitlength,page=8]{stlc.pdf}}%
  \end{picture}%
\endgroup%

%% file: pics/fullabs.pdf_tex
%% Creator: Inkscape 1.2.2 (732a01da63, 2022-12-09), www.inkscape.org
%% Accompanies image file 'fullabs.pdf' (pdf, eps, ps)
\begingroup%
  \makeatletter%
  \providecommand\color[2][]{%
    \errmessage{(Inkscape) Color is used for the text in Inkscape, but the package 'color.sty' is not loaded}%
    \renewcommand\color[2][]{}%
  }%
  \providecommand\transparent[1]{%
    \errmessage{(Inkscape) Transparency is used (non-zero) for the text in Inkscape, but the package 'transparent.sty' is not loaded}%
    \renewcommand\transparent[1]{}%
  }%
  \providecommand\rotatebox[2]{#2}%
  \newcommand*\fsize{\dimexpr\f@size pt\relax}%
  \newcommand*\lineheight[1]{\fontsize{\fsize}{#1\fsize}\selectfont}%
  \ifx\svgwidth\undefined%
    \setlength{\unitlength}{428.03691349bp}%
    \ifx\svgscale\undefined%
      \relax%
    \else%
      \setlength{\unitlength}{\unitlength * \real{\svgscale}}%
    \fi%
  \else%
    \setlength{\unitlength}{\svgwidth}%
  \fi%
  \global\let\svgwidth\undefined%
  \global\let\svgscale\undefined%
  \makeatother%
  \begin{picture}(1,0.12374568)%
    \lineheight{1}%
    \setlength\tabcolsep{0pt}%
    \put(0,0){\includegraphics[width=\unitlength,page=1]{fullabs.pdf}}%
    \put(0.90272607,0.06479134){\color[rgb]{0,0,0}\makebox(0,0)[lt]{\lineheight{1.25}\smash{\begin{tabular}[t]{l}$\seval u$\end{tabular}}}}%
    \put(0,0){\includegraphics[width=\unitlength,page=2]{fullabs.pdf}}%
  \end{picture}%
\endgroup%

%% file: pics/stlc2.pdf_tex
%% Creator: Inkscape 1.2.2 (732a01da63, 2022-12-09), www.inkscape.org
%% Accompanies image file 'stlc2.pdf' (pdf, eps, ps)
\begingroup%
  \makeatletter%
  \providecommand\color[2][]{%
    \errmessage{(Inkscape) Color is used for the text in Inkscape, but the package 'color.sty' is not loaded}%
    \renewcommand\color[2][]{}%
  }%
  \providecommand\transparent[1]{%
    \errmessage{(Inkscape) Transparency is used (non-zero) for the text in Inkscape, but the package 'transparent.sty' is not loaded}%
    \renewcommand\transparent[1]{}%
  }%
  \providecommand\rotatebox[2]{#2}%
  \newcommand*\fsize{\dimexpr\f@size pt\relax}%
  \newcommand*\lineheight[1]{\fontsize{\fsize}{#1\fsize}\selectfont}%
  \ifx\svgwidth\undefined%
    \setlength{\unitlength}{328.74981078bp}%
    \ifx\svgscale\undefined%
      \relax%
    \else%
      \setlength{\unitlength}{\unitlength * \real{\svgscale}}%
    \fi%
  \else%
    \setlength{\unitlength}{\svgwidth}%
  \fi%
  \global\let\svgwidth\undefined%
  \global\let\svgscale\undefined%
  \makeatother%
  \begin{picture}(1,0.15314448)%
    \lineheight{1}%
    \setlength\tabcolsep{0pt}%
    \put(0.61658579,0.13909977){\makebox(0,0)[rt]{\lineheight{1.25}\smash{\begin{tabular}[t]{r}$\seval{\Gamma_1,\Gamma_2\vdash \lambda x{:}T.u:T\to T'}=$\end{tabular}}}}%
    \put(0,0){\includegraphics[width=\unitlength,page=1]{stlc2.pdf}}%
    \put(0.7745094,0.05396609){\color[rgb]{0,0,0}\makebox(0,0)[t]{\lineheight{1.25}\smash{\begin{tabular}[t]{c}$\seval{\Gamma_1,x:T,\Gamma_2\vdash u:T'}$\end{tabular}}}}%
    \put(0,0){\includegraphics[width=\unitlength,page=2]{stlc2.pdf}}%
    \put(0.34590335,0.07774744){\color[rgb]{0,0,0}\makebox(0,0)[lt]{\lineheight{1.25}\smash{\begin{tabular}[t]{l}$\Gamma_1\times\Gamma_2$\end{tabular}}}}%
    \put(0.42098906,0.012493){\color[rgb]{0,0,0}\makebox(0,0)[lt]{\lineheight{1.25}\smash{\begin{tabular}[t]{l}$T$\end{tabular}}}}%
    \put(0,0){\includegraphics[width=\unitlength,page=3]{stlc2.pdf}}%
    \put(0.4736151,0.02966432){\color[rgb]{0,0,0}\makebox(0,0)[lt]{\lineheight{1.25}\smash{\begin{tabular}[t]{l}$\Gamma_1\times T\times\Gamma_2$\end{tabular}}}}%
  \end{picture}%
\endgroup%

%% file: pics/stlc3.pdf_tex
%% Creator: Inkscape 1.2.2 (732a01da63, 2022-12-09), www.inkscape.org
%% Accompanies image file 'stlc3.pdf' (pdf, eps, ps)
\begingroup%
  \makeatletter%
  \providecommand\color[2][]{%
    \errmessage{(Inkscape) Color is used for the text in Inkscape, but the package 'color.sty' is not loaded}%
    \renewcommand\color[2][]{}%
  }%
  \providecommand\transparent[1]{%
    \errmessage{(Inkscape) Transparency is used (non-zero) for the text in Inkscape, but the package 'transparent.sty' is not loaded}%
    \renewcommand\transparent[1]{}%
  }%
  \providecommand\rotatebox[2]{#2}%
  \newcommand*\fsize{\dimexpr\f@size pt\relax}%
  \newcommand*\lineheight[1]{\fontsize{\fsize}{#1\fsize}\selectfont}%
  \ifx\svgwidth\undefined%
    \setlength{\unitlength}{402.53338004bp}%
    \ifx\svgscale\undefined%
      \relax%
    \else%
      \setlength{\unitlength}{\unitlength * \real{\svgscale}}%
    \fi%
  \else%
    \setlength{\unitlength}{\svgwidth}%
  \fi%
  \global\let\svgwidth\undefined%
  \global\let\svgscale\undefined%
  \makeatother%
  \begin{picture}(1,0.08385591)%
    \lineheight{1}%
    \setlength\tabcolsep{0pt}%
    \put(0,0){\includegraphics[width=\unitlength,page=1]{stlc3.pdf}}%
    \put(0.50608465,0.07235861){\color[rgb]{0,0,0}\makebox(0,0)[lt]{\lineheight{1.25}\smash{\begin{tabular}[t]{l}$\Gamma_1\times\Gamma_2$\end{tabular}}}}%
    \put(0.5599545,0.01533878){\color[rgb]{0,0,0}\makebox(0,0)[lt]{\lineheight{1.25}\smash{\begin{tabular}[t]{l}$T$\end{tabular}}}}%
    \put(0.70354704,0.03308903){\color[rgb]{0,0,0}\makebox(0,0)[lt]{\lineheight{1.25}\smash{\begin{tabular}[t]{l}$\Gamma_1\times T\times\Gamma_2$\end{tabular}}}}%
    \put(0,0){\includegraphics[width=\unitlength,page=2]{stlc3.pdf}}%
    \put(0.23405752,0.06304258){\color[rgb]{0,0,0}\makebox(0,0)[lt]{\lineheight{1.25}\smash{\begin{tabular}[t]{l}$\Gamma_1\times\Gamma_2$\end{tabular}}}}%
    \put(0.29910657,0.0060228){\color[rgb]{0,0,0}\makebox(0,0)[lt]{\lineheight{1.25}\smash{\begin{tabular}[t]{l}$T$\end{tabular}}}}%
    \put(0,0){\includegraphics[width=\unitlength,page=3]{stlc3.pdf}}%
    \put(0.34208636,0.02377305){\color[rgb]{0,0,0}\makebox(0,0)[lt]{\lineheight{1.25}\smash{\begin{tabular}[t]{l}$\Gamma_1\times T\times\Gamma_2$\end{tabular}}}}%
    \put(0.46676669,0.0428655){\makebox(0,0)[lt]{\lineheight{1.25}\smash{\begin{tabular}[t]{l}$=$\end{tabular}}}}%
    \put(0,0){\includegraphics[width=\unitlength,page=4]{stlc3.pdf}}%
  \end{picture}%
\endgroup%

%% file: pics/idid.pdf_tex
%% Creator: Inkscape 1.1 (c68e22c387, 2021-05-23), www.inkscape.org
%% Accompanies image file 'idid.pdf' (pdf, eps, ps)
\begingroup%
  \makeatletter%
  \providecommand\color[2][]{%
    \errmessage{(Inkscape) Color is used for the text in Inkscape, but the package 'color.sty' is not loaded}%
    \renewcommand\color[2][]{}%
  }%
  \providecommand\transparent[1]{%
    \errmessage{(Inkscape) Transparency is used (non-zero) for the text in Inkscape, but the package 'transparent.sty' is not loaded}%
    \renewcommand\transparent[1]{}%
  }%
  \providecommand\rotatebox[2]{#2}%
  \newcommand*\fsize{\dimexpr\f@size pt\relax}%
  \newcommand*\lineheight[1]{\fontsize{\fsize}{#1\fsize}\selectfont}%
  \ifx\svgwidth\undefined%
    \setlength{\unitlength}{48.93749717bp}%
    \ifx\svgscale\undefined%
      \relax%
    \else%
      \setlength{\unitlength}{\unitlength * \real{\svgscale}}%
    \fi%
  \else%
    \setlength{\unitlength}{\svgwidth}%
  \fi%
  \global\let\svgwidth\undefined%
  \global\let\svgscale\undefined%
  \makeatother%
  \begin{picture}(1,0.54406134)%
    \lineheight{1}%
    \setlength\tabcolsep{0pt}%
    \put(0,0){\includegraphics[width=\unitlength,page=1]{idid.pdf}}%
  \end{picture}%
\endgroup%

%% file: pics/fnfx.pdf_tex
%% Creator: Inkscape 1.2.2 (732a01da63, 2022-12-09), www.inkscape.org
%% Accompanies image file 'fnfx.pdf' (pdf, eps, ps)
\begingroup%
  \makeatletter%
  \providecommand\color[2][]{%
    \errmessage{(Inkscape) Color is used for the text in Inkscape, but the package 'color.sty' is not loaded}%
    \renewcommand\color[2][]{}%
  }%
  \providecommand\transparent[1]{%
    \errmessage{(Inkscape) Transparency is used (non-zero) for the text in Inkscape, but the package 'transparent.sty' is not loaded}%
    \renewcommand\transparent[1]{}%
  }%
  \providecommand\rotatebox[2]{#2}%
  \newcommand*\fsize{\dimexpr\f@size pt\relax}%
  \newcommand*\lineheight[1]{\fontsize{\fsize}{#1\fsize}\selectfont}%
  \ifx\svgwidth\undefined%
    \setlength{\unitlength}{196.21240814bp}%
    \ifx\svgscale\undefined%
      \relax%
    \else%
      \setlength{\unitlength}{\unitlength * \real{\svgscale}}%
    \fi%
  \else%
    \setlength{\unitlength}{\svgwidth}%
  \fi%
  \global\let\svgwidth\undefined%
  \global\let\svgscale\undefined%
  \makeatother%
  \begin{picture}(1,0.44156351)%
    \lineheight{1}%
    \setlength\tabcolsep{0pt}%
    \put(0,0){\includegraphics[width=\unitlength,page=1]{fnfx.pdf}}%
  \end{picture}%
\endgroup%

%% file: pics/fnfx2.pdf_tex
%% Creator: Inkscape 1.2.2 (732a01da63, 2022-12-09), www.inkscape.org
%% Accompanies image file 'fnfx2.pdf' (pdf, eps, ps)
\begingroup%
  \makeatletter%
  \providecommand\color[2][]{%
    \errmessage{(Inkscape) Color is used for the text in Inkscape, but the package 'color.sty' is not loaded}%
    \renewcommand\color[2][]{}%
  }%
  \providecommand\transparent[1]{%
    \errmessage{(Inkscape) Transparency is used (non-zero) for the text in Inkscape, but the package 'transparent.sty' is not loaded}%
    \renewcommand\transparent[1]{}%
  }%
  \providecommand\rotatebox[2]{#2}%
  \newcommand*\fsize{\dimexpr\f@size pt\relax}%
  \newcommand*\lineheight[1]{\fontsize{\fsize}{#1\fsize}\selectfont}%
  \ifx\svgwidth\undefined%
    \setlength{\unitlength}{198.88646782bp}%
    \ifx\svgscale\undefined%
      \relax%
    \else%
      \setlength{\unitlength}{\unitlength * \real{\svgscale}}%
    \fi%
  \else%
    \setlength{\unitlength}{\svgwidth}%
  \fi%
  \global\let\svgwidth\undefined%
  \global\let\svgscale\undefined%
  \makeatother%
  \begin{picture}(1,0.43562661)%
    \lineheight{1}%
    \setlength\tabcolsep{0pt}%
    \put(0,0){\includegraphics[width=\unitlength,page=1]{fnfx2.pdf}}%
  \end{picture}%
\endgroup%

%% file: pics/lambprod.pdf_tex
%% Creator: Inkscape 1.2.2 (732a01da63, 2022-12-09), www.inkscape.org
%% Accompanies image file 'lambprod.pdf' (pdf, eps, ps)
\begingroup%
  \makeatletter%
  \providecommand\color[2][]{%
    \errmessage{(Inkscape) Color is used for the text in Inkscape, but the package 'color.sty' is not loaded}%
    \renewcommand\color[2][]{}%
  }%
  \providecommand\transparent[1]{%
    \errmessage{(Inkscape) Transparency is used (non-zero) for the text in Inkscape, but the package 'transparent.sty' is not loaded}%
    \renewcommand\transparent[1]{}%
  }%
  \providecommand\rotatebox[2]{#2}%
  \newcommand*\fsize{\dimexpr\f@size pt\relax}%
  \newcommand*\lineheight[1]{\fontsize{\fsize}{#1\fsize}\selectfont}%
  \ifx\svgwidth\undefined%
    \setlength{\unitlength}{151.61998867bp}%
    \ifx\svgscale\undefined%
      \relax%
    \else%
      \setlength{\unitlength}{\unitlength * \real{\svgscale}}%
    \fi%
  \else%
    \setlength{\unitlength}{\svgwidth}%
  \fi%
  \global\let\svgwidth\undefined%
  \global\let\svgscale\undefined%
  \makeatother%
  \begin{picture}(1,0.69623067)%
    \lineheight{1}%
    \setlength\tabcolsep{0pt}%
    \put(0,0){\includegraphics[width=\unitlength,page=1]{lambprod.pdf}}%
    \put(0.59820676,0.62450538){\makebox(0,0)[t]{\lineheight{1.25}\smash{\begin{tabular}[t]{c}$\seval{u_2}$\end{tabular}}}}%
    \put(0.24971974,0.55033296){\makebox(0,0)[lt]{\lineheight{1.25}\smash{\begin{tabular}[t]{l}$\Gamma$\end{tabular}}}}%
    \put(0,0){\includegraphics[width=\unitlength,page=2]{lambprod.pdf}}%
    \put(0.59820676,0.47610806){\makebox(0,0)[t]{\lineheight{1.25}\smash{\begin{tabular}[t]{c}$\seval{u_1}$\end{tabular}}}}%
    \put(0,0){\includegraphics[width=\unitlength,page=3]{lambprod.pdf}}%
    \put(0.79730642,0.23743545){\makebox(0,0)[t]{\lineheight{1.25}\smash{\begin{tabular}[t]{c}$\seval{u}$\end{tabular}}}}%
    \put(0,0){\includegraphics[width=\unitlength,page=4]{lambprod.pdf}}%
    \put(0.2708062,0.28543197){\makebox(0,0)[rt]{\lineheight{1.25}\smash{\begin{tabular}[t]{r}$\Gamma_3$\end{tabular}}}}%
    \put(0.2708062,0.22607241){\makebox(0,0)[rt]{\lineheight{1.25}\smash{\begin{tabular}[t]{r}$\Gamma_2$\end{tabular}}}}%
    \put(0.2708062,0.16671314){\makebox(0,0)[rt]{\lineheight{1.25}\smash{\begin{tabular}[t]{r}$\Gamma_1$\end{tabular}}}}%
    \put(0.54068295,0.14371377){\makebox(0,0)[rt]{\lineheight{1.25}\smash{\begin{tabular}[t]{r}$T_2$\end{tabular}}}}%
    \put(0.59214725,0.08674396){\makebox(0,0)[lt]{\lineheight{1.25}\smash{\begin{tabular}[t]{l}$T_1$\end{tabular}}}}%
    \put(0,0){\includegraphics[width=\unitlength,page=5]{lambprod.pdf}}%
  \end{picture}%
\endgroup%

%% file: pics/untypedlc.pdf_tex
%% Creator: Inkscape 1.1 (c68e22c387, 2021-05-23), www.inkscape.org
%% Accompanies image file 'untypedlc.pdf' (pdf, eps, ps)
\begingroup%
  \makeatletter%
  \providecommand\color[2][]{%
    \errmessage{(Inkscape) Color is used for the text in Inkscape, but the package 'color.sty' is not loaded}%
    \renewcommand\color[2][]{}%
  }%
  \providecommand\transparent[1]{%
    \errmessage{(Inkscape) Transparency is used (non-zero) for the text in Inkscape, but the package 'transparent.sty' is not loaded}%
    \renewcommand\transparent[1]{}%
  }%
  \providecommand\rotatebox[2]{#2}%
  \newcommand*\fsize{\dimexpr\f@size pt\relax}%
  \newcommand*\lineheight[1]{\fontsize{\fsize}{#1\fsize}\selectfont}%
  \ifx\svgwidth\undefined%
    \setlength{\unitlength}{203.63690432bp}%
    \ifx\svgscale\undefined%
      \relax%
    \else%
      \setlength{\unitlength}{\unitlength * \real{\svgscale}}%
    \fi%
  \else%
    \setlength{\unitlength}{\svgwidth}%
  \fi%
  \global\let\svgwidth\undefined%
  \global\let\svgscale\undefined%
  \makeatother%
  \begin{picture}(1,0.24307991)%
    \lineheight{1}%
    \setlength\tabcolsep{0pt}%
    \put(0,0){\includegraphics[width=\unitlength,page=1]{untypedlc.pdf}}%
    \put(0.35293068,0.19243781){\makebox(0,0)[t]{\lineheight{1.25}\smash{\begin{tabular}[t]{c}$\rho$\end{tabular}}}}%
    \put(0.63228729,0.19335879){\makebox(0,0)[t]{\lineheight{1.25}\smash{\begin{tabular}[t]{c}$\iota$\end{tabular}}}}%
    \put(0,0){\includegraphics[width=\unitlength,page=2]{untypedlc.pdf}}%
    \put(0.26412804,0.1910644){\makebox(0,0)[rt]{\lineheight{1.25}\smash{\begin{tabular}[t]{r}$U\Rightarrow U$\end{tabular}}}}%
    \put(0.73568495,0.19151746){\makebox(0,0)[lt]{\lineheight{1.25}\smash{\begin{tabular}[t]{l}$U\Rightarrow U$\end{tabular}}}}%
    \put(0.53959266,0.19105038){\makebox(0,0)[rt]{\lineheight{1.25}\smash{\begin{tabular}[t]{r}$U$\end{tabular}}}}%
    \put(0,0){\includegraphics[width=\unitlength,page=3]{untypedlc.pdf}}%
    \put(0.3308325,0.02670164){\makebox(0,0)[t]{\lineheight{1.25}\smash{\begin{tabular}[t]{c}$\rho$\end{tabular}}}}%
    \put(0.4665508,0.02762262){\makebox(0,0)[t]{\lineheight{1.25}\smash{\begin{tabular}[t]{c}$\iota$\end{tabular}}}}%
    \put(0,0){\includegraphics[width=\unitlength,page=4]{untypedlc.pdf}}%
    \put(0.5939964,0.02670228){\makebox(0,0)[t]{\lineheight{1.25}\smash{\begin{tabular}[t]{c}$=$\end{tabular}}}}%
    \put(0.70026448,0.05708686){\makebox(0,0)[t]{\lineheight{1.25}\smash{\begin{tabular}[t]{c}$\id_{U\Rightarrow U}$\end{tabular}}}}%
    \put(0,0){\includegraphics[width=\unitlength,page=5]{untypedlc.pdf}}%
  \end{picture}%
\endgroup%

%% file: pics/utlc.pdf_tex
%% Creator: Inkscape 1.2.2 (732a01da63, 2022-12-09), www.inkscape.org
%% Accompanies image file 'utlc.pdf' (pdf, eps, ps)
\begingroup%
  \makeatletter%
  \providecommand\color[2][]{%
    \errmessage{(Inkscape) Color is used for the text in Inkscape, but the package 'color.sty' is not loaded}%
    \renewcommand\color[2][]{}%
  }%
  \providecommand\transparent[1]{%
    \errmessage{(Inkscape) Transparency is used (non-zero) for the text in Inkscape, but the package 'transparent.sty' is not loaded}%
    \renewcommand\transparent[1]{}%
  }%
  \providecommand\rotatebox[2]{#2}%
  \newcommand*\fsize{\dimexpr\f@size pt\relax}%
  \newcommand*\lineheight[1]{\fontsize{\fsize}{#1\fsize}\selectfont}%
  \ifx\svgwidth\undefined%
    \setlength{\unitlength}{300.77610549bp}%
    \ifx\svgscale\undefined%
      \relax%
    \else%
      \setlength{\unitlength}{\unitlength * \real{\svgscale}}%
    \fi%
  \else%
    \setlength{\unitlength}{\svgwidth}%
  \fi%
  \global\let\svgwidth\undefined%
  \global\let\svgscale\undefined%
  \makeatother%
  \begin{picture}(1,0.37573373)%
    \lineheight{1}%
    \setlength\tabcolsep{0pt}%
    \put(0,0){\includegraphics[width=\unitlength,page=1]{utlc.pdf}}%
    \put(0.55074319,0.31828373){\color[rgb]{0,0,0}\makebox(0,0)[rt]{\lineheight{1.25}\smash{\begin{tabular}[t]{r}$\seval{\Gamma_1,x,\Gamma_2\vdash x}=$\end{tabular}}}}%
    \put(0,0){\includegraphics[width=\unitlength,page=2]{utlc.pdf}}%
    \put(0.55848742,0.26961103){\color[rgb]{0,0,0}\makebox(0,0)[lt]{\lineheight{1.25}\smash{\begin{tabular}[t]{l}$\Gamma_1$\end{tabular}}}}%
    \put(0.55913512,0.36038279){\color[rgb]{0,0,0}\makebox(0,0)[lt]{\lineheight{1.25}\smash{\begin{tabular}[t]{l}$\Gamma_2$\end{tabular}}}}%
    \put(0,0){\includegraphics[width=\unitlength,page=3]{utlc.pdf}}%
    \put(0.55503955,0.18112462){\makebox(0,0)[rt]{\lineheight{1.25}\smash{\begin{tabular}[t]{r}$\seval{\Gamma_1,\Gamma_2\vdash \lambda x.u}=$\end{tabular}}}}%
    \put(0,0){\includegraphics[width=\unitlength,page=4]{utlc.pdf}}%
    \put(0.55744784,0.13505849){\color[rgb]{0,0,0}\makebox(0,0)[lt]{\lineheight{1.25}\smash{\begin{tabular}[t]{l}$\Gamma_1$\end{tabular}}}}%
    \put(0.55875337,0.22559277){\color[rgb]{0,0,0}\makebox(0,0)[lt]{\lineheight{1.25}\smash{\begin{tabular}[t]{l}$\Gamma_2$\end{tabular}}}}%
    \put(0.56404583,0.05235325){\makebox(0,0)[rt]{\lineheight{1.25}\smash{\begin{tabular}[t]{r}$\seval{\Gamma\vdash u_2\,u_1}=$\end{tabular}}}}%
    \put(0,0){\includegraphics[width=\unitlength,page=5]{utlc.pdf}}%
    \put(0.76313542,0.07456979){\color[rgb]{0,0,0}\makebox(0,0)[t]{\lineheight{1.25}\smash{\begin{tabular}[t]{c}$\seval{\Gamma\vdash u_2}$\end{tabular}}}}%
    \put(0,0){\includegraphics[width=\unitlength,page=6]{utlc.pdf}}%
    \put(0.76427964,0.01824021){\color[rgb]{0,0,0}\makebox(0,0)[t]{\lineheight{1.25}\smash{\begin{tabular}[t]{c}$\seval{\Gamma\vdash u_1}$\end{tabular}}}}%
    \put(0,0){\includegraphics[width=\unitlength,page=7]{utlc.pdf}}%
    \put(0.75565491,0.18428604){\color[rgb]{0,0,0}\makebox(0,0)[t]{\lineheight{1.25}\smash{\begin{tabular}[t]{c}$\seval{\Gamma_1,x,\Gamma_2\vdash u}$\end{tabular}}}}%
  \end{picture}%
\endgroup%

%% file: pics/idid2.pdf_tex
%% Creator: Inkscape 1.2.2 (732a01da63, 2022-12-09), www.inkscape.org
%% Accompanies image file 'idid2.pdf' (pdf, eps, ps)
\begingroup%
  \makeatletter%
  \providecommand\color[2][]{%
    \errmessage{(Inkscape) Color is used for the text in Inkscape, but the package 'color.sty' is not loaded}%
    \renewcommand\color[2][]{}%
  }%
  \providecommand\transparent[1]{%
    \errmessage{(Inkscape) Transparency is used (non-zero) for the text in Inkscape, but the package 'transparent.sty' is not loaded}%
    \renewcommand\transparent[1]{}%
  }%
  \providecommand\rotatebox[2]{#2}%
  \newcommand*\fsize{\dimexpr\f@size pt\relax}%
  \newcommand*\lineheight[1]{\fontsize{\fsize}{#1\fsize}\selectfont}%
  \ifx\svgwidth\undefined%
    \setlength{\unitlength}{143.43748867bp}%
    \ifx\svgscale\undefined%
      \relax%
    \else%
      \setlength{\unitlength}{\unitlength * \real{\svgscale}}%
    \fi%
  \else%
    \setlength{\unitlength}{\svgwidth}%
  \fi%
  \global\let\svgwidth\undefined%
  \global\let\svgscale\undefined%
  \makeatother%
  \begin{picture}(1,0.18562093)%
    \lineheight{1}%
    \setlength\tabcolsep{0pt}%
    \put(0,0){\includegraphics[width=\unitlength,page=1]{idid2.pdf}}%
    \put(0.56210183,0.09214566){\makebox(0,0)[t]{\lineheight{1.25}\smash{\begin{tabular}[t]{c}$=$\end{tabular}}}}%
  \end{picture}%
\endgroup%

%% file: pics/omega.pdf_tex
%% Creator: Inkscape 1.2.2 (732a01da63, 2022-12-09), www.inkscape.org
%% Accompanies image file 'omega.pdf' (pdf, eps, ps)
\begingroup%
  \makeatletter%
  \providecommand\color[2][]{%
    \errmessage{(Inkscape) Color is used for the text in Inkscape, but the package 'color.sty' is not loaded}%
    \renewcommand\color[2][]{}%
  }%
  \providecommand\transparent[1]{%
    \errmessage{(Inkscape) Transparency is used (non-zero) for the text in Inkscape, but the package 'transparent.sty' is not loaded}%
    \renewcommand\transparent[1]{}%
  }%
  \providecommand\rotatebox[2]{#2}%
  \newcommand*\fsize{\dimexpr\f@size pt\relax}%
  \newcommand*\lineheight[1]{\fontsize{\fsize}{#1\fsize}\selectfont}%
  \ifx\svgwidth\undefined%
    \setlength{\unitlength}{112.68746575bp}%
    \ifx\svgscale\undefined%
      \relax%
    \else%
      \setlength{\unitlength}{\unitlength * \real{\svgscale}}%
    \fi%
  \else%
    \setlength{\unitlength}{\svgwidth}%
  \fi%
  \global\let\svgwidth\undefined%
  \global\let\svgscale\undefined%
  \makeatother%
  \begin{picture}(1,0.9999996)%
    \lineheight{1}%
    \setlength\tabcolsep{0pt}%
    \put(0,0){\includegraphics[width=\unitlength,page=1]{omega.pdf}}%
  \end{picture}%
\endgroup%

%% file: pics/sbst.pdf_tex
%% Creator: Inkscape 1.0.1 (c497b03c, 2020-09-10), www.inkscape.org
%% Accompanies image file 'sbst.pdf' (pdf, eps, ps)
\begingroup%
  \makeatletter%
  \providecommand\color[2][]{%
    \errmessage{(Inkscape) Color is used for the text in Inkscape, but the package 'color.sty' is not loaded}%
    \renewcommand\color[2][]{}%
  }%
  \providecommand\transparent[1]{%
    \errmessage{(Inkscape) Transparency is used (non-zero) for the text in Inkscape, but the package 'transparent.sty' is not loaded}%
    \renewcommand\transparent[1]{}%
  }%
  \providecommand\rotatebox[2]{#2}%
  \newcommand*\fsize{\dimexpr\f@size pt\relax}%
  \newcommand*\lineheight[1]{\fontsize{\fsize}{#1\fsize}\selectfont}%
  \ifx\svgwidth\undefined%
    \setlength{\unitlength}{322.1529145bp}%
    \ifx\svgscale\undefined%
      \relax%
    \else%
      \setlength{\unitlength}{\unitlength * \real{\svgscale}}%
    \fi%
  \else%
    \setlength{\unitlength}{\svgwidth}%
  \fi%
  \global\let\svgwidth\undefined%
  \global\let\svgscale\undefined%
  \makeatother%
  \begin{picture}(1,0.10589427)%
    \lineheight{1}%
    \setlength\tabcolsep{0pt}%
    \put(0,0){\includegraphics[width=\unitlength,page=1]{sbst.pdf}}%
    \put(0.76820949,0.05565319){\color[rgb]{0,0,0}\makebox(0,0)[t]{\lineheight{1.25}\smash{\begin{tabular}[t]{c}$\seval{\Gamma_1,x:T',\Gamma_2\vdash u:T}$\end{tabular}}}}%
    \put(0,0){\includegraphics[width=\unitlength,page=2]{sbst.pdf}}%
    \put(0.33548289,0.09156196){\color[rgb]{0,0,0}\makebox(0,0)[lt]{\lineheight{1.25}\smash{\begin{tabular}[t]{l}$\Gamma_1\times\Gamma_2$\end{tabular}}}}%
    \put(0.38658229,0.04317729){\color[rgb]{0,0,0}\makebox(0,0)[lt]{\lineheight{1.25}\smash{\begin{tabular}[t]{l}$T'$\end{tabular}}}}%
    \put(0,0){\includegraphics[width=\unitlength,page=3]{sbst.pdf}}%
    \put(0.46115359,0.03085378){\color[rgb]{0,0,0}\makebox(0,0)[lt]{\lineheight{1.25}\smash{\begin{tabular}[t]{l}$\Gamma_1\times T'\times\Gamma_2$\end{tabular}}}}%
    \put(0,0){\includegraphics[width=\unitlength,page=4]{sbst.pdf}}%
    \put(0.22907571,0.03171488){\color[rgb]{0,0,0}\makebox(0,0)[t]{\lineheight{1.25}\smash{\begin{tabular}[t]{c}$\seval{\Gamma_1,\Gamma_2\vdash v:T'}$\end{tabular}}}}%
    \put(0,0){\includegraphics[width=\unitlength,page=5]{sbst.pdf}}%
  \end{picture}%
\endgroup%

%% file: pics/compactcc.pdf_tex
%% Creator: Inkscape 1.1 (c68e22c387, 2021-05-23), www.inkscape.org
%% Accompanies image file 'compactcc.pdf' (pdf, eps, ps)
\begingroup%
  \makeatletter%
  \providecommand\color[2][]{%
    \errmessage{(Inkscape) Color is used for the text in Inkscape, but the package 'color.sty' is not loaded}%
    \renewcommand\color[2][]{}%
  }%
  \providecommand\transparent[1]{%
    \errmessage{(Inkscape) Transparency is used (non-zero) for the text in Inkscape, but the package 'transparent.sty' is not loaded}%
    \renewcommand\transparent[1]{}%
  }%
  \providecommand\rotatebox[2]{#2}%
  \newcommand*\fsize{\dimexpr\f@size pt\relax}%
  \newcommand*\lineheight[1]{\fontsize{\fsize}{#1\fsize}\selectfont}%
  \ifx\svgwidth\undefined%
    \setlength{\unitlength}{131.98966694bp}%
    \ifx\svgscale\undefined%
      \relax%
    \else%
      \setlength{\unitlength}{\unitlength * \real{\svgscale}}%
    \fi%
  \else%
    \setlength{\unitlength}{\svgwidth}%
  \fi%
  \global\let\svgwidth\undefined%
  \global\let\svgscale\undefined%
  \makeatother%
  \begin{picture}(1,0.22956453)%
    \lineheight{1}%
    \setlength\tabcolsep{0pt}%
    \put(-0.00458726,0.09567581){\color[rgb]{0,0,0}\makebox(0,0)[lt]{\lineheight{1.25}\smash{\begin{tabular}[t]{l}$\eta_A=$\end{tabular}}}}%
    \put(0,0){\includegraphics[width=\unitlength,page=1]{compactcc.pdf}}%
    \put(0.55619035,0.09223477){\color[rgb]{0,0,0}\makebox(0,0)[lt]{\lineheight{1.25}\smash{\begin{tabular}[t]{l}$\epsilon_A=$\end{tabular}}}}%
    \put(0,0){\includegraphics[width=\unitlength,page=2]{compactcc.pdf}}%
    \put(0.78354516,0.17809332){\color[rgb]{0,0,0}\makebox(0,0)[lt]{\lineheight{1.25}\smash{\begin{tabular}[t]{l}$A^*$\end{tabular}}}}%
    \put(0.77871598,0.01068387){\color[rgb]{0,0,0}\makebox(0,0)[lt]{\lineheight{1.25}\smash{\begin{tabular}[t]{l}$A$\end{tabular}}}}%
    \put(0.26270212,0.18292249){\color[rgb]{0,0,0}\makebox(0,0)[lt]{\lineheight{1.25}\smash{\begin{tabular}[t]{l}$A$\end{tabular}}}}%
    \put(0.26666649,0.01551288){\color[rgb]{0,0,0}\makebox(0,0)[lt]{\lineheight{1.25}\smash{\begin{tabular}[t]{l}$A^*$\end{tabular}}}}%
  \end{picture}%
\endgroup%

%% file: pics/compactcceq.pdf_tex
%% Creator: Inkscape 1.2.2 (b0a84865, 2022-12-01), www.inkscape.org
%% Accompanies image file 'compactcceq.pdf' (pdf, eps, ps)
\begingroup%
  \makeatletter%
  \providecommand\color[2][]{%
    \errmessage{(Inkscape) Color is used for the text in Inkscape, but the package 'color.sty' is not loaded}%
    \renewcommand\color[2][]{}%
  }%
  \providecommand\transparent[1]{%
    \errmessage{(Inkscape) Transparency is used (non-zero) for the text in Inkscape, but the package 'transparent.sty' is not loaded}%
    \renewcommand\transparent[1]{}%
  }%
  \providecommand\rotatebox[2]{#2}%
  \newcommand*\fsize{\dimexpr\f@size pt\relax}%
  \newcommand*\lineheight[1]{\fontsize{\fsize}{#1\fsize}\selectfont}%
  \ifx\svgwidth\undefined%
    \setlength{\unitlength}{170.49917802bp}%
    \ifx\svgscale\undefined%
      \relax%
    \else%
      \setlength{\unitlength}{\unitlength * \real{\svgscale}}%
    \fi%
  \else%
    \setlength{\unitlength}{\svgwidth}%
  \fi%
  \global\let\svgwidth\undefined%
  \global\let\svgscale\undefined%
  \makeatother%
  \begin{picture}(1,0.3119684)%
    \lineheight{1}%
    \setlength\tabcolsep{0pt}%
    \put(0.14049472,0.148506){\color[rgb]{0,0,0}\makebox(0,0)[lt]{\lineheight{1.25}\smash{\begin{tabular}[t]{l}$=$\end{tabular}}}}%
    \put(0,0){\includegraphics[width=\unitlength,page=1]{compactcceq.pdf}}%
    \put(0.67754381,0.14459594){\color[rgb]{0,0,0}\makebox(0,0)[lt]{\lineheight{1.25}\smash{\begin{tabular}[t]{l}$=$\end{tabular}}}}%
    \put(0.2353785,0.00984502){\color[rgb]{0,0,0}\makebox(0,0)[lt]{\lineheight{1.25}\smash{\begin{tabular}[t]{l}$A$\end{tabular}}}}%
    \put(0.02568663,0.17243202){\color[rgb]{0,0,0}\makebox(0,0)[lt]{\lineheight{1.25}\smash{\begin{tabular}[t]{l}$A$\end{tabular}}}}%
    \put(0.29235904,0.27586109){\color[rgb]{0,0,0}\makebox(0,0)[lt]{\lineheight{1.25}\smash{\begin{tabular}[t]{l}$A$\end{tabular}}}}%
    \put(0,0){\includegraphics[width=\unitlength,page=2]{compactcceq.pdf}}%
    \put(0.84329059,0.00827077){\color[rgb]{0,0,0}\makebox(0,0)[lt]{\lineheight{1.25}\smash{\begin{tabular}[t]{l}$A^*$\end{tabular}}}}%
    \put(0.75111922,0.27428671){\color[rgb]{0,0,0}\makebox(0,0)[lt]{\lineheight{1.25}\smash{\begin{tabular}[t]{l}$A^*$\end{tabular}}}}%
    \put(0.5767565,0.17492436){\color[rgb]{0,0,0}\makebox(0,0)[lt]{\lineheight{1.25}\smash{\begin{tabular}[t]{l}$A^*$\end{tabular}}}}%
    \put(0,0){\includegraphics[width=\unitlength,page=3]{compactcceq.pdf}}%
  \end{picture}%
\endgroup%

%% file: pics/compcc2moncc.pdf_tex
%% Creator: Inkscape 1.1 (c68e22c387, 2021-05-23), www.inkscape.org
%% Accompanies image file 'compcc2moncc.pdf' (pdf, eps, ps)
\begingroup%
  \makeatletter%
  \providecommand\color[2][]{%
    \errmessage{(Inkscape) Color is used for the text in Inkscape, but the package 'color.sty' is not loaded}%
    \renewcommand\color[2][]{}%
  }%
  \providecommand\transparent[1]{%
    \errmessage{(Inkscape) Transparency is used (non-zero) for the text in Inkscape, but the package 'transparent.sty' is not loaded}%
    \renewcommand\transparent[1]{}%
  }%
  \providecommand\rotatebox[2]{#2}%
  \newcommand*\fsize{\dimexpr\f@size pt\relax}%
  \newcommand*\lineheight[1]{\fontsize{\fsize}{#1\fsize}\selectfont}%
  \ifx\svgwidth\undefined%
    \setlength{\unitlength}{308.29833705bp}%
    \ifx\svgscale\undefined%
      \relax%
    \else%
      \setlength{\unitlength}{\unitlength * \real{\svgscale}}%
    \fi%
  \else%
    \setlength{\unitlength}{\svgwidth}%
  \fi%
  \global\let\svgwidth\undefined%
  \global\let\svgscale\undefined%
  \makeatother%
  \begin{picture}(1,0.10128756)%
    \lineheight{1}%
    \setlength\tabcolsep{0pt}%
    \put(0,0){\includegraphics[width=\unitlength,page=1]{compcc2moncc.pdf}}%
    \put(0.19082419,0.08131903){\color[rgb]{0,0,0}\makebox(0,0)[lt]{\lineheight{1.25}\smash{\begin{tabular}[t]{l}$A^*$\end{tabular}}}}%
    \put(0.30195004,0.05582037){\color[rgb]{0,0,0}\makebox(0,0)[lt]{\lineheight{1.25}\smash{\begin{tabular}[t]{l}$A^*\otimes B\stackrel{\text{def}}=A\multimap B$\end{tabular}}}}%
    \put(0.04842869,0.04341575){\color[rgb]{0,0,0}\makebox(0,0)[lt]{\lineheight{1.25}\smash{\begin{tabular}[t]{l}$A$\end{tabular}}}}%
    \put(0.19270968,0.03170013){\color[rgb]{0,0,0}\makebox(0,0)[lt]{\lineheight{1.25}\smash{\begin{tabular}[t]{l}$B$\end{tabular}}}}%
    \put(0.11687531,0.01977778){\color[rgb]{0,0,0}\makebox(0,0)[lt]{\lineheight{1.25}\smash{\begin{tabular}[t]{l}$f$\end{tabular}}}}%
    \put(0,0){\includegraphics[width=\unitlength,page=2]{compcc2moncc.pdf}}%
    \put(0.00520505,0.01984673){\color[rgb]{0,0,0}\makebox(0,0)[lt]{\lineheight{1.25}\smash{\begin{tabular}[t]{l}$C$\end{tabular}}}}%
  \end{picture}%
\endgroup%

%% file: pics/compcc2moncceval.pdf_tex
%% Creator: Inkscape 1.1 (c68e22c387, 2021-05-23), www.inkscape.org
%% Accompanies image file 'compcc2moncceval.pdf' (pdf, eps, ps)
\begingroup%
  \makeatletter%
  \providecommand\color[2][]{%
    \errmessage{(Inkscape) Color is used for the text in Inkscape, but the package 'color.sty' is not loaded}%
    \renewcommand\color[2][]{}%
  }%
  \providecommand\transparent[1]{%
    \errmessage{(Inkscape) Transparency is used (non-zero) for the text in Inkscape, but the package 'transparent.sty' is not loaded}%
    \renewcommand\transparent[1]{}%
  }%
  \providecommand\rotatebox[2]{#2}%
  \newcommand*\fsize{\dimexpr\f@size pt\relax}%
  \newcommand*\lineheight[1]{\fontsize{\fsize}{#1\fsize}\selectfont}%
  \ifx\svgwidth\undefined%
    \setlength{\unitlength}{243.87053953bp}%
    \ifx\svgscale\undefined%
      \relax%
    \else%
      \setlength{\unitlength}{\unitlength * \real{\svgscale}}%
    \fi%
  \else%
    \setlength{\unitlength}{\svgwidth}%
  \fi%
  \global\let\svgwidth\undefined%
  \global\let\svgscale\undefined%
  \makeatother%
  \begin{picture}(1,0.26046695)%
    \lineheight{1}%
    \setlength\tabcolsep{0pt}%
    \put(0,0){\includegraphics[width=\unitlength,page=1]{compcc2moncceval.pdf}}%
    \put(0.12092364,0.08248852){\color[rgb]{0,0,0}\makebox(0,0)[lt]{\lineheight{1.25}\smash{\begin{tabular}[t]{l}$f$\end{tabular}}}}%
    \put(0,0){\includegraphics[width=\unitlength,page=2]{compcc2moncceval.pdf}}%
    \put(0.12092364,0.22088169){\color[rgb]{0,0,0}\makebox(0,0)[lt]{\lineheight{1.25}\smash{\begin{tabular}[t]{l}$g$\end{tabular}}}}%
    \put(0,0){\includegraphics[width=\unitlength,page=3]{compcc2moncceval.pdf}}%
    \put(0.32170715,0.0049495){\makebox(0,0)[lt]{\lineheight{1.25}\smash{\begin{tabular}[t]{l}$\mathit{eval}$\end{tabular}}}}%
    \put(0.10632096,0.0049495){\makebox(0,0)[lt]{\lineheight{1.25}\smash{\begin{tabular}[t]{l}$\Lambda(f)$\end{tabular}}}}%
    \put(0,0){\includegraphics[width=\unitlength,page=4]{compcc2moncceval.pdf}}%
    \put(0.55842428,0.08248852){\color[rgb]{0,0,0}\makebox(0,0)[lt]{\lineheight{1.25}\smash{\begin{tabular}[t]{l}$f$\end{tabular}}}}%
    \put(0,0){\includegraphics[width=\unitlength,page=5]{compcc2moncceval.pdf}}%
    \put(0.55842428,0.22088169){\color[rgb]{0,0,0}\makebox(0,0)[lt]{\lineheight{1.25}\smash{\begin{tabular}[t]{l}$g$\end{tabular}}}}%
    \put(0,0){\includegraphics[width=\unitlength,page=6]{compcc2moncceval.pdf}}%
    \put(0.42805767,0.12545641){\makebox(0,0)[lt]{\lineheight{1.25}\smash{\begin{tabular}[t]{l}$=$\end{tabular}}}}%
    \put(0.91259017,0.12861961){\color[rgb]{0,0,0}\makebox(0,0)[lt]{\lineheight{1.25}\smash{\begin{tabular}[t]{l}$f$\end{tabular}}}}%
    \put(0,0){\includegraphics[width=\unitlength,page=7]{compcc2moncceval.pdf}}%
    \put(0.78759001,0.14399647){\color[rgb]{0,0,0}\makebox(0,0)[lt]{\lineheight{1.25}\smash{\begin{tabular}[t]{l}$g$\end{tabular}}}}%
    \put(0,0){\includegraphics[width=\unitlength,page=8]{compcc2moncceval.pdf}}%
    \put(0.67569738,0.12515138){\makebox(0,0)[lt]{\lineheight{1.25}\smash{\begin{tabular}[t]{l}$=$\end{tabular}}}}%
    \put(0,0){\includegraphics[width=\unitlength,page=9]{compcc2moncceval.pdf}}%
  \end{picture}%
\endgroup%

%% file: pics/ast.pdf_tex
%% Creator: Inkscape 1.1 (c68e22c387, 2021-05-23), www.inkscape.org
%% Accompanies image file 'ast.pdf' (pdf, eps, ps)
\begingroup%
  \makeatletter%
  \providecommand\color[2][]{%
    \errmessage{(Inkscape) Color is used for the text in Inkscape, but the package 'color.sty' is not loaded}%
    \renewcommand\color[2][]{}%
  }%
  \providecommand\transparent[1]{%
    \errmessage{(Inkscape) Transparency is used (non-zero) for the text in Inkscape, but the package 'transparent.sty' is not loaded}%
    \renewcommand\transparent[1]{}%
  }%
  \providecommand\rotatebox[2]{#2}%
  \newcommand*\fsize{\dimexpr\f@size pt\relax}%
  \newcommand*\lineheight[1]{\fontsize{\fsize}{#1\fsize}\selectfont}%
  \ifx\svgwidth\undefined%
    \setlength{\unitlength}{51.09040483bp}%
    \ifx\svgscale\undefined%
      \relax%
    \else%
      \setlength{\unitlength}{\unitlength * \real{\svgscale}}%
    \fi%
  \else%
    \setlength{\unitlength}{\svgwidth}%
  \fi%
  \global\let\svgwidth\undefined%
  \global\let\svgscale\undefined%
  \makeatother%
  \begin{picture}(1,0.96528721)%
    \lineheight{1}%
    \setlength\tabcolsep{0pt}%
    \put(0,0){\includegraphics[width=\unitlength,page=1]{ast.pdf}}%
    \put(0.10673885,0.42304818){\makebox(0,0)[lt]{\lineheight{1.25}\smash{\begin{tabular}[t]{l}$1$\end{tabular}}}}%
    \put(0,0){\includegraphics[width=\unitlength,page=2]{ast.pdf}}%
    \put(0.32693646,0.7900442){\makebox(0,0)[lt]{\lineheight{1.25}\smash{\begin{tabular}[t]{l}$+$\end{tabular}}}}%
    \put(0,0){\includegraphics[width=\unitlength,page=3]{ast.pdf}}%
    \put(0.32693734,0.05605215){\makebox(0,0)[lt]{\lineheight{1.25}\smash{\begin{tabular}[t]{l}$2$\end{tabular}}}}%
    \put(0,0){\includegraphics[width=\unitlength,page=4]{ast.pdf}}%
    \put(0.76733301,0.05605215){\makebox(0,0)[lt]{\lineheight{1.25}\smash{\begin{tabular}[t]{l}$3$\end{tabular}}}}%
    \put(0,0){\includegraphics[width=\unitlength,page=5]{ast.pdf}}%
    \put(0.54713408,0.42304818){\makebox(0,0)[lt]{\lineheight{1.25}\smash{\begin{tabular}[t]{l}$+$\end{tabular}}}}%
    \put(0,0){\includegraphics[width=\unitlength,page=6]{ast.pdf}}%
  \end{picture}%
\endgroup%

%% file: pics/ast2.pdf_tex
%% Creator: Inkscape 1.1 (c68e22c387, 2021-05-23), www.inkscape.org
%% Accompanies image file 'ast2.pdf' (pdf, eps, ps)
\begingroup%
  \makeatletter%
  \providecommand\color[2][]{%
    \errmessage{(Inkscape) Color is used for the text in Inkscape, but the package 'color.sty' is not loaded}%
    \renewcommand\color[2][]{}%
  }%
  \providecommand\transparent[1]{%
    \errmessage{(Inkscape) Transparency is used (non-zero) for the text in Inkscape, but the package 'transparent.sty' is not loaded}%
    \renewcommand\transparent[1]{}%
  }%
  \providecommand\rotatebox[2]{#2}%
  \newcommand*\fsize{\dimexpr\f@size pt\relax}%
  \newcommand*\lineheight[1]{\fontsize{\fsize}{#1\fsize}\selectfont}%
  \ifx\svgwidth\undefined%
    \setlength{\unitlength}{50.80617684bp}%
    \ifx\svgscale\undefined%
      \relax%
    \else%
      \setlength{\unitlength}{\unitlength * \real{\svgscale}}%
    \fi%
  \else%
    \setlength{\unitlength}{\svgwidth}%
  \fi%
  \global\let\svgwidth\undefined%
  \global\let\svgscale\undefined%
  \makeatother%
  \begin{picture}(1,0.97068738)%
    \lineheight{1}%
    \setlength\tabcolsep{0pt}%
    \put(0,0){\includegraphics[width=\unitlength,page=1]{ast2.pdf}}%
    \put(0.10733599,0.42541486){\makebox(0,0)[lt]{\lineheight{1.25}\smash{\begin{tabular}[t]{l}$x$\end{tabular}}}}%
    \put(0,0){\includegraphics[width=\unitlength,page=2]{ast2.pdf}}%
    \put(0.32876546,0.79446399){\makebox(0,0)[lt]{\lineheight{1.25}\smash{\begin{tabular}[t]{l}$+$\end{tabular}}}}%
    \put(0,0){\includegraphics[width=\unitlength,page=3]{ast2.pdf}}%
    \put(0.32876635,0.05636573){\makebox(0,0)[lt]{\lineheight{1.25}\smash{\begin{tabular}[t]{l}$2$\end{tabular}}}}%
    \put(0,0){\includegraphics[width=\unitlength,page=4]{ast2.pdf}}%
    \put(0.77162575,0.05636573){\makebox(0,0)[lt]{\lineheight{1.25}\smash{\begin{tabular}[t]{l}$x$\end{tabular}}}}%
    \put(0,0){\includegraphics[width=\unitlength,page=5]{ast2.pdf}}%
    \put(0.55019494,0.42541486){\makebox(0,0)[lt]{\lineheight{1.25}\smash{\begin{tabular}[t]{l}$+$\end{tabular}}}}%
    \put(0,0){\includegraphics[width=\unitlength,page=6]{ast2.pdf}}%
  \end{picture}%
\endgroup%

%% file: pics/ast2sd.pdf_tex
%% Creator: Inkscape 1.1 (c68e22c387, 2021-05-23), www.inkscape.org
%% Accompanies image file 'ast2sd.pdf' (pdf, eps, ps)
\begingroup%
  \makeatletter%
  \providecommand\color[2][]{%
    \errmessage{(Inkscape) Color is used for the text in Inkscape, but the package 'color.sty' is not loaded}%
    \renewcommand\color[2][]{}%
  }%
  \providecommand\transparent[1]{%
    \errmessage{(Inkscape) Transparency is used (non-zero) for the text in Inkscape, but the package 'transparent.sty' is not loaded}%
    \renewcommand\transparent[1]{}%
  }%
  \providecommand\rotatebox[2]{#2}%
  \newcommand*\fsize{\dimexpr\f@size pt\relax}%
  \newcommand*\lineheight[1]{\fontsize{\fsize}{#1\fsize}\selectfont}%
  \ifx\svgwidth\undefined%
    \setlength{\unitlength}{108.74999622bp}%
    \ifx\svgscale\undefined%
      \relax%
    \else%
      \setlength{\unitlength}{\unitlength * \real{\svgscale}}%
    \fi%
  \else%
    \setlength{\unitlength}{\svgwidth}%
  \fi%
  \global\let\svgwidth\undefined%
  \global\let\svgscale\undefined%
  \makeatother%
  \begin{picture}(1,0.31754052)%
    \lineheight{1}%
    \setlength\tabcolsep{0pt}%
    \put(0,0){\includegraphics[width=\unitlength,page=1]{ast2sd.pdf}}%
    \put(0.25949332,0.15758427){\color[rgb]{0,0,0}\makebox(0,0)[lt]{\lineheight{1.25}\smash{\begin{tabular}[t]{l}$2$\end{tabular}}}}%
    \put(0,0){\includegraphics[width=\unitlength,page=2]{ast2sd.pdf}}%
    \put(0.53535547,0.05413638){\color[rgb]{0,0,0}\makebox(0,0)[lt]{\lineheight{1.25}\smash{\begin{tabular}[t]{l}$+$\end{tabular}}}}%
    \put(0,0){\includegraphics[width=\unitlength,page=3]{ast2sd.pdf}}%
    \put(0.81120687,0.15758466){\color[rgb]{0,0,0}\makebox(0,0)[lt]{\lineheight{1.25}\smash{\begin{tabular}[t]{l}$+$\end{tabular}}}}%
    \put(0,0){\includegraphics[width=\unitlength,page=4]{ast2sd.pdf}}%
  \end{picture}%
\endgroup%

%% file: pics/ast3.pdf_tex
%% Creator: Inkscape 1.1 (c68e22c387, 2021-05-23), www.inkscape.org
%% Accompanies image file 'ast3.pdf' (pdf, eps, ps)
\begingroup%
  \makeatletter%
  \providecommand\color[2][]{%
    \errmessage{(Inkscape) Color is used for the text in Inkscape, but the package 'color.sty' is not loaded}%
    \renewcommand\color[2][]{}%
  }%
  \providecommand\transparent[1]{%
    \errmessage{(Inkscape) Transparency is used (non-zero) for the text in Inkscape, but the package 'transparent.sty' is not loaded}%
    \renewcommand\transparent[1]{}%
  }%
  \providecommand\rotatebox[2]{#2}%
  \newcommand*\fsize{\dimexpr\f@size pt\relax}%
  \newcommand*\lineheight[1]{\fontsize{\fsize}{#1\fsize}\selectfont}%
  \ifx\svgwidth\undefined%
    \setlength{\unitlength}{84.55619846bp}%
    \ifx\svgscale\undefined%
      \relax%
    \else%
      \setlength{\unitlength}{\unitlength * \real{\svgscale}}%
    \fi%
  \else%
    \setlength{\unitlength}{\svgwidth}%
  \fi%
  \global\let\svgwidth\undefined%
  \global\let\svgscale\undefined%
  \makeatother%
  \begin{picture}(1,0.80499058)%
    \lineheight{1}%
    \setlength\tabcolsep{0pt}%
    \put(0,0){\includegraphics[width=\unitlength,page=1]{ast3.pdf}}%
    \put(0.46363665,0.25561346){\makebox(0,0)[lt]{\lineheight{1.25}\smash{\begin{tabular}[t]{l}$x$\end{tabular}}}}%
    \put(0,0){\includegraphics[width=\unitlength,page=2]{ast3.pdf}}%
    \put(0.59668409,0.47735919){\makebox(0,0)[lt]{\lineheight{1.25}\smash{\begin{tabular}[t]{l}$+$\end{tabular}}}}%
    \put(0,0){\includegraphics[width=\unitlength,page=3]{ast3.pdf}}%
    \put(0.28623936,0.6991054){\makebox(0,0)[lt]{\lineheight{1.25}\smash{\begin{tabular}[t]{l}$\mathit{let}$\end{tabular}}}}%
    \put(0,0){\includegraphics[width=\unitlength,page=4]{ast3.pdf}}%
    \put(0.59668462,0.03386774){\makebox(0,0)[lt]{\lineheight{1.25}\smash{\begin{tabular}[t]{l}$2$\end{tabular}}}}%
    \put(0,0){\includegraphics[width=\unitlength,page=5]{ast3.pdf}}%
    \put(0.86277975,0.03386774){\makebox(0,0)[lt]{\lineheight{1.25}\smash{\begin{tabular}[t]{l}$x$\end{tabular}}}}%
    \put(0,0){\includegraphics[width=\unitlength,page=6]{ast3.pdf}}%
    \put(0.72973152,0.25561346){\makebox(0,0)[lt]{\lineheight{1.25}\smash{\begin{tabular}[t]{l}$+$\end{tabular}}}}%
    \put(0,0){\includegraphics[width=\unitlength,page=7]{ast3.pdf}}%
    \put(0.33058869,0.47735919){\makebox(0,0)[lt]{\lineheight{1.25}\smash{\begin{tabular}[t]{l}$0$\end{tabular}}}}%
    \put(0,0){\includegraphics[width=\unitlength,page=8]{ast3.pdf}}%
    \put(0.06449302,0.47735919){\makebox(0,0)[lt]{\lineheight{1.25}\smash{\begin{tabular}[t]{l}$x$\end{tabular}}}}%
    \put(0,0){\includegraphics[width=\unitlength,page=9]{ast3.pdf}}%
  \end{picture}%
\endgroup%

%% file: pics/ast3sd.pdf_tex
%% Creator: Inkscape 1.1 (c68e22c387, 2021-05-23), www.inkscape.org
%% Accompanies image file 'ast3sd.pdf' (pdf, eps, ps)
\begingroup%
  \makeatletter%
  \providecommand\color[2][]{%
    \errmessage{(Inkscape) Color is used for the text in Inkscape, but the package 'color.sty' is not loaded}%
    \renewcommand\color[2][]{}%
  }%
  \providecommand\transparent[1]{%
    \errmessage{(Inkscape) Transparency is used (non-zero) for the text in Inkscape, but the package 'transparent.sty' is not loaded}%
    \renewcommand\transparent[1]{}%
  }%
  \providecommand\rotatebox[2]{#2}%
  \newcommand*\fsize{\dimexpr\f@size pt\relax}%
  \newcommand*\lineheight[1]{\fontsize{\fsize}{#1\fsize}\selectfont}%
  \ifx\svgwidth\undefined%
    \setlength{\unitlength}{131.62500655bp}%
    \ifx\svgscale\undefined%
      \relax%
    \else%
      \setlength{\unitlength}{\unitlength * \real{\svgscale}}%
    \fi%
  \else%
    \setlength{\unitlength}{\svgwidth}%
  \fi%
  \global\let\svgwidth\undefined%
  \global\let\svgscale\undefined%
  \makeatother%
  \begin{picture}(1,0.26235539)%
    \lineheight{1}%
    \setlength\tabcolsep{0pt}%
    \put(0,0){\includegraphics[width=\unitlength,page=1]{ast3sd.pdf}}%
    \put(0.38818542,0.13019782){\color[rgb]{0,0,0}\makebox(0,0)[lt]{\lineheight{1.25}\smash{\begin{tabular}[t]{l}$2$\end{tabular}}}}%
    \put(0,0){\includegraphics[width=\unitlength,page=2]{ast3sd.pdf}}%
    \put(0.04630517,0.13019765){\color[rgb]{0,0,0}\makebox(0,0)[lt]{\lineheight{1.25}\smash{\begin{tabular}[t]{l}$0$\end{tabular}}}}%
    \put(0,0){\includegraphics[width=\unitlength,page=3]{ast3sd.pdf}}%
    \put(0.61610569,0.04472806){\color[rgb]{0,0,0}\makebox(0,0)[lt]{\lineheight{1.25}\smash{\begin{tabular}[t]{l}$+$\end{tabular}}}}%
    \put(0,0){\includegraphics[width=\unitlength,page=4]{ast3sd.pdf}}%
    \put(0.84401709,0.13019814){\color[rgb]{0,0,0}\makebox(0,0)[lt]{\lineheight{1.25}\smash{\begin{tabular}[t]{l}$+$\end{tabular}}}}%
    \put(0,0){\includegraphics[width=\unitlength,page=5]{ast3sd.pdf}}%
  \end{picture}%
\endgroup%

%% file: pics/ast4.pdf_tex
%% Creator: Inkscape 1.1 (c68e22c387, 2021-05-23), www.inkscape.org
%% Accompanies image file 'ast4.pdf' (pdf, eps, ps)
\begingroup%
  \makeatletter%
  \providecommand\color[2][]{%
    \errmessage{(Inkscape) Color is used for the text in Inkscape, but the package 'color.sty' is not loaded}%
    \renewcommand\color[2][]{}%
  }%
  \providecommand\transparent[1]{%
    \errmessage{(Inkscape) Transparency is used (non-zero) for the text in Inkscape, but the package 'transparent.sty' is not loaded}%
    \renewcommand\transparent[1]{}%
  }%
  \providecommand\rotatebox[2]{#2}%
  \newcommand*\fsize{\dimexpr\f@size pt\relax}%
  \newcommand*\lineheight[1]{\fontsize{\fsize}{#1\fsize}\selectfont}%
  \ifx\svgwidth\undefined%
    \setlength{\unitlength}{88.59036915bp}%
    \ifx\svgscale\undefined%
      \relax%
    \else%
      \setlength{\unitlength}{\unitlength * \real{\svgscale}}%
    \fi%
  \else%
    \setlength{\unitlength}{\svgwidth}%
  \fi%
  \global\let\svgwidth\undefined%
  \global\let\svgscale\undefined%
  \makeatother%
  \begin{picture}(1,0.97998108)%
    \lineheight{1}%
    \setlength\tabcolsep{0pt}%
    \put(0,0){\includegraphics[width=\unitlength,page=1]{ast4.pdf}}%
    \put(0.73883125,0.45562114){\makebox(0,0)[lt]{\lineheight{1.25}\smash{\begin{tabular}[t]{l}$x$\end{tabular}}}}%
    \put(0,0){\includegraphics[width=\unitlength,page=2]{ast4.pdf}}%
    \put(0.61184244,0.66726915){\makebox(0,0)[lt]{\lineheight{1.25}\smash{\begin{tabular}[t]{l}$+$\end{tabular}}}}%
    \put(0,0){\includegraphics[width=\unitlength,page=3]{ast4.pdf}}%
    \put(0.27320478,0.87891763){\makebox(0,0)[lt]{\lineheight{1.25}\smash{\begin{tabular}[t]{l}$\mathit{let}$\end{tabular}}}}%
    \put(0,0){\includegraphics[width=\unitlength,page=4]{ast4.pdf}}%
    \put(0.6118427,0.03232512){\makebox(0,0)[lt]{\lineheight{1.25}\smash{\begin{tabular}[t]{l}$x$\end{tabular}}}}%
    \put(0,0){\includegraphics[width=\unitlength,page=5]{ast4.pdf}}%
    \put(0.86582006,0.03232512){\makebox(0,0)[lt]{\lineheight{1.25}\smash{\begin{tabular}[t]{l}$2$\end{tabular}}}}%
    \put(0,0){\includegraphics[width=\unitlength,page=6]{ast4.pdf}}%
    \put(0.73883074,0.24397313){\makebox(0,0)[lt]{\lineheight{1.25}\smash{\begin{tabular}[t]{l}$+$\end{tabular}}}}%
    \put(0,0){\includegraphics[width=\unitlength,page=7]{ast4.pdf}}%
    \put(0.31553455,0.66726915){\makebox(0,0)[lt]{\lineheight{1.25}\smash{\begin{tabular}[t]{l}$0$\end{tabular}}}}%
    \put(0,0){\includegraphics[width=\unitlength,page=8]{ast4.pdf}}%
    \put(0.06155618,0.66726915){\makebox(0,0)[lt]{\lineheight{1.25}\smash{\begin{tabular}[t]{l}$x$\end{tabular}}}}%
    \put(0,0){\includegraphics[width=\unitlength,page=9]{ast4.pdf}}%
    \put(0.44252361,0.45562067){\makebox(0,0)[lt]{\lineheight{1.25}\smash{\begin{tabular}[t]{l}$\mathit{let}$\end{tabular}}}}%
    \put(0,0){\includegraphics[width=\unitlength,page=10]{ast4.pdf}}%
    \put(0.48485313,0.24397219){\makebox(0,0)[lt]{\lineheight{1.25}\smash{\begin{tabular}[t]{l}$1$\end{tabular}}}}%
    \put(0,0){\includegraphics[width=\unitlength,page=11]{ast4.pdf}}%
    \put(0.23087476,0.24397219){\makebox(0,0)[lt]{\lineheight{1.25}\smash{\begin{tabular}[t]{l}$x$\end{tabular}}}}%
    \put(0,0){\includegraphics[width=\unitlength,page=12]{ast4.pdf}}%
  \end{picture}%
\endgroup%

%% file: pics/ast4sd.pdf_tex
%% Creator: Inkscape 1.1 (c68e22c387, 2021-05-23), www.inkscape.org
%% Accompanies image file 'ast4sd.pdf' (pdf, eps, ps)
\begingroup%
  \makeatletter%
  \providecommand\color[2][]{%
    \errmessage{(Inkscape) Color is used for the text in Inkscape, but the package 'color.sty' is not loaded}%
    \renewcommand\color[2][]{}%
  }%
  \providecommand\transparent[1]{%
    \errmessage{(Inkscape) Transparency is used (non-zero) for the text in Inkscape, but the package 'transparent.sty' is not loaded}%
    \renewcommand\transparent[1]{}%
  }%
  \providecommand\rotatebox[2]{#2}%
  \newcommand*\fsize{\dimexpr\f@size pt\relax}%
  \newcommand*\lineheight[1]{\fontsize{\fsize}{#1\fsize}\selectfont}%
  \ifx\svgwidth\undefined%
    \setlength{\unitlength}{124.12499934bp}%
    \ifx\svgscale\undefined%
      \relax%
    \else%
      \setlength{\unitlength}{\unitlength * \real{\svgscale}}%
    \fi%
  \else%
    \setlength{\unitlength}{\svgwidth}%
  \fi%
  \global\let\svgwidth\undefined%
  \global\let\svgscale\undefined%
  \makeatother%
  \begin{picture}(1,0.45921438)%
    \lineheight{1}%
    \setlength\tabcolsep{0pt}%
    \put(0,0){\includegraphics[width=\unitlength,page=1]{ast4sd.pdf}}%
    \put(0.35121773,0.2286992){\color[rgb]{0,0,0}\makebox(0,0)[lt]{\lineheight{1.25}\smash{\begin{tabular}[t]{l}$2$\end{tabular}}}}%
    \put(0,0){\includegraphics[width=\unitlength,page=2]{ast4sd.pdf}}%
    \put(0.04910307,0.37975631){\color[rgb]{0,0,0}\makebox(0,0)[lt]{\lineheight{1.25}\smash{\begin{tabular}[t]{l}$0$\end{tabular}}}}%
    \put(0,0){\includegraphics[width=\unitlength,page=3]{ast4sd.pdf}}%
    \put(0.20016036,0.04743014){\color[rgb]{0,0,0}\makebox(0,0)[lt]{\lineheight{1.25}\smash{\begin{tabular}[t]{l}$1$\end{tabular}}}}%
    \put(0,0){\includegraphics[width=\unitlength,page=4]{ast4sd.pdf}}%
    \put(0.59290963,0.13806511){\color[rgb]{0,0,0}\makebox(0,0)[lt]{\lineheight{1.25}\smash{\begin{tabular}[t]{l}$+$\end{tabular}}}}%
    \put(0,0){\includegraphics[width=\unitlength,page=5]{ast4sd.pdf}}%
    \put(0.83459213,0.22869955){\color[rgb]{0,0,0}\makebox(0,0)[lt]{\lineheight{1.25}\smash{\begin{tabular}[t]{l}$+$\end{tabular}}}}%
    \put(0,0){\includegraphics[width=\unitlength,page=6]{ast4sd.pdf}}%
  \end{picture}%
\endgroup%

%% file: pics/astlam.pdf_tex
%% Creator: Inkscape 1.1 (c68e22c387, 2021-05-23), www.inkscape.org
%% Accompanies image file 'astlam.pdf' (pdf, eps, ps)
\begingroup%
  \makeatletter%
  \providecommand\color[2][]{%
    \errmessage{(Inkscape) Color is used for the text in Inkscape, but the package 'color.sty' is not loaded}%
    \renewcommand\color[2][]{}%
  }%
  \providecommand\transparent[1]{%
    \errmessage{(Inkscape) Transparency is used (non-zero) for the text in Inkscape, but the package 'transparent.sty' is not loaded}%
    \renewcommand\transparent[1]{}%
  }%
  \providecommand\rotatebox[2]{#2}%
  \newcommand*\fsize{\dimexpr\f@size pt\relax}%
  \newcommand*\lineheight[1]{\fontsize{\fsize}{#1\fsize}\selectfont}%
  \ifx\svgwidth\undefined%
    \setlength{\unitlength}{145.05969665bp}%
    \ifx\svgscale\undefined%
      \relax%
    \else%
      \setlength{\unitlength}{\unitlength * \real{\svgscale}}%
    \fi%
  \else%
    \setlength{\unitlength}{\svgwidth}%
  \fi%
  \global\let\svgwidth\undefined%
  \global\let\svgscale\undefined%
  \makeatother%
  \begin{picture}(1,0.3923112)%
    \lineheight{1}%
    \setlength\tabcolsep{0pt}%
    \put(0,0){\includegraphics[width=\unitlength,page=1]{astlam.pdf}}%
    \put(0.11514807,0.33059007){\makebox(0,0)[lt]{\lineheight{1.25}\smash{\begin{tabular}[t]{l}$\lambda$\end{tabular}}}}%
    \put(0,0){\includegraphics[width=\unitlength,page=2]{astlam.pdf}}%
    \put(0.11514869,0.07207615){\makebox(0,0)[lt]{\lineheight{1.25}\smash{\begin{tabular}[t]{l}$x$\end{tabular}}}}%
    \put(0,0){\includegraphics[width=\unitlength,page=3]{astlam.pdf}}%
    \put(0.27025689,0.07207615){\makebox(0,0)[lt]{\lineheight{1.25}\smash{\begin{tabular}[t]{l}$x$\end{tabular}}}}%
    \put(0,0){\includegraphics[width=\unitlength,page=4]{astlam.pdf}}%
    \put(0.1927024,0.20133311){\makebox(0,0)[lt]{\lineheight{1.25}\smash{\begin{tabular}[t]{l}$+$\end{tabular}}}}%
    \put(0,0){\includegraphics[width=\unitlength,page=5]{astlam.pdf}}%
    \put(0.03759297,0.20133282){\makebox(0,0)[lt]{\lineheight{1.25}\smash{\begin{tabular}[t]{l}$x$\end{tabular}}}}%
    \put(0,0){\includegraphics[width=\unitlength,page=6]{astlam.pdf}}%
    \put(0.76143375,0.33059007){\makebox(0,0)[lt]{\lineheight{1.25}\smash{\begin{tabular}[t]{l}$\lambda$\end{tabular}}}}%
    \put(0,0){\includegraphics[width=\unitlength,page=7]{astlam.pdf}}%
    \put(0.83898793,0.20133311){\makebox(0,0)[lt]{\lineheight{1.25}\smash{\begin{tabular}[t]{l}$+$\end{tabular}}}}%
    \put(0,0){\includegraphics[width=\unitlength,page=8]{astlam.pdf}}%
  \end{picture}%
\endgroup%

%% file: pics/astlam2.pdf_tex
%% Creator: Inkscape 1.1 (c68e22c387, 2021-05-23), www.inkscape.org
%% Accompanies image file 'astlam2.pdf' (pdf, eps, ps)
\begingroup%
  \makeatletter%
  \providecommand\color[2][]{%
    \errmessage{(Inkscape) Color is used for the text in Inkscape, but the package 'color.sty' is not loaded}%
    \renewcommand\color[2][]{}%
  }%
  \providecommand\transparent[1]{%
    \errmessage{(Inkscape) Transparency is used (non-zero) for the text in Inkscape, but the package 'transparent.sty' is not loaded}%
    \renewcommand\transparent[1]{}%
  }%
  \providecommand\rotatebox[2]{#2}%
  \newcommand*\fsize{\dimexpr\f@size pt\relax}%
  \newcommand*\lineheight[1]{\fontsize{\fsize}{#1\fsize}\selectfont}%
  \ifx\svgwidth\undefined%
    \setlength{\unitlength}{148.80973568bp}%
    \ifx\svgscale\undefined%
      \relax%
    \else%
      \setlength{\unitlength}{\unitlength * \real{\svgscale}}%
    \fi%
  \else%
    \setlength{\unitlength}{\svgwidth}%
  \fi%
  \global\let\svgwidth\undefined%
  \global\let\svgscale\undefined%
  \makeatother%
  \begin{picture}(1,0.63442463)%
    \lineheight{1}%
    \setlength\tabcolsep{0pt}%
    \put(0,0){\includegraphics[width=\unitlength,page=1]{astlam2.pdf}}%
    \put(0.08704648,0.32225915){\makebox(0,0)[lt]{\lineheight{1.25}\smash{\begin{tabular}[t]{l}$\lambda$\end{tabular}}}}%
    \put(0,0){\includegraphics[width=\unitlength,page=2]{astlam2.pdf}}%
    \put(0.03024237,0.58433848){\makebox(0,0)[lt]{\lineheight{1.25}\smash{\begin{tabular}[t]{l}$\mathit{apply}$\end{tabular}}}}%
    \put(0,0){\includegraphics[width=\unitlength,page=3]{astlam2.pdf}}%
    \put(0.60984174,0.58433876){\makebox(0,0)[lt]{\lineheight{1.25}\smash{\begin{tabular}[t]{l}$\mathit{apply}$\end{tabular}}}}%
    \put(0,0){\includegraphics[width=\unitlength,page=4]{astlam2.pdf}}%
    \put(0.16264627,0.1962595){\makebox(0,0)[lt]{\lineheight{1.25}\smash{\begin{tabular}[t]{l}$+$\end{tabular}}}}%
    \put(0,0){\includegraphics[width=\unitlength,page=5]{astlam2.pdf}}%
    \put(0.56584602,0.42305933){\makebox(0,0)[lt]{\lineheight{1.25}\smash{\begin{tabular}[t]{l}$\lambda$\end{tabular}}}}%
    \put(0,0){\includegraphics[width=\unitlength,page=6]{astlam2.pdf}}%
    \put(0.64144582,0.1962595){\makebox(0,0)[lt]{\lineheight{1.25}\smash{\begin{tabular}[t]{l}$+$\end{tabular}}}}%
    \put(0,0){\includegraphics[width=\unitlength,page=7]{astlam2.pdf}}%
    \put(0.76744567,0.42305905){\makebox(0,0)[lt]{\lineheight{1.25}\smash{\begin{tabular}[t]{l}$\lambda$\end{tabular}}}}%
    \put(0,0){\includegraphics[width=\unitlength,page=8]{astlam2.pdf}}%
  \end{picture}%
\endgroup%

%% file: pics/astlam3b.pdf_tex
%% Creator: Inkscape 1.0.1 (c497b03c, 2020-09-10), www.inkscape.org
%% Accompanies image file 'astlam3b.pdf' (pdf, eps, ps)
\begingroup%
  \makeatletter%
  \providecommand\color[2][]{%
    \errmessage{(Inkscape) Color is used for the text in Inkscape, but the package 'color.sty' is not loaded}%
    \renewcommand\color[2][]{}%
  }%
  \providecommand\transparent[1]{%
    \errmessage{(Inkscape) Transparency is used (non-zero) for the text in Inkscape, but the package 'transparent.sty' is not loaded}%
    \renewcommand\transparent[1]{}%
  }%
  \providecommand\rotatebox[2]{#2}%
  \newcommand*\fsize{\dimexpr\f@size pt\relax}%
  \newcommand*\lineheight[1]{\fontsize{\fsize}{#1\fsize}\selectfont}%
  \ifx\svgwidth\undefined%
    \setlength{\unitlength}{135.06854557bp}%
    \ifx\svgscale\undefined%
      \relax%
    \else%
      \setlength{\unitlength}{\unitlength * \real{\svgscale}}%
    \fi%
  \else%
    \setlength{\unitlength}{\svgwidth}%
  \fi%
  \global\let\svgwidth\undefined%
  \global\let\svgscale\undefined%
  \makeatother%
  \begin{picture}(1,0.44421896)%
    \lineheight{1}%
    \setlength\tabcolsep{0pt}%
    \put(0,0){\includegraphics[width=\unitlength,page=1]{astlam3b.pdf}}%
    \put(0.067261,0.33017737){\makebox(0,0)[lt]{\lineheight{1.25}\smash{\begin{tabular}[t]{l}$C$\end{tabular}}}}%
    \put(0.56497074,0.08111891){\makebox(0,0)[lt]{\lineheight{1.25}\smash{\begin{tabular}[t]{l}$\delta$\end{tabular}}}}%
    \put(0.06739678,0.21939409){\makebox(0,0)[lt]{\lineheight{1.25}\smash{\begin{tabular}[t]{l}$\lambda$\end{tabular}}}}%
    \put(0.54012603,0.33180639){\makebox(0,0)[lt]{\lineheight{1.25}\smash{\begin{tabular}[t]{l}$\lambda$\end{tabular}}}}%
    \put(0.76209959,0.33180639){\makebox(0,0)[lt]{\lineheight{1.25}\smash{\begin{tabular}[t]{l}$\lambda$\end{tabular}}}}%
    \put(0.73365723,0.07976113){\makebox(0,0)[lt]{\lineheight{1.25}\smash{\begin{tabular}[t]{l}$\delta$\end{tabular}}}}%
  \end{picture}%
\endgroup%

%% file: pics/astlamsd.pdf_tex
%% Creator: Inkscape 1.2.2 (732a01da63, 2022-12-09), www.inkscape.org
%% Accompanies image file 'astlamsd.pdf' (pdf, eps, ps)
\begingroup%
  \makeatletter%
  \providecommand\color[2][]{%
    \errmessage{(Inkscape) Color is used for the text in Inkscape, but the package 'color.sty' is not loaded}%
    \renewcommand\color[2][]{}%
  }%
  \providecommand\transparent[1]{%
    \errmessage{(Inkscape) Transparency is used (non-zero) for the text in Inkscape, but the package 'transparent.sty' is not loaded}%
    \renewcommand\transparent[1]{}%
  }%
  \providecommand\rotatebox[2]{#2}%
  \newcommand*\fsize{\dimexpr\f@size pt\relax}%
  \newcommand*\lineheight[1]{\fontsize{\fsize}{#1\fsize}\selectfont}%
  \ifx\svgwidth\undefined%
    \setlength{\unitlength}{242.20487445bp}%
    \ifx\svgscale\undefined%
      \relax%
    \else%
      \setlength{\unitlength}{\unitlength * \real{\svgscale}}%
    \fi%
  \else%
    \setlength{\unitlength}{\svgwidth}%
  \fi%
  \global\let\svgwidth\undefined%
  \global\let\svgscale\undefined%
  \makeatother%
  \begin{picture}(1,0.20244654)%
    \lineheight{1}%
    \setlength\tabcolsep{0pt}%
    \put(0,0){\includegraphics[width=\unitlength,page=1]{astlamsd.pdf}}%
    \put(0.24857754,0.09761329){\makebox(0,0)[lt]{\lineheight{1.25}\smash{\begin{tabular}[t]{l}${\mathit{apply}}$\end{tabular}}}}%
    \put(0,0){\includegraphics[width=\unitlength,page=2]{astlamsd.pdf}}%
    \put(0.10999667,0.09554603){\makebox(0,0)[lt]{\lineheight{1.25}\smash{\begin{tabular}[t]{l}$+$\end{tabular}}}}%
    \put(0,0){\includegraphics[width=\unitlength,page=3]{astlamsd.pdf}}%
    \put(0.44140676,0.09373914){\makebox(0,0)[lt]{\lineheight{1.25}\smash{\begin{tabular}[t]{l}$=$\end{tabular}}}}%
    \put(0,0){\includegraphics[width=\unitlength,page=4]{astlamsd.pdf}}%
    \put(0.74789673,0.09761382){\makebox(0,0)[lt]{\lineheight{1.25}\smash{\begin{tabular}[t]{l}${\mathit{apply}}$\end{tabular}}}}%
    \put(0,0){\includegraphics[width=\unitlength,page=5]{astlamsd.pdf}}%
    \put(0.62092775,0.14199449){\makebox(0,0)[lt]{\lineheight{1.25}\smash{\begin{tabular}[t]{l}$+$\end{tabular}}}}%
    \put(0,0){\includegraphics[width=\unitlength,page=6]{astlamsd.pdf}}%
    \put(0.62092793,0.03361493){\makebox(0,0)[lt]{\lineheight{1.25}\smash{\begin{tabular}[t]{l}$+$\end{tabular}}}}%
    \put(0,0){\includegraphics[width=\unitlength,page=7]{astlamsd.pdf}}%
  \end{picture}%
\endgroup%

%% file: pics/compact-closed.pdf_tex
%% Creator: Inkscape 1.0.1 (c497b03c, 2020-09-10), www.inkscape.org
%% Accompanies image file 'compact-closed.pdf' (pdf, eps, ps)
\begingroup%
  \makeatletter%
  \providecommand\color[2][]{%
    \errmessage{(Inkscape) Color is used for the text in Inkscape, but the package 'color.sty' is not loaded}%
    \renewcommand\color[2][]{}%
  }%
  \providecommand\transparent[1]{%
    \errmessage{(Inkscape) Transparency is used (non-zero) for the text in Inkscape, but the package 'transparent.sty' is not loaded}%
    \renewcommand\transparent[1]{}%
  }%
  \providecommand\rotatebox[2]{#2}%
  \newcommand*\fsize{\dimexpr\f@size pt\relax}%
  \newcommand*\lineheight[1]{\fontsize{\fsize}{#1\fsize}\selectfont}%
  \ifx\svgwidth\undefined%
    \setlength{\unitlength}{245.34452714bp}%
    \ifx\svgscale\undefined%
      \relax%
    \else%
      \setlength{\unitlength}{\unitlength * \real{\svgscale}}%
    \fi%
  \else%
    \setlength{\unitlength}{\svgwidth}%
  \fi%
  \global\let\svgwidth\undefined%
  \global\let\svgscale\undefined%
  \makeatother%
  \begin{picture}(1,0.21405171)%
    \lineheight{1}%
    \setlength\tabcolsep{0pt}%
    \put(0,0){\includegraphics[width=\unitlength,page=1]{compact-closed.pdf}}%
    \put(0.76595832,0.08451683){\color[rgb]{0,0,0}\makebox(0,0)[lt]{\lineheight{1.25}\smash{\begin{tabular}[t]{l}$=$\end{tabular}}}}%
    \put(0.60004051,0.19523248){\makebox(0,0)[lt]{\lineheight{1.25}\smash{\begin{tabular}[t]{l}$\id_{A^*}$\end{tabular}}}}%
    \put(0,0){\includegraphics[width=\unitlength,page=2]{compact-closed.pdf}}%
    \put(0.67646359,0.01181694){\makebox(0,0)[lt]{\lineheight{1.25}\smash{\begin{tabular}[t]{l}$\id_{A^*}$\end{tabular}}}}%
    \put(0.62055572,0.04738237){\makebox(0,0)[lt]{\lineheight{1.25}\smash{\begin{tabular}[t]{l}$\eta$\end{tabular}}}}%
    \put(0.70614937,0.12380545){\makebox(0,0)[lt]{\lineheight{1.25}\smash{\begin{tabular}[t]{l}$\epsilon$\end{tabular}}}}%
    \put(0.83855587,0.10351889){\makebox(0,0)[lt]{\lineheight{1.25}\smash{\begin{tabular}[t]{l}$\id_{A^*}$\end{tabular}}}}%
    \put(0,0){\includegraphics[width=\unitlength,page=3]{compact-closed.pdf}}%
    \put(0.16985798,0.08451683){\color[rgb]{0,0,0}\makebox(0,0)[lt]{\lineheight{1.25}\smash{\begin{tabular}[t]{l}$=$\end{tabular}}}}%
    \put(0.07730645,0.1952324){\makebox(0,0)[lt]{\lineheight{1.25}\smash{\begin{tabular}[t]{l}$\id_{A}$\end{tabular}}}}%
    \put(0,0){\includegraphics[width=\unitlength,page=4]{compact-closed.pdf}}%
    \put(0.01311082,0.01181694){\makebox(0,0)[lt]{\lineheight{1.25}\smash{\begin{tabular}[t]{l}$\id_{A}$\end{tabular}}}}%
    \put(0.02445538,0.13297531){\makebox(0,0)[lt]{\lineheight{1.25}\smash{\begin{tabular}[t]{l}$\eta$\end{tabular}}}}%
    \put(0.11004903,0.03821251){\makebox(0,0)[lt]{\lineheight{1.25}\smash{\begin{tabular}[t]{l}$\epsilon$\end{tabular}}}}%
    \put(0.24245561,0.10351889){\makebox(0,0)[lt]{\lineheight{1.25}\smash{\begin{tabular}[t]{l}$\id_{A}$\end{tabular}}}}%
    \put(0,0){\includegraphics[width=\unitlength,page=5]{compact-closed.pdf}}%
  \end{picture}%
\endgroup%

%% file: pics/rosetta-currying.pdf_tex
%% Creator: Inkscape 1.2.2 (732a01da63, 2022-12-09), www.inkscape.org
%% Accompanies image file 'rosetta-currying.pdf' (pdf, eps, ps)
\begingroup%
  \makeatletter%
  \providecommand\color[2][]{%
    \errmessage{(Inkscape) Color is used for the text in Inkscape, but the package 'color.sty' is not loaded}%
    \renewcommand\color[2][]{}%
  }%
  \providecommand\transparent[1]{%
    \errmessage{(Inkscape) Transparency is used (non-zero) for the text in Inkscape, but the package 'transparent.sty' is not loaded}%
    \renewcommand\transparent[1]{}%
  }%
  \providecommand\rotatebox[2]{#2}%
  \newcommand*\fsize{\dimexpr\f@size pt\relax}%
  \newcommand*\lineheight[1]{\fontsize{\fsize}{#1\fsize}\selectfont}%
  \ifx\svgwidth\undefined%
    \setlength{\unitlength}{233.07016991bp}%
    \ifx\svgscale\undefined%
      \relax%
    \else%
      \setlength{\unitlength}{\unitlength * \real{\svgscale}}%
    \fi%
  \else%
    \setlength{\unitlength}{\svgwidth}%
  \fi%
  \global\let\svgwidth\undefined%
  \global\let\svgscale\undefined%
  \makeatother%
  \begin{picture}(1,0.1629328)%
    \lineheight{1}%
    \setlength\tabcolsep{0pt}%
    \put(0,0){\includegraphics[width=\unitlength,page=1]{rosetta-currying.pdf}}%
    \put(0.5027611,0.05445336){\color[rgb]{0,0,0}\makebox(0,0)[lt]{\lineheight{1.25}\smash{\begin{tabular}[t]{l}$X$\end{tabular}}}}%
    \put(0.66805563,0.0957746){\color[rgb]{0,0,0}\makebox(0,0)[lt]{\lineheight{1.25}\smash{\begin{tabular}[t]{l}$Y$\end{tabular}}}}%
    \put(0,0){\includegraphics[width=\unitlength,page=2]{rosetta-currying.pdf}}%
    \put(0.50532972,0.10966985){\color[rgb]{0,0,0}\makebox(0,0)[lt]{\lineheight{1.25}\smash{\begin{tabular}[t]{l}$A$\end{tabular}}}}%
    \put(0.72359373,0.09567281){\color[rgb]{0,0,0}\makebox(0,0)[lt]{\lineheight{1.25}\smash{\begin{tabular}[t]{l}$X\multimap Y$\end{tabular}}}}%
    \put(0,0){\includegraphics[width=\unitlength,page=3]{rosetta-currying.pdf}}%
    \put(0.60376621,0.07315189){\color[rgb]{0,0,0}\makebox(0,0)[lt]{\lineheight{1.25}\smash{\begin{tabular}[t]{l}$f$\end{tabular}}}}%
    \put(0,0){\includegraphics[width=\unitlength,page=4]{rosetta-currying.pdf}}%
    \put(0.40101916,0.09177052){\makebox(0,0)[lt]{\lineheight{1.25}\smash{\begin{tabular}[t]{l}vs.\end{tabular}}}}%
    \put(0.31275649,0.00605037){\color[rgb]{0,0,0}\makebox(0,0)[lt]{\lineheight{1.25}\smash{\begin{tabular}[t]{l}$X$\end{tabular}}}}%
    \put(0.31300584,0.11146422){\color[rgb]{0,0,0}\makebox(0,0)[lt]{\lineheight{1.25}\smash{\begin{tabular}[t]{l}$Y$\end{tabular}}}}%
    \put(0,0){\includegraphics[width=\unitlength,page=5]{rosetta-currying.pdf}}%
    \put(0.03873211,0.12535946){\color[rgb]{0,0,0}\makebox(0,0)[lt]{\lineheight{1.25}\smash{\begin{tabular}[t]{l}$A$\end{tabular}}}}%
    \put(0,0){\includegraphics[width=\unitlength,page=6]{rosetta-currying.pdf}}%
    \put(0.1371686,0.08884151){\color[rgb]{0,0,0}\makebox(0,0)[lt]{\lineheight{1.25}\smash{\begin{tabular}[t]{l}$f$\end{tabular}}}}%
    \put(0,0){\includegraphics[width=\unitlength,page=7]{rosetta-currying.pdf}}%
  \end{picture}%
\endgroup%

%% file: pics/rosetta-beta.pdf_tex
%% Creator: Inkscape 1.2.2 (732a01da63, 2022-12-09), www.inkscape.org
%% Accompanies image file 'rosetta-beta.pdf' (pdf, eps, ps)
\begingroup%
  \makeatletter%
  \providecommand\color[2][]{%
    \errmessage{(Inkscape) Color is used for the text in Inkscape, but the package 'color.sty' is not loaded}%
    \renewcommand\color[2][]{}%
  }%
  \providecommand\transparent[1]{%
    \errmessage{(Inkscape) Transparency is used (non-zero) for the text in Inkscape, but the package 'transparent.sty' is not loaded}%
    \renewcommand\transparent[1]{}%
  }%
  \providecommand\rotatebox[2]{#2}%
  \newcommand*\fsize{\dimexpr\f@size pt\relax}%
  \newcommand*\lineheight[1]{\fontsize{\fsize}{#1\fsize}\selectfont}%
  \ifx\svgwidth\undefined%
    \setlength{\unitlength}{211.13385318bp}%
    \ifx\svgscale\undefined%
      \relax%
    \else%
      \setlength{\unitlength}{\unitlength * \real{\svgscale}}%
    \fi%
  \else%
    \setlength{\unitlength}{\svgwidth}%
  \fi%
  \global\let\svgwidth\undefined%
  \global\let\svgscale\undefined%
  \makeatother%
  \begin{picture}(1,0.25159362)%
    \lineheight{1}%
    \setlength\tabcolsep{0pt}%
    \put(0,0){\includegraphics[width=\unitlength,page=1]{rosetta-beta.pdf}}%
    \put(0.29122157,0.08218297){\color[rgb]{0,0,0}\makebox(0,0)[lt]{\lineheight{1.25}\smash{\begin{tabular}[t]{l}$X$\end{tabular}}}}%
    \put(0.29149682,0.21111408){\color[rgb]{0,0,0}\makebox(0,0)[lt]{\lineheight{1.25}\smash{\begin{tabular}[t]{l}$Y$\end{tabular}}}}%
    \put(0,0){\includegraphics[width=\unitlength,page=2]{rosetta-beta.pdf}}%
    \put(0.04275635,0.21011845){\color[rgb]{0,0,0}\makebox(0,0)[lt]{\lineheight{1.25}\smash{\begin{tabular}[t]{l}$A$\end{tabular}}}}%
    \put(0,0){\includegraphics[width=\unitlength,page=3]{rosetta-beta.pdf}}%
    \put(0.15142016,0.16980636){\color[rgb]{0,0,0}\makebox(0,0)[lt]{\lineheight{1.25}\smash{\begin{tabular}[t]{l}$f$\end{tabular}}}}%
    \put(0,0){\includegraphics[width=\unitlength,page=4]{rosetta-beta.pdf}}%
    \put(0.40576026,0.00543939){\makebox(0,0)[lt]{\lineheight{1.25}\smash{\begin{tabular}[t]{l}$\mathit{eval}$\end{tabular}}}}%
    \put(0,0){\includegraphics[width=\unitlength,page=5]{rosetta-beta.pdf}}%
    \put(0.53283734,0.10861752){\makebox(0,0)[lt]{\lineheight{1.25}\smash{\begin{tabular}[t]{l}$=$\end{tabular}}}}%
    \put(0,0){\includegraphics[width=\unitlength,page=6]{rosetta-beta.pdf}}%
    \put(0.58834674,0.0725725){\color[rgb]{0,0,0}\makebox(0,0)[lt]{\lineheight{1.25}\smash{\begin{tabular}[t]{l}$X$\end{tabular}}}}%
    \put(0.83740972,0.1952211){\color[rgb]{0,0,0}\makebox(0,0)[lt]{\lineheight{1.25}\smash{\begin{tabular}[t]{l}$Y$\end{tabular}}}}%
    \put(0,0){\includegraphics[width=\unitlength,page=7]{rosetta-beta.pdf}}%
    \put(0.58866924,0.21056003){\color[rgb]{0,0,0}\makebox(0,0)[lt]{\lineheight{1.25}\smash{\begin{tabular}[t]{l}$A$\end{tabular}}}}%
    \put(0,0){\includegraphics[width=\unitlength,page=8]{rosetta-beta.pdf}}%
    \put(0.73754118,0.17024794){\color[rgb]{0,0,0}\makebox(0,0)[lt]{\lineheight{1.25}\smash{\begin{tabular}[t]{l}$f$\end{tabular}}}}%
    \put(0,0){\includegraphics[width=\unitlength,page=9]{rosetta-beta.pdf}}%
  \end{picture}%
\endgroup%

%% file: pics/rosetta-beta2.pdf_tex
%% Creator: Inkscape 1.2.2 (732a01da63, 2022-12-09), www.inkscape.org
%% Accompanies image file 'rosetta-beta2.pdf' (pdf, eps, ps)
\begingroup%
  \makeatletter%
  \providecommand\color[2][]{%
    \errmessage{(Inkscape) Color is used for the text in Inkscape, but the package 'color.sty' is not loaded}%
    \renewcommand\color[2][]{}%
  }%
  \providecommand\transparent[1]{%
    \errmessage{(Inkscape) Transparency is used (non-zero) for the text in Inkscape, but the package 'transparent.sty' is not loaded}%
    \renewcommand\transparent[1]{}%
  }%
  \providecommand\rotatebox[2]{#2}%
  \newcommand*\fsize{\dimexpr\f@size pt\relax}%
  \newcommand*\lineheight[1]{\fontsize{\fsize}{#1\fsize}\selectfont}%
  \ifx\svgwidth\undefined%
    \setlength{\unitlength}{236.24999433bp}%
    \ifx\svgscale\undefined%
      \relax%
    \else%
      \setlength{\unitlength}{\unitlength * \real{\svgscale}}%
    \fi%
  \else%
    \setlength{\unitlength}{\svgwidth}%
  \fi%
  \global\let\svgwidth\undefined%
  \global\let\svgscale\undefined%
  \makeatother%
  \begin{picture}(1,0.21052746)%
    \lineheight{1}%
    \setlength\tabcolsep{0pt}%
    \put(0,0){\includegraphics[width=\unitlength,page=1]{rosetta-beta2.pdf}}%
    \put(0.22253824,0.07597094){\color[rgb]{0,0,0}\makebox(0,0)[lt]{\lineheight{1.25}\smash{\begin{tabular}[t]{l}$X$\end{tabular}}}}%
    \put(0.20032578,0.1721055){\color[rgb]{0,0,0}\makebox(0,0)[lt]{\lineheight{1.25}\smash{\begin{tabular}[t]{l}$Y$\end{tabular}}}}%
    \put(0,0){\includegraphics[width=\unitlength,page=2]{rosetta-beta2.pdf}}%
    \put(0.0022773,0.17497916){\color[rgb]{0,0,0}\makebox(0,0)[lt]{\lineheight{1.25}\smash{\begin{tabular}[t]{l}$A$\end{tabular}}}}%
    \put(0,0){\includegraphics[width=\unitlength,page=3]{rosetta-beta2.pdf}}%
    \put(0.13532239,0.13782978){\color[rgb]{0,0,0}\makebox(0,0)[lt]{\lineheight{1.25}\smash{\begin{tabular}[t]{l}$f$\end{tabular}}}}%
    \put(0.43563909,0.00324074){\makebox(0,0)[lt]{\lineheight{1.25}\smash{\begin{tabular}[t]{l}$\mathit{eval}$\end{tabular}}}}%
    \put(0,0){\includegraphics[width=\unitlength,page=4]{rosetta-beta2.pdf}}%
    \put(0.59477381,0.13997844){\makebox(0,0)[lt]{\lineheight{1.25}\smash{\begin{tabular}[t]{l}$=$\end{tabular}}}}%
    \put(0,0){\includegraphics[width=\unitlength,page=5]{rosetta-beta2.pdf}}%
    \put(0.633541,0.01572779){\color[rgb]{0,0,0}\makebox(0,0)[lt]{\lineheight{1.25}\smash{\begin{tabular}[t]{l}$X$\end{tabular}}}}%
    \put(0.85949444,0.16014801){\color[rgb]{0,0,0}\makebox(0,0)[lt]{\lineheight{1.25}\smash{\begin{tabular}[t]{l}$Y$\end{tabular}}}}%
    \put(0,0){\includegraphics[width=\unitlength,page=6]{rosetta-beta2.pdf}}%
    \put(0.63719798,0.17385623){\color[rgb]{0,0,0}\makebox(0,0)[lt]{\lineheight{1.25}\smash{\begin{tabular}[t]{l}$A$\end{tabular}}}}%
    \put(0,0){\includegraphics[width=\unitlength,page=7]{rosetta-beta2.pdf}}%
    \put(0.77024309,0.13782978){\color[rgb]{0,0,0}\makebox(0,0)[lt]{\lineheight{1.25}\smash{\begin{tabular}[t]{l}$f$\end{tabular}}}}%
    \put(0,0){\includegraphics[width=\unitlength,page=8]{rosetta-beta2.pdf}}%
    \put(0.29183296,0.15280548){\color[rgb]{0,0,0}\makebox(0,0)[lt]{\lineheight{1.25}\smash{\begin{tabular}[t]{l}$X\multimap Y$\end{tabular}}}}%
    \put(0,0){\includegraphics[width=\unitlength,page=9]{rosetta-beta2.pdf}}%
    \put(0.34729909,0.05800422){\color[rgb]{0,0,0}\makebox(0,0)[lt]{\lineheight{1.25}\smash{\begin{tabular}[t]{l}$X$\end{tabular}}}}%
    \put(0.48491714,0.1528637){\color[rgb]{0,0,0}\makebox(0,0)[lt]{\lineheight{1.25}\smash{\begin{tabular}[t]{l}$Y$\end{tabular}}}}%
    \put(0,0){\includegraphics[width=\unitlength,page=10]{rosetta-beta2.pdf}}%
  \end{picture}%
\endgroup%

%% file: pics/diagruleLHS.pdf_tex
%% Creator: Inkscape 1.2 (dc2aeda, 2022-05-15), www.inkscape.org
%% Accompanies image file 'diagruleLHS.pdf' (pdf, eps, ps)
\begingroup%
  \makeatletter%
  \providecommand\color[2][]{%
    \errmessage{(Inkscape) Color is used for the text in Inkscape, but the package 'color.sty' is not loaded}%
    \renewcommand\color[2][]{}%
  }%
  \providecommand\transparent[1]{%
    \errmessage{(Inkscape) Transparency is used (non-zero) for the text in Inkscape, but the package 'transparent.sty' is not loaded}%
    \renewcommand\transparent[1]{}%
  }%
  \providecommand\rotatebox[2]{#2}%
  \newcommand*\fsize{\dimexpr\f@size pt\relax}%
  \newcommand*\lineheight[1]{\fontsize{\fsize}{#1\fsize}\selectfont}%
  \ifx\svgwidth\undefined%
    \setlength{\unitlength}{74.84348409bp}%
    \ifx\svgscale\undefined%
      \relax%
    \else%
      \setlength{\unitlength}{\unitlength * \real{\svgscale}}%
    \fi%
  \else%
    \setlength{\unitlength}{\svgwidth}%
  \fi%
  \global\let\svgwidth\undefined%
  \global\let\svgscale\undefined%
  \makeatother%
  \begin{picture}(1,0.42139631)%
    \lineheight{1}%
    \setlength\tabcolsep{0pt}%
    \put(0,0){\includegraphics[width=\unitlength,page=1]{diagruleLHS.pdf}}%
    \put(0.25814083,0.28075502){\color[rgb]{0,0,0}\makebox(0,0)[lt]{\lineheight{1.25}\smash{\begin{tabular}[t]{l}$f$\end{tabular}}}}%
    \put(0,0){\includegraphics[width=\unitlength,page=2]{diagruleLHS.pdf}}%
    \put(0.70908203,0.20058844){\color[rgb]{0,0,0}\makebox(0,0)[lt]{\lineheight{1.25}\smash{\begin{tabular}[t]{l}$g$\end{tabular}}}}%
    \put(0.01475092,0.3391411){\color[rgb]{0,0,0}\makebox(0,0)[lt]{\lineheight{1.25}\smash{\begin{tabular}[t]{l}$A$\end{tabular}}}}%
    \put(0.49002117,0.33219368){\color[rgb]{0,0,0}\makebox(0,0)[lt]{\lineheight{1.25}\smash{\begin{tabular}[t]{l}$B$\end{tabular}}}}%
    \put(0.0039902,0.01884146){\color[rgb]{0,0,0}\makebox(0,0)[lt]{\lineheight{1.25}\smash{\begin{tabular}[t]{l}$C$\end{tabular}}}}%
  \end{picture}%
\endgroup%

%% file: pics/diagruleRHS.pdf_tex
%% Creator: Inkscape 1.2 (dc2aeda, 2022-05-15), www.inkscape.org
%% Accompanies image file 'diagruleRHS.pdf' (pdf, eps, ps)
\begingroup%
  \makeatletter%
  \providecommand\color[2][]{%
    \errmessage{(Inkscape) Color is used for the text in Inkscape, but the package 'color.sty' is not loaded}%
    \renewcommand\color[2][]{}%
  }%
  \providecommand\transparent[1]{%
    \errmessage{(Inkscape) Transparency is used (non-zero) for the text in Inkscape, but the package 'transparent.sty' is not loaded}%
    \renewcommand\transparent[1]{}%
  }%
  \providecommand\rotatebox[2]{#2}%
  \newcommand*\fsize{\dimexpr\f@size pt\relax}%
  \newcommand*\lineheight[1]{\fontsize{\fsize}{#1\fsize}\selectfont}%
  \ifx\svgwidth\undefined%
    \setlength{\unitlength}{85.14618832bp}%
    \ifx\svgscale\undefined%
      \relax%
    \else%
      \setlength{\unitlength}{\unitlength * \real{\svgscale}}%
    \fi%
  \else%
    \setlength{\unitlength}{\svgwidth}%
  \fi%
  \global\let\svgwidth\undefined%
  \global\let\svgscale\undefined%
  \makeatother%
  \begin{picture}(1,0.33054957)%
    \lineheight{1}%
    \setlength\tabcolsep{0pt}%
    \put(0.01255473,0.25824725){\color[rgb]{0,0,0}\makebox(0,0)[lt]{\lineheight{1.25}\smash{\begin{tabular}[t]{l}$A$\end{tabular}}}}%
    \put(0.82263909,0.18574457){\color[rgb]{0,0,0}\makebox(0,0)[lt]{\lineheight{1.25}\smash{\begin{tabular}[t]{l}$B$\end{tabular}}}}%
    \put(0.43527522,0.18652408){\color[rgb]{0,0,0}\makebox(0,0)[lt]{\lineheight{1.25}\smash{\begin{tabular}[t]{l}$A$\end{tabular}}}}%
    \put(0,0){\includegraphics[width=\unitlength,page=1]{diagruleRHS.pdf}}%
    \put(0.63866684,0.13418491){\color[rgb]{0,0,0}\makebox(0,0)[lt]{\lineheight{1.25}\smash{\begin{tabular}[t]{l}$f$\end{tabular}}}}%
    \put(0,0){\includegraphics[width=\unitlength,page=2]{diagruleRHS.pdf}}%
    \put(0.0077931,0.08685069){\color[rgb]{0,0,0}\makebox(0,0)[lt]{\lineheight{1.25}\smash{\begin{tabular}[t]{l}$C$\end{tabular}}}}%
    \put(0,0){\includegraphics[width=\unitlength,page=3]{diagruleRHS.pdf}}%
    \put(0.23439839,0.13418554){\color[rgb]{0,0,0}\makebox(0,0)[lt]{\lineheight{1.25}\smash{\begin{tabular}[t]{l}$h$\end{tabular}}}}%
  \end{picture}%
\endgroup%

%% file: pics/diagrewlhs.pdf_tex
%% Creator: Inkscape 1.2 (dc2aeda, 2022-05-15), www.inkscape.org
%% Accompanies image file 'diagrewlhs.pdf' (pdf, eps, ps)
\begingroup%
  \makeatletter%
  \providecommand\color[2][]{%
    \errmessage{(Inkscape) Color is used for the text in Inkscape, but the package 'color.sty' is not loaded}%
    \renewcommand\color[2][]{}%
  }%
  \providecommand\transparent[1]{%
    \errmessage{(Inkscape) Transparency is used (non-zero) for the text in Inkscape, but the package 'transparent.sty' is not loaded}%
    \renewcommand\transparent[1]{}%
  }%
  \providecommand\rotatebox[2]{#2}%
  \newcommand*\fsize{\dimexpr\f@size pt\relax}%
  \newcommand*\lineheight[1]{\fontsize{\fsize}{#1\fsize}\selectfont}%
  \ifx\svgwidth\undefined%
    \setlength{\unitlength}{162.60987723bp}%
    \ifx\svgscale\undefined%
      \relax%
    \else%
      \setlength{\unitlength}{\unitlength * \real{\svgscale}}%
    \fi%
  \else%
    \setlength{\unitlength}{\svgwidth}%
  \fi%
  \global\let\svgwidth\undefined%
  \global\let\svgscale\undefined%
  \makeatother%
  \begin{picture}(1,0.33444716)%
    \lineheight{1}%
    \setlength\tabcolsep{0pt}%
    \put(0,0){\includegraphics[width=\unitlength,page=1]{diagrewlhs.pdf}}%
    \put(0.78628582,0.0979363){\color[rgb]{0,0,0}\makebox(0,0)[lt]{\lineheight{1.25}\smash{\begin{tabular}[t]{l}$g$\end{tabular}}}}%
    \put(0.00543012,0.29658804){\color[rgb]{0,0,0}\makebox(0,0)[lt]{\lineheight{1.25}\smash{\begin{tabular}[t]{l}$A$\end{tabular}}}}%
    \put(0.00756193,0.07274992){\color[rgb]{0,0,0}\makebox(0,0)[lt]{\lineheight{1.25}\smash{\begin{tabular}[t]{l}$A$\end{tabular}}}}%
    \put(0.89545297,0.29339032){\color[rgb]{0,0,0}\makebox(0,0)[lt]{\lineheight{1.25}\smash{\begin{tabular}[t]{l}$A$\end{tabular}}}}%
    \put(0,0){\includegraphics[width=\unitlength,page=2]{diagrewlhs.pdf}}%
    \put(0.23421463,0.27024492){\color[rgb]{0,0,0}\makebox(0,0)[lt]{\lineheight{1.25}\smash{\begin{tabular}[t]{l}$f$\end{tabular}}}}%
    \put(0,0){\includegraphics[width=\unitlength,page=3]{diagrewlhs.pdf}}%
    \put(0.4176543,0.15822498){\color[rgb]{0,0,0}\makebox(0,0)[lt]{\lineheight{1.25}\smash{\begin{tabular}[t]{l}$i$\end{tabular}}}}%
    \put(0,0){\includegraphics[width=\unitlength,page=4]{diagrewlhs.pdf}}%
    \put(0.90712984,0.13563776){\color[rgb]{0,0,0}\makebox(0,0)[lt]{\lineheight{1.25}\smash{\begin{tabular}[t]{l}$B$\end{tabular}}}}%
    \put(0.00720154,0.18253717){\color[rgb]{0,0,0}\makebox(0,0)[lt]{\lineheight{1.25}\smash{\begin{tabular}[t]{l}$C$\end{tabular}}}}%
    \put(0,0){\includegraphics[width=\unitlength,page=5]{diagrewlhs.pdf}}%
  \end{picture}%
\endgroup%

%% file: pics/ruleappLHS.pdf_tex
%% Creator: Inkscape 1.2 (dc2aeda, 2022-05-15), www.inkscape.org
%% Accompanies image file 'ruleappLHS.pdf' (pdf, eps, ps)
\begingroup%
  \makeatletter%
  \providecommand\color[2][]{%
    \errmessage{(Inkscape) Color is used for the text in Inkscape, but the package 'color.sty' is not loaded}%
    \renewcommand\color[2][]{}%
  }%
  \providecommand\transparent[1]{%
    \errmessage{(Inkscape) Transparency is used (non-zero) for the text in Inkscape, but the package 'transparent.sty' is not loaded}%
    \renewcommand\transparent[1]{}%
  }%
  \providecommand\rotatebox[2]{#2}%
  \newcommand*\fsize{\dimexpr\f@size pt\relax}%
  \newcommand*\lineheight[1]{\fontsize{\fsize}{#1\fsize}\selectfont}%
  \ifx\svgwidth\undefined%
    \setlength{\unitlength}{157.21183968bp}%
    \ifx\svgscale\undefined%
      \relax%
    \else%
      \setlength{\unitlength}{\unitlength * \real{\svgscale}}%
    \fi%
  \else%
    \setlength{\unitlength}{\svgwidth}%
  \fi%
  \global\let\svgwidth\undefined%
  \global\let\svgscale\undefined%
  \makeatother%
  \begin{picture}(1,0.34314705)%
    \lineheight{1}%
    \setlength\tabcolsep{0pt}%
    \put(0,0){\includegraphics[width=\unitlength,page=1]{ruleappLHS.pdf}}%
    \put(0.43171168,0.22089214){\color[rgb]{0,0,0}\makebox(0,0)[lt]{\lineheight{1.25}\smash{\begin{tabular}[t]{l}$g$\end{tabular}}}}%
    \put(0.00660996,0.30398799){\color[rgb]{0,0,0}\makebox(0,0)[lt]{\lineheight{1.25}\smash{\begin{tabular}[t]{l}$A$\end{tabular}}}}%
    \put(0.00881496,0.07246402){\color[rgb]{0,0,0}\makebox(0,0)[lt]{\lineheight{1.25}\smash{\begin{tabular}[t]{l}$A$\end{tabular}}}}%
    \put(0.90404043,0.24776065){\color[rgb]{0,0,0}\makebox(0,0)[lt]{\lineheight{1.25}\smash{\begin{tabular}[t]{l}$A$\end{tabular}}}}%
    \put(0,0){\includegraphics[width=\unitlength,page=2]{ruleappLHS.pdf}}%
    \put(0.24325,0.27674049){\color[rgb]{0,0,0}\makebox(0,0)[lt]{\lineheight{1.25}\smash{\begin{tabular}[t]{l}$f$\end{tabular}}}}%
    \put(0,0){\includegraphics[width=\unitlength,page=3]{ruleappLHS.pdf}}%
    \put(0.63314052,0.03598558){\color[rgb]{0,0,0}\makebox(0,0)[lt]{\lineheight{1.25}\smash{\begin{tabular}[t]{l}$i$\end{tabular}}}}%
    \put(0,0){\includegraphics[width=\unitlength,page=4]{ruleappLHS.pdf}}%
    \put(0.90178577,0.06143912){\color[rgb]{0,0,0}\makebox(0,0)[lt]{\lineheight{1.25}\smash{\begin{tabular}[t]{l}$B$\end{tabular}}}}%
    \put(0.0084422,0.18602127){\color[rgb]{0,0,0}\makebox(0,0)[lt]{\lineheight{1.25}\smash{\begin{tabular}[t]{l}$C$\end{tabular}}}}%
    \put(0,0){\includegraphics[width=\unitlength,page=5]{ruleappLHS.pdf}}%
  \end{picture}%
\endgroup%

%% file: pics/ruleappRHS.pdf_tex
%% Creator: Inkscape 1.2 (dc2aeda, 2022-05-15), www.inkscape.org
%% Accompanies image file 'ruleappRHS.pdf' (pdf, eps, ps)
\begingroup%
  \makeatletter%
  \providecommand\color[2][]{%
    \errmessage{(Inkscape) Color is used for the text in Inkscape, but the package 'color.sty' is not loaded}%
    \renewcommand\color[2][]{}%
  }%
  \providecommand\transparent[1]{%
    \errmessage{(Inkscape) Transparency is used (non-zero) for the text in Inkscape, but the package 'transparent.sty' is not loaded}%
    \renewcommand\transparent[1]{}%
  }%
  \providecommand\rotatebox[2]{#2}%
  \newcommand*\fsize{\dimexpr\f@size pt\relax}%
  \newcommand*\lineheight[1]{\fontsize{\fsize}{#1\fsize}\selectfont}%
  \ifx\svgwidth\undefined%
    \setlength{\unitlength}{157.2118427bp}%
    \ifx\svgscale\undefined%
      \relax%
    \else%
      \setlength{\unitlength}{\unitlength * \real{\svgscale}}%
    \fi%
  \else%
    \setlength{\unitlength}{\svgwidth}%
  \fi%
  \global\let\svgwidth\undefined%
  \global\let\svgscale\undefined%
  \makeatother%
  \begin{picture}(1,0.34314718)%
    \lineheight{1}%
    \setlength\tabcolsep{0pt}%
    \put(0,0){\includegraphics[width=\unitlength,page=1]{ruleappRHS.pdf}}%
    \put(0.24280637,0.2234112){\color[rgb]{0,0,0}\makebox(0,0)[lt]{\lineheight{1.25}\smash{\begin{tabular}[t]{l}$h$\end{tabular}}}}%
    \put(0.00660998,0.30398813){\color[rgb]{0,0,0}\makebox(0,0)[lt]{\lineheight{1.25}\smash{\begin{tabular}[t]{l}$A$\end{tabular}}}}%
    \put(0.00881498,0.07246429){\color[rgb]{0,0,0}\makebox(0,0)[lt]{\lineheight{1.25}\smash{\begin{tabular}[t]{l}$A$\end{tabular}}}}%
    \put(0.90404044,0.24776092){\color[rgb]{0,0,0}\makebox(0,0)[lt]{\lineheight{1.25}\smash{\begin{tabular}[t]{l}$A$\end{tabular}}}}%
    \put(0,0){\includegraphics[width=\unitlength,page=2]{ruleappRHS.pdf}}%
    \put(0.44223028,0.21629077){\color[rgb]{0,0,0}\makebox(0,0)[lt]{\lineheight{1.25}\smash{\begin{tabular}[t]{l}$f$\end{tabular}}}}%
    \put(0,0){\includegraphics[width=\unitlength,page=3]{ruleappRHS.pdf}}%
    \put(0.63314053,0.03598572){\color[rgb]{0,0,0}\makebox(0,0)[lt]{\lineheight{1.25}\smash{\begin{tabular}[t]{l}$i$\end{tabular}}}}%
    \put(0,0){\includegraphics[width=\unitlength,page=4]{ruleappRHS.pdf}}%
    \put(0.90178577,0.06143939){\color[rgb]{0,0,0}\makebox(0,0)[lt]{\lineheight{1.25}\smash{\begin{tabular}[t]{l}$B$\end{tabular}}}}%
    \put(0.00844222,0.18602127){\color[rgb]{0,0,0}\makebox(0,0)[lt]{\lineheight{1.25}\smash{\begin{tabular}[t]{l}$C$\end{tabular}}}}%
    \put(0,0){\includegraphics[width=\unitlength,page=5]{ruleappRHS.pdf}}%
  \end{picture}%
\endgroup%

%% file: pics/diagl.pdf_tex
%% Creator: Inkscape 1.2 (dc2aeda, 2022-05-15), www.inkscape.org
%% Accompanies image file 'diagl.pdf' (pdf, eps, ps)
\begingroup%
  \makeatletter%
  \providecommand\color[2][]{%
    \errmessage{(Inkscape) Color is used for the text in Inkscape, but the package 'color.sty' is not loaded}%
    \renewcommand\color[2][]{}%
  }%
  \providecommand\transparent[1]{%
    \errmessage{(Inkscape) Transparency is used (non-zero) for the text in Inkscape, but the package 'transparent.sty' is not loaded}%
    \renewcommand\transparent[1]{}%
  }%
  \providecommand\rotatebox[2]{#2}%
  \newcommand*\fsize{\dimexpr\f@size pt\relax}%
  \newcommand*\lineheight[1]{\fontsize{\fsize}{#1\fsize}\selectfont}%
  \ifx\svgwidth\undefined%
    \setlength{\unitlength}{51.3670127bp}%
    \ifx\svgscale\undefined%
      \relax%
    \else%
      \setlength{\unitlength}{\unitlength * \real{\svgscale}}%
    \fi%
  \else%
    \setlength{\unitlength}{\svgwidth}%
  \fi%
  \global\let\svgwidth\undefined%
  \global\let\svgscale\undefined%
  \makeatother%
  \begin{picture}(1,0.34302614)%
    \lineheight{1}%
    \setlength\tabcolsep{0pt}%
    \put(0,0){\includegraphics[width=\unitlength,page=1]{diagl.pdf}}%
    \put(0.70600577,0.21740018){\color[rgb]{0,0,0}\makebox(0,0)[lt]{\lineheight{1.25}\smash{\begin{tabular}[t]{l}$B$\end{tabular}}}}%
    \put(0.3883638,0.13508095){\color[rgb]{0,0,0}\makebox(0,0)[lt]{\lineheight{1.25}\smash{\begin{tabular}[t]{l}$l$\end{tabular}}}}%
    \put(0.06336321,0.19838368){\color[rgb]{0,0,0}\makebox(0,0)[lt]{\lineheight{1.25}\smash{\begin{tabular}[t]{l}$A$\end{tabular}}}}%
  \end{picture}%
\endgroup%

%% file: pics/diagr.pdf_tex
%% Creator: Inkscape 1.2 (dc2aeda, 2022-05-15), www.inkscape.org
%% Accompanies image file 'diagr.pdf' (pdf, eps, ps)
\begingroup%
  \makeatletter%
  \providecommand\color[2][]{%
    \errmessage{(Inkscape) Color is used for the text in Inkscape, but the package 'color.sty' is not loaded}%
    \renewcommand\color[2][]{}%
  }%
  \providecommand\transparent[1]{%
    \errmessage{(Inkscape) Transparency is used (non-zero) for the text in Inkscape, but the package 'transparent.sty' is not loaded}%
    \renewcommand\transparent[1]{}%
  }%
  \providecommand\rotatebox[2]{#2}%
  \newcommand*\fsize{\dimexpr\f@size pt\relax}%
  \newcommand*\lineheight[1]{\fontsize{\fsize}{#1\fsize}\selectfont}%
  \ifx\svgwidth\undefined%
    \setlength{\unitlength}{51.36700893bp}%
    \ifx\svgscale\undefined%
      \relax%
    \else%
      \setlength{\unitlength}{\unitlength * \real{\svgscale}}%
    \fi%
  \else%
    \setlength{\unitlength}{\svgwidth}%
  \fi%
  \global\let\svgwidth\undefined%
  \global\let\svgscale\undefined%
  \makeatother%
  \begin{picture}(1,0.34302617)%
    \lineheight{1}%
    \setlength\tabcolsep{0pt}%
    \put(0,0){\includegraphics[width=\unitlength,page=1]{diagr.pdf}}%
    \put(0.70600575,0.21740041){\color[rgb]{0,0,0}\makebox(0,0)[lt]{\lineheight{1.25}\smash{\begin{tabular}[t]{l}$B$\end{tabular}}}}%
    \put(0.38836418,0.13508096){\color[rgb]{0,0,0}\makebox(0,0)[lt]{\lineheight{1.25}\smash{\begin{tabular}[t]{l}$r$\end{tabular}}}}%
    \put(0.06336314,0.1983837){\color[rgb]{0,0,0}\makebox(0,0)[lt]{\lineheight{1.25}\smash{\begin{tabular}[t]{l}$A$\end{tabular}}}}%
  \end{picture}%
\endgroup%

%% file: pics/diagf.pdf_tex
%% Creator: Inkscape 1.2 (dc2aeda, 2022-05-15), www.inkscape.org
%% Accompanies image file 'diagf.pdf' (pdf, eps, ps)
\begingroup%
  \makeatletter%
  \providecommand\color[2][]{%
    \errmessage{(Inkscape) Color is used for the text in Inkscape, but the package 'color.sty' is not loaded}%
    \renewcommand\color[2][]{}%
  }%
  \providecommand\transparent[1]{%
    \errmessage{(Inkscape) Transparency is used (non-zero) for the text in Inkscape, but the package 'transparent.sty' is not loaded}%
    \renewcommand\transparent[1]{}%
  }%
  \providecommand\rotatebox[2]{#2}%
  \newcommand*\fsize{\dimexpr\f@size pt\relax}%
  \newcommand*\lineheight[1]{\fontsize{\fsize}{#1\fsize}\selectfont}%
  \ifx\svgwidth\undefined%
    \setlength{\unitlength}{51.2606641bp}%
    \ifx\svgscale\undefined%
      \relax%
    \else%
      \setlength{\unitlength}{\unitlength * \real{\svgscale}}%
    \fi%
  \else%
    \setlength{\unitlength}{\svgwidth}%
  \fi%
  \global\let\svgwidth\undefined%
  \global\let\svgscale\undefined%
  \makeatother%
  \begin{picture}(1,0.32848456)%
    \lineheight{1}%
    \setlength\tabcolsep{0pt}%
    \put(0.69213638,0.20709246){\color[rgb]{0,0,0}\makebox(0,0)[lt]{\lineheight{1.25}\smash{\begin{tabular}[t]{l}$D$\end{tabular}}}}%
    \put(0.04870817,0.20838726){\color[rgb]{0,0,0}\makebox(0,0)[lt]{\lineheight{1.25}\smash{\begin{tabular}[t]{l}$C$\end{tabular}}}}%
    \put(0,0){\includegraphics[width=\unitlength,page=1]{diagf.pdf}}%
    \put(0.38655091,0.12144984){\color[rgb]{0,0,0}\makebox(0,0)[lt]{\lineheight{1.25}\smash{\begin{tabular}[t]{l}$f$\end{tabular}}}}%
    \put(0,0){\includegraphics[width=\unitlength,page=2]{diagf.pdf}}%
  \end{picture}%
\endgroup%

%% file: pics/diaglc.pdf_tex
%% Creator: Inkscape 1.2 (dc2aeda, 2022-05-15), www.inkscape.org
%% Accompanies image file 'diaglc.pdf' (pdf, eps, ps)
\begingroup%
  \makeatletter%
  \providecommand\color[2][]{%
    \errmessage{(Inkscape) Color is used for the text in Inkscape, but the package 'color.sty' is not loaded}%
    \renewcommand\color[2][]{}%
  }%
  \providecommand\transparent[1]{%
    \errmessage{(Inkscape) Transparency is used (non-zero) for the text in Inkscape, but the package 'transparent.sty' is not loaded}%
    \renewcommand\transparent[1]{}%
  }%
  \providecommand\rotatebox[2]{#2}%
  \newcommand*\fsize{\dimexpr\f@size pt\relax}%
  \newcommand*\lineheight[1]{\fontsize{\fsize}{#1\fsize}\selectfont}%
  \ifx\svgwidth\undefined%
    \setlength{\unitlength}{51.69525735bp}%
    \ifx\svgscale\undefined%
      \relax%
    \else%
      \setlength{\unitlength}{\unitlength * \real{\svgscale}}%
    \fi%
  \else%
    \setlength{\unitlength}{\svgwidth}%
  \fi%
  \global\let\svgwidth\undefined%
  \global\let\svgscale\undefined%
  \makeatother%
  \begin{picture}(1,0.60651501)%
    \lineheight{1}%
    \setlength\tabcolsep{0pt}%
    \put(0,0){\includegraphics[width=\unitlength,page=1]{diaglc.pdf}}%
    \put(0.68995277,0.48742735){\color[rgb]{0,0,0}\makebox(0,0)[lt]{\lineheight{1.25}\smash{\begin{tabular}[t]{l}$A$\end{tabular}}}}%
    \put(0,0){\includegraphics[width=\unitlength,page=2]{diaglc.pdf}}%
    \put(0.69215538,0.20048226){\color[rgb]{0,0,0}\makebox(0,0)[lt]{\lineheight{1.25}\smash{\begin{tabular}[t]{l}$Q$\end{tabular}}}}%
    \put(0.38011304,0.25522493){\color[rgb]{0,0,0}\makebox(0,0)[lt]{\lineheight{1.25}\smash{\begin{tabular}[t]{l}$l_C$\end{tabular}}}}%
    \put(0.03394477,0.32969566){\color[rgb]{0,0,0}\makebox(0,0)[lt]{\lineheight{1.25}\smash{\begin{tabular}[t]{l}$C$\end{tabular}}}}%
  \end{picture}%
\endgroup%

%% file: pics/diagrc.pdf_tex
%% Creator: Inkscape 1.2 (dc2aeda, 2022-05-15), www.inkscape.org
%% Accompanies image file 'diagrc.pdf' (pdf, eps, ps)
\begingroup%
  \makeatletter%
  \providecommand\color[2][]{%
    \errmessage{(Inkscape) Color is used for the text in Inkscape, but the package 'color.sty' is not loaded}%
    \renewcommand\color[2][]{}%
  }%
  \providecommand\transparent[1]{%
    \errmessage{(Inkscape) Transparency is used (non-zero) for the text in Inkscape, but the package 'transparent.sty' is not loaded}%
    \renewcommand\transparent[1]{}%
  }%
  \providecommand\rotatebox[2]{#2}%
  \newcommand*\fsize{\dimexpr\f@size pt\relax}%
  \newcommand*\lineheight[1]{\fontsize{\fsize}{#1\fsize}\selectfont}%
  \ifx\svgwidth\undefined%
    \setlength{\unitlength}{51.74764663bp}%
    \ifx\svgscale\undefined%
      \relax%
    \else%
      \setlength{\unitlength}{\unitlength * \real{\svgscale}}%
    \fi%
  \else%
    \setlength{\unitlength}{\svgwidth}%
  \fi%
  \global\let\svgwidth\undefined%
  \global\let\svgscale\undefined%
  \makeatother%
  \begin{picture}(1,0.59858259)%
    \lineheight{1}%
    \setlength\tabcolsep{0pt}%
    \put(0,0){\includegraphics[width=\unitlength,page=1]{diagrc.pdf}}%
    \put(0.6950336,0.34823829){\color[rgb]{0,0,0}\makebox(0,0)[lt]{\lineheight{1.25}\smash{\begin{tabular}[t]{l}$D$\end{tabular}}}}%
    \put(0,0){\includegraphics[width=\unitlength,page=2]{diagrc.pdf}}%
    \put(0.03833691,0.20027929){\color[rgb]{0,0,0}\makebox(0,0)[lt]{\lineheight{1.25}\smash{\begin{tabular}[t]{l}$Q$\end{tabular}}}}%
    \put(0.37972806,0.25496665){\color[rgb]{0,0,0}\makebox(0,0)[lt]{\lineheight{1.25}\smash{\begin{tabular}[t]{l}$r_C$\end{tabular}}}}%
    \put(0.03391025,0.47961549){\color[rgb]{0,0,0}\makebox(0,0)[lt]{\lineheight{1.25}\smash{\begin{tabular}[t]{l}$B$\end{tabular}}}}%
  \end{picture}%
\endgroup%

%% file: pics/rewcontextLHS.pdf_tex
%% Creator: Inkscape 1.2 (dc2aeda, 2022-05-15), www.inkscape.org
%% Accompanies image file 'rewcontextLHS.pdf' (pdf, eps, ps)
\begingroup%
  \makeatletter%
  \providecommand\color[2][]{%
    \errmessage{(Inkscape) Color is used for the text in Inkscape, but the package 'color.sty' is not loaded}%
    \renewcommand\color[2][]{}%
  }%
  \providecommand\transparent[1]{%
    \errmessage{(Inkscape) Transparency is used (non-zero) for the text in Inkscape, but the package 'transparent.sty' is not loaded}%
    \renewcommand\transparent[1]{}%
  }%
  \providecommand\rotatebox[2]{#2}%
  \newcommand*\fsize{\dimexpr\f@size pt\relax}%
  \newcommand*\lineheight[1]{\fontsize{\fsize}{#1\fsize}\selectfont}%
  \ifx\svgwidth\undefined%
    \setlength{\unitlength}{134.55800788bp}%
    \ifx\svgscale\undefined%
      \relax%
    \else%
      \setlength{\unitlength}{\unitlength * \real{\svgscale}}%
    \fi%
  \else%
    \setlength{\unitlength}{\svgwidth}%
  \fi%
  \global\let\svgwidth\undefined%
  \global\let\svgscale\undefined%
  \makeatother%
  \begin{picture}(1,0.23382956)%
    \lineheight{1}%
    \setlength\tabcolsep{0pt}%
    \put(0,0){\includegraphics[width=\unitlength,page=1]{rewcontextLHS.pdf}}%
    \put(0.2650701,0.18807779){\color[rgb]{0,0,0}\makebox(0,0)[lt]{\lineheight{1.25}\smash{\begin{tabular}[t]{l}$A$\end{tabular}}}}%
    \put(0,0){\includegraphics[width=\unitlength,page=2]{rewcontextLHS.pdf}}%
    \put(0.41259817,0.05557772){\color[rgb]{0,0,0}\makebox(0,0)[lt]{\lineheight{1.25}\smash{\begin{tabular}[t]{l}$Q$\end{tabular}}}}%
    \put(0.14603412,0.09886895){\color[rgb]{0,0,0}\makebox(0,0)[lt]{\lineheight{1.25}\smash{\begin{tabular}[t]{l}$l_C$\end{tabular}}}}%
    \put(0.01304123,0.12747954){\color[rgb]{0,0,0}\makebox(0,0)[lt]{\lineheight{1.25}\smash{\begin{tabular}[t]{l}$C$\end{tabular}}}}%
    \put(0,0){\includegraphics[width=\unitlength,page=3]{rewcontextLHS.pdf}}%
    \put(0.88271754,0.13392374){\color[rgb]{0,0,0}\makebox(0,0)[lt]{\lineheight{1.25}\smash{\begin{tabular}[t]{l}$D$\end{tabular}}}}%
    \put(0.76145908,0.0980538){\color[rgb]{0,0,0}\makebox(0,0)[lt]{\lineheight{1.25}\smash{\begin{tabular}[t]{l}$r_C$\end{tabular}}}}%
    \put(0.60617127,0.18444813){\color[rgb]{0,0,0}\makebox(0,0)[lt]{\lineheight{1.25}\smash{\begin{tabular}[t]{l}$B$\end{tabular}}}}%
    \put(0,0){\includegraphics[width=\unitlength,page=4]{rewcontextLHS.pdf}}%
    \put(0.46208899,0.15193028){\color[rgb]{0,0,0}\makebox(0,0)[lt]{\lineheight{1.25}\smash{\begin{tabular}[t]{l}$l$\end{tabular}}}}%
  \end{picture}%
\endgroup%

%% file: pics/rewcontextRHS.pdf_tex
%% Creator: Inkscape 1.2 (dc2aeda, 2022-05-15), www.inkscape.org
%% Accompanies image file 'rewcontextRHS.pdf' (pdf, eps, ps)
\begingroup%
  \makeatletter%
  \providecommand\color[2][]{%
    \errmessage{(Inkscape) Color is used for the text in Inkscape, but the package 'color.sty' is not loaded}%
    \renewcommand\color[2][]{}%
  }%
  \providecommand\transparent[1]{%
    \errmessage{(Inkscape) Transparency is used (non-zero) for the text in Inkscape, but the package 'transparent.sty' is not loaded}%
    \renewcommand\transparent[1]{}%
  }%
  \providecommand\rotatebox[2]{#2}%
  \newcommand*\fsize{\dimexpr\f@size pt\relax}%
  \newcommand*\lineheight[1]{\fontsize{\fsize}{#1\fsize}\selectfont}%
  \ifx\svgwidth\undefined%
    \setlength{\unitlength}{134.55800782bp}%
    \ifx\svgscale\undefined%
      \relax%
    \else%
      \setlength{\unitlength}{\unitlength * \real{\svgscale}}%
    \fi%
  \else%
    \setlength{\unitlength}{\svgwidth}%
  \fi%
  \global\let\svgwidth\undefined%
  \global\let\svgscale\undefined%
  \makeatother%
  \begin{picture}(1,0.2338296)%
    \lineheight{1}%
    \setlength\tabcolsep{0pt}%
    \put(0,0){\includegraphics[width=\unitlength,page=1]{rewcontextRHS.pdf}}%
    \put(0.2650701,0.18807783){\color[rgb]{0,0,0}\makebox(0,0)[lt]{\lineheight{1.25}\smash{\begin{tabular}[t]{l}$A$\end{tabular}}}}%
    \put(0,0){\includegraphics[width=\unitlength,page=2]{rewcontextRHS.pdf}}%
    \put(0.41259817,0.05557772){\color[rgb]{0,0,0}\makebox(0,0)[lt]{\lineheight{1.25}\smash{\begin{tabular}[t]{l}$Q$\end{tabular}}}}%
    \put(0.14603412,0.09886899){\color[rgb]{0,0,0}\makebox(0,0)[lt]{\lineheight{1.25}\smash{\begin{tabular}[t]{l}$l_C$\end{tabular}}}}%
    \put(0.01304123,0.12747954){\color[rgb]{0,0,0}\makebox(0,0)[lt]{\lineheight{1.25}\smash{\begin{tabular}[t]{l}$C$\end{tabular}}}}%
    \put(0,0){\includegraphics[width=\unitlength,page=3]{rewcontextRHS.pdf}}%
    \put(0.88271754,0.13392374){\color[rgb]{0,0,0}\makebox(0,0)[lt]{\lineheight{1.25}\smash{\begin{tabular}[t]{l}$D$\end{tabular}}}}%
    \put(0.76145908,0.0980538){\color[rgb]{0,0,0}\makebox(0,0)[lt]{\lineheight{1.25}\smash{\begin{tabular}[t]{l}$r_C$\end{tabular}}}}%
    \put(0.60617127,0.18444813){\color[rgb]{0,0,0}\makebox(0,0)[lt]{\lineheight{1.25}\smash{\begin{tabular}[t]{l}$B$\end{tabular}}}}%
    \put(0,0){\includegraphics[width=\unitlength,page=4]{rewcontextRHS.pdf}}%
    \put(0.4509414,0.16307794){\color[rgb]{0,0,0}\makebox(0,0)[lt]{\lineheight{1.25}\smash{\begin{tabular}[t]{l}$r$\end{tabular}}}}%
  \end{picture}%
\endgroup%

%% file: pics/symmetry.pdf_tex
%% Creator: Inkscape 1.2 (dc2aeda, 2022-05-15), www.inkscape.org
%% Accompanies image file 'symmetry.pdf' (pdf, eps, ps)
\begingroup%
  \makeatletter%
  \providecommand\color[2][]{%
    \errmessage{(Inkscape) Color is used for the text in Inkscape, but the package 'color.sty' is not loaded}%
    \renewcommand\color[2][]{}%
  }%
  \providecommand\transparent[1]{%
    \errmessage{(Inkscape) Transparency is used (non-zero) for the text in Inkscape, but the package 'transparent.sty' is not loaded}%
    \renewcommand\transparent[1]{}%
  }%
  \providecommand\rotatebox[2]{#2}%
  \newcommand*\fsize{\dimexpr\f@size pt\relax}%
  \newcommand*\lineheight[1]{\fontsize{\fsize}{#1\fsize}\selectfont}%
  \ifx\svgwidth\undefined%
    \setlength{\unitlength}{46.0184658bp}%
    \ifx\svgscale\undefined%
      \relax%
    \else%
      \setlength{\unitlength}{\unitlength * \real{\svgscale}}%
    \fi%
  \else%
    \setlength{\unitlength}{\svgwidth}%
  \fi%
  \global\let\svgwidth\undefined%
  \global\let\svgscale\undefined%
  \makeatother%
  \begin{picture}(1,0.83604446)%
    \lineheight{1}%
    \setlength\tabcolsep{0pt}%
    \put(0.66259006,0.0733413){\color[rgb]{0,0,0}\makebox(0,0)[lt]{\lineheight{1.25}\smash{\begin{tabular}[t]{l}$A$\end{tabular}}}}%
    \put(0.02418045,0.70226627){\color[rgb]{0,0,0}\makebox(0,0)[lt]{\lineheight{1.25}\smash{\begin{tabular}[t]{l}$A$\end{tabular}}}}%
    \put(0,0){\includegraphics[width=\unitlength,page=1]{symmetry.pdf}}%
    \put(0.02212769,0.07336198){\color[rgb]{0,0,0}\makebox(0,0)[lt]{\lineheight{1.25}\smash{\begin{tabular}[t]{l}$B$\end{tabular}}}}%
    \put(0.67183597,0.70226627){\color[rgb]{0,0,0}\makebox(0,0)[lt]{\lineheight{1.25}\smash{\begin{tabular}[t]{l}$B$\end{tabular}}}}%
    \put(0,0){\includegraphics[width=\unitlength,page=2]{symmetry.pdf}}%
  \end{picture}%
\endgroup%

%% file: pics/opf.pdf_tex
%% Creator: Inkscape 1.2 (dc2aeda, 2022-05-15), www.inkscape.org
%% Accompanies image file 'opf.pdf' (pdf, eps, ps)
\begingroup%
  \makeatletter%
  \providecommand\color[2][]{%
    \errmessage{(Inkscape) Color is used for the text in Inkscape, but the package 'color.sty' is not loaded}%
    \renewcommand\color[2][]{}%
  }%
  \providecommand\transparent[1]{%
    \errmessage{(Inkscape) Transparency is used (non-zero) for the text in Inkscape, but the package 'transparent.sty' is not loaded}%
    \renewcommand\transparent[1]{}%
  }%
  \providecommand\rotatebox[2]{#2}%
  \newcommand*\fsize{\dimexpr\f@size pt\relax}%
  \newcommand*\lineheight[1]{\fontsize{\fsize}{#1\fsize}\selectfont}%
  \ifx\svgwidth\undefined%
    \setlength{\unitlength}{73.70321424bp}%
    \ifx\svgscale\undefined%
      \relax%
    \else%
      \setlength{\unitlength}{\unitlength * \real{\svgscale}}%
    \fi%
  \else%
    \setlength{\unitlength}{\svgwidth}%
  \fi%
  \global\let\svgwidth\undefined%
  \global\let\svgscale\undefined%
  \makeatother%
  \begin{picture}(1,0.59361787)%
    \lineheight{1}%
    \setlength\tabcolsep{0pt}%
    \put(0.48435362,0.14403292){\color[rgb]{0,0,0}\makebox(0,0)[lt]{\lineheight{1.25}\smash{\begin{tabular}[t]{l}$B_m$\end{tabular}}}}%
    \put(0.48298389,0.51009007){\color[rgb]{0,0,0}\makebox(0,0)[lt]{\lineheight{1.25}\smash{\begin{tabular}[t]{l}$B_1$\end{tabular}}}}%
    \put(-0.00385463,0.14493345){\color[rgb]{0,0,0}\makebox(0,0)[lt]{\lineheight{1.25}\smash{\begin{tabular}[t]{l}$A_n$\end{tabular}}}}%
    \put(0,0){\includegraphics[width=\unitlength,page=1]{opf.pdf}}%
    \put(0.28217721,0.25020682){\color[rgb]{0,0,0}\makebox(0,0)[lt]{\lineheight{1.25}\smash{\begin{tabular}[t]{l}$f$\end{tabular}}}}%
    \put(0,0){\includegraphics[width=\unitlength,page=2]{opf.pdf}}%
    \put(-0.00821498,0.50527549){\color[rgb]{0,0,0}\makebox(0,0)[lt]{\lineheight{1.25}\smash{\begin{tabular}[t]{l}$A_1$\end{tabular}}}}%
    \put(0,0){\includegraphics[width=\unitlength,page=3]{opf.pdf}}%
    \put(0.05402122,0.28514185){\color[rgb]{0,0,0}\makebox(0,0)[lt]{\lineheight{1.25}\smash{\begin{tabular}[t]{l}$\vdots$\end{tabular}}}}%
    \put(0.52316013,0.2886759){\color[rgb]{0,0,0}\makebox(0,0)[lt]{\lineheight{1.25}\smash{\begin{tabular}[t]{l}$\vdots$\end{tabular}}}}%
  \end{picture}%
\endgroup%

%% file: pics/seqcomp.pdf_tex
%% Creator: Inkscape 1.2 (dc2aeda, 2022-05-15), www.inkscape.org
%% Accompanies image file 'seqcomp.pdf' (pdf, eps, ps)
\begingroup%
  \makeatletter%
  \providecommand\color[2][]{%
    \errmessage{(Inkscape) Color is used for the text in Inkscape, but the package 'color.sty' is not loaded}%
    \renewcommand\color[2][]{}%
  }%
  \providecommand\transparent[1]{%
    \errmessage{(Inkscape) Transparency is used (non-zero) for the text in Inkscape, but the package 'transparent.sty' is not loaded}%
    \renewcommand\transparent[1]{}%
  }%
  \providecommand\rotatebox[2]{#2}%
  \newcommand*\fsize{\dimexpr\f@size pt\relax}%
  \newcommand*\lineheight[1]{\fontsize{\fsize}{#1\fsize}\selectfont}%
  \ifx\svgwidth\undefined%
    \setlength{\unitlength}{87.60628346bp}%
    \ifx\svgscale\undefined%
      \relax%
    \else%
      \setlength{\unitlength}{\unitlength * \real{\svgscale}}%
    \fi%
  \else%
    \setlength{\unitlength}{\svgwidth}%
  \fi%
  \global\let\svgwidth\undefined%
  \global\let\svgscale\undefined%
  \makeatother%
  \begin{picture}(1,0.18303784)%
    \lineheight{1}%
    \setlength\tabcolsep{0pt}%
    \put(0,0){\includegraphics[width=\unitlength,page=1]{seqcomp.pdf}}%
    \put(0.22618066,0.07106314){\color[rgb]{0,0,0}\makebox(0,0)[lt]{\lineheight{1.25}\smash{\begin{tabular}[t]{l}$f$\end{tabular}}}}%
    \put(0,0){\includegraphics[width=\unitlength,page=2]{seqcomp.pdf}}%
    \put(0.71773303,0.07430592){\color[rgb]{0,0,0}\makebox(0,0)[lt]{\lineheight{1.25}\smash{\begin{tabular}[t]{l}$g$\end{tabular}}}}%
    \put(0,0){\includegraphics[width=\unitlength,page=3]{seqcomp.pdf}}%
  \end{picture}%
\endgroup%

%% file: pics/parcomp.pdf_tex
%% Creator: Inkscape 1.2 (dc2aeda, 2022-05-15), www.inkscape.org
%% Accompanies image file 'parcomp.pdf' (pdf, eps, ps)
\begingroup%
  \makeatletter%
  \providecommand\color[2][]{%
    \errmessage{(Inkscape) Color is used for the text in Inkscape, but the package 'color.sty' is not loaded}%
    \renewcommand\color[2][]{}%
  }%
  \providecommand\transparent[1]{%
    \errmessage{(Inkscape) Transparency is used (non-zero) for the text in Inkscape, but the package 'transparent.sty' is not loaded}%
    \renewcommand\transparent[1]{}%
  }%
  \providecommand\rotatebox[2]{#2}%
  \newcommand*\fsize{\dimexpr\f@size pt\relax}%
  \newcommand*\lineheight[1]{\fontsize{\fsize}{#1\fsize}\selectfont}%
  \ifx\svgwidth\undefined%
    \setlength{\unitlength}{44.5432252bp}%
    \ifx\svgscale\undefined%
      \relax%
    \else%
      \setlength{\unitlength}{\unitlength * \real{\svgscale}}%
    \fi%
  \else%
    \setlength{\unitlength}{\svgwidth}%
  \fi%
  \global\let\svgwidth\undefined%
  \global\let\svgscale\undefined%
  \makeatother%
  \begin{picture}(1,0.86385541)%
    \lineheight{1}%
    \setlength\tabcolsep{0pt}%
    \put(0,0){\includegraphics[width=\unitlength,page=1]{parcomp.pdf}}%
    \put(0.44484559,0.65000471){\color[rgb]{0,0,0}\makebox(0,0)[lt]{\lineheight{1.25}\smash{\begin{tabular}[t]{l}$f$\end{tabular}}}}%
    \put(0,0){\includegraphics[width=\unitlength,page=2]{parcomp.pdf}}%
    \put(0.44484559,0.13976438){\color[rgb]{0,0,0}\makebox(0,0)[lt]{\lineheight{1.25}\smash{\begin{tabular}[t]{l}$g$\end{tabular}}}}%
    \put(0,0){\includegraphics[width=\unitlength,page=3]{parcomp.pdf}}%
  \end{picture}%
\endgroup%

%% file: pics/rewriting-pics/exhyp.pdf_tex
%% Creator: Inkscape 1.2.2 (732a01da63, 2022-12-09), www.inkscape.org
%% Accompanies image file 'exhyp.pdf' (pdf, eps, ps)
\begingroup%
  \makeatletter%
  \providecommand\color[2][]{%
    \errmessage{(Inkscape) Color is used for the text in Inkscape, but the package 'color.sty' is not loaded}%
    \renewcommand\color[2][]{}%
  }%
  \providecommand\transparent[1]{%
    \errmessage{(Inkscape) Transparency is used (non-zero) for the text in Inkscape, but the package 'transparent.sty' is not loaded}%
    \renewcommand\transparent[1]{}%
  }%
  \providecommand\rotatebox[2]{#2}%
  \newcommand*\fsize{\dimexpr\f@size pt\relax}%
  \newcommand*\lineheight[1]{\fontsize{\fsize}{#1\fsize}\selectfont}%
  \ifx\svgwidth\undefined%
    \setlength{\unitlength}{83.33082487bp}%
    \ifx\svgscale\undefined%
      \relax%
    \else%
      \setlength{\unitlength}{\unitlength * \real{\svgscale}}%
    \fi%
  \else%
    \setlength{\unitlength}{\svgwidth}%
  \fi%
  \global\let\svgwidth\undefined%
  \global\let\svgscale\undefined%
  \makeatother%
  \begin{picture}(1,0.28510587)%
    \lineheight{1}%
    \setlength\tabcolsep{0pt}%
    \put(0,0){\includegraphics[width=\unitlength,page=1]{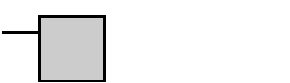}}%
    \put(0.20472892,0.09104651){\color[rgb]{0,0,0}\makebox(0,0)[lt]{\lineheight{1.25}\smash{\begin{tabular}[t]{l}$f$\end{tabular}}}}%
    \put(0.04240441,0.21122844){\color[rgb]{0,0,0}\makebox(0,0)[lt]{\lineheight{1.25}\smash{\begin{tabular}[t]{l}$A$\end{tabular}}}}%
    \put(0,0){\includegraphics[width=\unitlength,page=2]{exhyp.pdf}}%
    \put(0.58335756,0.10401457){\color[rgb]{0,0,0}\makebox(0,0)[lt]{\lineheight{1.25}\smash{\begin{tabular}[t]{l}$g$\end{tabular}}}}%
    \put(0,0){\includegraphics[width=\unitlength,page=3]{exhyp.pdf}}%
    \put(0.03663252,0.07652855){\color[rgb]{0,0,0}\makebox(0,0)[lt]{\lineheight{1.25}\smash{\begin{tabular}[t]{l}$B$\end{tabular}}}}%
    \put(0.3934305,0.14281833){\color[rgb]{0,0,0}\makebox(0,0)[lt]{\lineheight{1.25}\smash{\begin{tabular}[t]{l}$C$\end{tabular}}}}%
    \put(0.79612277,0.1479183){\color[rgb]{0,0,0}\makebox(0,0)[lt]{\lineheight{1.25}\smash{\begin{tabular}[t]{l}$B$\end{tabular}}}}%
  \end{picture}%
\endgroup%

%% file: pics/opsem.pdf_tex
%% Creator: Inkscape 1.2.2 (732a01da63, 2022-12-09), www.inkscape.org
%% Accompanies image file 'opsem.pdf' (pdf, eps, ps)
\begingroup%
  \makeatletter%
  \providecommand\color[2][]{%
    \errmessage{(Inkscape) Color is used for the text in Inkscape, but the package 'color.sty' is not loaded}%
    \renewcommand\color[2][]{}%
  }%
  \providecommand\transparent[1]{%
    \errmessage{(Inkscape) Transparency is used (non-zero) for the text in Inkscape, but the package 'transparent.sty' is not loaded}%
    \renewcommand\transparent[1]{}%
  }%
  \providecommand\rotatebox[2]{#2}%
  \newcommand*\fsize{\dimexpr\f@size pt\relax}%
  \newcommand*\lineheight[1]{\fontsize{\fsize}{#1\fsize}\selectfont}%
  \ifx\svgwidth\undefined%
    \setlength{\unitlength}{187.50001663bp}%
    \ifx\svgscale\undefined%
      \relax%
    \else%
      \setlength{\unitlength}{\unitlength * \real{\svgscale}}%
    \fi%
  \else%
    \setlength{\unitlength}{\svgwidth}%
  \fi%
  \global\let\svgwidth\undefined%
  \global\let\svgscale\undefined%
  \makeatother%
  \begin{picture}(1,1.08558431)%
    \lineheight{1}%
    \setlength\tabcolsep{0pt}%
    \put(0,0){\includegraphics[width=\unitlength,page=1]{opsem.pdf}}%
    \put(0.53063679,0.69370658){\makebox(0,0)[lt]{\lineheight{1.25}\smash{\begin{tabular}[t]{l}$\stackrel\beta\rew$\end{tabular}}}}%
    \put(0,0){\includegraphics[width=\unitlength,page=2]{opsem.pdf}}%
    \put(0.53063679,0.51050055){\makebox(0,0)[lt]{\lineheight{1.25}\smash{\begin{tabular}[t]{l}$\stackrel{{C}{1}}\rew$\end{tabular}}}}%
    \put(0.53063679,0.27219287){\makebox(0,0)[lt]{\lineheight{1.25}\smash{\begin{tabular}[t]{l}$\stackrel{{C}{2}}\rew$\end{tabular}}}}%
    \put(0.53063679,0.07104792){\makebox(0,0)[lt]{\lineheight{1.25}\smash{\begin{tabular}[t]{l}$\stackrel{V}\rew$\end{tabular}}}}%
    \put(0.53063679,1.02906326){\makebox(0,0)[lt]{\lineheight{1.25}\smash{\begin{tabular}[t]{l}$\stackrel{{S}{1}}\rew$\end{tabular}}}}%
    \put(0.53063679,0.91036087){\makebox(0,0)[lt]{\lineheight{1.25}\smash{\begin{tabular}[t]{l}$\stackrel{{S}{2}}\rew$\end{tabular}}}}%
    \put(0,0){\includegraphics[width=\unitlength,page=3]{opsem.pdf}}%
  \end{picture}%
\endgroup%

%% file: pics/ptrnotsd.pdf_tex
%% Creator: Inkscape 1.2.2 (732a01da63, 2022-12-09), www.inkscape.org
%% Accompanies image file 'ptrnotsd.pdf' (pdf, eps, ps)
\begingroup%
  \makeatletter%
  \providecommand\color[2][]{%
    \errmessage{(Inkscape) Color is used for the text in Inkscape, but the package 'color.sty' is not loaded}%
    \renewcommand\color[2][]{}%
  }%
  \providecommand\transparent[1]{%
    \errmessage{(Inkscape) Transparency is used (non-zero) for the text in Inkscape, but the package 'transparent.sty' is not loaded}%
    \renewcommand\transparent[1]{}%
  }%
  \providecommand\rotatebox[2]{#2}%
  \newcommand*\fsize{\dimexpr\f@size pt\relax}%
  \newcommand*\lineheight[1]{\fontsize{\fsize}{#1\fsize}\selectfont}%
  \ifx\svgwidth\undefined%
    \setlength{\unitlength}{72.32670457bp}%
    \ifx\svgscale\undefined%
      \relax%
    \else%
      \setlength{\unitlength}{\unitlength * \real{\svgscale}}%
    \fi%
  \else%
    \setlength{\unitlength}{\svgwidth}%
  \fi%
  \global\let\svgwidth\undefined%
  \global\let\svgscale\undefined%
  \makeatother%
  \begin{picture}(1,0.43743995)%
    \lineheight{1}%
    \setlength\tabcolsep{0pt}%
    \put(0,0){\includegraphics[width=\unitlength,page=1]{ptrnotsd.pdf}}%
    \put(0.41579886,0.17956065){\color[rgb]{0,0,0}\makebox(0,0)[lt]{\lineheight{1.25}\smash{\begin{tabular}[t]{l}$=$\end{tabular}}}}%
    \put(0,0){\includegraphics[width=\unitlength,page=2]{ptrnotsd.pdf}}%
  \end{picture}%
\endgroup%

%% file: pics/crule.pdf_tex
%% Creator: Inkscape 1.1 (c68e22c387, 2021-05-23), www.inkscape.org
%% Accompanies image file 'crule.pdf' (pdf, eps, ps)
\begingroup%
  \makeatletter%
  \providecommand\color[2][]{%
    \errmessage{(Inkscape) Color is used for the text in Inkscape, but the package 'color.sty' is not loaded}%
    \renewcommand\color[2][]{}%
  }%
  \providecommand\transparent[1]{%
    \errmessage{(Inkscape) Transparency is used (non-zero) for the text in Inkscape, but the package 'transparent.sty' is not loaded}%
    \renewcommand\transparent[1]{}%
  }%
  \providecommand\rotatebox[2]{#2}%
  \newcommand*\fsize{\dimexpr\f@size pt\relax}%
  \newcommand*\lineheight[1]{\fontsize{\fsize}{#1\fsize}\selectfont}%
  \ifx\svgwidth\undefined%
    \setlength{\unitlength}{310.02341952bp}%
    \ifx\svgscale\undefined%
      \relax%
    \else%
      \setlength{\unitlength}{\unitlength * \real{\svgscale}}%
    \fi%
  \else%
    \setlength{\unitlength}{\svgwidth}%
  \fi%
  \global\let\svgwidth\undefined%
  \global\let\svgscale\undefined%
  \makeatother%
  \begin{picture}(1,0.19490423)%
    \lineheight{1}%
    \setlength\tabcolsep{0pt}%
    \put(0,0){\includegraphics[width=\unitlength,page=1]{crule.pdf}}%
    \put(0.89927346,0.09211324){\makebox(0,0)[lt]{\lineheight{1.25}\smash{\begin{tabular}[t]{l}$\seval u$\end{tabular}}}}%
    \put(0.75251053,0.09534095){\makebox(0,0)[lt]{\lineheight{1.25}\smash{\begin{tabular}[t]{l}$\seval v$\end{tabular}}}}%
    \put(0.74562534,0.03163525){\makebox(0,0)[lt]{\lineheight{1.25}\smash{\begin{tabular}[t]{l}$\seval{v'}$\end{tabular}}}}%
    \put(0,0){\includegraphics[width=\unitlength,page=2]{crule.pdf}}%
    \put(0.82024735,0.10904883){\makebox(0,0)[lt]{\lineheight{1.25}\smash{\begin{tabular}[t]{l}$x$\end{tabular}}}}%
    \put(0.64687333,0.06631002){\makebox(0,0)[lt]{\lineheight{1.25}\smash{\begin{tabular}[t]{l}$x$\end{tabular}}}}%
    \put(0.82266656,0.05421411){\makebox(0,0)[lt]{\lineheight{1.25}\smash{\begin{tabular}[t]{l}$x'$\end{tabular}}}}%
    \put(0.6468734,0.03647352){\makebox(0,0)[lt]{\lineheight{1.25}\smash{\begin{tabular}[t]{l}$x'$\end{tabular}}}}%
    \put(0,0){\includegraphics[width=\unitlength,page=3]{crule.pdf}}%
    \put(0.48317581,0.09211324){\makebox(0,0)[lt]{\lineheight{1.25}\smash{\begin{tabular}[t]{l}$\seval u$\end{tabular}}}}%
    \put(0.33641285,0.09534095){\makebox(0,0)[lt]{\lineheight{1.25}\smash{\begin{tabular}[t]{l}$\seval v$\end{tabular}}}}%
    \put(0.20856894,0.03163525){\makebox(0,0)[lt]{\lineheight{1.25}\smash{\begin{tabular}[t]{l}$\seval{v'}$\end{tabular}}}}%
    \put(0,0){\includegraphics[width=\unitlength,page=4]{crule.pdf}}%
    \put(0.40334326,0.11066163){\makebox(0,0)[lt]{\lineheight{1.25}\smash{\begin{tabular}[t]{l}$x$\end{tabular}}}}%
    \put(0.08562527,0.06227787){\makebox(0,0)[lt]{\lineheight{1.25}\smash{\begin{tabular}[t]{l}$x$\end{tabular}}}}%
    \put(0.64687326,0.12114473){\makebox(0,0)[lt]{\lineheight{1.25}\smash{\begin{tabular}[t]{l}$\Gamma$\end{tabular}}}}%
    \put(0.08239973,0.09695306){\makebox(0,0)[lt]{\lineheight{1.25}\smash{\begin{tabular}[t]{l}$\Gamma$\end{tabular}}}}%
    \put(0.40414959,0.03244179){\makebox(0,0)[lt]{\lineheight{1.25}\smash{\begin{tabular}[t]{l}$x'$\end{tabular}}}}%
    \put(0.08481883,0.00341142){\makebox(0,0)[lt]{\lineheight{1.25}\smash{\begin{tabular}[t]{l}$x'$\end{tabular}}}}%
    \put(0,0){\includegraphics[width=\unitlength,page=5]{crule.pdf}}%
    \put(0.59284523,0.08405075){\makebox(0,0)[lt]{\lineheight{1.25}\smash{\begin{tabular}[t]{l}$=$\end{tabular}}}}%
    \put(-0.00146474,0.18001118){\makebox(0,0)[lt]{\lineheight{1.25}\smash{\begin{tabular}[t]{l}$\seval{\bigl(u[x/v]\bigr)[x'/v']}=$\end{tabular}}}}%
  \end{picture}%
\endgroup%

%% file: pics/delta.pdf_tex
%% Creator: Inkscape 1.2.2 (732a01da63, 2022-12-09), www.inkscape.org
%% Accompanies image file 'delta.pdf' (pdf, eps, ps)
\begingroup%
  \makeatletter%
  \providecommand\color[2][]{%
    \errmessage{(Inkscape) Color is used for the text in Inkscape, but the package 'color.sty' is not loaded}%
    \renewcommand\color[2][]{}%
  }%
  \providecommand\transparent[1]{%
    \errmessage{(Inkscape) Transparency is used (non-zero) for the text in Inkscape, but the package 'transparent.sty' is not loaded}%
    \renewcommand\transparent[1]{}%
  }%
  \providecommand\rotatebox[2]{#2}%
  \newcommand*\fsize{\dimexpr\f@size pt\relax}%
  \newcommand*\lineheight[1]{\fontsize{\fsize}{#1\fsize}\selectfont}%
  \ifx\svgwidth\undefined%
    \setlength{\unitlength}{174.58413024bp}%
    \ifx\svgscale\undefined%
      \relax%
    \else%
      \setlength{\unitlength}{\unitlength * \real{\svgscale}}%
    \fi%
  \else%
    \setlength{\unitlength}{\svgwidth}%
  \fi%
  \global\let\svgwidth\undefined%
  \global\let\svgscale\undefined%
  \makeatother%
  \begin{picture}(1,0.617653)%
    \lineheight{1}%
    \setlength\tabcolsep{0pt}%
    \put(0,0){\includegraphics[width=\unitlength,page=1]{delta.pdf}}%
    \put(0.02505978,0.19307007){\makebox(0,0)[lt]{\lineheight{1.25}\smash{\begin{tabular}[t]{l}$m$\end{tabular}}}}%
    \put(0.03078752,0.06419287){\makebox(0,0)[lt]{\lineheight{1.25}\smash{\begin{tabular}[t]{l}$n$\end{tabular}}}}%
    \put(0.20978446,0.43364299){\makebox(0,0)[lt]{\lineheight{1.25}\smash{\begin{tabular}[t]{l}$\mathit{op}$\end{tabular}}}}%
    \put(0.68520008,0.43221099){\makebox(0,0)[lt]{\lineheight{1.25}\smash{\begin{tabular}[t]{l}$\mathit{op}$\end{tabular}}}}%
    \put(0,0){\includegraphics[width=\unitlength,page=2]{delta.pdf}}%
    \put(0.20978446,0.56252056){\makebox(0,0)[lt]{\lineheight{1.25}\smash{\begin{tabular}[t]{l}$\mathit{p}$\end{tabular}}}}%
    \put(0.62076123,0.56108869){\makebox(0,0)[lt]{\lineheight{1.25}\smash{\begin{tabular}[t]{l}$\mathit{p}$\end{tabular}}}}%
    \put(0.20978446,0.30619729){\makebox(0,0)[lt]{\lineheight{1.25}\smash{\begin{tabular}[t]{l}$\mathit{op}$\end{tabular}}}}%
    \put(0.68520008,0.3047653){\makebox(0,0)[lt]{\lineheight{1.25}\smash{\begin{tabular}[t]{l}$\mathit{op}$\end{tabular}}}}%
    \put(0.63221683,0.12576748){\makebox(0,0)[lt]{\lineheight{1.25}\smash{\begin{tabular}[t]{l}$p$\end{tabular}}}}%
    \put(0,0){\includegraphics[width=\unitlength,page=3]{delta.pdf}}%
    \put(0.20119261,0.12863172){\makebox(0,0)[lt]{\lineheight{1.25}\smash{\begin{tabular}[t]{l}$\mathit{op}$\end{tabular}}}}%
    \put(0.59641754,0.00691376){\makebox(0,0)[lt]{\lineheight{1.25}\smash{\begin{tabular}[t]{l}$p=\mathit{op}(m,n)$\end{tabular}}}}%
    \put(0.42008837,0.56261062){\makebox(0,0)[lt]{\lineheight{1.25}\smash{\begin{tabular}[t]{l}$\stackrel{V}\rew$\end{tabular}}}}%
    \put(0.42008837,0.43230018){\makebox(0,0)[lt]{\lineheight{1.25}\smash{\begin{tabular}[t]{l}$\stackrel{{S}{1}}\rew$\end{tabular}}}}%
    \put(0.42008837,0.30771848){\makebox(0,0)[lt]{\lineheight{1.25}\smash{\begin{tabular}[t]{l}$\stackrel{{S}{2}}\rew$\end{tabular}}}}%
    \put(0.42008837,0.1301539){\makebox(0,0)[lt]{\lineheight{1.25}\smash{\begin{tabular}[t]{l}$\stackrel{\delta}\rew$\end{tabular}}}}%
  \end{picture}%
\endgroup%

%% file: pics/vdelta.pdf_tex
%% Creator: Inkscape 1.2.2 (732a01da63, 2022-12-09), www.inkscape.org
%% Accompanies image file 'vdelta.pdf' (pdf, eps, ps)
\begingroup%
  \makeatletter%
  \providecommand\color[2][]{%
    \errmessage{(Inkscape) Color is used for the text in Inkscape, but the package 'color.sty' is not loaded}%
    \renewcommand\color[2][]{}%
  }%
  \providecommand\transparent[1]{%
    \errmessage{(Inkscape) Transparency is used (non-zero) for the text in Inkscape, but the package 'transparent.sty' is not loaded}%
    \renewcommand\transparent[1]{}%
  }%
  \providecommand\rotatebox[2]{#2}%
  \newcommand*\fsize{\dimexpr\f@size pt\relax}%
  \newcommand*\lineheight[1]{\fontsize{\fsize}{#1\fsize}\selectfont}%
  \ifx\svgwidth\undefined%
    \setlength{\unitlength}{174.58413024bp}%
    \ifx\svgscale\undefined%
      \relax%
    \else%
      \setlength{\unitlength}{\unitlength * \real{\svgscale}}%
    \fi%
  \else%
    \setlength{\unitlength}{\svgwidth}%
  \fi%
  \global\let\svgwidth\undefined%
  \global\let\svgscale\undefined%
  \makeatother%
  \begin{picture}(1,0.2482014)%
    \lineheight{1}%
    \setlength\tabcolsep{0pt}%
    \put(0,0){\includegraphics[width=\unitlength,page=1]{vdelta.pdf}}%
    \put(0.02505978,0.19307007){\makebox(0,0)[lt]{\lineheight{1.25}\smash{\begin{tabular}[t]{l}$m$\end{tabular}}}}%
    \put(0.03078752,0.06419287){\makebox(0,0)[lt]{\lineheight{1.25}\smash{\begin{tabular}[t]{l}$n$\end{tabular}}}}%
    \put(0.63221683,0.12576748){\makebox(0,0)[lt]{\lineheight{1.25}\smash{\begin{tabular}[t]{l}$p$\end{tabular}}}}%
    \put(0,0){\includegraphics[width=\unitlength,page=2]{vdelta.pdf}}%
    \put(0.20119261,0.12863172){\makebox(0,0)[lt]{\lineheight{1.25}\smash{\begin{tabular}[t]{l}$\mathit{op}$\end{tabular}}}}%
    \put(0.59641754,0.00691376){\makebox(0,0)[lt]{\lineheight{1.25}\smash{\begin{tabular}[t]{l}$p=\mathit{op}(m,n)$\end{tabular}}}}%
    \put(0.4186565,0.12728991){\makebox(0,0)[lt]{\lineheight{1.25}\smash{\begin{tabular}[t]{l}$\rew$\end{tabular}}}}%
  \end{picture}%
\endgroup%

%% file: pics/ite.pdf_tex
%% Creator: Inkscape 1.2.2 (732a01da63, 2022-12-09), www.inkscape.org
%% Accompanies image file 'ite.pdf' (pdf, eps, ps)
\begingroup%
  \makeatletter%
  \providecommand\color[2][]{%
    \errmessage{(Inkscape) Color is used for the text in Inkscape, but the package 'color.sty' is not loaded}%
    \renewcommand\color[2][]{}%
  }%
  \providecommand\transparent[1]{%
    \errmessage{(Inkscape) Transparency is used (non-zero) for the text in Inkscape, but the package 'transparent.sty' is not loaded}%
    \renewcommand\transparent[1]{}%
  }%
  \providecommand\rotatebox[2]{#2}%
  \newcommand*\fsize{\dimexpr\f@size pt\relax}%
  \newcommand*\lineheight[1]{\fontsize{\fsize}{#1\fsize}\selectfont}%
  \ifx\svgwidth\undefined%
    \setlength{\unitlength}{307.50000378bp}%
    \ifx\svgscale\undefined%
      \relax%
    \else%
      \setlength{\unitlength}{\unitlength * \real{\svgscale}}%
    \fi%
  \else%
    \setlength{\unitlength}{\svgwidth}%
  \fi%
  \global\let\svgwidth\undefined%
  \global\let\svgscale\undefined%
  \makeatother%
  \begin{picture}(1,0.2236321)%
    \lineheight{1}%
    \setlength\tabcolsep{0pt}%
    \put(0,0){\includegraphics[width=\unitlength,page=1]{ite.pdf}}%
    \put(0.04959354,0.18089454){\makebox(0,0)[lt]{\lineheight{1.25}\smash{\begin{tabular}[t]{l}$u$\end{tabular}}}}%
    \put(0.03333326,0.01991884){\makebox(0,0)[lt]{\lineheight{1.25}\smash{\begin{tabular}[t]{l}$\mathit{true}$\end{tabular}}}}%
    \put(0.35528472,0.09390203){\makebox(0,0)[lt]{\lineheight{1.25}\smash{\begin{tabular}[t]{l}$t$\end{tabular}}}}%
    \put(0,0){\includegraphics[width=\unitlength,page=2]{ite.pdf}}%
    \put(0.16422769,0.09308958){\makebox(0,0)[lt]{\lineheight{1.25}\smash{\begin{tabular}[t]{l}$\mathit{\mathit{ite}}$\end{tabular}}}}%
    \put(0,0){\includegraphics[width=\unitlength,page=3]{ite.pdf}}%
    \put(0.05447159,0.09552864){\makebox(0,0)[lt]{\lineheight{1.25}\smash{\begin{tabular}[t]{l}$t$\end{tabular}}}}%
    \put(0,0){\includegraphics[width=\unitlength,page=4]{ite.pdf}}%
    \put(0.62520271,0.18089454){\makebox(0,0)[lt]{\lineheight{1.25}\smash{\begin{tabular}[t]{l}$u$\end{tabular}}}}%
    \put(0.60617465,0.01991884){\makebox(0,0)[lt]{\lineheight{1.25}\smash{\begin{tabular}[t]{l}$\mathit{false}$\end{tabular}}}}%
    \put(0.93089439,0.09390203){\makebox(0,0)[lt]{\lineheight{1.25}\smash{\begin{tabular}[t]{l}$u$\end{tabular}}}}%
    \put(0,0){\includegraphics[width=\unitlength,page=5]{ite.pdf}}%
    \put(0.73983693,0.09308958){\makebox(0,0)[lt]{\lineheight{1.25}\smash{\begin{tabular}[t]{l}$\mathit{\mathit{ite}}$\end{tabular}}}}%
    \put(0,0){\includegraphics[width=\unitlength,page=6]{ite.pdf}}%
    \put(0.63008084,0.09552864){\makebox(0,0)[lt]{\lineheight{1.25}\smash{\begin{tabular}[t]{l}$t$\end{tabular}}}}%
    \put(0,0){\includegraphics[width=\unitlength,page=7]{ite.pdf}}%
    \put(0.26161356,0.09180407){\makebox(0,0)[lt]{\lineheight{1.25}\smash{\begin{tabular}[t]{l}$\rew$\end{tabular}}}}%
    \put(0.84052138,0.09353997){\makebox(0,0)[lt]{\lineheight{1.25}\smash{\begin{tabular}[t]{l}$\rew$\end{tabular}}}}%
  \end{picture}%
\endgroup%

%% file: pics/iteg.pdf_tex
%% Creator: Inkscape 1.2.2 (732a01da63, 2022-12-09), www.inkscape.org
%% Accompanies image file 'iteg.pdf' (pdf, eps, ps)
\begingroup%
  \makeatletter%
  \providecommand\color[2][]{%
    \errmessage{(Inkscape) Color is used for the text in Inkscape, but the package 'color.sty' is not loaded}%
    \renewcommand\color[2][]{}%
  }%
  \providecommand\transparent[1]{%
    \errmessage{(Inkscape) Transparency is used (non-zero) for the text in Inkscape, but the package 'transparent.sty' is not loaded}%
    \renewcommand\transparent[1]{}%
  }%
  \providecommand\rotatebox[2]{#2}%
  \newcommand*\fsize{\dimexpr\f@size pt\relax}%
  \newcommand*\lineheight[1]{\fontsize{\fsize}{#1\fsize}\selectfont}%
  \ifx\svgwidth\undefined%
    \setlength{\unitlength}{155.62248823bp}%
    \ifx\svgscale\undefined%
      \relax%
    \else%
      \setlength{\unitlength}{\unitlength * \real{\svgscale}}%
    \fi%
  \else%
    \setlength{\unitlength}{\svgwidth}%
  \fi%
  \global\let\svgwidth\undefined%
  \global\let\svgscale\undefined%
  \makeatother%
  \begin{picture}(1,0.44188262)%
    \lineheight{1}%
    \setlength\tabcolsep{0pt}%
    \put(0,0){\includegraphics[width=\unitlength,page=1]{iteg.pdf}}%
    \put(0.09799363,0.35743594){\makebox(0,0)[lt]{\lineheight{1.25}\smash{\begin{tabular}[t]{l}$u$\end{tabular}}}}%
    \put(0.06586438,0.03935835){\makebox(0,0)[lt]{\lineheight{1.25}\smash{\begin{tabular}[t]{l}$\mathit{true}$\end{tabular}}}}%
    \put(0,0){\includegraphics[width=\unitlength,page=2]{iteg.pdf}}%
    \put(0.32450332,0.183939){\makebox(0,0)[lt]{\lineheight{1.25}\smash{\begin{tabular}[t]{l}$\mathit{\mathit{ite}}$\end{tabular}}}}%
    \put(0,0){\includegraphics[width=\unitlength,page=3]{iteg.pdf}}%
    \put(0.10763235,0.18875843){\makebox(0,0)[lt]{\lineheight{1.25}\smash{\begin{tabular}[t]{l}$t$\end{tabular}}}}%
    \put(0,0){\includegraphics[width=\unitlength,page=4]{iteg.pdf}}%
    \put(0.84499359,0.35743594){\makebox(0,0)[lt]{\lineheight{1.25}\smash{\begin{tabular}[t]{l}$u$\end{tabular}}}}%
    \put(0,0){\includegraphics[width=\unitlength,page=5]{iteg.pdf}}%
    \put(0.85463217,0.18875843){\makebox(0,0)[lt]{\lineheight{1.25}\smash{\begin{tabular}[t]{l}$t$\end{tabular}}}}%
    \put(0,0){\includegraphics[width=\unitlength,page=6]{iteg.pdf}}%
    \put(0.57759352,0.18599629){\makebox(0,0)[lt]{\lineheight{1.25}\smash{\begin{tabular}[t]{l}$\rew$\end{tabular}}}}%
  \end{picture}%
\endgroup%

%% file: pics/rec.pdf_tex
%% Creator: Inkscape 1.2.2 (732a01da63, 2022-12-09), www.inkscape.org
%% Accompanies image file 'rec.pdf' (pdf, eps, ps)
\begingroup%
  \makeatletter%
  \providecommand\color[2][]{%
    \errmessage{(Inkscape) Color is used for the text in Inkscape, but the package 'color.sty' is not loaded}%
    \renewcommand\color[2][]{}%
  }%
  \providecommand\transparent[1]{%
    \errmessage{(Inkscape) Transparency is used (non-zero) for the text in Inkscape, but the package 'transparent.sty' is not loaded}%
    \renewcommand\transparent[1]{}%
  }%
  \providecommand\rotatebox[2]{#2}%
  \newcommand*\fsize{\dimexpr\f@size pt\relax}%
  \newcommand*\lineheight[1]{\fontsize{\fsize}{#1\fsize}\selectfont}%
  \ifx\svgwidth\undefined%
    \setlength{\unitlength}{236.09437795bp}%
    \ifx\svgscale\undefined%
      \relax%
    \else%
      \setlength{\unitlength}{\unitlength * \real{\svgscale}}%
    \fi%
  \else%
    \setlength{\unitlength}{\svgwidth}%
  \fi%
  \global\let\svgwidth\undefined%
  \global\let\svgscale\undefined%
  \makeatother%
  \begin{picture}(1,0.11337309)%
    \lineheight{1}%
    \setlength\tabcolsep{0pt}%
    \put(0,0){\includegraphics[width=\unitlength,page=1]{rec.pdf}}%
    \put(0.08047645,0.03448486){\makebox(0,0)[lt]{\lineheight{1.25}\smash{\begin{tabular}[t]{l}$u$\end{tabular}}}}%
    \put(0,0){\includegraphics[width=\unitlength,page=2]{rec.pdf}}%
    \put(0.22978106,0.03766142){\makebox(0,0)[lt]{\lineheight{1.25}\smash{\begin{tabular}[t]{l}${\mathit{rec}}$\end{tabular}}}}%
    \put(0,0){\includegraphics[width=\unitlength,page=3]{rec.pdf}}%
    \put(0.63639832,0.03448486){\makebox(0,0)[lt]{\lineheight{1.25}\smash{\begin{tabular}[t]{l}$u$\end{tabular}}}}%
    \put(0,0){\includegraphics[width=\unitlength,page=4]{rec.pdf}}%
    \put(0.90641722,0.0344845){\makebox(0,0)[lt]{\lineheight{1.25}\smash{\begin{tabular}[t]{l}$u$\end{tabular}}}}%
    \put(0,0){\includegraphics[width=\unitlength,page=5]{rec.pdf}}%
    \put(0.73645589,0.03614468){\makebox(0,0)[lt]{\lineheight{1.25}\smash{\begin{tabular}[t]{l}${\mathit{rec}}$\end{tabular}}}}%
    \put(0.39173788,0.03667451){\makebox(0,0)[lt]{\lineheight{1.25}\smash{\begin{tabular}[t]{l}$\rew$\end{tabular}}}}%
  \end{picture}%
\endgroup%

%% file: pics/callcc.pdf_tex
%% Creator: Inkscape 1.2.2 (732a01da63, 2022-12-09), www.inkscape.org
%% Accompanies image file 'callcc.pdf' (pdf, eps, ps)
\begingroup%
  \makeatletter%
  \providecommand\color[2][]{%
    \errmessage{(Inkscape) Color is used for the text in Inkscape, but the package 'color.sty' is not loaded}%
    \renewcommand\color[2][]{}%
  }%
  \providecommand\transparent[1]{%
    \errmessage{(Inkscape) Transparency is used (non-zero) for the text in Inkscape, but the package 'transparent.sty' is not loaded}%
    \renewcommand\transparent[1]{}%
  }%
  \providecommand\rotatebox[2]{#2}%
  \newcommand*\fsize{\dimexpr\f@size pt\relax}%
  \newcommand*\lineheight[1]{\fontsize{\fsize}{#1\fsize}\selectfont}%
  \ifx\svgwidth\undefined%
    \setlength{\unitlength}{306.67281165bp}%
    \ifx\svgscale\undefined%
      \relax%
    \else%
      \setlength{\unitlength}{\unitlength * \real{\svgscale}}%
    \fi%
  \else%
    \setlength{\unitlength}{\svgwidth}%
  \fi%
  \global\let\svgwidth\undefined%
  \global\let\svgscale\undefined%
  \makeatother%
  \begin{picture}(1,0.23401706)%
    \lineheight{1}%
    \setlength\tabcolsep{0pt}%
    \put(0,0){\includegraphics[width=\unitlength,page=1]{callcc.pdf}}%
    \put(0.09594199,0.06323234){\makebox(0,0)[lt]{\lineheight{1.25}\smash{\begin{tabular}[t]{l}$u$\end{tabular}}}}%
    \put(0,0){\includegraphics[width=\unitlength,page=2]{callcc.pdf}}%
    \put(0.19865728,0.0656777){\makebox(0,0)[lt]{\lineheight{1.25}\smash{\begin{tabular}[t]{l}${\mathit{call{/}cc}}$\end{tabular}}}}%
    \put(0,0){\includegraphics[width=\unitlength,page=3]{callcc.pdf}}%
    \put(0.34697467,0.06633325){\makebox(0,0)[lt]{\lineheight{1.25}\smash{\begin{tabular}[t]{l}$t$\end{tabular}}}}%
    \put(0,0){\includegraphics[width=\unitlength,page=4]{callcc.pdf}}%
    \put(0.0928848,0.17328434){\makebox(0,0)[lt]{\lineheight{1.25}\smash{\begin{tabular}[t]{l}$u$\end{tabular}}}}%
    \put(0,0){\includegraphics[width=\unitlength,page=5]{callcc.pdf}}%
    \put(0.20538223,0.1757297){\makebox(0,0)[lt]{\lineheight{1.25}\smash{\begin{tabular}[t]{l}${\mathit{abort}}$\end{tabular}}}}%
    \put(0,0){\includegraphics[width=\unitlength,page=6]{callcc.pdf}}%
    \put(0.5330933,0.17328434){\makebox(0,0)[lt]{\lineheight{1.25}\smash{\begin{tabular}[t]{l}$u$\end{tabular}}}}%
    \put(0,0){\includegraphics[width=\unitlength,page=7]{callcc.pdf}}%
    \put(0.34391748,0.17638525){\makebox(0,0)[lt]{\lineheight{1.25}\smash{\begin{tabular}[t]{l}$t$\end{tabular}}}}%
    \put(0,0){\includegraphics[width=\unitlength,page=8]{callcc.pdf}}%
    \put(0.70978845,0.02899379){\makebox(0,0)[lt]{\lineheight{1.25}\smash{\begin{tabular}[t]{l}${\mathit{abort}}$\end{tabular}}}}%
    \put(0,0){\includegraphics[width=\unitlength,page=9]{callcc.pdf}}%
    \put(0.854079,0.03877621){\makebox(0,0)[lt]{\lineheight{1.25}\smash{\begin{tabular}[t]{l}$u$\end{tabular}}}}%
    \put(0,0){\includegraphics[width=\unitlength,page=10]{callcc.pdf}}%
    \put(0.62821913,0.02964935){\makebox(0,0)[lt]{\lineheight{1.25}\smash{\begin{tabular}[t]{l}$t$\end{tabular}}}}%
    \put(0,0){\includegraphics[width=\unitlength,page=11]{callcc.pdf}}%
    \put(0.94614736,0.05410519){\makebox(0,0)[lt]{\lineheight{1.25}\smash{\begin{tabular}[t]{l}$t$\end{tabular}}}}%
    \put(0,0){\includegraphics[width=\unitlength,page=12]{callcc.pdf}}%
    \put(0.43020939,0.17539741){\makebox(0,0)[lt]{\lineheight{1.25}\smash{\begin{tabular}[t]{l}$\rew$\end{tabular}}}}%
    \put(0.43020939,0.06733591){\makebox(0,0)[lt]{\lineheight{1.25}\smash{\begin{tabular}[t]{l}$\rew$\end{tabular}}}}%
  \end{picture}%
\endgroup%

%% file: pics/store.pdf_tex
%% Creator: Inkscape 1.2.2 (732a01da63, 2022-12-09), www.inkscape.org
%% Accompanies image file 'store.pdf' (pdf, eps, ps)
\begingroup%
  \makeatletter%
  \providecommand\color[2][]{%
    \errmessage{(Inkscape) Color is used for the text in Inkscape, but the package 'color.sty' is not loaded}%
    \renewcommand\color[2][]{}%
  }%
  \providecommand\transparent[1]{%
    \errmessage{(Inkscape) Transparency is used (non-zero) for the text in Inkscape, but the package 'transparent.sty' is not loaded}%
    \renewcommand\transparent[1]{}%
  }%
  \providecommand\rotatebox[2]{#2}%
  \newcommand*\fsize{\dimexpr\f@size pt\relax}%
  \newcommand*\lineheight[1]{\fontsize{\fsize}{#1\fsize}\selectfont}%
  \ifx\svgwidth\undefined%
    \setlength{\unitlength}{212.41774488bp}%
    \ifx\svgscale\undefined%
      \relax%
    \else%
      \setlength{\unitlength}{\unitlength * \real{\svgscale}}%
    \fi%
  \else%
    \setlength{\unitlength}{\svgwidth}%
  \fi%
  \global\let\svgwidth\undefined%
  \global\let\svgscale\undefined%
  \makeatother%
  \begin{picture}(1,0.21537759)%
    \lineheight{1}%
    \setlength\tabcolsep{0pt}%
    \put(0,0){\includegraphics[width=\unitlength,page=1]{store.pdf}}%
    \put(0.2632556,0.07709219){\makebox(0,0)[lt]{\lineheight{1.25}\smash{\begin{tabular}[t]{l}${\mathit{newvar}}$\end{tabular}}}}%
    \put(0,0){\includegraphics[width=\unitlength,page=2]{store.pdf}}%
    \put(0.12289573,0.07473505){\makebox(0,0)[lt]{\lineheight{1.25}\smash{\begin{tabular}[t]{l}$u$\end{tabular}}}}%
    \put(0,0){\includegraphics[width=\unitlength,page=3]{store.pdf}}%
    \put(0.88201286,0.0923884){\makebox(0,0)[lt]{\lineheight{1.25}\smash{\begin{tabular}[t]{l}$u$\end{tabular}}}}%
    \put(0,0){\includegraphics[width=\unitlength,page=4]{store.pdf}}%
    \put(0.63485864,0.02177276){\makebox(0,0)[lt]{\lineheight{1.25}\smash{\begin{tabular}[t]{l}$a$\end{tabular}}}}%
    \put(0,0){\includegraphics[width=\unitlength,page=5]{store.pdf}}%
    \put(0.75843616,0.07473444){\makebox(0,0)[lt]{\lineheight{1.25}\smash{\begin{tabular}[t]{l}$\mapsto$\end{tabular}}}}%
    \put(0,0){\includegraphics[width=\unitlength,page=6]{store.pdf}}%
    \put(0.63485864,0.12769591){\makebox(0,0)[lt]{\lineheight{1.25}\smash{\begin{tabular}[t]{l}$v$\end{tabular}}}}%
    \put(0,0){\includegraphics[width=\unitlength,page=7]{store.pdf}}%
    \put(0.48927084,0.07516042){\makebox(0,0)[lt]{\lineheight{1.25}\smash{\begin{tabular}[t]{l}$\rew$\end{tabular}}}}%
  \end{picture}%
\endgroup%

%% file: pics/deref.pdf_tex
%% Creator: Inkscape 1.2.2 (732a01da63, 2022-12-09), www.inkscape.org
%% Accompanies image file 'deref.pdf' (pdf, eps, ps)
\begingroup%
  \makeatletter%
  \providecommand\color[2][]{%
    \errmessage{(Inkscape) Color is used for the text in Inkscape, but the package 'color.sty' is not loaded}%
    \renewcommand\color[2][]{}%
  }%
  \providecommand\transparent[1]{%
    \errmessage{(Inkscape) Transparency is used (non-zero) for the text in Inkscape, but the package 'transparent.sty' is not loaded}%
    \renewcommand\transparent[1]{}%
  }%
  \providecommand\rotatebox[2]{#2}%
  \newcommand*\fsize{\dimexpr\f@size pt\relax}%
  \newcommand*\lineheight[1]{\fontsize{\fsize}{#1\fsize}\selectfont}%
  \ifx\svgwidth\undefined%
    \setlength{\unitlength}{161.62495402bp}%
    \ifx\svgscale\undefined%
      \relax%
    \else%
      \setlength{\unitlength}{\unitlength * \real{\svgscale}}%
    \fi%
  \else%
    \setlength{\unitlength}{\svgwidth}%
  \fi%
  \global\let\svgwidth\undefined%
  \global\let\svgscale\undefined%
  \makeatother%
  \begin{picture}(1,0.23665872)%
    \lineheight{1}%
    \setlength\tabcolsep{0pt}%
    \put(0,0){\includegraphics[width=\unitlength,page=1]{deref.pdf}}%
    \put(0.02706868,0.02861514){\makebox(0,0)[lt]{\lineheight{1.25}\smash{\begin{tabular}[t]{l}$a$\end{tabular}}}}%
    \put(0,0){\includegraphics[width=\unitlength,page=2]{deref.pdf}}%
    \put(0.18948181,0.09822073){\makebox(0,0)[lt]{\lineheight{1.25}\smash{\begin{tabular}[t]{l}$\mapsto$\end{tabular}}}}%
    \put(0,0){\includegraphics[width=\unitlength,page=3]{deref.pdf}}%
    \put(0.02706868,0.16782605){\makebox(0,0)[lt]{\lineheight{1.25}\smash{\begin{tabular}[t]{l}$v$\end{tabular}}}}%
    \put(0,0){\includegraphics[width=\unitlength,page=4]{deref.pdf}}%
    \put(0.86233561,0.09822073){\makebox(0,0)[lt]{\lineheight{1.25}\smash{\begin{tabular}[t]{l}$v$\end{tabular}}}}%
    \put(0,0){\includegraphics[width=\unitlength,page=5]{deref.pdf}}%
    \put(0.3739511,0.1013189){\makebox(0,0)[lt]{\lineheight{1.25}\smash{\begin{tabular}[t]{l}${\mathit{deref}}$\end{tabular}}}}%
    \put(0.66486732,0.10403359){\makebox(0,0)[lt]{\lineheight{1.25}\smash{\begin{tabular}[t]{l}$\rew$\end{tabular}}}}%
  \end{picture}%
\endgroup%

%% file: pics/asg.pdf_tex
%% Creator: Inkscape 1.2.2 (732a01da63, 2022-12-09), www.inkscape.org
%% Accompanies image file 'asg.pdf' (pdf, eps, ps)
\begingroup%
  \makeatletter%
  \providecommand\color[2][]{%
    \errmessage{(Inkscape) Color is used for the text in Inkscape, but the package 'color.sty' is not loaded}%
    \renewcommand\color[2][]{}%
  }%
  \providecommand\transparent[1]{%
    \errmessage{(Inkscape) Transparency is used (non-zero) for the text in Inkscape, but the package 'transparent.sty' is not loaded}%
    \renewcommand\transparent[1]{}%
  }%
  \providecommand\rotatebox[2]{#2}%
  \newcommand*\fsize{\dimexpr\f@size pt\relax}%
  \newcommand*\lineheight[1]{\fontsize{\fsize}{#1\fsize}\selectfont}%
  \ifx\svgwidth\undefined%
    \setlength{\unitlength}{209.12215058bp}%
    \ifx\svgscale\undefined%
      \relax%
    \else%
      \setlength{\unitlength}{\unitlength * \real{\svgscale}}%
    \fi%
  \else%
    \setlength{\unitlength}{\svgwidth}%
  \fi%
  \global\let\svgwidth\undefined%
  \global\let\svgscale\undefined%
  \makeatother%
  \begin{picture}(1,0.52361727)%
    \lineheight{1}%
    \setlength\tabcolsep{0pt}%
    \put(0,0){\includegraphics[width=\unitlength,page=1]{asg.pdf}}%
    \put(0.3392508,0.36402864){\makebox(0,0)[lt]{\lineheight{1.25}\smash{\begin{tabular}[t]{l}${\mathit{assign}}$\end{tabular}}}}%
    \put(0,0){\includegraphics[width=\unitlength,page=2]{asg.pdf}}%
    \put(0.02092087,0.25523329){\makebox(0,0)[lt]{\lineheight{1.25}\smash{\begin{tabular}[t]{l}$a$\end{tabular}}}}%
    \put(0,0){\includegraphics[width=\unitlength,page=3]{asg.pdf}}%
    \put(0.14644566,0.3090296){\makebox(0,0)[lt]{\lineheight{1.25}\smash{\begin{tabular}[t]{l}$\mapsto$\end{tabular}}}}%
    \put(0,0){\includegraphics[width=\unitlength,page=4]{asg.pdf}}%
    \put(0.02092087,0.3628257){\makebox(0,0)[lt]{\lineheight{1.25}\smash{\begin{tabular}[t]{l}$v$\end{tabular}}}}%
    \put(0,0){\includegraphics[width=\unitlength,page=5]{asg.pdf}}%
    \put(0.02092046,0.47041832){\makebox(0,0)[lt]{\lineheight{1.25}\smash{\begin{tabular}[t]{l}$v'$\end{tabular}}}}%
    \put(0,0){\includegraphics[width=\unitlength,page=6]{asg.pdf}}%
    \put(0.72027269,0.36282591){\makebox(0,0)[lt]{\lineheight{1.25}\smash{\begin{tabular}[t]{l}$v_0$\end{tabular}}}}%
    \put(0,0){\includegraphics[width=\unitlength,page=7]{asg.pdf}}%
    \put(0.28990232,0.02211588){\makebox(0,0)[lt]{\lineheight{1.25}\smash{\begin{tabular}[t]{l}$a$\end{tabular}}}}%
    \put(0,0){\includegraphics[width=\unitlength,page=8]{asg.pdf}}%
    \put(0.41542742,0.07591219){\makebox(0,0)[lt]{\lineheight{1.25}\smash{\begin{tabular}[t]{l}$\mapsto$\end{tabular}}}}%
    \put(0,0){\includegraphics[width=\unitlength,page=9]{asg.pdf}}%
    \put(0.28990232,0.12970829){\makebox(0,0)[lt]{\lineheight{1.25}\smash{\begin{tabular}[t]{l}$v$\end{tabular}}}}%
    \put(0,0){\includegraphics[width=\unitlength,page=10]{asg.pdf}}%
    \put(0.72027269,0.02211588){\makebox(0,0)[lt]{\lineheight{1.25}\smash{\begin{tabular}[t]{l}$a$\end{tabular}}}}%
    \put(0,0){\includegraphics[width=\unitlength,page=11]{asg.pdf}}%
    \put(0.84579769,0.07591219){\makebox(0,0)[lt]{\lineheight{1.25}\smash{\begin{tabular}[t]{l}$\mapsto$\end{tabular}}}}%
    \put(0,0){\includegraphics[width=\unitlength,page=12]{asg.pdf}}%
    \put(0.72027269,0.12970829){\makebox(0,0)[lt]{\lineheight{1.25}\smash{\begin{tabular}[t]{l}$v'$\end{tabular}}}}%
    \put(0,0){\includegraphics[width=\unitlength,page=13]{asg.pdf}}%
    \put(0.57399288,0.083777){\makebox(0,0)[lt]{\lineheight{1.25}\smash{\begin{tabular}[t]{l}$\rew$\end{tabular}}}}%
    \put(0.56565236,0.36854607){\makebox(0,0)[lt]{\lineheight{1.25}\smash{\begin{tabular}[t]{l}$\rew$\end{tabular}}}}%
  \end{picture}%
\endgroup%

%% file: pics/atom1.pdf_tex
%% Creator: Inkscape 1.1 (c68e22c387, 2021-05-23), www.inkscape.org
%% Accompanies image file 'atom1.pdf' (pdf, eps, ps)
\begingroup%
  \makeatletter%
  \providecommand\color[2][]{%
    \errmessage{(Inkscape) Color is used for the text in Inkscape, but the package 'color.sty' is not loaded}%
    \renewcommand\color[2][]{}%
  }%
  \providecommand\transparent[1]{%
    \errmessage{(Inkscape) Transparency is used (non-zero) for the text in Inkscape, but the package 'transparent.sty' is not loaded}%
    \renewcommand\transparent[1]{}%
  }%
  \providecommand\rotatebox[2]{#2}%
  \newcommand*\fsize{\dimexpr\f@size pt\relax}%
  \newcommand*\lineheight[1]{\fontsize{\fsize}{#1\fsize}\selectfont}%
  \ifx\svgwidth\undefined%
    \setlength{\unitlength}{184.12498062bp}%
    \ifx\svgscale\undefined%
      \relax%
    \else%
      \setlength{\unitlength}{\unitlength * \real{\svgscale}}%
    \fi%
  \else%
    \setlength{\unitlength}{\svgwidth}%
  \fi%
  \global\let\svgwidth\undefined%
  \global\let\svgscale\undefined%
  \makeatother%
  \begin{picture}(1,0.45213806)%
    \lineheight{1}%
    \setlength\tabcolsep{0pt}%
    \put(0,0){\includegraphics[width=\unitlength,page=1]{atom1.pdf}}%
    \put(0.4779163,0.20970322){\makebox(0,0)[lt]{\lineheight{1.25}\smash{\begin{tabular}[t]{l}$=$\end{tabular}}}}%
    \put(0,0){\includegraphics[width=\unitlength,page=2]{atom1.pdf}}%
    \put(0.36999303,0.20841755){\makebox(0,0)[lt]{\lineheight{1.25}\smash{\begin{tabular}[t]{l}$t$\end{tabular}}}}%
    \put(0,0){\includegraphics[width=\unitlength,page=3]{atom1.pdf}}%
    \put(0.02376101,0.26951735){\makebox(0,0)[lt]{\lineheight{1.25}\smash{\begin{tabular}[t]{l}$a$\end{tabular}}}}%
    \put(0,0){\includegraphics[width=\unitlength,page=4]{atom1.pdf}}%
    \put(0.16632765,0.33061715){\makebox(0,0)[lt]{\lineheight{1.25}\smash{\begin{tabular}[t]{l}$\mapsto$\end{tabular}}}}%
    \put(0,0){\includegraphics[width=\unitlength,page=5]{atom1.pdf}}%
    \put(0.02376101,0.39171672){\makebox(0,0)[lt]{\lineheight{1.25}\smash{\begin{tabular}[t]{l}$v$\end{tabular}}}}%
    \put(0,0){\includegraphics[width=\unitlength,page=6]{atom1.pdf}}%
    \put(0.0237609,0.02511837){\makebox(0,0)[lt]{\lineheight{1.25}\smash{\begin{tabular}[t]{l}$a$\end{tabular}}}}%
    \put(0,0){\includegraphics[width=\unitlength,page=7]{atom1.pdf}}%
    \put(0.16632765,0.08621818){\makebox(0,0)[lt]{\lineheight{1.25}\smash{\begin{tabular}[t]{l}$\mapsto$\end{tabular}}}}%
    \put(0,0){\includegraphics[width=\unitlength,page=8]{atom1.pdf}}%
    \put(0.0237609,0.14731775){\makebox(0,0)[lt]{\lineheight{1.25}\smash{\begin{tabular}[t]{l}$v$\end{tabular}}}}%
    \put(0,0){\includegraphics[width=\unitlength,page=9]{atom1.pdf}}%
    \put(0.89952481,0.20841778){\makebox(0,0)[lt]{\lineheight{1.25}\smash{\begin{tabular}[t]{l}$t$\end{tabular}}}}%
    \put(0,0){\includegraphics[width=\unitlength,page=10]{atom1.pdf}}%
    \put(0.57365912,0.14731798){\makebox(0,0)[lt]{\lineheight{1.25}\smash{\begin{tabular}[t]{l}$a$\end{tabular}}}}%
    \put(0,0){\includegraphics[width=\unitlength,page=11]{atom1.pdf}}%
    \put(0.71622588,0.20841778){\makebox(0,0)[lt]{\lineheight{1.25}\smash{\begin{tabular}[t]{l}$\mapsto$\end{tabular}}}}%
    \put(0,0){\includegraphics[width=\unitlength,page=12]{atom1.pdf}}%
    \put(0.57365912,0.26951735){\makebox(0,0)[lt]{\lineheight{1.25}\smash{\begin{tabular}[t]{l}$v$\end{tabular}}}}%
    \put(0,0){\includegraphics[width=\unitlength,page=13]{atom1.pdf}}%
  \end{picture}%
\endgroup%

%% file: pics/atom2.pdf_tex
%% Creator: Inkscape 1.1 (c68e22c387, 2021-05-23), www.inkscape.org
%% Accompanies image file 'atom2.pdf' (pdf, eps, ps)
\begingroup%
  \makeatletter%
  \providecommand\color[2][]{%
    \errmessage{(Inkscape) Color is used for the text in Inkscape, but the package 'color.sty' is not loaded}%
    \renewcommand\color[2][]{}%
  }%
  \providecommand\transparent[1]{%
    \errmessage{(Inkscape) Transparency is used (non-zero) for the text in Inkscape, but the package 'transparent.sty' is not loaded}%
    \renewcommand\transparent[1]{}%
  }%
  \providecommand\rotatebox[2]{#2}%
  \newcommand*\fsize{\dimexpr\f@size pt\relax}%
  \newcommand*\lineheight[1]{\fontsize{\fsize}{#1\fsize}\selectfont}%
  \ifx\svgwidth\undefined%
    \setlength{\unitlength}{56.62498111bp}%
    \ifx\svgscale\undefined%
      \relax%
    \else%
      \setlength{\unitlength}{\unitlength * \real{\svgscale}}%
    \fi%
  \else%
    \setlength{\unitlength}{\svgwidth}%
  \fi%
  \global\let\svgwidth\undefined%
  \global\let\svgscale\undefined%
  \makeatother%
  \begin{picture}(1,0.27814653)%
    \lineheight{1}%
    \setlength\tabcolsep{0pt}%
    \put(0,0){\includegraphics[width=\unitlength,page=1]{atom2.pdf}}%
    \put(0.67328896,0.08167709){\makebox(0,0)[lt]{\lineheight{1.25}\smash{\begin{tabular}[t]{l}$t$\end{tabular}}}}%
    \put(0,0){\includegraphics[width=\unitlength,page=2]{atom2.pdf}}%
    \put(0.07726265,0.08167632){\makebox(0,0)[lt]{\lineheight{1.25}\smash{\begin{tabular}[t]{l}$v$\end{tabular}}}}%
    \put(0,0){\includegraphics[width=\unitlength,page=3]{atom2.pdf}}%
  \end{picture}%
\endgroup%

%% file: pics/garbage.pdf_tex
%% Creator: Inkscape 1.2.2 (732a01da63, 2022-12-09), www.inkscape.org
%% Accompanies image file 'garbage.pdf' (pdf, eps, ps)
\begingroup%
  \makeatletter%
  \providecommand\color[2][]{%
    \errmessage{(Inkscape) Color is used for the text in Inkscape, but the package 'color.sty' is not loaded}%
    \renewcommand\color[2][]{}%
  }%
  \providecommand\transparent[1]{%
    \errmessage{(Inkscape) Transparency is used (non-zero) for the text in Inkscape, but the package 'transparent.sty' is not loaded}%
    \renewcommand\transparent[1]{}%
  }%
  \providecommand\rotatebox[2]{#2}%
  \newcommand*\fsize{\dimexpr\f@size pt\relax}%
  \newcommand*\lineheight[1]{\fontsize{\fsize}{#1\fsize}\selectfont}%
  \ifx\svgwidth\undefined%
    \setlength{\unitlength}{162.75003554bp}%
    \ifx\svgscale\undefined%
      \relax%
    \else%
      \setlength{\unitlength}{\unitlength * \real{\svgscale}}%
    \fi%
  \else%
    \setlength{\unitlength}{\svgwidth}%
  \fi%
  \global\let\svgwidth\undefined%
  \global\let\svgscale\undefined%
  \makeatother%
  \begin{picture}(1,0.41935492)%
    \lineheight{1}%
    \setlength\tabcolsep{0pt}%
    \put(0,0){\includegraphics[width=\unitlength,page=1]{garbage.pdf}}%
    \put(0.32758746,0.10061876){\makebox(0,0)[lt]{\lineheight{1.25}\smash{\begin{tabular}[t]{l}${\mathit{deref}}$\end{tabular}}}}%
    \put(0,0){\includegraphics[width=\unitlength,page=2]{garbage.pdf}}%
    \put(0.04416266,0.02841732){\makebox(0,0)[lt]{\lineheight{1.25}\smash{\begin{tabular}[t]{l}$v$\end{tabular}}}}%
    \put(0,0){\includegraphics[width=\unitlength,page=3]{garbage.pdf}}%
    \put(0.80453156,0.02841759){\makebox(0,0)[lt]{\lineheight{1.25}\smash{\begin{tabular}[t]{l}$v$\end{tabular}}}}%
    \put(0,0){\includegraphics[width=\unitlength,page=4]{garbage.pdf}}%
    \put(0.80453156,0.14362485){\makebox(0,0)[lt]{\lineheight{1.25}\smash{\begin{tabular}[t]{l}$v$\end{tabular}}}}%
    \put(0,0){\includegraphics[width=\unitlength,page=5]{garbage.pdf}}%
    \put(0.30454603,0.3540751){\makebox(0,0)[lt]{\lineheight{1.25}\smash{\begin{tabular}[t]{l}${\mathit{deref}}$\end{tabular}}}}%
    \put(0,0){\includegraphics[width=\unitlength,page=6]{garbage.pdf}}%
    \put(0.02688169,0.35099807){\makebox(0,0)[lt]{\lineheight{1.25}\smash{\begin{tabular}[t]{l}$v$\end{tabular}}}}%
    \put(0,0){\includegraphics[width=\unitlength,page=7]{garbage.pdf}}%
    \put(0.77834626,0.36025068){\makebox(0,0)[lt]{\lineheight{1.25}\smash{\begin{tabular}[t]{l}$???$\end{tabular}}}}%
    \put(0.29608275,0.23271884){\makebox(0,0)[lt]{\lineheight{1.25}\smash{\begin{tabular}[t]{l}$=$\end{tabular}}}}%
    \put(0.61955753,0.36209881){\makebox(0,0)[lt]{\lineheight{1.25}\smash{\begin{tabular}[t]{l}$\rew$\end{tabular}}}}%
    \put(0.61649566,0.08345778){\makebox(0,0)[lt]{\lineheight{1.25}\smash{\begin{tabular}[t]{l}$\rew$\end{tabular}}}}%
  \end{picture}%
\endgroup%

%% file: pics/unifgraph3.pdf_tex
%% Creator: Inkscape 1.1 (c68e22c387, 2021-05-23), www.inkscape.org
%% Accompanies image file 'unifgraph3.pdf' (pdf, eps, ps)
\begingroup%
  \makeatletter%
  \providecommand\color[2][]{%
    \errmessage{(Inkscape) Color is used for the text in Inkscape, but the package 'color.sty' is not loaded}%
    \renewcommand\color[2][]{}%
  }%
  \providecommand\transparent[1]{%
    \errmessage{(Inkscape) Transparency is used (non-zero) for the text in Inkscape, but the package 'transparent.sty' is not loaded}%
    \renewcommand\transparent[1]{}%
  }%
  \providecommand\rotatebox[2]{#2}%
  \newcommand*\fsize{\dimexpr\f@size pt\relax}%
  \newcommand*\lineheight[1]{\fontsize{\fsize}{#1\fsize}\selectfont}%
  \ifx\svgwidth\undefined%
    \setlength{\unitlength}{171.70559709bp}%
    \ifx\svgscale\undefined%
      \relax%
    \else%
      \setlength{\unitlength}{\unitlength * \real{\svgscale}}%
    \fi%
  \else%
    \setlength{\unitlength}{\svgwidth}%
  \fi%
  \global\let\svgwidth\undefined%
  \global\let\svgscale\undefined%
  \makeatother%
  \begin{picture}(1,0.46839944)%
    \lineheight{1}%
    \setlength\tabcolsep{0pt}%
    \put(0,0){\includegraphics[width=\unitlength,page=1]{unifgraph3.pdf}}%
    \put(0.13212579,0.24059585){\makebox(0,0)[lt]{\lineheight{1.25}\smash{\begin{tabular}[t]{l}$\alpha_1$\end{tabular}}}}%
    \put(0,0){\includegraphics[width=\unitlength,page=2]{unifgraph3.pdf}}%
    \put(-0.00264466,0.24059585){\makebox(0,0)[lt]{\lineheight{1.25}\smash{\begin{tabular}[t]{l}$\alpha_0$\end{tabular}}}}%
    \put(0,0){\includegraphics[width=\unitlength,page=3]{unifgraph3.pdf}}%
    \put(0.45408342,0.24059585){\makebox(0,0)[lt]{\lineheight{1.25}\smash{\begin{tabular}[t]{l}$\alpha_3$\end{tabular}}}}%
    \put(0,0){\includegraphics[width=\unitlength,page=4]{unifgraph3.pdf}}%
    \put(0.27563213,0.24059585){\makebox(0,0)[lt]{\lineheight{1.25}\smash{\begin{tabular}[t]{l}$\alpha_2$\end{tabular}}}}%
    \put(0,0){\includegraphics[width=\unitlength,page=5]{unifgraph3.pdf}}%
    \put(0.80224797,0.24059585){\makebox(0,0)[lt]{\lineheight{1.25}\smash{\begin{tabular}[t]{l}$\alpha_5$\end{tabular}}}}%
    \put(0,0){\includegraphics[width=\unitlength,page=6]{unifgraph3.pdf}}%
    \put(0.65000498,0.24059585){\makebox(0,0)[lt]{\lineheight{1.25}\smash{\begin{tabular}[t]{l}$\alpha_4$\end{tabular}}}}%
    \put(0,0){\includegraphics[width=\unitlength,page=7]{unifgraph3.pdf}}%
    \put(0.32062703,0.00615949){\makebox(0,0)[lt]{\lineheight{1.25}\smash{\begin{tabular}[t]{l}$N$\end{tabular}}}}%
    \put(0,0){\includegraphics[width=\unitlength,page=8]{unifgraph3.pdf}}%
    \put(0.47302028,0.00615949){\makebox(0,0)[lt]{\lineheight{1.25}\smash{\begin{tabular}[t]{l}$B$\end{tabular}}}}%
    \put(0,0){\includegraphics[width=\unitlength,page=9]{unifgraph3.pdf}}%
    \put(0.08041451,0.44040542){\makebox(0,0)[lt]{\lineheight{1.25}\smash{\begin{tabular}[t]{l}$\alpha_1\to\alpha_2$\end{tabular}}}}%
    \put(0,0){\includegraphics[width=\unitlength,page=10]{unifgraph3.pdf}}%
    \put(0.53851597,0.44150926){\makebox(0,0)[lt]{\lineheight{1.25}\smash{\begin{tabular}[t]{l}$\alpha_4\to\alpha_5$\end{tabular}}}}%
    \put(0,0){\includegraphics[width=\unitlength,page=11]{unifgraph3.pdf}}%
  \end{picture}%
\endgroup%

%% file: pics/sbstlam2.pdf_tex
%% Creator: Inkscape 1.2.2 (732a01da63, 2022-12-09), www.inkscape.org
%% Accompanies image file 'sbstlam2.pdf' (pdf, eps, ps)
\begingroup%
  \makeatletter%
  \providecommand\color[2][]{%
    \errmessage{(Inkscape) Color is used for the text in Inkscape, but the package 'color.sty' is not loaded}%
    \renewcommand\color[2][]{}%
  }%
  \providecommand\transparent[1]{%
    \errmessage{(Inkscape) Transparency is used (non-zero) for the text in Inkscape, but the package 'transparent.sty' is not loaded}%
    \renewcommand\transparent[1]{}%
  }%
  \providecommand\rotatebox[2]{#2}%
  \newcommand*\fsize{\dimexpr\f@size pt\relax}%
  \newcommand*\lineheight[1]{\fontsize{\fsize}{#1\fsize}\selectfont}%
  \ifx\svgwidth\undefined%
    \setlength{\unitlength}{245.2748315bp}%
    \ifx\svgscale\undefined%
      \relax%
    \else%
      \setlength{\unitlength}{\unitlength * \real{\svgscale}}%
    \fi%
  \else%
    \setlength{\unitlength}{\svgwidth}%
  \fi%
  \global\let\svgwidth\undefined%
  \global\let\svgscale\undefined%
  \makeatother%
  \begin{picture}(1,0.20639529)%
    \lineheight{1}%
    \setlength\tabcolsep{0pt}%
    \put(0,0){\includegraphics[width=\unitlength,page=1]{sbstlam2.pdf}}%
    \put(0.90138827,0.10109402){\color[rgb]{0,0,0}\makebox(0,0)[t]{\lineheight{1.25}\smash{\begin{tabular}[t]{c}$u$\end{tabular}}}}%
    \put(0,0){\includegraphics[width=\unitlength,page=2]{sbstlam2.pdf}}%
    \put(0.33021321,0.10024113){\color[rgb]{0,0,0}\makebox(0,0)[t]{\lineheight{1.25}\smash{\begin{tabular}[t]{c}$u$\end{tabular}}}}%
    \put(0,0){\includegraphics[width=\unitlength,page=3]{sbstlam2.pdf}}%
    \put(0.16149972,0.0989903){\color[rgb]{0,0,0}\makebox(0,0)[t]{\lineheight{1.25}\smash{\begin{tabular}[t]{c}$v$\end{tabular}}}}%
    \put(0,0){\includegraphics[width=\unitlength,page=4]{sbstlam2.pdf}}%
    \put(0.44969158,0.0988569){\makebox(0,0)[lt]{\lineheight{1.25}\smash{\begin{tabular}[t]{l}$=$\end{tabular}}}}%
    \put(0,0){\includegraphics[width=\unitlength,page=5]{sbstlam2.pdf}}%
    \put(0.6367589,0.0684215){\color[rgb]{0,0,0}\makebox(0,0)[t]{\lineheight{1.25}\smash{\begin{tabular}[t]{c}$v$\end{tabular}}}}%
    \put(0,0){\includegraphics[width=\unitlength,page=6]{sbstlam2.pdf}}%
  \end{picture}%
\endgroup%

%% file: pics/sbstlam3.pdf_tex
%% Creator: Inkscape 1.2.2 (732a01da63, 2022-12-09), www.inkscape.org
%% Accompanies image file 'sbstlam3.pdf' (pdf, eps, ps)
\begingroup%
  \makeatletter%
  \providecommand\color[2][]{%
    \errmessage{(Inkscape) Color is used for the text in Inkscape, but the package 'color.sty' is not loaded}%
    \renewcommand\color[2][]{}%
  }%
  \providecommand\transparent[1]{%
    \errmessage{(Inkscape) Transparency is used (non-zero) for the text in Inkscape, but the package 'transparent.sty' is not loaded}%
    \renewcommand\transparent[1]{}%
  }%
  \providecommand\rotatebox[2]{#2}%
  \newcommand*\fsize{\dimexpr\f@size pt\relax}%
  \newcommand*\lineheight[1]{\fontsize{\fsize}{#1\fsize}\selectfont}%
  \ifx\svgwidth\undefined%
    \setlength{\unitlength}{284.74239685bp}%
    \ifx\svgscale\undefined%
      \relax%
    \else%
      \setlength{\unitlength}{\unitlength * \real{\svgscale}}%
    \fi%
  \else%
    \setlength{\unitlength}{\svgwidth}%
  \fi%
  \global\let\svgwidth\undefined%
  \global\let\svgscale\undefined%
  \makeatother%
  \begin{picture}(1,0.22520354)%
    \lineheight{1}%
    \setlength\tabcolsep{0pt}%
    \put(0,0){\includegraphics[width=\unitlength,page=1]{sbstlam3.pdf}}%
    \put(0.29715181,0.10478447){\color[rgb]{0,0,0}\makebox(0,0)[t]{\lineheight{1.25}\smash{\begin{tabular}[t]{c}$u$\end{tabular}}}}%
    \put(0,0){\includegraphics[width=\unitlength,page=2]{sbstlam3.pdf}}%
    \put(0.1518232,0.10370702){\color[rgb]{0,0,0}\makebox(0,0)[t]{\lineheight{1.25}\smash{\begin{tabular}[t]{c}$v$\end{tabular}}}}%
    \put(0,0){\includegraphics[width=\unitlength,page=3]{sbstlam3.pdf}}%
    \put(0.39860722,0.1033761){\makebox(0,0)[lt]{\lineheight{1.25}\smash{\begin{tabular}[t]{l}$=$\end{tabular}}}}%
    \put(0,0){\includegraphics[width=\unitlength,page=4]{sbstlam3.pdf}}%
    \put(0.6000572,0.13112432){\color[rgb]{0,0,0}\makebox(0,0)[t]{\lineheight{1.25}\smash{\begin{tabular}[t]{c}$u$\end{tabular}}}}%
    \put(0,0){\includegraphics[width=\unitlength,page=5]{sbstlam3.pdf}}%
    \put(0.7708036,0.08526895){\color[rgb]{0,0,0}\makebox(0,0)[t]{\lineheight{1.25}\smash{\begin{tabular}[t]{c}$v$\end{tabular}}}}%
    \put(0,0){\includegraphics[width=\unitlength,page=6]{sbstlam3.pdf}}%
    \put(0.48659376,0.10264887){\makebox(0,0)[lt]{\lineheight{1.25}\smash{\begin{tabular}[t]{l}$w$\end{tabular}}}}%
    \put(0,0){\includegraphics[width=\unitlength,page=7]{sbstlam3.pdf}}%
    \put(0.89161899,0.13235641){\makebox(0,0)[lt]{\lineheight{1.25}\smash{\begin{tabular}[t]{l}$w'$\end{tabular}}}}%
    \put(0,0){\includegraphics[width=\unitlength,page=8]{sbstlam3.pdf}}%
  \end{picture}%
\endgroup%

%% file: pics/let.pdf_tex
%% Creator: Inkscape 1.2.2 (732a01da63, 2022-12-09), www.inkscape.org
%% Accompanies image file 'let.pdf' (pdf, eps, ps)
\begingroup%
  \makeatletter%
  \providecommand\color[2][]{%
    \errmessage{(Inkscape) Color is used for the text in Inkscape, but the package 'color.sty' is not loaded}%
    \renewcommand\color[2][]{}%
  }%
  \providecommand\transparent[1]{%
    \errmessage{(Inkscape) Transparency is used (non-zero) for the text in Inkscape, but the package 'transparent.sty' is not loaded}%
    \renewcommand\transparent[1]{}%
  }%
  \providecommand\rotatebox[2]{#2}%
  \newcommand*\fsize{\dimexpr\f@size pt\relax}%
  \newcommand*\lineheight[1]{\fontsize{\fsize}{#1\fsize}\selectfont}%
  \ifx\svgwidth\undefined%
    \setlength{\unitlength}{298.50002929bp}%
    \ifx\svgscale\undefined%
      \relax%
    \else%
      \setlength{\unitlength}{\unitlength * \real{\svgscale}}%
    \fi%
  \else%
    \setlength{\unitlength}{\svgwidth}%
  \fi%
  \global\let\svgwidth\undefined%
  \global\let\svgscale\undefined%
  \makeatother%
  \begin{picture}(1,0.12735924)%
    \lineheight{1}%
    \setlength\tabcolsep{0pt}%
    \put(0,0){\includegraphics[width=\unitlength,page=1]{let.pdf}}%
    \put(0.30376388,0.03065542){\color[rgb]{0,0,0}\makebox(0,0)[t]{\lineheight{1.25}\smash{\begin{tabular}[t]{c}$let$\end{tabular}}}}%
    \put(0,0){\includegraphics[width=\unitlength,page=2]{let.pdf}}%
    \put(0.10895708,0.07046334){\color[rgb]{0,0,0}\makebox(0,0)[t]{\lineheight{1.25}\smash{\begin{tabular}[t]{c}$v$\end{tabular}}}}%
    \put(0,0){\includegraphics[width=\unitlength,page=3]{let.pdf}}%
    \put(0.76222347,0.07046334){\color[rgb]{0,0,0}\makebox(0,0)[t]{\lineheight{1.25}\smash{\begin{tabular}[t]{c}$v$\end{tabular}}}}%
    \put(0,0){\includegraphics[width=\unitlength,page=4]{let.pdf}}%
    \put(0.46216614,0.03127154){\makebox(0,0)[lt]{\lineheight{1.25}\smash{\begin{tabular}[t]{l}$\rew$\end{tabular}}}}%
    \put(0,0){\includegraphics[width=\unitlength,page=5]{let.pdf}}%
  \end{picture}%
\endgroup%

%% file: pics/rad2ut.pdf_tex
%% Creator: Inkscape 1.2.2 (732a01da63, 2022-12-09), www.inkscape.org
%% Accompanies image file 'rad2ut.pdf' (pdf, eps, ps)
\begingroup%
  \makeatletter%
  \providecommand\color[2][]{%
    \errmessage{(Inkscape) Color is used for the text in Inkscape, but the package 'color.sty' is not loaded}%
    \renewcommand\color[2][]{}%
  }%
  \providecommand\transparent[1]{%
    \errmessage{(Inkscape) Transparency is used (non-zero) for the text in Inkscape, but the package 'transparent.sty' is not loaded}%
    \renewcommand\transparent[1]{}%
  }%
  \providecommand\rotatebox[2]{#2}%
  \newcommand*\fsize{\dimexpr\f@size pt\relax}%
  \newcommand*\lineheight[1]{\fontsize{\fsize}{#1\fsize}\selectfont}%
  \ifx\svgwidth\undefined%
    \setlength{\unitlength}{434.16217411bp}%
    \ifx\svgscale\undefined%
      \relax%
    \else%
      \setlength{\unitlength}{\unitlength * \real{\svgscale}}%
    \fi%
  \else%
    \setlength{\unitlength}{\svgwidth}%
  \fi%
  \global\let\svgwidth\undefined%
  \global\let\svgscale\undefined%
  \makeatother%
  \begin{picture}(1,0.17204762)%
    \lineheight{1}%
    \setlength\tabcolsep{0pt}%
    \put(0.23110256,0.07972226){\color[rgb]{0,0,0}\makebox(0,0)[lt]{\lineheight{1.25}\smash{\begin{tabular}[t]{l}$=$\end{tabular}}}}%
    \put(0,0){\includegraphics[width=\unitlength,page=1]{rad2ut.pdf}}%
    \put(0.05326722,0.02271046){\color[rgb]{0,0,0}\makebox(0,0)[lt]{\lineheight{1.25}\smash{\begin{tabular}[t]{l}$\leftarrow$\end{tabular}}}}%
    \put(0,0){\includegraphics[width=\unitlength,page=2]{rad2ut.pdf}}%
    \put(0.12351647,0.07937238){\color[rgb]{0,0,0}\makebox(0,0)[t]{\lineheight{1.25}\smash{\begin{tabular}[t]{c}$f$\end{tabular}}}}%
    \put(0,0){\includegraphics[width=\unitlength,page=3]{rad2ut.pdf}}%
    \put(0.29945248,0.06318965){\color[rgb]{0,0,0}\makebox(0,0)[lt]{\lineheight{1.25}\smash{\begin{tabular}[t]{l}$\leftarrow$\end{tabular}}}}%
    \put(0,0){\includegraphics[width=\unitlength,page=4]{rad2ut.pdf}}%
    \put(0.36970171,0.07839297){\color[rgb]{0,0,0}\makebox(0,0)[t]{\lineheight{1.25}\smash{\begin{tabular}[t]{c}$f$\end{tabular}}}}%
    \put(0,0){\includegraphics[width=\unitlength,page=5]{rad2ut.pdf}}%
    \put(0.74164775,0.10013307){\color[rgb]{0,0,0}\makebox(0,0)[lt]{\lineheight{1.25}\smash{\begin{tabular}[t]{l}$=$\end{tabular}}}}%
    \put(0,0){\includegraphics[width=\unitlength,page=6]{rad2ut.pdf}}%
    \put(0.56381245,0.02930172){\color[rgb]{0,0,0}\makebox(0,0)[lt]{\lineheight{1.25}\smash{\begin{tabular}[t]{l}$\rightarrow$\end{tabular}}}}%
    \put(0,0){\includegraphics[width=\unitlength,page=7]{rad2ut.pdf}}%
    \put(0.63406173,0.09978318){\color[rgb]{0,0,0}\makebox(0,0)[t]{\lineheight{1.25}\smash{\begin{tabular}[t]{c}$f$\end{tabular}}}}%
    \put(0,0){\includegraphics[width=\unitlength,page=8]{rad2ut.pdf}}%
    \put(0.83763711,0.08360045){\color[rgb]{0,0,0}\makebox(0,0)[lt]{\lineheight{1.25}\smash{\begin{tabular}[t]{l}$\rightarrow$\end{tabular}}}}%
    \put(0,0){\includegraphics[width=\unitlength,page=9]{rad2ut.pdf}}%
    \put(0.90788635,0.09880377){\color[rgb]{0,0,0}\makebox(0,0)[t]{\lineheight{1.25}\smash{\begin{tabular}[t]{c}$f$\end{tabular}}}}%
    \put(0,0){\includegraphics[width=\unitlength,page=10]{rad2ut.pdf}}%
  \end{picture}%
\endgroup%